\documentclass[11pt]{article}

\usepackage[T1]{fontenc}
\usepackage[english]{babel}
\usepackage[letterpaper,margin=1in]{geometry}
\usepackage{enumitem}

\usepackage{nicefrac}

\usepackage{amsmath}
\usepackage{amssymb}
\usepackage{amsthm}
\usepackage{bbm}
\usepackage{thmtools}
\usepackage{graphicx}
\graphicspath{{./}}
\usepackage{physics}
\usepackage{booktabs}
\usepackage{multirow}
\usepackage{pgfplots}
\usepackage{standalone}
\usepackage[compat=newest]{yquant}
\pgfplotsset{compat=1.18}
\usepgfplotslibrary{colormaps}
\usepackage[colorlinks=true, allcolors=teal, hypertexnames=false]{hyperref}
\usepackage[capitalise]{cleveref}
\usepackage{url}
\usepackage{csquotes}
\usepackage{bm}
\usepackage{comment}
\usepackage[dvipsnames]{xcolor}
\usepackage{enumitem}
\setlist[enumerate,itemize]{itemsep=0.4ex,topsep=0.8ex,leftmargin=1.5em}
\usepackage{tikz}
\usetikzlibrary{shapes.geometric, positioning, arrows.meta}

\usepackage{stmaryrd}
\newcommand{\gcomm}[2]{\llbracket #1,#2 \rrbracket}
\definecolor{charcoal}{RGB}{47, 53, 66}       %
\definecolor{bordergray}{RGB}{128,0,128}     %
\definecolor{softgray}{RGB}{230,230,250}    %
\definecolor{crimson}{RGB}{179,27,27}       %
\definecolor{crimsonlight}{RGB}{255, 234, 234}%
\definecolor{arrowcolor}{RGB}{38,67,72}

\usepackage{titlesec}
\titleclass{\subsubsubsection}{straight}[\subsection]

\newcounter{subsubsubsection}[subsubsection]
\renewcommand\thesubsubsubsection{\thesubsubsection.\arabic{subsubsubsection}}

\titleformat{\subsubsubsection}
{\normalfont\normalsize\bfseries}{\thesubsubsubsection}{1em}{}
\titlespacing*{\subsubsubsection}
{0pt}{3.25ex plus 1ex minus .2ex}{1.5ex plus .2ex}
\makeatletter
\def\toclevel@subsubsubsection{4}
\def\l@subsubsubsection{\@dottedtocline{4}{7em}{4em}}
\makeatother
\usepackage[giveninits=true,maxbibnames=99,style=alphabetic,maxalphanames=4,minalphanames=3,isbn=false]{biblatex}

\AtEveryBibitem{%
    \clearlist{language}%
    \iffieldequalstr{eprinttype}{arxiv}{\clearfield{doi}}{}%
}
\renewbibmacro*{doi+eprint+url}{%
    \iftoggle{bbx:doi}
    {\printfield{doi}}
    {}%
    \newunit\newblock
    \ifboolexpr{togl {bbx:eprint} and test {\iffieldundef{doi}}}
    {\usebibmacro{eprint}}
    {}%
    \newunit\newblock
    \ifboolexpr{togl {bbx:url} and test {\iffieldundef{doi}}  and test {\iffieldundef{eprint}}}
    {\usebibmacro{url+urldate}}
{}}

\usepackage{algorithm}
\usepackage{algpseudocodex}

\usepackage{quantikz}
\usepackage{tikz}
\usepackage{complexity}
\usepackage{totcount}
\newtotcounter{todoitems}
\usepackage{xcolor}

\let\epsilon\varepsilon

\newcommand{\model}{\textup{BEQC}}

\NewDocumentCommand{\modelparam}{m m m}{%
  \IfBlankTF{#1}
    {\model(#2,#3)}
    {\model_{\mathcal{#1}}(#2,#3)}%
}
\newcommand{\modelname}{Bosonic Energy-Preserving Quantum Computation}
\newcommand{\thermalmodel}{ThermalBEQC}

\newcommand{\hD}{\hat{D}}

\newcommand{\hh}{\hat{H}}
\newcommand{\II}{\mathbb I}

\newcommand{\hN}{\hat{N}}

\newcommand{\hP}{\hat{p}}
\newcommand{\ph}{\hat{P}_{\hat{H}}}
\newcommand{\hQ}{\hat{q}}

\newcommand{\hS}{\hat{S}}

\newcommand{\vm}{\bm{m}}
\newcommand{\vn}{\bm{n}}
\newcommand{\LL}{\hat{L}^{(N_{\mathrm{max}})}}

\newclass{\CVBQP}{CVBQP}

\newclass{\PrecQMA}{PrecQMA}
\newclass{\TOWER}{TOWER}
\newclass{\PTOWER}{PTOWER}
\newclass{\AzPP}{A_0PP}
\newclass{\Pp}{P}
\newclass{\QPp}{QP}
\newclass{\BQPSPACE}{BQPSPACE}

\NewDocumentCommand{\BEQP}{m m}{%
  \IfBlankTF{#1}
    {\IfBlankTF{#2}
      {\class{BEQP}}                        %
      {\class{BEQP}(#2)}}                   %
    {\IfBlankTF{#2}
      {\class{BEQP}_{\mathcal{#1}}}         %
      {\class{BEQP}_{\mathcal{#1}}(#2)}}%
}

\NewDocumentCommand{\ThermalBEQP}{m m}{%
  \IfBlankTF{#1}
    {\IfBlankTF{#2}
      {\class{ThermalBEQP}}                        %
      {\class{ThermalBEQP}[#2]}}                   %
    {\IfBlankTF{#2}
      {\class{ThermalBEQP}_{\mathcal{#1}}}         %
      {\class{ThermalBEQP}_{\mathcal{#1}}[#2]}}%
}

\newclass{\PrecBEQP}{PreciseBEQP}

\newcommand{\ad}{\mathrm{ad}}

\newcommand\calH{\mathcal{H}}

\newcommand\calS{\mathcal{S}}
\newcommand\calT{\mathcal{T}}
\newcommand\calG{\mathcal{G}}
\newcommand{\GLO}{\calG_{\mathrm{LO}}}
\newcommand{\GKerr}{\calG_{\textup{Kerr}}}
\newcommand{\GcrossKerr}{\calG^{\times}_{\textup{Kerr}}}

\newcommand\calD{\mathcal{D}}

\newcommand\calC{\mathcal{C}}

\newcommand\calU{\mathcal{U}}

\newcommand{\wtH}{\widetilde{H}}

\newcommand\Span{\mathrm{span}}
\newcommand\NN{\mathbb{N}}
\newcommand\CC{\mathbb{C}}
\newcommand\RR{\mathbb{R}}
\newcommand\ZZ{\mathbb{Z}}

\newcommand\uuarr{\mathbin{\upuparrows}}

\newcommand{\bfn}{\bm{n}}
\newcommand{\bfm}{\bm{m}}
\newcommand{\bfsigma}{\bm{\sigma}}

\newcommand{\Ayes}{A_{\mathrm{yes}}}
\newcommand{\Ano}{A_{\mathrm{no}}}

\newcommand\wtC{\widetilde{C}}
\newcommand\wtcalC{\widetilde{\calC}}

\newcommand\Q{{\hat q}}
\renewcommand\P{{\hat p}}
\newcommand\N{{\hat n}} %
\renewcommand\a{{\hat a}}
\newcommand\Var{\mathrm{Var}}
\newcommand\diag{\mathrm{diag}}

\newcommand{\calN}{{\mathcal{N}}}

\renewcommand{\poly}{{\mathrm{poly}}}

\newcommand{\Dfin}{{\calD_{\mathrm{fin}}}}

\newcommand{\End}{\mathrm{End}}

\newcommand\bd{\hat{b}^\dagger}
\renewcommand\ad{\hat{a}^\dagger}

\newcommand{\Cinv}{\calC_{\mathrm{inv}}}
\newcommand\adag{\hat{a}^\dagger}

\newcommand{\adj}{\mathrm{ad}}

\renewcommand\b{{\hat{b}}}
\newcommand\sfX{\mathsf{X}}
\newcommand\Y{{\hat y}}
\newcommand\Ran{\mathrm{Ran}}

\newcommand{\Halg}{H_{\mathrm{alg}}}

\newcommand{\nb}{\bar{n}}
\newcommand{\lz}{\lambda_L}
\newcommand{\la}{\lambda_H}

\newcommand{\DATASWAP}{\textsf{DATA-SWAP}}
\newcommand{\Anm}{\hat{A}_{\vn, \vm}}
\newcommand{\Tnm}{\hat{T}_{\vn, \vm}}

\makeatletter
\newcommand{\subalign}[1]{%
  \vcenter{%
    \Let@ \restore@math@cr \default@tag
    \baselineskip\fontdimen10 \scriptfont\tw@
    \advance\baselineskip\fontdimen12 \scriptfont\tw@
    \lineskip\thr@@\fontdimen8 \scriptfont\thr@@
    \lineskiplimit\lineskip
    \ialign{\hfil$\m@th\scriptstyle##$&$\m@th\scriptstyle{}##$\hfil\crcr
      #1\crcr
    }%
  }%
}
\makeatother

\algrenewcommand{\algorithmicfunction}{\texttt{function}}
\algrenewcommand{\algorithmicif}{\texttt{if}}
\algrenewcommand{\algorithmicthen}{\texttt{then}}
\algrenewcommand{\algorithmicelse}{\texttt{else}}
\algrenewcommand{\algorithmicreturn}{\texttt{return}}
\algrenewcommand{\algorithmicend}{\texttt{end}}

\theoremstyle{plain}
\newtheorem{theorem}{Theorem}[section]

\newtheorem{claim}[theorem]{Claim}
\newtheorem{lem}[theorem]{Lemma}

\newtheorem{proposition}[theorem]{Proposition}
\newtheorem{corollary}[theorem]{Corollary}

\theoremstyle{definition}
\newtheorem{remark}[theorem]{Remark}
\newtheorem{definition}[theorem]{Definition}

\crefname{lem}{Lemma}{Lemmas}
\crefname{claim}{Claim}{Claims}
\crefname{claim}{Claim}{Claims}
\crefname{paragraph}{Paragraph}{Paragraphs}
\crefname{subsubsubsection}{Section}{Sections}

\newcommand\Ntot{\hN}
\newcommand\bfa{\bm{\hat a}}
\newcommand\be{\begin{equation}}
\newcommand\ee{\end{equation}}
\newcommand\SevSkip{\par\medskip\noindent}
\newcommand\proofstep[1]{\par\medskip\noindent\emph{#1}}

\usepackage{authblk}

\title{A physical and universal model of bosonic computations\\
with Solovay--Kitaev theorem}

\author[1]{Dorian Rudolph}
\author[2]{Arsalan Motamedi}
\author[1]{Dhruva Sambrani}
\author[1]{Hamid Reza Naeij}
\author[3]{\authorcr Ulysse Chabaud}
\author[1]{Sevag Gharibian}
\author[4]{Saeed Mehraban}

\affil[1]{\small Paderborn University and PhoQS, Warburger Stra{\ss}e 100, 33098 Paderborn, Germany}
\affil[2]{University of Toronto, Ontario, Canada}
\affil[3]{DIENS, \'Ecole Normale Sup\'erieure, PSL University, CNRS, INRIA, 45 rue d’Ulm, Paris, 75005, France}
\affil[4]{Tufts University, Medford, MA, USA}

\date{\today}

\begin{document}

\maketitle

\begin{abstract}
Bosonic quantum systems are among the leading architectures for quantum information processing, offering continuous-variable degrees of freedom with strong error-correction capabilities. However, standard bosonic quantum computation models such as the Lloyd--Braunstein~\cite{lloyd_quantum_1999} and hybrid oscillator-qubit models~\cite{brenner2025trading,liu2026hybrid} permit dramatic energy growth, leading to unphysical computational power and the breakdown of fundamental algorithmic tools such as universal and efficient compilation \cite{brenner2025factoring,CGMMNRS25}. To address this, we introduce a new model of bosonic quantum computation, \modelname\ (\model), in which energy is treated as a computational resource. Namely, energy is supplied solely through input coherent states, and all gates are generated by \textit{energy-preserving} Hamiltonians. Thus, by construction, dramatic energy growth is impossible, making the model physically grounded.
    
We next show that \model\ is a computationally robust and universal model in many respects, including: (1. Computational power) \model\ efficiently simulates all polynomial-energy %
computations in existing models, and exactly recovers $\mathsf{BQP}$ in the polynomial energy setting. It further admits several complexity-theoretic upper bounds when varying the energy, precision, and space parameters of the model. (2. Universal gate sets and state synthesis) \model\ has natural universal gate sets based on linear optics and Kerr interactions. In particular, we obtain a Solovay--Kitaev theorem which circumvents previous no-go results. We give various applications, including (a) a protocol for engineering GKP states with rigorous preparation guarantees, (b) Fock state preparation to exponential precision, and (c) native Fock space simulation of any qubit-based unitary. %
\end{abstract}

\tableofcontents

\bigskip

\begin{quote}
    \emph{The theory of computation has traditionally been studied almost entirely in the abstract, as a topic in pure mathematics. This is to miss the point of it. Computers are physical objects, and computations are physical processes. What computers can or cannot compute is determined by the laws of physics alone, and not by pure mathematics.\\
        \vspace{-7mm}
        \begin{flushright}
            David Deutsch~\cite{deutschFabricRealityScience1997}
        \end{flushright}
    }
\end{quote}
\section{Introduction}

The limits of computation are physical, most evident in the Noisy Intermediate Scale Quantum (NISQ) computing era \cite{preskill2018quantum} (transitioning into the ``early fault-tolerant'' \cite{acharyaQuantumErrorCorrection2025,dasuComputingManyEncoded2026, eisert2025mind} era), which has required increasingly close collaboration between quantum hardware and quantum algorithm designers. This, in turn, has put a spotlight on a fundamental theme in quantum computing: while most theoretical quantum algorithm design has traditionally taken place in the \emph{Discrete-Variable (DV)} setting (i.e.\ on finite-dimensional qu\emph{d}its), in practice many physical systems are of \emph{Continuous-Variable (CV)} nature (i.e.\ infinite-dimensional with continuous degrees of freedom, such as bosons). From a computational/platform perspective, this hence yields three possibilities for system design: ``Purely'' CV systems (e.g.\ quantum photonics, such as in Gaussian Boson Sampling experiments~\cite{aaronsonComputationalComplexityLinear2011,hamilton2017gaussian,zhongExperimentalGaussianBoson2019}), ``purely'' DV systems (e.g.\ superconducting qubits as in Random Circuit Sampling experiments~\cite{aruteQuantumSupremacyUsing2019,boulandComplexityVerificationQuantum2019}, which are CV systems masquerading as DV ones), and hybrid CV-DV systems (e.g.\ trapped ion, neutral atom, and even superconducting; see \cite{liu2026hybrid} for a survey). 

In this work, our focus is the first of these models --- CV systems, which, as far as \emph{physics} is concerned, aims to directly exploit the natural behavior of the underlying quantum system employed. Our goal is to define a new model, \modelname\ (\model), based on CV degrees of freedom, yet which canonically captures \emph{all three} types of models above in one physically and mathematically well-defined model.

\SevSkip
\textbf{The Lloyd--Braunstein model (LB)~\cite{lloyd_quantum_1999}.} The precursor of \model\ is the Lloyd--Braunstein CV model of CV quantum computing~\cite{lloyd_quantum_1999}. The latter is generally viewed as a standard model of universal bosonic computation in the literature~\cite{braunstein2005quantum, weedbrook2012gaussian}, defined as follows:
\begin{itemize}
    \item (State space) Quantum states live in $(\mathbb C^{\infty})^{\otimes n}\cong\ell^2 (\NN_0^{n},\CC)$, the set of square-summable complex sequences with an $n$-dimensional index set.

    \item (Initial state) The vacuum state on $n$ modes, denoted $\ket{0^n}$, analogous to the all-zeroes state on $n$ qubits.

    \item (Gates) Unitaries $e^{-iHt}$ acting non-trivially on $k\in O(1)$ modes, where $t$ is evolution time and Hamiltonian $H$ is an $O(1)$-degree polynomial in the position $\hQ$ and momentum operators $\hP$ (\cref{def:quadratures}), each monomial having constant coefficients. The total runtime of a sequence of $L$ gates $\prod_{j=1}^L e^{-it_jH_j}$ is $\sum_{j=1}^L t_j$. We assume that all gates are efficiently uniformly generated\footnote{Roughly, given input $x\in \{0,1\}^{\varsigma}$, a classical description of the sequence of bosonic gates to be applied can be generated in $\poly(\varsigma)$ time.}.

    \item (Measurement) A designated output mode is measured via the position operator $\hat q$, or the number operator $\hat n$, which counts the number of photons in said mode. The total photon-number operator is denoted $\hN$.
\end{itemize}
A common gate set in LB is the set of Gaussian gates (displacement (\cref{def:displacement}), squeezing (\cref{def:squeezing}), phase shift) and the cubic phase gate, generated by the Hamiltonian $\hQ^3$. These are in some sense analogous to Clifford versus $T$ gates in stabilizer theory: Gaussian gates\footnote{Formally, Gaussian gates are precisely those generated by any Hamiltonian of degree at most $2$. This includes the basic harmonic oscillator $\hat H_{\rm osc} = \frac{1}{2}(\hQ^2+\hP^2)$, which generates the Fourier transform operation as $e^{i\pi/2 \hat{H}_{\rm osc}}$.} are easy to both simulate (starting from the vacuum state) and experimentally implement (e.g.\ via photonics), whereas the cubic phase gate (which, being generated by a Hamiltonian of degree $3$, is non-Gaussian) is hard to implement. Taken together, these two gate families yield computations that are believed to be hard to simulate on a classical computer~\cite{CJMMM26}.
Of course, a gate being ``hard to implement'' is difficult to formally prove, and strongly hardware-dependent. Instead, we ask a more fundamental question: \emph{Is the LB model ``physical'', and if not, how can this be addressed?}

Before continuing, we remark that \model\ will capture more than just LB. In particular, hybrid models combining CV and DV degrees of freedom have recently been defined~\cite{brenner2025factoring,brenner2025trading,liu2026hybrid}, where couplings between CV and DV degrees of freedom (such as boson-fermion interactions) are used to instantiate non-Gaussian operations. Using a qubit-to-oscillator embedding developed in \cite{arzani2025can}, one can embed CV-DV interactions in purely CV computations. As a result, a comprehensive study of purely CV models also captures the power of quantum computation using either DV degrees of freedom or hybrid CV-DV degrees of freedom. Indeed, in \cref{sec:sim_existing_models} we formally show how to simulate such hybrid models.

\SevSkip\textbf{What makes a model physical?} To formally argue whether a model is ``physical'', we turn the quote of Deutsch on its head: \emph{Assuming LB can be implemented to high precision in the lab, what would it allow us to compute?} If the answer is ``unreasonable things'', one has formal evidence that said model is not physical. Underpinning our discussion will be the fact that LB is known in physics folklore to have a serious shortcoming --- it permits unbounded \emph{energy} (formally, average photon number (\cref{def:avg-photon-number})). This, in turn, manifests computationally in two ways which we focus on here: The ability to compute ``unreasonable things'', and the loss of fundamental robustness properties, such as universality and low-overhead gate compilation via fixed gate sets. Thus, LB is neither physical nor a robust computational model.

\SevSkip\emph{1. Computing ``unreasonable things''.} Already in 2005, Braunstein and van Loock~\cite{braunstein2005quantum} showed LB can attain arbitrarily high dense coding rates on fixed system sizes. More recently, \cite{CJMMM26,CGMMNRS25} showed that with the standard Gaussian + $\hQ^3$ gate set, evolution time $t$ in LB can spike the system's energy at a rate \emph{doubly exponential} in $t$. Even worse, for other polynomial gate sets, even \emph{constant} $t$ suffices to attain \emph{infinite} energy (i.e.\ formally, the average photon number diverges)~\cite{CGMMNRS25}. This energy can be formally harnessed to compute (wildly) unreasonable things ---
for example, $O(k)$-mode gates suffice to efficiently solve $\textsf{NTIME}(\exp^{(k)})$~\cite{CGMMNRS25}, i.e.\ non-deterministic computations whose runtime scales as a \emph{tower of exponentials} of height $k$.%

\SevSkip\emph{2. Loss of computational robustness.} Unbounded energy also makes obtaining universal gate sets challenging. Note that \cite{lloyd_quantum_1999} claims efficient compilation of arbitrary polynomial Hamiltonians into the Gaussian + $\Q^3$ gate set, to which we give an explicit counterexample (\cref{cor:uuarr}). The issue arises upon use of Trotter splitting, e.g.
$e^{itA}e^{itB}e^{-itA}e^{-itB} = e^{-t^2[A,B]} + O(t^3)$,
which does not generally hold for unbounded operators. If, in addition, we desire low-overhead gate compilation \emph{\`a la} Solovay--Kitaev, the situation is worse --- for the Gaussian + high-degree polynomial gates, a strong version of a bosonic Solovay--Kitaev (SK) theorem is impossible~\cite{CGMMNRS25}. In fact, in \cref{cor:uuarr} we show that even with an \emph{exponential} number of Gaussian + $\Q^3$ gates, one cannot surpass inverse exponential overlap with the output of a specific family of poly-size circuits. 

To date, universal gate sets for LB are either known only for certain subsets of gates, such as Gaussian gates~\cite{becker2021energy}, or require an infinite\footnote{Formally, the degree of the Hamiltonians in the universal set of~\cite{arzani2025can} scales with the desired approximation error~\cite{arzani2025can}.} gate set to generate all physical unitaries (i.e.\ unitaries which map finite-energy states to finite-energy states) within any desired precision~\cite{arzani2025can}. The lack of a fixed universal gate set is naturally problematic; for example, there exists an artificial gate set with which LB can simulate \textsf{BQP} (i.e.\ poly-time quantum computation), but the same is \emph{not} known for the standard Gaussian + $\Q^3$ set. This raises the question: \emph{Could bosonic computational power be fundamentally hardware-dependent?}

\subsection{Results}\label{sscn:results}
To address these shortcomings, we propose a new bosonic computational model, \modelname\ (\model), which is not only physical (i.e.\ inherently avoids energy blowups), but is also computationally powerful yet reasonable and robust (i.e.\ captures \textsf{BQP}, allows for Trotter splitting, permits natural fixed-size universal gate sets, as well as low-overhead circuit compilation via Solovay--Kitaev-like theorems, etc). 
Throughout, by ``polynomial Hamiltonian'', we mean constant degree as in the LB model, unless otherwise specified, and by energy we mean average photon number.

\subsubsection{Defining the model, \model: Energy from a resource state} \label{sssn:definingmodel}

\noindent\emph{The challenge.} \emph{A priori}, isolating the ``source'' of energy blowup in a bosonic circuit is difficult. Consider, for example, the Fourier transform $\hQ^2+\hP^2$ and cubic phase gate $\hQ^3$. Starting with a vacuum state, which has zero energy by definition, neither $\hQ^2+\hP^2$ nor $\hQ^3$ alone spikes system energy. However, simply \emph{alternating} the two does\footnote{This is the simple construction behind the doubly exponential energy lower bound~\cite{CGMMNRS25} mentioned above.}~\cite{CJMMM26,CGMMNRS25}. We also emphasize that simply imposing a promise on an energy bound is not a satisfactory solution, because (i) just reordering gates can cause a dramatic change in energy and (ii) knowing if energy diverges is itself undecidable~\cite{CGMMNRS25}. To circumvent this, we take inspiration from the study of resource ``magic'' states~\cite{knillFaultTolerantPostselectedQuantum2004}, and combine two ideas:

\SevSkip\emph{1. Energy-preserving gates.} From physics, we use the natural idea of allowing only \emph{energy-preserving} gates, which formally are photon-number-preserving unitaries $U$, meaning $[U,\hN]=0$ for photon number operator $\hN$. This is analogous to restricting a DV circuit to Clifford gates (which recall do not alter non-stabilizer magic), since energy-preserving gates leave the system energy unchanged\footnote{Another helpful analogy is performing computation based on symmetric gates, i.e.\ gates which commute with a symmetry group (cf.~\cite{lidar1998decoherence,divincenzo2000universal} where ``$\mathrm{SU}(2)$ symmetric'' exchange interactions can perform universal computation on singlet states as resource states); in our case we consider gates that commute with the group generated by total photon number $\hN$.}.

\SevSkip\emph{2. Isolating energy as an input resource.} Since gates no longer alter energy, all energy must come from input states. Thus, in addition to vacuum modes, we allow a single \emph{coherent state} (\cref{def:coherent-state}) input $\ket{\alpha}$ with amplitude $\alpha\in\mathbb{C}$. Such states are not only easy to prepare, but their expected energy has a simple formula: $\bra{\alpha} \N \ket{\alpha} = \abs{\alpha}^2$. In this sense, coherent states for us play a role analogous to resource magic states in the DV setting.

\SevSkip In sum, in \model\ the energy is fully controlled via parameter $\alpha$ at the input level, yielding a model which (1) allows precise quantification of energy and its impact on computational power, and (2) is a ``safe'' algorithmic model to work in. \model\ is formally defined in \Cref{def:model}, and an illustration is given in \cref{fig:BEQC}.

\SevSkip\emph{Remark on $\ket{\alpha}$:} Without loss of generality, one can replace $\ket{\alpha}$ with any $m$-mode coherent state $\ket{\bm{\beta}}=\ket{\beta_1}\otimes\cdots\otimes\ket{\beta_m}$, so long as the total energy is equal, i.e.\ $|\alpha|^2=\norm{\bm{\beta}}^2$ (\cref{lem:coher-gen}). Moreover, a number of our simulation results below hold even if the resource state is prepared imperfectly, e.g.\ undergoes thermal and/or displacement noise.

\begin{figure*}[t]
    \centering
    \includegraphics[width=\linewidth]{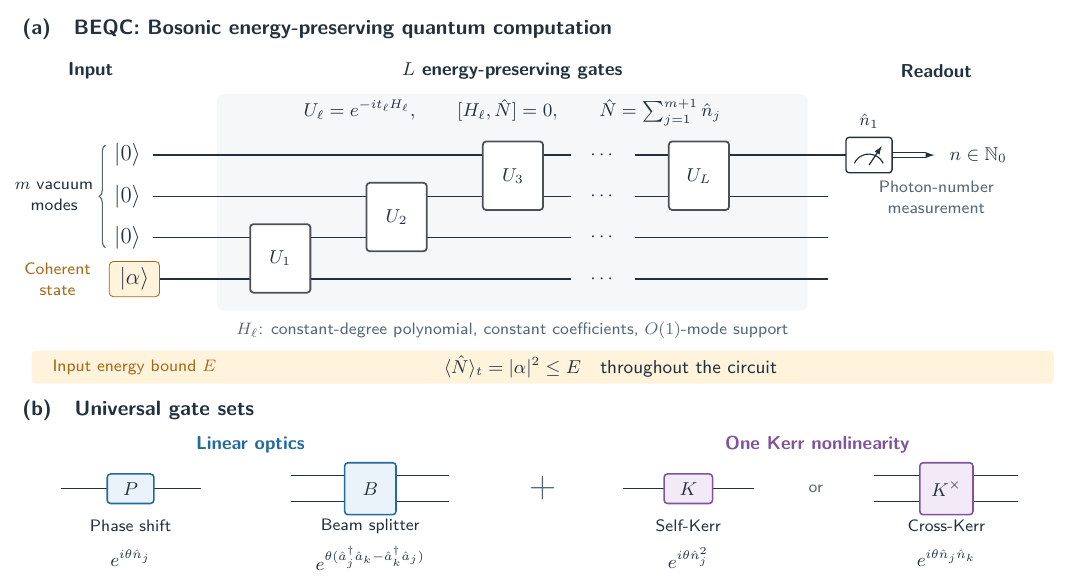}
    \caption{Schematic depiction of \model. (a) A bosonic energy-preserving computation consists of a sequence of energy-preserving gates applied to an input vacuum state, together with a coherent state supplying the energy for the computation, while the outcome is obtained by measuring the photon number of one output mode. (b) Two universal sets of gates for \model, both of which obey an energy-preserving Solovay--Kitaev theorem (\cref{thm:informal_cutoff-SK}).}
    \label{fig:BEQC}
\end{figure*}

\subsubsection{Computational properties}\label{ssscn:properties}
Having defined \model, we next show it is computationally powerful, yet reasonable and robust.

\SevSkip\textbf{a. For the computer science audience: Rigorous and reasonable computational power.} At a minimum, any universal quantum computational model should recover poly-time quantum computation in the DV model, \BQP:

\begin{theorem}[Informal; see \cref{thm:bqp-universal} and \cref{thm:in_BQP}]\label{thm:bqp-complete_informal}
    The promise class $\BEQP{}{}$ of problems solvable by \model\ with polynomial time, polynomial energy, and constant precision equals \BQP.
\end{theorem}

\noindent This result is robust in that it holds (1) even if the input coherent states are replaced with input thermal states (\cref{cor:thermalBQP}), or (2) if the resource state $\ket{\bm{\alpha}}:=\ket{\alpha}^{\otimes m}$ is perturbed. To avoid confusion, we also stress the difference between acronyms \model\ and $\BEQP{}{}$ --- the first is the general computational model (\emph{C} stands for computation), and the second is the promise class\footnote{It suffices at this point to think of a promise class as one where all computational problems have a YES or NO answer.} in \model\ analogous to BQP (\emph{P} stands for polynomial time and energy). %

In fact, our model is flexible beyond $\BEQP{}{}$ --- one can vary parameters for \model\ such as time, energy, space, and precision, and subsequently analyze how computational power changes. For example, $\BEQP{}{}$ with both logarithmic\footnote{As standard in complexity theory, throughout this paper resource estimates such as time, energy, and space are given by functions in the \emph{input} size of a problem. In this paper, input size is denoted $\varsigma$ (\Cref{def:model}), whereas $n$ typically denotes photon count.} energy and logarithmically many modes is in P, whereas the stronger variant of $\BEQP{}{}$ with high precision, denoted $\PrecBEQP$, satisfies $\PrecBEQP\subseteq\PSPACE$ (see \cref{thm:classical_upper_bounds_precise} and \Cref{tab:tradeoffs} in \Cref{sec:bound_complexity}). The other two parameters one can adjust are the \emph{locality} of gates allowed and \emph{degree} (recall gates are polynomials in the position and momentum operators) --- this yields a perhaps surprising jump in complexity:

\begin{theorem}[Informal, see \cref{claim:pspace-complete}]\label{thm:PSPACE-complete_informal}
    $\BEQP{}{}$ with non-local, polynomial degree Hamiltonians is $\PSPACE$-complete.
\end{theorem}

\noindent \cref{thm:PSPACE-complete_informal} aligns with the intuition that, while simulation of \emph{sparse} Hamiltonians with efficient sparse access and polynomially bounded $t\|H\|$ is in \textsf{BQP} for the DV model~\cite{aharonovAdiabaticQuantumState2003,berryEfficientQuantumAlgorithms2007}, simulating non-sparse $H$~\cite{childs2010limitations} or evolution for exponentially long times can in general require time exponential in the system size~\cite{berryEfficientQuantumAlgorithms2007}. Our construction effectively fast-forwards such long-time evolutions while keeping the photon number polynomially bounded. As a corollary, we obtain that unitaries generated via non-local, polynomial degree $H$ cannot be efficiently compiled into unitaries generated by local $O(1)$-degree $H$ unless $\BQP=\PSPACE$ (\cref{cor:no-compile}). In sum, \model\ captures a range of complexity classes, as one would expect of a robust bosonic computational framework. A natural open question is which other complexity classes can be captured via completeness results in this manner.

\SevSkip\textbf{b. For the physics audience: Any ``physical'' computation can be simulated by \model.} We next demonstrate the expressiveness of \model\ by showing it efficiently simulates any  ``physical'' computation in a wide variety of models: The LB model~\cite{lloyd_quantum_1999}, hybrid CV-DV~\cite{brenner2025factoring,liu2026hybrid}, DV/qubits, fusion-based quantum computing~\cite{bartolucci2023fusion}, and fermionic\footnote{For clarity, this is achieved here via the Jordan-Wigner~\cite{jordan1928paulische} mapping and simulation of the resulting qubit system.} computation~\cite{bravyi2002fermionic}.  For clarity, by ``physical'' computation, we mean with polynomial energy throughout the computation. In this section, we discuss only the LB model; details for other models are in \Cref{sec:sim_existing_models}. 

To simulate physical LB computations, for pedagogical reasons we use the standard physics technique of parametric approximation~\cite{mollow1967quantum,mollow1967quantum2}, which requires a polynomial bound on higher moments of the number operator. Later, however, our Solovay--Kitaev theorem (\cref{thm:informal_cutoff-SK}) will largely subsume this parametric construction, requiring only polynomial energy. For parametric approximation, we provide a rigorous treatment of the error incurred in its use for general polynomial Hamiltonians:

\begin{theorem}[Informal, see \cref{thm:passive-compilation}]\label{thm:simulation_informal}
    Let $H$ be a polynomial Hamiltonian in the LB model acting on register $A$. Then, in \model\ there exists an efficiently computable Hamiltonian $\wtH$ acting on $AB$ such that for $\alpha\geq 1$ and for all $\ket{\psi}$,
    \begin{equation}\label{eqn:1}
        \norm{\left(e^{-it(H_A\otimes\II_B)} - e^{-it\wtH}\right)\ket{\psi}_A\ket{\alpha}_B}
        \in O\left(\frac{M t}{\alpha}\right),
    \end{equation}
    for $M$ the energy moment bound (\cref{def:energy-moment}) achievable by $H$ in time interval $[0,t]$ (formally defined in \cref{eq:supsE}).
\end{theorem}
\noindent Thus, for any finite $M$ and $t$, choosing the input energy $|\alpha|^2$ large enough in \model\ yields an arbitrarily good approximation. Note that $M$ can be uncomputable for misbehaved $H$~\cite{CGMMNRS25} (i.e.\ there exist $H$ for which $M$ can diverge in time $t$). The simulation of \Cref{thm:simulation_informal} is moreover robust to simultaneous thermal and displacement error on the coherent state (\cref{thm:robust-parametric-approximation}), with the approximation error in \cref{eqn:1} incurring an additional multiplicative $\sqrt{\bar{n}+1}$ factor, where $\bar{n}$ is the thermal state's mean photon number.

\SevSkip\textbf{c. Universal gate sets and a Solovay--Kitaev theorem.} We have so far seen that \model\ is flexible enough to capture a range of quantum computing models (\BQP). To be useful as an \emph{algorithmic framework}, however, one next requires a finite universal gate set, which can be used to compile an arbitrary local energy-preserving $H$ in the lab. 

Our main result here is to show that a  clean universal gate set exists for \model: The set of energy-preserving linear optics gates (beam splitters and phase shifters; \cref{def:linear-optics}), together with either the self-Kerr gate or the cross-Kerr gate, which have Hamiltonian $\hat n^2$ and $\hat n_1\hat n_2$, respectively. Here, Kerr gates are non-linear optical gates generated by the Kerr effect, where the refractive index of a material depends on the intensity of the light passing through it (see \cite[Chapter 4]{walls1995quantum} and \cite[Section VII.A.2]{braunstein2005quantum}); they are energy-preserving as required, since their Hamiltonians commute with $\hN$. Denoting the Kerr and cross-Kerr based-gate sets by $\GKerr$ and $\GcrossKerr$ (\cref{eq:GKerr}), respectively, we now state our Solovay--Kitaev theorem for \model:

\begin{theorem}[Informal, see \cref{thm:cutoff-SK}]\label{thm:informal_cutoff-SK}
Every energy-preserving unitary $U$ over $m$ modes can be uniformly approximated up to precision $\varepsilon$ on the space of states with total photon-number at most $K$ by a circuit $C_U$ with gates in $\GKerr$ or\/ $\GcrossKerr$. Whenever $m$ or $K$ is constant, the circuit $C_U$ has $\poly(m,K,\log(1/\varepsilon))$ gates and can be constructed in $\poly(m,K,\log(1/\epsilon))$ operations given access to the coefficients of $U$ in the number basis.
\end{theorem}

\noindent This result answers the question posed at the end of the introduction, showing that the computational power of \model\ is not hardware-dependent, in the sense that any local gate set can be efficiently compiled into a single, universal gate set.

Before presenting applications of this energy-preserving Solovay--Kitaev theorem, a few remarks are in order. While it is well known that linear-optical gates and Kerr gates are universal for dual-rail qubits \cite{Chuang_1995} and over any fixed photon-number sectors \cite{PhysRevLett.119.220502}, \cref{thm:informal_cutoff-SK} shows that these gates are universal over the whole space of states with less than $K$ \textit{photons} (i.e.\ jointly over all fixed photon number sectors with photon number less than $K$) but more importantly satisfy a Solovay--Kitaev theorem for energy-preserving unitaries. We note that the formal version of \cref{thm:informal_cutoff-SK} which we prove actually yields a stronger result: independently of the photon-number cutoff $K$ and the number of modes $m$, the length of the sequence of gates scales with the desired precision as $\log^{c_{\rm SK}}(1/\varepsilon)$, where $c_{\rm SK}<4$ is an absolute constant. We also show in \cref{app:finite-alphabet} that both universal gate sets $\GKerr$ and $\GcrossKerr$ can be reduced to finite universal gate sets by picking suitable irrational parameters.

\SevSkip\textbf{d. Circuit compilation beyond local and energy-preserving gates.} As a first application of \cref{thm:informal_cutoff-SK}, we obtain two unitary compilation results: one for energy-preserving unitaries that are not necessarily \emph{local} (\cref{thm:compile_new_passive}), and one for unitaries that are not necessarily \emph{energy-preserving} (\cref{thm:compile_new}). Both results assume that the target unitary $U$ possesses an efficient qubit description, $V$ (\cref{def:circdesc}). Roughly, the latter says that $U$ can be simulated by first applying a simple isometry which flattens Fock states into unary-encoded qubit states, and then applying $V$ on qubits.

\begin{theorem}[informal, \Cref{thm:compile_new_passive} and \Cref{thm:compile_new}]\label{thm:compile_more_informal}
Let $U$ be an $m$-mode unitary which preserves the space of states with total photon-number at most $K$, such that $U$ has a $\delta$-approximate qubit circuit description with $L$ gates over $q$ qubits. Then, there exist a bosonic unitary circuit $C_U$ and $\alpha\in\mathbb R$ which (1) $\varepsilon$-approximates $U$ over the space of states with total photon-number at most $K$, (2) uses $\poly(m,K,L,\log(1/\varepsilon))$ gates, and
\begin{itemize}
    \item [(a)] (\Cref{{thm:compile_new_passive}}) if $U$ is non-local and number-preserving, $C_U$ uses gates from $\GKerr$ or $\GcrossKerr$, and coherent state $\ket\alpha$ with energy $O(mK+L +\log(1/\varepsilon))$.
    \item [(b)] (\Cref{thm:compile_new}) if $U$ is local but non-number-preserving, $C_U$ uses gates generated by local constant-degree energy-preserving polynomial Hamiltonians, and coherent state $\ket\alpha$ with energy $O(q^2/\varepsilon^2)$.
\end{itemize}
\end{theorem}
\noindent Part (a) of \Cref{thm:compile_more_informal} allows us to compile more general number-preserving unitaries than \cref{thm:informal_cutoff-SK}, allowing us to simulate e.g.\ unitary dynamics generated by the Bose--Hubbard Hamiltonian at total photon number $K$ (discussed shortly below). Part (b) can handle even \emph{non-energy-preserving} $U$, but the price to pay is the required energy scales polynomially with the desired precision instead of logarithmically, and that the gates used include gates generated by local constant-degree Hamiltonians beyond $\GKerr$ or $\GcrossKerr$. One application of (b) is that our parametric approximation from \Cref{thm:simulation_informal} now applies in a broader setting: Instead of requiring energy moments to be reasonably bounded throughout the computation (as in \Cref{thm:simulation_informal}), it now suffices for the mean energy $E$ to be polynomially bounded throughout the computation, as long as the target unitary has an efficient qubit description. Further applications are discussed in \Cref{ssscn:nonnumberpreserving}.

\SevSkip\textbf{e. State synthesis.} As a second application of \cref{thm:informal_cutoff-SK}, we obtain deterministic state-preparation within inverse-polynomial trace distance for two fundamental families of states: finite superpositions of Fock states, and high-quality Gottesman--Kitaev--Preskill (GKP) states efficiently.  Importantly, there are no intermediate measurements or postselections required. This is in contrast with usual protocols for GKP state preparation \cite{brenner2025complexity}. We begin with Fock state preparation:

\begin{theorem}[Fock state synthesis with a fixed number of modes (informal; \Cref{thm:new-finite})]\label{thm:new-finite_informal}
Let \(|\psi\rangle\) be any $m$-mode state of total photon number at most $K$.
There is a circuit on \(m+2\) modes whose reduced target output \(\rho\)
satisfies \(\frac12\norm{\rho-|\psi\rangle\langle\psi|}_1\le\epsilon\).
The resources required are (1) energy $\poly(K,1/\epsilon)$, (2) gate count $\poly(m,K, 1/\epsilon)\poly(K,1/\epsilon)^{O(m)}$, (3) given the target amplitudes to sufficient precision,
classical compilation takes $\poly(m,K, 1/\epsilon)\poly(K,1/\epsilon)^{O(m)}
\log(1/\epsilon)$ time.
\end{theorem}
\noindent Thus, assuming a constant number of modes $m$, $\ket{\psi}$ can be prepared efficiently in \model. Next, we show how to prepare high-quality GKP states, $|G_{\kappa,\Delta}\rangle$ (see \cref{eq:new-gkp} for definition):

\begin{restatable}[Deterministic GKP-state synthesis]{theorem}{GKPprep}\label{thm:new-gkp}
For every \(0<\kappa,\Delta\le1\) and \(0<\epsilon<1/2\), a circuit
on three bosonic modes, with one coherent input and two vacuum inputs,
prepares a reduced state \(\rho\) satisfying
$
\frac12\norm{\rho-
 |G_{\kappa,\Delta}\rangle\langle G_{\kappa,\Delta}|}_1\le\epsilon.
$
It uses only linear optics and either self-Kerr or cross-Kerr gates,
with gate count polynomial in \(\kappa^{-1},\Delta^{-1},\epsilon^{-1}\).
The circuit is deterministic, and a sufficient coherent-state energy is
\begin{equation}\label{eq:new-gkp-energy}
 O\!\left(
 \frac{(\kappa^{-2}+\Delta^{-2})^2}{\epsilon^2}
 \log^2\!\frac C{\kappa\epsilon}
 \right).
\end{equation}
\end{restatable}

\SevSkip\textbf{f. Application to Bose--Hubbard dynamics.} Finally, as a concrete illustration of our framework, in Appendix~\ref{app:bose-hubbard}, we study the quantum simulation of Bose--Hubbard dynamics. We first show that the standard Bose--Hubbard Hamiltonian, being energy-preserving, can be efficiently decomposed into energy-preserving gates on bounded-particle sectors, with polynomial complexity in the number of modes, evolution time, energy cutoff,
and inverse precision. We further extend the construction to states with bounded expected photon number. Finally, we consider an active, coherently driven Bose--Hubbard Hamiltonian and show how its dynamics can be approximated within \model\ by introducing a coherent-state ancilla, yielding an explicit tradeoff between the supplied coherent-state energy and the simulation error.

\subsection{Open questions}
\label{sec:open-questions}

Defining a standard model of computation over infinite-dimensional bosonic Hilbert spaces is significantly more challenging than in finite dimensions. In this work, we proposed \model, a model of quantum computation built from energy-preserving bosonic operations supplemented by a single resource state providing the energy for the computation. The resulting framework satisfies many of the properties one would expect from a standard computational model, including the existence of universal gate sets, efficient compilation, and robustness under natural variations of the model. It is universal in the sense that, when supplied with sufficiently energetic resource states, it can efficiently simulate other bosonic computational models, as well as qubit-based computations, fermionic computations, and interactions between bosonic, qubit, and fermionic degrees of freedom.

Having established a well-defined framework, a number of natural questions arise regarding the computational power, limitations, and complexity-theoretic structure of bosonic quantum computation.

\begin{itemize}
\item \textbf{Fault tolerance and experimental prospects.}
The next major question raised by this work is whether \model\ can be made fault tolerant and implemented in a practically realistic architecture. Recent work has shown that fault tolerance against general noise can be achieved for continuous-variable systems using concatenated GKP--qubit error-correcting codes~\cite{matsuura2024continuous}. A key lesson from such constructions is that careful control of energy is essential: without it, small errors can populate arbitrarily high-energy sectors and undermine the assumptions required for fault-tolerance analyses. Energy control is therefore not merely a technical convenience, but an important structural requirement for scalable bosonic computation.

A natural advantage of our model is that energy is controlled by construction: computation is built from energy-preserving operations, while energy is supplied explicitly through resource states. Together with the bosonic Solovay--Kitaev theorem developed in this work, this makes the framework particularly well-suited for studying practical implementations. In particular, once a suitable experimentally accessible universal gate set is identified, the Solovay--Kitaev theorem provides a systematic way to compile arbitrary operations in the model into that linear optical and Kerr gate set with controlled accuracy and overhead. Thus, the same structure that gives a robust complexity-theoretic definition also provides a natural route toward concrete fault-tolerant and experimental architectures.

From an experimental perspective, and also theory of fault tolerance, one would ideally realize universal computation using only experimentally accessible Gaussian operations, suitable non-Gaussian resource states, and adaptive measurements. This leads to a concrete open question: can nonlinear gates such as the Kerr gate be injected using only linear-optical operations and measurements on appropriately chosen resource states? More generally, can the non-Gaussian operations required for universality be shifted entirely into offline resource-state preparation?

Answering these questions, together with the efficient compilation guarantees established here, would help connect the abstract computational model to practical architectures for fault-tolerant bosonic quantum computation and clarify which energy resource states and experimentally available operations are sufficient for scalable universal bosonic quantum computation. 
    
\item \textbf{Variants of the model.}
A desirable feature of any standard computational model is robustness under natural variations of its definition. In this work, we establish such robustness in several directions, but a number of directions remain open.
\begin{itemize}
\item \emph{Ballistic versus gate-based computation.}
Is the computational power of the model unchanged if, instead of a gate-based computation, we consider unit-time evolution under a many-body, energy-preserving Hamiltonian supplied with energy resource states, sometimes referred to as a ballistic model? In finite-dimensional qubit systems, Feynman-type constructions imply an equivalence between such Hamiltonian models and $\BQP$. For general bosonic systems with active Hamiltonians, however, one should not expect an analogous equivalence without further restrictions. It is therefore natural to ask whether the equivalence is restored in the energy-preserving setting when energy resource states are supplied explicitly.

\item \emph{Homodyne measurements and phase sensitivity.}
Does allowing homodyne measurements change the computational power of the model? In particular, when phase-sensitive measurements are available, can pure coherent resource states enable computations that are inaccessible in the thermal model, whose resource states contain no phase information? Understanding this distinction may clarify whether phase coherence constitutes a computational resource beyond energy alone.

\end{itemize}
    \item \textbf{Restricted models of bosonic computation.}
Another natural direction is to understand the computational power of bosonic models under simultaneous restrictions on resources such as circuit depth, energy, number of modes, ancillas, and noise. For instance, what is the power of shallow bosonic circuits whose total energy is bounded polynomially? Could such a model compute Parity in constant depth when supplied with long-range gates, in contrast with what is believed to be possible for qubit $\mathsf{QAC}^0$ circuits~\cite{nadimpalli2024pauli, anshu2025computational, moore1999quantum}? More generally, it would be interesting to fully characterize the tradeoffs between the fundamental resources of bosonic computation---time (or circuit depth), space (number of modes and ancillas), and energy. In particular, can additional energy or ancillary modes compensate for limited circuit depth, and conversely, what lower bounds relate these resources?

A complementary question concerns the effect of noise. How does the computational power of these restricted models change in the presence of physically motivated noise, such as photon loss or dephasing, and how do the required energy, depth, and ancillary resources scale with the noise strength? Developing such tradeoffs could provide a more refined complexity theory of bosonic computation beyond the distinction between polynomial- and exponential-time computation.

    \item \textbf{Bosonic Hamiltonian complexity.}
A broad open direction is to develop a more systematic theory of Hamiltonian complexity for bosonic systems, particularly for physically natural low-degree Hamiltonians beyond Bose--Hubbard~\cite{childs2014bose}. Basic questions include identifying the complexity of bosonic analogues of local-Hamiltonian and other $\mathsf{QMA}$-type problems, as well as understanding how these complexities depend on restrictions such as interaction degree, locality, number of modes, and energy.

Closely related questions arise in the study of finite-temperature properties. What is the complexity of approximating Gibbs states, partition functions, or thermodynamic observables for restricted families of bosonic Hamiltonians? More generally, developing such results would help clarify the broader landscape of bosonic quantum complexity theory and identify which phenomena from finite-dimensional quantum complexity extend to the infinite-dimensional setting, which require explicit energy constraints, and which are genuinely new to bosonic systems.

    \item \textbf{Energy as a computational measure:} One way to view the main message of \cite{CGMMNRS25} is that energy, as a fundamental physical observable, is also a fundamental computational resource along with time and space, which in the bosonic setting specifically requires careful management. Note that the undecidability of whether the computation consumes finite energy in a bosonic quantum computation shown in \cite{CGMMNRS25} does not immediately violate Blum's second axiom for a complexity measure which requires the value of the measure to be decidable~\cite{Blum1967,BlumAxiomsWikipedia}. The situation is very similar to time complexity: deciding if a Turing machine halts in finite time is undecidable but deciding if a specific Turing machine halts within $3$ steps is decidable.
    We, however, conjecture that, given a finite value $c$, deciding if the value of expected energy is $\leq c$ is also undecidable. Assuming this conjecture, this is yet another indication of ``unphysicality'' of the standard model by Lloyd and Braunstein because energy, which happens to be an important resource in a bosonic computation, is not a suitable complexity measure according to Blum's criteria. We note that the energy-preserving model studied in this work does not suffer the same problem, because energy is provided to the input as a resource with a well-defined value. However, one can encounter similar undecidable problems if we consider problems such as: ``does there exist an energy level above which the output of a given \model\ circuit is close to a given state such as vacuum?''
    
    \item \textbf{Fast-forwarding of energy:} An instance for which we have not pinned down the complexity has exponential energy, polynomial space, and $O(1)$ precision. We know that for one oscillator \cite{jain2026efficient} we can simulate exponential energy input states in $\BQP$. It is interesting to ask whether we can generalize this to other energy-preserving Hamiltonians over multiple modes, which we leave as an open question.

    \item \textbf{Generalising the energy-preserving SK theorem.} Our SK theorem (\cref{thm:cutoff-SK}) is phrased in terms of an approximation on a strong cutoff subspace. A natural extension would be to ask for a similar SK theorem, but valid for smooth cutoff, i.e.\ in ECDN. Here simply cutting off and applying our SK does not work, because the cutoff would need to scale as $K=E/\varepsilon^2$, leading to a polynomial dependency on $1/\varepsilon$. However, such an approach works as long as the states throughout the computation have bounded exponential-energy, i.e.\ there exists $s>1$ such that $s^{\hN}$ has bounded expectation value $s^E$, since such states can be approximated to $\varepsilon$ precision with a truncation $K=O(E+\log(1/\varepsilon))$ \cite{upreti2026exponentially}. As shown in this work, these include Gaussian + Kerr acting on coherent states, and it is an interesting open question to characterize which gate sets generate states with bounded exponential-energy, starting from coherent states. Another limitation of \Cref{thm:cutoff-SK} is that it is prescribed for compiling bosonic operators into the particular Kerr + linear optical gate set. It is desirable to generalize this theorem for other natural gate sets.
    
    \item \textbf{Connections to resource theories:} \model\ is reminiscent of the resource theory of athermality, which is a source of inspiration for considering energy conservation. This would suggest that there might be a connection between resource engines and computation. Particularly:
        \begin{itemize}
            \item Resource theory of athermality: In the thermal model, we see that we need a temperature difference for computation, and in the RT of athermality, we similarly require a temperature difference to extract work.
            \item Resource theories of coherent thermodynamics: The primary shortcoming of the RT of athermality is that it ignores the cost of losing coherence. Similarly, if \model\ is stronger than the thermal model for some other measurement protocol, this would further strengthen the concept that resource engines can also compute.
        \end{itemize}
\end{itemize}

\subsection{Acknowledgements}

We thank Dominic Berry for helpful discussions. UC acknowledges inspiring discussions with Varun Upreti, Nicolás Quesada, Lukas Brenner, Robert K\"onig, Libor Caha, Robert Salzmann, Mohammed Ayyash, and funding from the European Union's Horizon Europe Framework Programme (EIC Pathfinder Challenge project Veriqub) under Grant Agreement No.~101114899. SG thanks David Gross for guidance and feedback, and apologizes for getting him kicked out of the Deutsche Bahn waiting room. SG, DR, and DS acknowledge support from the Deutsche Forschungsgemeinschaft
(DFG), project 563388236 (Bridge-QS, SPP 2514). SG additionally acknowledges
support from DFG project 572703436 (QPUP), EU QuantERA/DFG project 583918116
(SDPCODE), and BMFTR (PhoQuant). SM acknowledges inspiring discussions with Eugene Tang, Peter Love, and funding from the National Science Foundation (NSF CCF-2013062 and CCF-2534876). DS and HN acknowledge support from the PhoQuant project, funded by BMFTR.

\medskip

\noindent \textit{AI disclosure.} The authors acknowledge the use of AI tools to assist with manuscript preparation and the development and verification of mathematical proofs. 
The research question and approach were developed by the authors, who take full responsibility for the results presented in the manuscript.

\medskip

\noindent \textit{List of contributions.} All authors contributed to the development and verification of this work. DR, AM, DS, and HN led the development of the technical details. UC, SG, and SM contributed to conceptual and technical developments and supervised the project.

\section{Preliminaries}
\label{scn:preliminaries}
In this section, we introduce the necessary preliminary material. While we employ terminology which originates from quantum optics, we note that all concepts and results are equally valid for any bosonic system beyond photonics. For background on CV quantum information and quantum optics, see~\cite{weedbrook2012gaussian, braunstein2005quantum, walls1995quantum}.

\subsection{Definitions}
\label{sscn:defs}

\begin{definition}[Bosonic mode and Fock space]\label{def:bosonic_mode}
    A single bosonic mode is characterized by an infinite-dimensional Hilbert space $\mathcal{H}$ and a pair of mutually adjoint ladder operators---the annihilation operator $\hat{a}$ and the creation operator $\hat{a}^\dagger$---which satisfy the canonical commutation relation (CCR):
    \begin{equation}
        \label{eq:ccr}
        [\hat{a}, \hat{a}^\dagger] = \II,
    \end{equation}
    where $\II$ is the identity operator. The state space and associated operators are defined as follows:
    \begin{enumerate}
        \item \textbf{Vacuum state}: The unique ground state $\ket{0} \in \mathcal{H}$ defined by the relation:
            \begin{equation}
                \label{eq:vacuum}
                \hat{a}\ket{0} = 0.
            \end{equation}

        \item \textbf{Fock states}: The orthonormal basis $\{\ket{n}\}_{n=0}^\infty$ spanning $\mathcal{H}$ ($\braket{m}{n} = \delta_{m,n}$), where each state is generated via:
            \begin{equation}
                \label{eq:Fock_generation}
                \ket{n} = \frac{(\hat{a}^\dagger)^n}{\sqrt{n!}}\ket{0}.
            \end{equation}
            The action of the ladder operators on these states is given by:
            \begin{equation}
                \label{eq:ladder_action}
                \hat{a}\ket{n} = \sqrt{n}\ket{n-1}, \quad \hat{a}^\dagger\ket{n} = \sqrt{n+1}\ket{n+1}.
            \end{equation}

        \item \textbf{Number operator}: The self-adjoint operator $\hat{n} = \hat{a}^\dagger \hat{a}$ which counts the excitations in the mode, satisfying the eigenvalue equation:
            \begin{equation}
                \label{eq:number_op}
                \hat{n}\ket{n} = n\ket{n}.
            \end{equation}
    \end{enumerate}
    An arbitrary pure single-mode bosonic state $\ket{\psi} \in \mathcal{H}$ is represented as a superposition of these states, $\ket{\psi} = \sum_{n=0}^\infty c_n \ket{n}$, subject to the normalization constraint $\sum_{n=0}^\infty |c_n|^2 = 1$.

    A multimode bosonic Hilbert space is defined by taking a tensor product of single-mode bosonic spaces. The canonical commutation relations generalize to
    \begin{equation}
        [\hat{a}_i, \hat{a}_j] = 0, \qquad [\hat{a}_i, \hat{a}_j^\dagger] = \delta_{ij}\II.
    \end{equation}
\end{definition}

\begin{definition}[Total photon number operator]\label{def:total-photon-number}
    For an $m$-mode bosonic system, the \emph{total photon number operator} is
    \begin{equation}
        \hat{N}\coloneqq \sum_{j=1}^m \hat{n}_j = \sum_{j=1}^m \hat{a}_j^\dagger \hat{a}_j,
    \end{equation}
    where $\hat{n}_j \equiv \II^{\otimes(j-1)} \otimes \hat{n} \otimes \II^{\otimes(m-j)}$ is the number operator on mode $j$.
\end{definition}
\begin{definition}[Energy]\label{def:avg-photon-number}
    The \emph{energy} (or \emph{average photon number}) of a state $\ket{\psi}$ is $E(\ket{\psi}) \coloneqq \bra{\psi}\hat{N}\ket{\psi}$ (\cref{def:total-photon-number}).
\end{definition}
\begin{definition}[Coherent state]\label{def:coherent-state}
    For $\alpha \in \mathbb{C}$, a single-mode coherent state with parameter $\alpha$ is defined as
    \begin{equation}
        \ket{\alpha} \coloneqq e^{-|\alpha|^2/2}\sum_{n=0}^\infty \frac{\alpha^n}{\sqrt{n!}} \ket{n}.
    \end{equation}
    It is the eigenstate of the annihilation operator: $\hat{a}\ket{\alpha} = \alpha\ket{\alpha}$.
    The parameter $\alpha$ is called the \emph{displacement amplitude} since $\ket{\alpha} = \hat{D}(\alpha)\ket{0}$ (\cref{def:displacement}) and $\bra{\alpha}\hat{n}\ket{\alpha} = |\alpha|^2$.
    A multimode coherent state is $\ket{\bm{\alpha}} = \bigotimes_{j} \ket{\alpha_j}$ for $\bm{\alpha} = (\alpha_1,\dots,\alpha_m)$.
\end{definition}

\begin{definition}[Position and momentum operators]
    \label{def:quadratures}
    For a single bosonic mode with annihilation and creation operators $\hat{a}$ and $\hat{a}^\dagger$, the \emph{position} (quadrature) operator $\hQ$ and \emph{momentum} (quadrature) operator $\hP$ are defined as
    \begin{equation}
        \Q \coloneqq \frac{\hat{a} + \hat{a}^\dagger}{\sqrt{2}}, \qquad \P \coloneqq \frac{\hat{a} - \hat{a}^\dagger}{i\sqrt{2}}.
    \end{equation}
    These satisfy the canonical commutation relation $[\Q, \P] = i\II$. In terms of quadratures, the number operator is $\hat{n} = (\Q^2 + \P^2 - \II)/2$.
\end{definition}

\begin{definition}[Displacement operator]
    \label{def:displacement}
    For $\alpha \in \mathbb{C}$, the \emph{displacement operator} is defined as
    \begin{equation}
        \hat{D}(\alpha) \coloneqq e^{\alpha \hat{a}^\dagger - \alpha^* \hat{a}}.
    \end{equation}
    It generates coherent states from vacuum: $\ket{\alpha} = \hat{D}(\alpha)\ket{0}$. The displacement operator is unitary and satisfies $\hat{D}(\alpha)^\dagger \hat{a} \hat{D}(\alpha) = \hat{a} + \alpha$.
\end{definition}

\begin{definition}[Squeezing operator]
    \label{def:squeezing}
    For $\xi = r e^{i\phi} \in \mathbb{C}$, the \emph{squeezing operator} is defined as
    \begin{equation}
        \hat{S}(\xi) \coloneqq e^{\frac{1}{2}(\xi^* \hat{a}^2 - \xi \hat{a}^{\dagger 2})}.
    \end{equation}
    Squeezing is a Gaussian gate that does not preserve photon number. It satisfies $\hat{S}(\xi)^\dagger \hat{a} \hat{S}(\xi) = \hat{a} \cosh r - e^{i\phi} \hat{a}^\dagger \sinh r$.
\end{definition}

\begin{definition}[Linear optics]
    \label{def:linear-optics}
    \emph{Linear optics} (also called \emph{passive Gaussian unitaries}) refers to the group of unitaries generated by Hamiltonians quadratic in creation and annihilation operators that preserve total photon number. On $m$ modes, linear optics is generated by:
    \begin{itemize}
        \item \textbf{Phase shifters}: $e^{i\theta \hat{n}_j}$ acting on mode $j$, which apply a phase $e^{i\theta}$ to the annihilation operator.
        \item \textbf{Beam splitters}: $e^{\theta(\hat{a}_j^\dagger \hat{a}_k - \hat{a}_j \hat{a}_k^\dagger)}$ acting on modes $j,k$, which mix two modes via a unitary rotation.
    \end{itemize}
    Equivalently, linear optics on $m$ modes is isomorphic to the unitary group $U(m)$, acting on the mode operators via $\hat{a}_j \mapsto \sum_k U_{jk} \hat{a}_k$ for $U \in U(m)$~\cite{kok2007linear}. Linear optics preserves the total photon number operator $\hat{N}$.
\end{definition}

\begin{definition} [Polynomial Hamiltonian]
    A multimode bosonic Hamiltonian $H$ is said to be a \emph{polynomial Hamiltonian in the creation and annihilation operators} if it can be written as
\be
H
=
\sum_{\bm r,\bm s}
h_{\bm r,\bm s}
\prod_{j=1}^{m}
\left(\hat a_j^\dagger\right)^{r_j} \hat a_j^{s_j},
\ee
where $\bm r=(r_1,\ldots,r_m)$ and
$\bm s=(s_1,\ldots,s_m)$ are vectors of nonnegative integers and the sum contains finitely many terms. The \emph{degree} of the monomial
\be
\prod_{j=1}^{m}
\left(\hat a_j^\dagger\right)^{r_j} \hat a_j^{s_j}
\ee
is its total degree
\be
|\bm r|+|\bm s|
=
\sum_{j=1}^{m}(r_j+s_j),
\ee
and the degree of $H$ is the maximum degree among its monomials with nonzero coefficient:
\be
\deg(H)
=
\max_{h_{\bm r,\bm s}\neq 0}
\sum_{j=1}^{m}(r_j+s_j).
\ee
\end{definition}

\begin{definition}[Wick (normal) ordering]
    \label{def:wick-ordering}
    An operator polynomial in creation and annihilation operators is in \emph{Wick ordered} (or \emph{normal ordered}) form if all creation operators appear to the left of all annihilation operators. For a single mode, a Wick ordered monomial has the form $\hat{a}^{\dagger k} \hat{a}^l$ for $k,l \geq 0$. For multiple modes, a Wick ordered monomial has the form $\prod_j \hat{a}_j^{\dagger k_j} \prod_j \hat{a}_j^{l_j}$. Any polynomial Hamiltonian can be brought to Wick ordered form using the canonical commutation relations.
\end{definition}

\begin{definition}[Energy-preserving Hamiltonians]
    \label{def:energy-preserving}
    A Hamiltonian $H$ on $m$ modes is \emph{energy-preserving} (or \emph{photon-number-preserving}) if it commutes with the total photon number operator: $[H, \hat{N}] = 0$. Equivalently, $H$ is energy-preserving if and only if in its Wick ordered expansion, every monomial has an equal number of creation and annihilation operators (see \cref{lem:passive-wick} below). Energy-preserving unitaries satisfy $U \hat{N} U^\dagger = \hat{N}$ and thus cannot change the energy of any state.
\end{definition}

\begin{lem}[Energy-preserving polynomial Hamiltonians]
    \label{lem:passive-wick}
    Hamiltonians of the form \be H=\sum_{m,n} c_{m,n}\hat{h}_{m,n},\ee where $\hat{h}_{m,n}$ are in normal order with $m$ creation operators and $n$ annihilation operators, are energy-preserving if and only if there are only terms of the form $\hat h_{m,m}$, i.e.\ the numbers of creation and annihilation operators are equal.
\end{lem}

\begin{proof}
    Without loss of generality, assume some arbitrary term $\hat h_{m,n}$ in normal order, $\hat h_{m,n} =\hat  A^\dagger\hat B$, where $\hat A^\dagger = \hat a^\dagger_{p_1}\hat a^\dagger_{p_2}\dots \hat a^\dagger_{p_m}$ and $\hat B = \hat a_{q_1}\hat a_{q_2}\dots \hat a_{q_n}$. If $\hat h_{m,n}$ is not in normal order form, then it can be brought to that ordering by using the commutation relations of the creation and annihilation operators.

    We will use the following identity:
\be[X, YZ] = [X, Y]Z + Y[X, Z].\ee
    First, from repeated application of the identity, we obtain
    \be[\hat N, \hat A^\dagger] = [\hat N, \hat a^\dagger_{p_1}\hat a^\dagger_{p_2}\dots \hat a^\dagger_{p_m}] = \sum_{r=1}^m \hat a^\dagger_{p_1}\dots [\hat N, \hat a_{p_r}^\dagger] \dots \hat a^\dagger_{p_m}.\ee
    Since $[\hat N, \hat a^\dagger_{x}] = \hat a^\dagger_x$, each term in the sum is equal to $\hat A^\dagger$, so

    \be[\hat N, \hat A^\dagger] = m\hat A^\dagger.\ee
    Similarly, from $[\hat N, \hat a_x] = -\hat a_x$, we get
    \be[\hat N, \hat B] = -n\hat B.\ee
    Thus, by another application of the identity,
    \be
    [\hat N, \hat h_{m,n}] = [\hat N, \hat A^\dagger \hat B] = m\hat A^\dagger \hat B - n\hat A^\dagger \hat B = (m-n)\hat h_{m,n}.
    \ee
    And thus,
    \be[\hat N, H] = \sum_{m,n} c_{m,n} (m-n) \hat h_{m,n}.\ee
    Finally, since normally ordered monomials are linearly independent,

    \be[\hat N,H] = 0 \iff c_{m,n} = 0,\quad \forall m \neq n,\ee
    for all the terms.
\end{proof}

\begin{definition}[Schwartz space]
    \label{def:schwartz-space}
    The \emph{Schwartz space} $\calS(\RR^m)$ is the space of smooth functions on $\RR^m$ that are rapidly decreasing together with all their derivatives. Formally, $f \in \calS(\RR^m)$ if $f \in C^\infty(\RR^m)$ and for all multi-indices $\bm\alpha, \bm\beta$, $\sup_{\bm x \in \RR^m} |x^{\bm\alpha} \partial^{\bm\beta} f(\bm x)| < \infty$. In the bosonic setting, we identify $\calS(\RR^m)$ with the subspace of $L^2(\RR^m)$ consisting of wavefunctions that are Schwartz functions in the position representation. Equivalently, $\calS(\RR^m) = \bigcap_{k \in \NN} \calD(\hat{N}^k)$, where $\calD(\cdot)$ denotes the domain of an operator~\cite{CGMMNRS25}. Gaussian gates preserve the Schwartz space. Whether a unitary generated by a polynomial Hamiltonian of degree $\geq 3$ preserves the Schwartz space is undecidable in general~\cite{CGMMNRS25}; the gates used in this work are chosen to preserve it.
\end{definition}

\begin{definition}[Thermal states]
    \label{def:thermal-state}
    For a single mode with mean photon number $\bar{n} \geq 0$, the \emph{thermal state} is the mixed state
    \begin{equation}
        \rho_{\mathrm{th}}(\bar{n}) \coloneqq \frac{1}{\bar{n}+1} \sum_{n=0}^\infty \left(\frac{\bar{n}}{\bar{n}+1}\right)^n \ket{n}\!\bra{n}.
    \end{equation}
    Thermal states have zero displacement and variance $\bar{n} + \frac{1}{2}$ in each quadrature. A multimode thermal state is a tensor product of single-mode thermal states.
\end{definition}

\begin{definition}[Gaussian states]
\label{def:gaussian-state}
     A Gaussian quantum state is a state which can be prepared from a thermal state using polynomial Hamiltonians of degree at most $2$. Examples include vacuum states, coherent states (\cref{def:coherent-state}), squeezed states, and thermal states (\cref{def:thermal-state}).
     In particular, pure Gaussian states over a single mode are those obtained by applying displacement (\cref{def:displacement}) and squeezing (\cref{def:squeezing}) operations to the vacuum state.
\end{definition}

\begin{definition}[Gaussian gates]
    \label{def:gaussian-gates}
    \emph{Gaussian gates} are unitaries generated by Hamiltonians that are polynomials of degree at most $2$ in the creation and annihilation operators (equivalently, in the quadrature operators $\Q, \P$). On $m$ modes, Gaussian gates include: displacement (\cref{def:displacement}), squeezing (\cref{def:squeezing}), phase shift, beam splitters (\cref{def:linear-optics}), and the Fourier transform $e^{i\frac\pi4(\Q^2 + \P^2)}$. Compositions of these operations generate arbitrary Gaussian unitaries. Gaussian gates map Gaussian states to Gaussian states and preserve the Schwartz space.
\end{definition}

\begin{definition}[Kerr gate]
    \label{def:kerr-gate}
    The \emph{Kerr gate} (or \emph{self-Kerr gate}) on mode $j$ is the unitary $e^{i\theta \hat{n}_j^2}$ generated by the squared number operator $\hat{n}_j^2$, where $\theta \in \RR$ is the evolution time. The Kerr gate is energy-preserving since $[\hat{n}_j^2, \hat{N}] = 0$. The \emph{cross-Kerr gate} on modes $j, k$ is $e^{i\theta \hat{n}_j \hat{n}_k}$, which is also energy-preserving.
\end{definition}

\begin{definition}[Cutoff operator seminorm]\label{def:cutoff-norm}
    Let \(\Pi_{\le K}\) be the spectral projector onto the subspace with total photon number at most \(K\). For an operator \(A\), the \emph{cutoff operator seminorm} (or \emph{energy-cutoff seminorm}) is
    \begin{equation}
        \|A\|_{(K)}
        :=
        \|\Pi_{\leq K} A \Pi_{\leq K}\|_\infty.
    \end{equation}
    This quantity is the operator norm of the compression of $A$ to the
    subspace of states with total photon number at most $K$, i.e., $\Pi_{\leq K} A \Pi_{\leq K}$. Note that $\|\bullet\|_{(K)}$ is not a norm because it can send nonzero $A$ to zero. It is, however, a norm for operators supported on the cut-off subspace.
\end{definition}

We always write $\Pi_{\le K}$ for the photon-cutoff projector, reserving $\Pi_n$ for the projector onto exactly $n$ photons; thus $\Pi_{\le K}=\sum_{n=0}^K\Pi_n$. Hereafter, we use $K$ to denote the particle-number cutoff imposed by this projector. For a mean-energy bound we instead use $E$, i.e., to impose an energy bound $E$ on a density matrix $\rho$ we write $\Tr (\rho \hat N) \leq E$.

\begin{definition}[Trotter splitting]
    \label{def:trotter}
    For Hamiltonians $A$ and $B$, the \emph{Trotter product formula} (or \emph{Trotter splitting}) approximates $e^{-it(A+B)}$ by alternating short evolutions under $A$ and $B$:
    \begin{equation}
        e^{-it(A+B)} = \lim_{n \to \infty} \left(e^{-itA/n} e^{-itB/n}\right)^n.
    \end{equation}
    For bounded operators or under suitable domain conditions, the \emph{first-order Trotter error} scales as $O(t^2/n)$ for $n$ steps. Higher-order formulas such as
    \begin{equation}
        e^{-it(A+B)} \approx \left (e^{-itA/(2n)} e^{-itB/n} e^{-itA/(2n)}\right )^n
    \end{equation}
    achieve $O(t^3/n^2)$ error. For unbounded operators, convergence may fail in operator norm but can hold in strong operator topology under appropriate conditions~\cite{Tamura2000,burgarth2023state}.
\end{definition}

\begin{definition}[Dual-rail encoding~\cite{knill2001scheme}]
    \label{def:dual-rail}
    \emph{Dual-rail encoding} represents a qubit using two bosonic modes: the logical states are $\ket{0}_L = \ket{0}\ket{1}$ and $\ket{1}_L = \ket{1}\ket{0}$ in the Fock basis, i.e., a single photon in the second or first mode, respectively. This encoding preserves total photon number (one photon per logical qubit) and allows implementation of single-qubit gates via linear optics. It has been used in many previous works, including \cite{knill2001scheme}, to achieve universal quantum computation using bosons. An alternative \emph{single-rail encoding} uses $\ket{0}_L = \ket{0}$ and $\ket{1}_L = \ket{1}$ on a single mode.
\end{definition}

\paragraph{Basic notation:} We specify natural numbers with $\NN = \{1,2, \ldots\}$ and $\NN_0$ to denote $\{0\} \cup \NN$.
We use italic bold symbols, such as $\bm n$, for vectors and multi-indices. For an occupation vector, $|\bm n|=\sum_j n_j$ denotes total photon number; for a bit string, it denotes Hamming weight. The identity operator is $\II$. An unsubscripted norm denotes the Hilbert-space norm for vectors and the operator norm for operators (also written $\norm{\cdot}_\infty$). We write $\norm{\cdot}_{\mathrm F}$ for the Frobenius norm, $\norm{\cdot}_1$ for the trace norm, and $\frac12\norm{\rho-\sigma}_1$ for trace distance.
We use Knuth's double up-arrow notation to denote tetration. In particular,
$2 \uuarr n$
denotes a tower of \(n\) copies of \(2\), defined recursively by
\be
2 \uuarr 1 = 2,
\qquad
2 \uuarr (n+1) = 2^{\,2 \uuarr n}.
\ee

\subsection{Energy-constrained norms}\label{sscn:energynorm}

\begin{definition}[Energy-constrained diamond norm]\label{def:ecd-norm}
    For a linear Hermiticity-preserving map \(\Phi\) acting on trace-class operators, the \emph{energy-constrained diamond norm} (ECDN) is
    \begin{equation}
        \|\Phi\|_{\diamond,E}
        :=
        \sup_{\rho_{RA}:\ \Tr(\hat N_A\rho_A)\le E}
        \left\|
        (\mathrm{id}_R\otimes \Phi)(\rho_{RA})
        \right\|_1 ,
    \end{equation}
    where \(\hat N_A\) denotes the total photon-number operator on the bosonic system \(A\), and the supremum is over all auxiliary systems \(R\) and all states \(\rho_{RA}\) whose marginal on \(A\) has expected total photon number at most \(E\). The ECDN avoids pathologies of the standard diamond norm in infinite dimensions~\cite{winter2017energy}.
\end{definition}

We need the following definition when we study parametric approximations in \Cref{sec:parametric-approx2}.
\begin{definition}[Energy moment bound]
    \label{def:energy-moment}
    For a state $\ket{\psi}$ evolving under a polynomial Hamiltonian $H$ as $\ket{\psi(s)} = e^{-isH}\ket{\psi}$, we define the \emph{energy moment bound} over time interval $[0,t]$ as
    \begin{equation}
        M(t)  \coloneqq \sup_{s \in [0,t]} \norm{(1+\hat{N})^{d/2}\ket{\psi(s)}},
    \end{equation}
    where $d$ is the active degree of the Hamiltonian (maximum degree of an active monomial) and $\hat{N}$ is the total photon number operator. This bound controls the growth of energy during evolution and appears in simulation error bounds.
\end{definition}

The following elementary lemma relates the two notions. It is often useful because the ECDN allows states with high-energy tails, whereas
\(\|\cdot\|_{(K)}\) only controls the operator on a sharp cutoff subspace. Throughout we assume integer cut-off $K$.

\begin{lem}[From cutoff norm to energy-constrained diamond norm]
    \label{lem:cutoff-to-ec-diamond}
    Let \(U,V\) be unitaries, and let
    \be
        \mathcal U(\rho)=U\rho U^\dagger,
        \qquad
        \mathcal V(\rho)=V\rho V^\dagger .
    \ee
    Then for every \(K>E\),
    \begin{equation}
        \|\mathcal U-\mathcal V\|_{\diamond,E}
        \le
        2\|(U-V)\Pi_{\le K}\|_\infty
        +
        4\sqrt{\frac{E}{K+1}} .
    \end{equation}
    In particular, if \(U\) and \(V\) preserve the cutoff subspace \(\Ran(\Pi_{\le K})\), then
    \be
        \|(U-V)\Pi_{\le K}\|_\infty
        =
        \|\Pi_{\le K}(U-V)\Pi_{\le K}\|_\infty
        =
        \|U-V\|_{(K)},
    \ee
    and hence
    \begin{equation}
        \|\mathcal U-\mathcal V\|_{\diamond,E}
        \le
        2\|U-V\|_{(K)}
        +
        4\sqrt{\frac{E}{K+1}} .
    \end{equation}
\end{lem}

\begin{proof}
    Let \(\rho_{RA}\) be any state satisfying
    \(\Tr(\hN_A\rho_A)\le E\), and define
    \be
        P:=\II_R\otimes \Pi_{\leq K},
        \qquad
        \rho_K:=P\rho_{RA}P .
    \ee
    Since the spectrum of \(\hN\) is contained in \(\{0,1,2,\ldots\}\), Markov's inequality (see~\cite{CJMMM26}) gives
    \be
        \Tr[(\II-\Pi_{\le K})\rho_A]
        \le
        \frac{E}{K+1}.
    \ee
    By the gentle measurement lemma~\cite{wilde2013quantum, winter2017energy},
    \begin{equation}
        \|\rho_{RA}-\rho_K\|_1
        \le
        2\sqrt{\frac{E}{K+1}} .
    \end{equation}
    Now write \(\Delta:=\mathrm{id}_R\otimes(\mathcal U-\mathcal V)\). Since \(\mathcal U\) and
    \(\mathcal V\) are channels, \(\|\Delta(X)\|_1\le 2\|X\|_1\). Therefore
    \be
        \|\Delta(\rho_{RA})\|_1
        \le
        \|\Delta(\rho_K)\|_1
        +
        2\|\rho_{RA}-\rho_K\|_1 .
    \ee
    For the first term, using
    \be
        U\rho_K U^\dagger - V\rho_K V^\dagger
        =
        (U-V)\rho_K U^\dagger
        +
        V\rho_K(U^\dagger-V^\dagger),
    \ee
    and the fact that \(\rho_K\) is supported on \(\Ran(\Pi_{\le K})\) on the system register, we get
    \be
        \|\Delta(\rho_K)\|_1
        \le
        2\|(U-V)\Pi_{\le K}\|_\infty \|\rho_K\|_1
        \le
        2\|(U-V)\Pi_{\le K}\|_\infty .
    \ee
    Combining the two estimates gives
    \be
        \|\Delta(\rho_{RA})\|_1
        \le
        2\|(U-V)\Pi_{\le K}\|_\infty
        +
        4\sqrt{\frac{E}{K+1}} .
    \ee
    Taking the supremum over all \(\rho_{RA}\) with \(\Tr(\hN_A\rho_A)\le E\) proves the claim.
\end{proof}

\begin{lem}[Cutoff norm of energy-preserving monomials]
    \label{lem:cutoff-monomial-bound}
    Let
    \be    M_{\bm\alpha,\bm\beta}=(\hat a^\dagger)^{\bm\alpha} \hat a^{\bm\beta}
    \ee
    be energy-preserving, with $\bm\alpha,\bm\beta \in \mathbb{Z}_{\geq 0}^m$ with degree $d = 2r$, i.e., \(|\bm\alpha|=|\bm\beta|=r\). Then
    \begin{equation}
        \|M_{\bm\alpha,\bm\beta}\|_{(K)}
        =
        \|\Pi_{\le K} M_{\bm\alpha,\bm\beta}\Pi_{\le K}\|_\infty
        \le
        K^r ,
    \end{equation}
    for all $K\in\mathbb N$. Consequently, if
    \be
        A=\lambda M_{\bm\alpha,\bm\beta}+\overline{\lambda}M_{\bm\alpha,\bm\beta}^\dagger ,
    \ee
    for $\lambda\in\mathbb C$, then
    \begin{equation}
        \|A\|_{(K)}
        \le
        2|\lambda|K^r
        =
        O\!\left(|\lambda| K^{d/2}\right).
    \end{equation}
\end{lem}

\begin{proof}
    It suffices to bound
    \(\|M_{\bm\alpha,\bm\beta}\Pi_{\le K}\|_\infty\). We first note that, on the subspace
    with total photon number at most \(K\), each annihilation operator contributes at
    most a factor of order \(\sqrt K\). More precisely,
    \be
        \|\hat a_j\Pi_{\le K}\|_\infty \le \sqrt K,
    \ee
    and hence, for a product of \(r\) annihilation operators,
    \be
        \|\hat a^{\bm\beta} \Pi_{\le K}\|_\infty \le K^{r/2}.
    \ee
    After applying \(\hat a^{\bm\beta}\), the total photon number is still at most \(K-r\). (If $r > K$, the annihilation product vanishes on the cutoff, so the upper bound trivially holds.)
    Applying \(r\) creation operators can increase the total photon number by at most
    \(r\), so
    \be
        \|(\hat a^\dagger)^{\bm\alpha} \Pi_{\le J}\|_\infty
        \le
        (J+r)^{r/2},
    \ee
    for any $J\ge0$. Therefore
    \be
        \|M_{\bm\alpha,\bm\beta}\Pi_{\le K}\|_\infty
        =
        \|(\hat a^\dagger)^{\bm\alpha} \hat a^{\bm\beta} \Pi_{\le K}\|_\infty
        \le
        K^{r/2}K^{r/2}
        \le
        K^r .
    \ee
    Since
    \be
        \|\Pi_{\le K} M_{\bm\alpha,\bm\beta}\Pi_{\le K}\|_\infty
        \le
        \|M_{\bm\alpha,\bm\beta}\Pi_{\le K}\|_\infty,
    \ee
    we obtain the first claim. The bound for
    \(A=\lambda M_{\bm\alpha,\bm\beta}+\overline{\lambda}M_{\bm\alpha,\bm\beta}^\dagger\)
    then follows from the triangle inequality and the same estimate applied to
    \(M_{\bm\alpha,\bm\beta}^\dagger\).
\end{proof}

\section{Why energy-preserving gates? The issues of active gate decompositions}\label{sec:pathology}

Lloyd and Braunstein \cite{lloyd_quantum_1999} showed that linear optics, the Gaussian gates (displacement, squeezing, phase shift), and a single higher-order gate such as $\N^2$ or $\Q^3$ generate the entire Lie algebra of symmetric polynomials in position and momentum operators $\Q_j,\P_j$.
They then apply the equation
\begin{equation}\label{eq:LB99-2}
    e^{iAt}e^{iBt}e^{-iAt}e^{-iBt} = e^{-t^2[A,B]} + O(t^3),
\end{equation}
to argue that any operator in the Lie algebra can be approximated using the ``Gaussian + $\N^2$ or $\Q^3$'' gate set.
For unbounded operators such as polynomials in $\Q$ and $\P$, this formula does not generally hold with an $O(t^3)$ error in operator norm.
A related issue arises for Trotter splitting: \cref{sec:no-convergence} gives an example where it fails to converge in operator norm, even for linear-optical Hamiltonians.
Additionally, one needs to take into account the \emph{energy} of the output state of a circuit.
We give an example in \cref{sec:not-enough-energy} where a circuit of $k$ gates achieves energy at least $2\uuarr\Omega(k)$ (where $\uuarr$ is Knuth's double up-arrow notation; see \cref{scn:preliminaries}), whereas each Gaussian gate of bounded evolution time can increase energy $E$ to at most $O(E+1)$.
In this example, most of the amplitude weight is concentrated on states with huge photon number.

These results do not rule out efficient compilation under suitable energy bounds.
In \cref{sec:cutoff-sk}, we give efficient compilation results below a photon-number cutoff using linear optics and Kerr gates. Energy-preserving Hamiltonians are particularly nice here, because they can be written as a direct sum of finite-dimensional operators acting on the $N$-photon subspace for all $N\in\NN_0$.
In \cref{sec:parametric-approx}, we show how to simulate non-energy-preserving circuits under energy-moment bounds using energy-preserving gates supplied with energy resource states.

This leaves the question of whether an energy-constrained variant of \cite{lloyd_quantum_1999} can be made rigorous.
Recently, \cite{Becker2025} have shown rigorous convergence rates for Trotter splitting of operators $e^{A+L}$, where $L$ is relatively $A$-bounded, as well as for energy-limited unitary dynamics.
It would be interesting to extend these methods to commutator formulas such as \cref{eq:LB99-2}.
Trotter splitting is also used implicitly in \cite{lloyd_quantum_1999} to write elements of the Lie algebra in terms of single commutator monomials, which are then decomposed with \cref{eq:LB99-2}.\footnote{At least that is our interpretation since there is no explicit discussion of Trotter splitting.}

Unfortunately, relative boundedness is non-trivial and does not hold in general.
There are also domain issues: $A+B$ can be essentially self-adjoint even when $A$ individually is not.
Note that (essential) self-adjointness is undecidable for general polynomial Hamiltonians \cite{CGMMNRS25}.
Even when the Trotter formula converges, one needs to be careful to avoid pathological cases.
In \cref{sec:too-much-energy}, we give an example where the Trotter formula converges in strong operator topology, but the energy of the output state is doubly exponential in the number of steps, while the exact evolution has bounded energy.
This means that almost all the energy must live in heavy tails, so it might be possible to gently truncate between Trotter steps by inserting measurements.

In conclusion, formulating an energy-constrained version of the LB model will require considerable technical effort to avoid pathological cases.
In our paper we are able to circumvent these issues by taking all energy from a single source of coherent light on which we act only via energy-preserving operations (see \cref{scn:model-definition}).

\begin{remark}
    Prior work (e.g.\ \cite{kalajdzievski2021exact,CJMMM26,CGMMNRS25}) relies on the Campbell identity
    \begin{equation}
        e^{X} Y e^{-X} = \sum_{n=0}^\infty\frac{\adj^n_X(Y)}{n!}, \qquad \adj_X(Y) \coloneq [X,Y]
    \end{equation}
    to work with bosonic gates in the Heisenberg picture, which is a special case of the Baker--Campbell--Hausdorff formula.
    Using such identities with unbounded operators requires great care due to issues of domains, continuity, and differentiability.
    For Gaussian gates they are generally unproblematic, but once gates like $\Q^3$ are added, we should justify the formula more rigorously.
    In \cref{app:heisenberg}, we provide such a justification.
    In summary, as long as we restrict ourselves to the Schwartz space, the Heisenberg identities in the following and in prior works can be justified.
    This is sufficient as all gates we use here preserve the Schwartz space,\footnote{Though generally, it is \emph{undecidable} whether a given gate preserves the Schwartz space \cite{CGMMNRS25}.}
    apart from the number-controlled squeezing gate in \cref{prop:uuarr}.
\end{remark}

\subsection{No convergence in operator norm}\label{sec:no-convergence}

It is well known that the Trotter product formula does not generally converge in operator norm for unbounded operators (see e.g.~\cite{Tamura2000}).
A bosonic example is provided by \cite{burgarth2023state} with $A= \frac12(\Q^2+\P^2)$ and $B = \frac12(\Q\P+\P\Q)$, where the Trotter error approaches $2$, the maximum distance of two unitaries in operator norm.
Here, we provide an example with just linear optics.

\begin{proposition}
    Let $A = \a_1^\dagger \a_1 - \a_2^\dagger \a_2$ and $B = \a_1^\dagger\a_2 + \a_2^\dagger \a_1$.
    Define
    \begin{equation}
        \calU \coloneq e^{-i(A+B)},\qquad \calT_n \coloneq \bigl(e^{-iA/n}e^{-iB/n}\bigr)^n.
    \end{equation}
    Then for all $\ket{\psi}\in L^2(\RR^2)$,
    \begin{equation}\label{eq:lim:TnUpsi}
        \lim_{n\to\infty} (\calT_n - \calU)\ket{\psi} = 0,
    \end{equation}
    but
    \begin{equation}\label{eq:lim:TnU}
        \lim_{n\to\infty}\norm{\calT_n - \calU} = 2.
    \end{equation}
\end{proposition}
\begin{proof}
    Convergence in strong operator topology \cref{eq:lim:TnUpsi} is straightforward since $A,B$ preserve photon number and are bounded in each sector.

    To prove the failure of operator-norm convergence in \cref{eq:lim:TnU}, we compute the action of $\calU$ and $\calT_n$ in the Heisenberg formalism.
    We first show that the formula never holds exactly in the single-photon subspace, and then we turn a tiny error for one photon into an asymptotically maximal error for many photons.

    Define the row vector $\bfa^\dagger \coloneq (\ad_1,\ad_2)$.
    Recall the CCR $[\a_j,\ad_k] = \delta_{jk}$ and $[\a_j,\a_k]=0$.
    The commutators with $A$ and $B$ are
    \begin{equation}
        [A,\ad_1] = \ad_1,\qquad [A,\ad_2]=-\ad_2,\qquad [B,\a_1^\dagger] = \a_2^\dagger,\qquad[B,\a_2^\dagger]=\a_1^\dagger.
    \end{equation}
    Consider $Y(s) \coloneq e^{-isA}\bfa^\dagger e^{isA}$.
    Then
    \begin{equation}\label{eq:ODE:Ys}
        \frac{d}{ds} Y(s) = -i \left(e^{-isA}[A,\a_1^\dagger]e^{isA}, e^{-isA}[A,\a_2^\dagger]e^{isA}\right) = -i Y(s) \sigma_z,\qquad Y(0) = \bfa^\dagger,
    \end{equation}
    where, as always,
    \begin{equation}
        \sigma_x =
        \begin{pmatrix}
            0&1\\1&0
        \end{pmatrix},\qquad\sigma_y =
        \begin{pmatrix}
            0&-i\\i&0
        \end{pmatrix},\qquad \sigma_z =
        \begin{pmatrix}
            1&0\\0&-1
        \end{pmatrix}.
    \end{equation}
    The solution to the ODE in \cref{eq:ODE:Ys} is given by
    \begin{equation}
        e^{-isA}\bfa^\dagger e^{isA} =  \bfa^\dagger e^{-is\sigma_z}.
    \end{equation}
    With the same argument, we can also compute the Heisenberg evolutions of $B$ and $A+B$:\footnote{This proof does not require \cref{app:heisenberg} since we only need to consider finite-dimensional subspaces.}
    \begin{equation}\label{eq:UT-Heisenberg}
        \calU \bfa^\dagger \calU^\dagger = \bfa^\dagger U,\quad U \coloneq e^{-i(\sigma_z+\sigma_x)},\qquad\qquad \calT_n \bfa^\dagger \calT_n^\dagger = \bfa^\dagger U_n,\quad U_n \coloneq \bigl(e^{-i\sigma_z/n}e^{-i\sigma_x/n}\bigr)^n.
    \end{equation}
    Writing the exponentials in terms of Paulis gives
    \begin{equation}\label{eq:exp-pauli}
        e^{-it\sigma_z} = \cos t\,\II - i\sin t\,\sigma_z,\qquad e^{-it\sigma_x} = \cos t\,\II - i\sin t\,\sigma_x,
    \end{equation}
    and since $M \coloneq \sigma_x+\sigma_z$ satisfies $M^2 = 2\II$,
    \begin{equation}
        \begin{aligned}\label{eq:u}
            U &= e^{-iM} = \sum_{k=0}^\infty \frac{(-1)^k2^k}{(2k)!}\II - i\sum_{k=0}^\infty\frac{(-1)^k2^k}{(2k+1)!}M\\
            &= \cos \sqrt{2}\,\II - \frac{i\sin\sqrt2}{\sqrt2}(\sigma_z+\sigma_x).
        \end{aligned}
    \end{equation}
    Next, let $S_t \coloneq e^{-it\sigma_z}e^{-it\sigma_x}$, so $U_n = S_{1/n}^n$.
    Then, by \cref{eq:exp-pauli} and $\sigma_z\sigma_x = i\sigma_y$, we have
    \begin{equation}
        \begin{aligned}
            S_t &= (\cos t\,\II - i\sin t\, \sigma_z)(\cos t\,\II - i\sin t\,\sigma_x) \\
            &= \cos^2t\,\II - i\bigl(\cos t\sin t\,(\sigma_x + \sigma_z) + \sin^2t\,\sigma_y\bigr).
        \end{aligned}
    \end{equation}
    The next step is to write $S_t$ in the form of \cref{eq:exp-pauli} so that we can easily compute its $n$-th power.
    Define
    \begin{equation}
        \alpha_t \coloneq \arccos(\cos^2 t),\qquad \bfm_t \coloneq \frac{(\cos t,\sin t,\cos t)}{\sqrt{1+\cos^2 t}},\qquad \bfsigma \coloneq (\sigma_x,\sigma_y,\sigma_z).
    \end{equation}
    For $0<t\le1$, we have $\sin\alpha_t = \sqrt{1-\cos^4t}=\sin t\,\sqrt{1+\cos^2t}$, so
    \begin{equation}
        S_t = \cos \alpha_t \,\II - i\sin\alpha_t \, (\bfm_t \cdot \bfsigma),
    \end{equation}
    where $\bfm_t \cdot \bfsigma$ is the dot-product of two row vectors.
    Since the Pauli matrices anticommute,
    \begin{equation}
        (\bfm_t\cdot\bfsigma)^2 =\norm{\bfm_t}^2\II = \frac{2\cos^2 t + \sin^2 t}{1+\cos^2 t}\II = \II.
    \end{equation}
    Thus, we have analogously to \cref{eq:exp-pauli},
    \begin{equation}\label{eq:Un-pauli}
        U_n = S_{1/n}^n = \cos(n\alpha_{1/n})\,\II - i\sin(n\alpha_{1/n})\, (\bfm_{1/n}\cdot\bfsigma).
    \end{equation}
    Using the expansion $\arccos(\cos^2 t) = \sqrt{2}t + O(t^3)$ as $t\downarrow0$, with $t=1/n$,
    \begin{equation}
        n\alpha_{1/n}\ \to\ \sqrt{2},\qquad \bfm_{1/n}\ \to\ \frac{(1,0,1)}{\sqrt 2},
    \end{equation}
    and, recalling \cref{eq:u}, we also have
    \begin{equation}
        U_n\ \to \ \cos\sqrt{2}\,\II -\frac{i\sin\sqrt{2}}{\sqrt2}(\sigma_x + \sigma_z) = U
    \end{equation}
    without ever achieving equality for any finite $n$ since $U_n$ in \cref{eq:Un-pauli} has a nonzero $\sigma_y$ component. Indeed, $\cos^2t\ge\cos(2t)$ gives $0<n\alpha_{1/n}\le2<\pi$, so $\sin(n\alpha_{1/n})\ne0$.

    Observe that the computation so far was effectively in the single-particle subspace.
    It remains to turn a tiny error for one photon into a large error for many photons.
    Let $R_n = U^\dagger U_n$.
    Then $R_n\to\II$, but $R_n\ne\II$ for all $n\in\NN$.
    Hence, we can pick a sequence $0\ne\theta_n\to0$, such that $e^{i\theta_n}$ is an eigenvalue of $R_n$ with (unit) eigenvector $\bm v_n\in\CC^2$.
    Define the supermode creation operator
    \begin{equation}
        \b_n^\dagger \coloneq \bfa^\dagger\bm v_n = v_{n,1} \ad_1 + v_{n,2}\ad_2,
    \end{equation}
    and define the $K_n$-photon state
    \begin{equation}
        \ket{\psi_n} \coloneq \frac{\b_n^{\dagger K_n}}{\sqrt{K_n!}}\ket{0},\qquad K_n \coloneq \left\lfloor\frac{\pi}{\abs{\theta_n}}\right\rfloor.
    \end{equation}
    The choice of $K_n$ is so that
    \begin{equation}\label{eq:limKn}
        \lim_{n\to\infty} \abs{K_n\theta_n} = \pi,
    \end{equation}
    which follows from $K_n\abs{\theta_n} \le \pi < K_n\abs{\theta_n}+\abs{\theta_n}$ and $\theta_n\to0$.
    Using $\calU\ket{0}=\calT_n\ket{0}=\ket{0}$ and \cref{eq:UT-Heisenberg}, we obtain
    \begin{equation}
        \calU \ket{\psi_n} = \frac{(\calU\b_n^\dagger\calU^\dagger)^{K_n}}{\sqrt{K_n!}}\ket{0} = \frac{(\bfa^\dagger U\bm v_n)^{K_n}}{\sqrt{K_n!}}\ket{0},\qquad \calT_n\ket{\psi_n} = \frac{(\bfa^\dagger U_n\bm v_n)^{K_n}}{\sqrt{K_n!}}\ket{0}.
    \end{equation}
    Since $\bm v_n$ is an eigenvector of $R_n=U^\dagger U_n$, we have
    \begin{equation}
        U_n\bm v_n = U R_n\bm v_n = e^{i\theta_n} U\bm v_n.
    \end{equation}
    Thus,
    \begin{equation}
        \calT_n\ket{\psi_n} = e^{iK_n\theta_n}\frac{(\bfa^\dagger U\bm v_n)^{K_n}}{\sqrt{K_n!}}\ket{0} = e^{iK_n\theta_n}\calU\ket{\psi_n},
    \end{equation}
    and therefore
    \begin{equation}
        \norm{\calT_n - \calU} \ge  \norm{(\calT_n - \calU)\ket{\psi_n}} = \abs{e^{iK_n\theta_n} - 1} \to 2,
    \end{equation}
    as by \cref{eq:limKn} $e^{iK_n\theta_n}\to e^{\pm i\pi}=-1$.
\end{proof}

\subsection{Doubly exponential energy despite convergence}\label{sec:too-much-energy}

In the following example, Trotter splitting converges in strong operator topology, but energy diverges extremely fast.
This is similar to \cite[Lemma 3.4]{CGMMNRS25}, where alternating Fourier transforms with $\Q^3$ gates yields doubly exponential energy.
Here, we use $\Q^4$ instead, so that $\Q^2+\P^2+\Q^4$ is semibounded, giving a simple proof for essential self-adjointness.
Note that the Gaussian and cubic gate set can synthesize $\Q^4$ exactly, as shown in \cite{kalajdzievski2021exact}.

\begin{proposition}
    Let $t>0$, $A = \Q^2 + \P^2$, $B = \Q^4$, and $H = A + B$.
    Trotter splitting converges, i.e., for all $\ket{\psi}\in L^2(\RR)$
    \begin{equation}\label{eq:HAB-trotter}
        \lim_{n\to\infty} \left(e^{itA/n}e^{itB/n}\right)^n\ket{\psi} = e^{it(A+B)}\ket{\psi}.
    \end{equation}
    However, energy with vacuum state input diverges doubly exponentially fast in $n$, i.e.,
    \begin{equation}
        \ev{\N}{\psi_n} \ge 2^{2^{\Omega(n)}},\qquad \ket{\psi_n}\coloneq \left(e^{itA/n}e^{itB/n}\right)^n\ket{0},
    \end{equation}
    whereas $\bra{0}e^{-itH}\N e^{itH}\ket{0} = O(1)$ for all $t$.
\end{proposition}
\begin{proof}
    By \cite[Theorem 7.15]{Schmdgen2012}, any semibounded symmetric operator is essentially self-adjoint if the span of its Stieltjes vectors \cite[Definition 7.1]{Schmdgen2012} is dense in the Hilbert space.
    Here, $A,B,H$ are clearly bounded from below and symmetric.
    It remains to prove that Fock states are Stieltjes vectors, i.e., for any Fock state $\ket{k}$,
    \begin{equation}\label{eq:stieltjes}
        \sum_{n=1}^\infty\norm*{H^n\ket{k}}^{-1/2n} = \infty.
    \end{equation}
    For any $k$, there exists a constant $C_k$, such that
    \begin{equation}
        \norm{H^n\ket{k}} \le C_k^n n^{2n},
    \end{equation}
    which follows from the fact that $H$ has degree $4$ and
    \begin{equation}
        \norm*{\a^{\dagger 4n}\ket{k}} \le \sqrt{k+4n}^{4n}.
    \end{equation}
    Since $(C_k^n n^{2n})^{-1/2n} = 1/(\sqrt{C_k}n)$, the series in \cref{eq:stieltjes} diverges.
    By the same argument, $A, B$ are also essentially self-adjoint on $\Dfin$.

    Hence, $A+B$ is also essentially self-adjoint on $\calD(A)\cap\calD(B)$, which yields \cref{eq:HAB-trotter} with \cite[Theorem VIII.31]{Reed1981-kj}.
    Since $H \ge \N$ as a quadratic form, we have
    \begin{equation}
        \bra{0}e^{-itH}\N e^{itH}\ket{0} \le \bra{0}e^{-itH}H e^{itH}\ket{0} = \ev{H}{0} = O(1).
    \end{equation}
    It remains to prove the energy divergence in Trotterization.

    On the Schwartz space $\calS(\RR)$, we have for any $\tau=t/n$
    \begin{equation}\label{eq:update-B}
        e^{-i\tau \Q^4} \Q e^{i\tau \Q^4} = \Q,\qquad e^{-i\tau \Q^4} \P e^{i\tau \Q^4} = \P + 4\tau\Q^3.
    \end{equation}
    The latter follows because Baker--Campbell--Hausdorff terminates with $[\Q^4,\Q^3]=0$, and
    \begin{equation}
        e^{-i\tau \Q^4} \P e^{i\tau \Q^4} = \P + (-i\tau)[\Q^4,\P] = \P + 4\tau\Q^3.
    \end{equation}
    $\tau A$ acts like a rotation with angle $\theta\coloneq 2\tau$ \cite[Eq. (34)]{weedbrook2012gaussian}:
    \begin{equation}\label{eq:update-A}
        e^{-i\tau A} \Q e^{i\tau A} = \Q \cos\theta - \P \sin\theta,\qquad  e^{-i\tau A} \P e^{i\tau A} = \P\cos\theta + \Q\sin\theta.
    \end{equation}
    Combining \cref{eq:update-B} and \cref{eq:update-A}, we get a recurrence relation for updating the observables at each Trotter step.
    See \cref{thm:Heisenberg} for a formal justification of this Heisenberg evolution.

    Let $c=\cos\theta$ and $s=\sin\theta$. Let $\Q_k,\P_k$ be the observables after $k$ Trotter steps with $\Q_0=\Q$ and $\P_0 = \P$.
    Each step first applies Rule \cref{eq:update-A} and then Rule \cref{eq:update-B}:
    \begin{equation}
        \Q_{k+1} = c\Q_{k} - s\P_k - 4\tau s\Q_k^3,\qquad \P_{k+1} = c\P_k + s\Q_k + 4\tau c\Q_k^3.
    \end{equation}
    After $n$ steps, we have
    \begin{equation}\label{eq:Xn}
        \Q_n = \alpha_n \Q^{3^n} + f_n(\Q,\P),\qquad \alpha_n = \gamma^{(3^n-1)/2},\qquad \gamma = -4\tau s,
    \end{equation}
    where $f_n(\Q,\P)$ is a polynomial of degree $<3^n$,
    and $\alpha_k$ solves the recurrence $\alpha_1=\gamma,\alpha_{k+1}=\gamma\alpha_k^3$.

    Then we use
    \begin{equation}
        \ev*{\N}{\psi_n} \ge \frac12 \ev*{\Q^2}{\psi_n}-\frac12.
    \end{equation}
    \cref{eq:Xn} gives
    \begin{equation}
        \ev*{\Q^2}{\psi_n}  = \ev*{\Q_{n}^2}{0} = \norm{\Q_n\ket{0}}^2 \ge  \norm{\Pi_{\ge 3^n} \Q_n\ket{0}}^2,
    \end{equation}
    where $\Pi_{\ge 3^n}$ is the projector onto the subspace spanned by number states with $\ge 3^n$ photons.
    This projector kills every term except $\a^{\dagger 3^n}$, so
    \begin{equation}
        \norm{\Pi_{\ge 3^n} \Q_n\ket{0}}^2 = \abs{\gamma}^{3^n-1} \frac{(3^n)!}{2^{3^n}} \in 2^{2^{\Omega(n)}},
    \end{equation}
    using Stirling's formula and noting that $s \ge \tau = t/n$ for large $n$, so $\abs{\gamma} \in \Theta(1/n^2)$ for fixed $t$.
\end{proof}

\subsection{A no-go for efficient approximate synthesis of active gates}\label{sec:not-enough-energy}

We can give a rigorous no-go for synthesizing arbitrary bosonic circuits with Gaussian and $\Q^3$ or $\N^2$ gates efficiently, when no energy bounds are promised a priori.
For that, note that $n$ Gaussian and $\N^2$ gates with bounded evolution time can achieve energy at most $2^{O(n)}$ when starting in the vacuum state \cite{CGMMNRS25}.
With the $\Q^3$ gate, at most doubly exponential energy $2^{2^{O(n)}}$ can be achieved \cite{CJMMM26}.
In contrast, \cite{CGMMNRS25} gives an example where iterated number-controlled squeezing gates produce photon numbers at least as large as a power tower of height $\Omega(n)$ with high probability.

We write $\bm1_S$ for the indicator function of $S$ (and the corresponding spectral projector when applied to an operator). In particular, $\bm1_S(\N_j)$ projects onto states with photon number in $S$ in mode $j$. For mode $j$, define $\Pi_{\le K}^{(j)}\coloneq\bm1_{[0,K]}(\N_j)$ and $\Pi_{\ge K}^{(j)}\coloneq\bm1_{[K,\infty)}(\N_j)$, acting as the identity on all other modes. Thus, $\ev{\Pi_{\le K}^{(j)}}{\psi}$ is the probability of measuring at most $K$ photons in mode $j$ in the state $\ket{\psi}$.

\begin{proposition}[{\cite[Lemma 4.23]{CGMMNRS25}}]\label{prop:uuarr}
    There exists a family of circuits $C_n$ consisting of $n$ number-controlled two-mode squeezing gates of the form $\exponential(\N_0(\a_1\a_2-\a_1^\dagger\a_2^\dagger))$ and Gaussian gates of total runtime $O(n)$, such that the output state $\ket{\psi_n}\coloneq C_n\ket{0}$ on vacuum input satisfies
    \begin{equation}
        \ev{\Pi_{\ge 2\uuarr\Omega(n)}^{\mathrm{out}}}{\psi_n} \ge 1-2^{-\Omega(n)}.
    \end{equation}
    In words, the output mode has at least $2\uuarr\Omega(n)$ photons with probability at least $1-2^{-\Omega(n)}$.
\end{proposition}

Hence, the circuit family $C_n$ of \cref{prop:uuarr} cannot be approximated by even $\exp(n)$ Gaussian and $\Q^3$ gates of constant runtime, since they can only produce states of energy $2^{2^{O(\exp(n))}}$.

\begin{corollary}\label{cor:uuarr}
    For any circuit family $\wtC_n$ of $\exp(n)$ Gaussian and $\Q^3$ gates of bounded evolution time, we have
    \begin{equation}
        \abs{\ev{C_n^\dagger \wtC_n}{0}} < 2^{-\Omega(n)},
    \end{equation}
    i.e.\ the overlap of the output states of $C_n$ and $\wtC_n$ on vacuum input is exponentially small.
\end{corollary}
\begin{proof}
    The claim follows from \cref{prop:uuarr} and a Markov bound on the number distribution of $\wtC_n\ket{0}$.
\end{proof}

This means that even a circuit of exponential depth can only achieve exponentially small overlap with the output state of $C_n$.

A critical reader might now object that the number-controlled squeezing gate is not well motivated physically.
While this example rules out efficient compilation of arbitrary polynomial gates to the cubic gate set, one could still hope to efficiently synthesize the more natural $\Q^3$ gate using Gaussian and Kerr gates.
Unfortunately, we can also give a counterexample to that case.

\cite{CGMMNRS25} shows that Gaussian and $\Q^3$ gates can produce states with doubly exponential energy.
However, we do not know whether this energy just comes from a fat tail as in our example in \cref{sec:too-much-energy}, or the state has constant probability mass on high occupation numbers as in \cref{prop:uuarr}.
Therefore, our next example shows how to construct a state that has most of its support on doubly exponential photon numbers using just Gaussian and $\Q^3$ gates.

\begin{theorem}\label{thm:repeated-squaring}
    There exists a circuit $\calC_n$ of $O(n^2)$ Gaussian and $\Q^3$ gates of bounded evolution time on $n+1$ modes whose output state $\ket{\psi_n} \coloneq \calC_n \ket{0}^{\otimes(n+1)}$ satisfies
    \begin{equation}\label{eq:N-tailbound}
        \ev{\Pi_{\le 2^{2^n-O(1)}}^{(n)}}{\psi_n} \le e^{-\Omega(n^2)}.
    \end{equation}
    In words, the probability of measuring at most $2^{2^n-O(1)}$ photons in the output mode $n$ is at most $e^{-\Omega(n^2)}$.
\end{theorem}

This immediately yields the following obstruction to synthesis with Gaussian and Kerr gates.

\begin{corollary}
    For any circuit family $\wtcalC_n$ of $2^{o(n)}$ Gaussian and Kerr gates of bounded evolution time, we have
    \begin{equation}
        \abs{\ev{\calC_n^\dagger \wtcalC_n}{0}} < 2^{-\Omega(n)},
    \end{equation}
    i.e.\ the overlap of the output states of $\calC_n$ and $\wtcalC_n$ on vacuum input is exponentially small.
\end{corollary}

\begin{proof}[Proof of \cref{thm:repeated-squaring}]
    To obtain an obstruction using only Gaussian and cubic gates, we replace the number-controlled squeezing gate of \cref{prop:uuarr} by a position-controlled displacement gate
    \begin{equation}
        G_{jk} \coloneq \exp(-i\Q_j^2 \P_k)
    \end{equation}
    which can be synthesized exactly using a constant number of bounded-time Gaussian and cubic gates \cite{kalajdzievski2021exact}.
    $G_{jk}$ displaces the target mode $k$ by $x^2$ if the control mode $j$ has position $x$.
    This gives doubly exponential growth in the number of steps.

    \begin{figure}[t]
        \centering
        \includegraphics[]{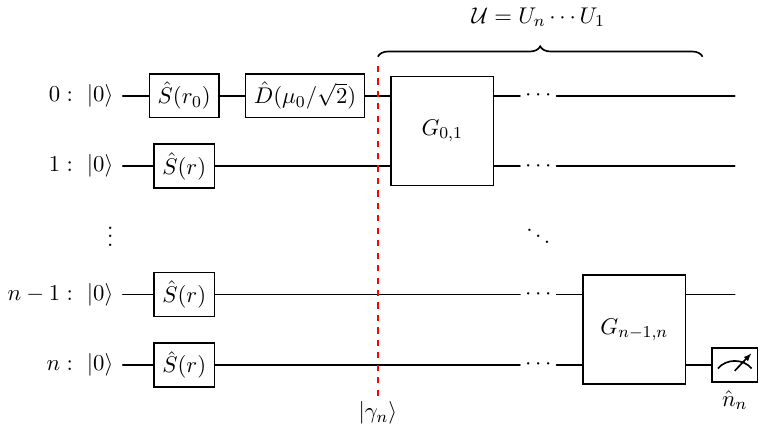}
        \caption{Circuit $\calC_n$ from \cref{thm:repeated-squaring}. The Gaussian preparation stage first squeezes all modes, then displaces mode $0$.
        $\calU$ then applies a sequence of $G_{j-1,j}=\exponential(-i\Q_{j-1}^2\P_j)$ to perform repeated squaring on position, as $G_{j-1,j}$ displaces mode $j$ by the squared position of mode $j-1$.}
        \label{fig:repeated-squaring}
    \end{figure}
    We construct the circuit $\calC_n$ as depicted in \cref{fig:repeated-squaring}.
    Label the modes $0,\dots,n$.
    Define $\calU \coloneq U_n\dotsm U_1$ with $U_j = G_{j-1,j}$.
    $\calU$ essentially performs repeated squaring on the position.
    To satisfy the tail bound in \cref{eq:N-tailbound}, the full circuit $\calC_n$ first prepares a position-squeezed input state $\ket{\gamma_n}$ so that the position of each mode has small variance.

    We choose the input parameters as follows. The $n+1$ modes are initialized with position distributions
    \begin{equation}\label{eq:initial-distribution}
        Q_0 \sim \calN(\mu_0,\sigma_0^2),\quad \mu_j\coloneq 2^{2^j} \quad(j \in \{0,\dots,n\}),\qquad Q_j\sim\calN(0,\sigma^2)\quad (j\in[n]),
    \end{equation}
    where the variance parameters are
    \begin{equation}
        \sigma_0 \coloneq \frac{\mu_0b_0}{n},\qquad \sigma \coloneq \frac{\mu_1 b_1}{n},\qquad b_0 \coloneq \frac{3^{-n}}{16},\qquad b_{j}\coloneq\frac{3^{j-n}}{16(n+1)}\quad(j\in[n]),
    \end{equation}
    and the $b_j$ are bounds that will be used later.
    Recall that the squeezed vacuum state $\hS(r)\ket{0}$ with $\hS(r) = \exponential(\frac{r}2(\a^2-\a^{\dagger2}))$ has variance \cite[Eq. (2.1.35a)]{Schumaker1986}
    \begin{equation}
        \Var(\Q) = \frac{e^{-2r}}2,
    \end{equation}
    so we need $r = -\frac12\log(2\sigma^2)$ to achieve variance $\sigma^2$.
    Here we need a squeezing parameter of $O(n)$ on each mode, achieved using $O(n)$ constant-strength squeezing gates per mode and hence $O(n^2)$ gates in total.
    Since mode $0$ is the first control, we increase position $\Q_0$ with a displacement gate $\hD(\mu_0/\sqrt{2})=e^{-i\mu_0\P_0}$.

    We now analyze $\calU$ in the Heisenberg picture. The action of $U_j$ is as follows, with $A_j \coloneq i\Q_{j-1}^2\P_j$:
    \begin{equation}
        U_j^\dagger \, \Q_j\, U_j = e^{A_j}\, \Q_j\, e^{-A_j} = \Q_j + [A_j,\Q_j] + \frac12[A_j,[A_j,\Q_j]]+\dotsm = \Q_j + \Q_{j-1}^2
    \end{equation}
    as $[\P_j,\Q_j] = -i$ and thus $[A_j,\Q_j] = \Q_{j-1}^2$, with all higher commutators vanishing. The identity is justified on the Schwartz space by \cref{cor:Heisenberg-mixed}.
    Let $V_j=U_j\cdots U_1$, with $V_0=\II$, and define $\Y_0=\Q_0$ and $\Y_{j+1}=\Y_j^2+\Q_{j+1}$. Since $V_j$ leaves mode $j+1$ untouched, induction gives
    \begin{equation}
        V_{j+1}^\dagger\Q_{j+1}V_{j+1}
        =V_j^\dagger(\Q_{j+1}+\Q_j^2)V_j
        =\Q_{j+1}+\Y_j^2=\Y_{j+1}.
    \end{equation}
    Later gates leave $\Q_j$ unchanged, so $\calU^\dagger\Q_j\calU=\Y_j$. In the position representation, $\Y_j$ therefore acts by multiplication by $F_j(x)$, where $F_0(x)=x_0$ and $F_{j+1}(x)=x_{j+1}+F_j(x)^2$. The input position density $\abs{\gamma_n(x)}^2$ defines independent Gaussian coordinates $Q_0,\ldots,Q_n$ with distributions given in \cref{eq:initial-distribution}. Setting $Y_j=F_j(Q_0,\ldots,Q_n)$, we obtain $Y_0=Q_0$ and $Y_{j+1}=Q_{j+1}+Y_j^2$.
    For any Borel set $S$,
    \begin{equation}
        \ev{\bm1_S(\Y_j)}{\gamma_n}
        =\int_{\RR^{n+1}}\bm1_S(F_j(x))\abs{\gamma_n(x)}^2\,dx
        =\Pr_{\gamma_n}[Y_j\in S],
    \end{equation}
    where $\Pr_{\gamma_n}$ denotes probability with respect to the joint position distribution with density $|\gamma_n(x)|^2$.

    For $R_n\coloneq \frac78\mu_n$, using $\ket{\psi_n}=\calU\ket{\gamma_n}$, define
    \begin{equation}
        \begin{aligned}
            p_n &\coloneq \ev{\bm1_{[R_n,\infty)}(\abs{\Q_n})}{\psi_n}
            = \ev{\calU^\dagger\bm1_{[R_n,\infty)}(\abs{\Q_n})\calU}{\gamma_n} \\
            &= \ev{\bm1_{[R_n,\infty)}(\abs*{\calU^\dagger\Q_n\calU})}{\gamma_n}
            = \ev{\bm1_{[R_n,\infty)}(\abs{\Y_n})}{\gamma_n}
            = \Pr_{\gamma_n}[\abs{Y_n}\ge R_n].
        \end{aligned}
    \end{equation}
    Thus, $p_n$ is the probability of measuring position of absolute value $\ge R_n$ in the output mode.

    \begin{claim}\label{claim:pn-bound}
       We have $p_n \ge 1-e^{-n^2/3}$ for sufficiently large $n$.
    \end{claim}
    \begin{proof}
        Probabilities involving $Q_j$ and $Y_j$ below are with respect to the joint distribution in $\ket{\gamma_n}$.
        By the standard Gaussian tail bound for $Z\sim\calN(0,1)$ \cite[Proposition 2.1.2]{vershynin2026high},
        \begin{equation}
            \Pr[\abs{Z}\ge t] \le 2e^{-t^2/2}.
        \end{equation}
        So we have
        \begin{equation}\label{eq:bound-pr-X0}
            \Pr[\frac{\abs{Q_0-\mu_0}}{\mu_0}>b_0] = \Pr[\frac{\abs{Y_0-\mu_0}}{\sigma_0}>\frac{\mu_0b_0}{\sigma_0}] = \Pr[\abs{Z}>n] \le 2e^{-n^2/2},
        \end{equation}
        and analogously for $j\in[n]$
        \begin{equation}\label{eq:bound-pr-Xj}
            \Pr[\frac{\abs{Q_j}}{\mu_j}>b_j] = \Pr[\frac{\abs{Q_j}}{\sigma}>\frac{\mu_jb_j}{\sigma}] \le \Pr[\abs{Z} > n] \le 2e^{-n^2/2}.
        \end{equation}
        Define the ``good'' event 
        \begin{equation}
            E_n \coloneq \left\{ \frac{\abs{Q_0-\mu_0}}{\mu_0}\le b_0,\ \frac{\abs{Q_j}}{\mu_j}\le b_j\;\forall j\in[n]  \right\},
        \end{equation}
        where the position variables are not ``too far'' from their expectation.
        Then by the union bound, $\Pr[E_n] \ge 1-e^{-n^2/3}$ for sufficiently large $n$.
        So it only remains to verify that $E_n$ implies $\abs{Y_n}\ge R_n$.

        To track the relative error, define for $j\in\{0,\dots,n\}$,
        \begin{equation}
            \Delta_j \coloneq \frac{\abs{Y_j - \mu_j}}{\mu_j}.
        \end{equation}
        Conditioned on $E_n$, we will prove the bounds $\Delta_j \le \delta_j$ with 
        \begin{equation}
            \delta_j \coloneq 3^j b_0 + \sum_{k=1}^j 3^{j-k}b_k
            =3^{j-n}\left(\frac{1}{16}+\frac{j}{16(n+1)}\right)\le\frac18
            \qquad(0\le j\le n)
        \end{equation}
        by induction.
        The base case is $\Delta_0\le b_0=\delta_0$. Assuming $\Delta_j\le\delta_j\le1/8$, we have $\Delta_j\le1$, so $\abs{Y_j}\le(1+\Delta_j)\mu_j$ and
        \begin{equation}
            \abs{Y_j^2 - \mu_j^2} = \abs{Y_j - \mu_j}\cdot\abs{Y_j+\mu_j} \le \Delta_j\mu_j\cdot (2+\Delta_j)\mu_j \le 3\Delta_j\mu_j^2 = 3\Delta_j\mu_{j+1}.
        \end{equation}
        Thus, for $j<n$,
        \begin{equation}
            \Delta_{j+1} = \frac{\abs{Y_{j+1} - \mu_{j+1}}}{\mu_{j+1}} = \frac{\abs*{(Y_{j}^2 + Q_{j+1}) - \mu^2_{j}}}{\mu_{j+1}} \le \frac{\abs*{Y_j^2 -\mu_j^2}}{\mu_{j+1}} + \frac{\abs{Q_{j+1}}}{\mu_{j+1}} \le 3\Delta_{j} + b_{j+1},
        \end{equation}
        conditioned on $E_n$.
        Since $3\delta_j+b_{j+1}=\delta_{j+1}$, this gives $\Delta_{j+1}\le\delta_{j+1}\le1/8$, completing the induction.
        In particular, $Y_n \ge \frac{7}{8}\mu_n = R_n$ conditioned on $E_n$.
    \end{proof}

    The last step of the proof is to turn the position tail bound of \cref{claim:pn-bound} into a number tail bound.
    \begin{claim}
        Let $K_n \coloneq \lfloor\mu_n/100\rfloor$. Then
        \begin{equation}
            \ev{\Pi_{\le K_n}^{(n)}}{\psi_n} \le e^{-\Omega(n^2)}.
        \end{equation}
    \end{claim}
    \begin{proof}
        Write
        \begin{equation}
            \ket{\psi_n} = \ket{\phi} + \ket{\chi_n},\qquad \ket{\phi} \coloneq \Pi_{\le K_n}^{(n)}\ket{\psi_n},\quad\ket{\chi_n} \coloneq \Pi_{>K_n}^{(n)}\ket{\psi_n}.
        \end{equation}
        Then
        \begin{equation}\label{eq:sqrtpn}
            \sqrt{p_n} = \norm{\bm1_{[R_n,\infty)}(\abs{\Q_n})\ket{\psi_n}} \le \norm{\bm1_{[R_n,\infty)}(\abs{\Q_n})\ket\phi} + \norm{\ket{\chi_n}},
        \end{equation}
        where $\norm{\ket{\chi_n}}^2 = \ev{\Pi_{>K_n}^{(n)}}{\psi_n}$.
        Additionally, since $\Q_n^2 \le \Q_n^2+\P_n^2 =2\N_n+1$ as quadratic forms, and $\ket{\phi} \in \Ran(\Pi_{\le K_n}^{(n)})$, we have
        \begin{equation}
            \norm{\bm1_{[R_n,\infty)}(\abs{\Q_n})\ket\phi}^2 \le \frac{\ev{\Q_n^2}{\phi}}{R_n^2} \le \frac{2K_n + 1}{R_n^2} \eqcolon a_n.
        \end{equation}
        The first inequality is Markov's inequality for the position distribution, obtained by applying the spectral theorem to the scalar bound
        \(\bm1_{[R_n,\infty)}(|x|)\le x^2/R_n^2\), valid for all \(x\in\mathbb R\).
        Inserting into \cref{eq:sqrtpn} gives $\norm{\ket{\chi_n}} \ge \sqrt{p_n} - \sqrt{a_n}$ and thus
        \begin{equation}
            \ev{\Pi_{>K_n}^{(n)}}{\psi_n} = \norm{\ket{\chi_n}}^2 \ge (\sqrt{p_n} -\sqrt{a_n})^2.
        \end{equation}
        We bound
        \begin{equation}
            a_n \le \frac{2\mu_n/100+1}{(7\mu_n/8)^2} \le \frac{2}{\mu_n}.
        \end{equation}
        Thus, $\sqrt{a_n} \le 2^{-2^{n-1}+1/2} \le \frac14 e^{-n^2/3}$ for large $n$ and
        \begin{equation}
            \sqrt{p_n} -\sqrt{a_n} \ge \sqrt{1-e^{-n^2/3}} - \frac14 e^{-n^2/3} \ge 1 - \frac54 e^{-n^2/3}.
        \end{equation}
        Finally, we have
        \begin{equation}
            \ev{\Pi_{\le K_n}^{(n)}}{\psi_n} \le 1-\left(1 - \frac54 e^{-n^2/3}\right)^2 \le 3 e^{-n^2/3}.
        \end{equation}
    \end{proof}
    \cref{eq:N-tailbound} follows with $K_n = \lfloor\mu_n/100\rfloor=2^{2^n-O(1)}$, which completes the proof.
\end{proof}

\section{\modelname\ (\model)\ and its \BQP\ analogue, \texorpdfstring{$\BEQP{}{}$}{BEQP}}
\label{scn:model-definition}
Having discussed the shortcomings of non-energy-preserving bosonic computation as in the LB model (\cref{sec:pathology}), we now define our energy-preserving model, \model\ (\cref{sscn:modeldef}), along with its corresponding complexity class $\BEQP{}{}$ analogous to \BQP\ (\cref{sscn:BEQ}).

\subsection{\modelname\ (\model)}
\label{sscn:modeldef}
For the following, we utilize the definitions of \cref{scn:preliminaries}.

\begin{definition}[\modelname\ ($\modelparam{G}{E}{L}$)]\label{def:model}~
\begin{itemize}
    \item (Input) String $x\in \{0,1\}^*$, whose length is denoted $\varsigma$.
    \item (Parameters) 
    \begin{itemize}
        \item Gate set $\calG$ of energy-preserving polynomial Hamiltonians. If not specified, assumed to be $O(1)$ degree with constant coefficients and acting non-trivially on $O(1)$ modes.
        \item Energy bound $E$ on coherent state energy. If not specified, assumed to be $\poly(\varsigma)$.
        \item Circuit depth $L$. If not specified, assumed to be $\poly(\varsigma)$.
    \end{itemize}
    
    \item (State space) $(\mathbb C^{\infty})^{\otimes(m+1)}\cong\ell^2 (\NN_0^{m+1},\CC)$.

    \item (Initial state\footnote{Unlike the qubit setting, preparing the Fock state $\ket{1}$ is non-trivial in the CV setting. Thus, instead of encoding the input $x$ directly into the initial state as done in DV computation, without loss of generality we assume $x$ is input to the efficient algorithm for uniformly generating the gates in our circuit. If needed, one can use \Cref{prop:prepare-Fock-states} to explicitly prepare a Fock/binary encoding of the input.}) $\ket{0}^{\otimes m}\otimes \ket{\alpha}$ with $\ket{0}$ the vacuum state and $\ket{\alpha}$ a coherent state with amplitude $\alpha\in\CC$ such that $\abs{\alpha}^2\leq {E}$. We assume that the amplitude is efficiently uniformly generated given input $x$. Having a single coherent state $\ket{\alpha}$ is without loss of generality (\cref{lem:coher-gen}).

    \item (Gates) Unitaries $e^{-iHt}$ for $H\in\calG$, where $t\ge0$ is evolution time. The total runtime of a sequence of $L$ gates $\prod_{j=1}^L e^{-it_jH_j}$ is $\sum_{j=1}^{L} t_j$. We assume that the generator and time descriptions of all gates are efficiently uniformly generated given input $x$. When no explicit runtime $t_j$ is specified, we assume $t\in \Theta(1)$, in which case we are more interested in the depth $L$ of the circuit.

    \item (Measurement) The first mode is the designated output mode, and is measured at the end of the computation via the number operator $\hat n$.
\end{itemize}
\end{definition}
We use ``gate set'' interchangeably for a family of Hamiltonian generators and its associated unitary gates; circuit descriptions specify the generators and evolution times.
\noindent \emph{Remarks:} (1) As desired, the energy throughout a $\modelparam{G}{E}{L}$ is, by construction, bounded by $\abs{\alpha}^2\leq E$. (2) We will make heavy use of the following gate sets in this work: 
\begin{equation}\label{eq:GKerr}
    \begin{aligned}
        \GLO &\coloneq \Bigl\{e^{i\theta\hat n_j}\Bigm|\theta\in\RR\Bigr\}
        \cup\Bigl\{e^{\theta(\ad_j\a_k-\ad_k\a_j)}\Bigm|\theta\in\RR,\ j<k\Bigr\},\\
        \GKerr &\coloneq \GLO\cup\Bigl\{e^{i\theta\hat n_j^2}\Bigm|\theta\in\RR\Bigr\},\\
        \GcrossKerr &\coloneq \GLO\cup\Bigl\{e^{i\theta\hat n_j\hat n_k}\Bigm|\theta\in\RR,\ j<k\Bigr\}.
    \end{aligned}
\end{equation}
(3) The following lemma justifies our use of a single coherent state input:

\begin{lem}[Input energy redistribution]
\label{lem:coher-gen}
Let $\ket{\bm{\alpha}}$ with $\bm{\alpha}\in\mathbb C^{m_1}$ be the input coherent state over $m_1$ modes. Any coherent state $\ket{\bm\beta}$ with $\bm\beta\in\mathbb C^{m_2}$ over arbitrary $m_2$ modes can be generated using only linear optics (\cref{def:linear-optics}) with the condition that $\norm{\bm{\alpha}} =\norm{\bm{\beta}}$. Equivalently, coherent states of equal total energy are equivalent up to a linear-optical transformation. The construction is also efficient.
\end{lem}
\begin{proof}
This is a standard result in quantum optics \cite{he2008coherent}. Note that in the case where $m_1>m_2$, we end up with additional $m_1-m_2$ vacuum modes next to $\ket{\bm\beta}$, and if $m_1<m_2$, we pad $\ket{\bm\alpha}$ with extra vacuum ancillas.
\end{proof}

\subsection{The complexity class \texorpdfstring{$\BEQP{}{}$}{BEQP}}
\label{sscn:BEQ}

 Similar to the definition of the class \CVBQP\ \cite{CGMMNRS25} for the LB model, a central focus of this work is the promise class version of $\modelparam{G}{E}{\poly}$ (i.e. polynomial circuit depth), which we denote $\BEQP{}{}$:

\begin{definition}[$\BEQP{G}{E}$]
    \label{def:cc-beq}
    Let $\mathcal{G}$ be a family of energy-preserving polynomial Hamiltonians as in \cref{def:model}, and $x\in\{0,1\}^*$ an arbitrary string of length $\varsigma$. A promise problem $A=(\Ayes,\Ano) \in \BEQP{G}{E}$ if there exists a poly-time uniform family of \model\ circuits $\{C_x\}$ over $\mathcal G$ of total runtime $\le \poly(\varsigma)$, as well as efficiently computable integer bounds $0<a<b\le \poly(\varsigma)$ satisfying the following. Denote the output state by
    \be
    \ket{\psi_x} = C_x \ket{\bm{0}}\ket{\alpha},
    \ee
    with $\abs{\alpha}^2\leq E$. If $E$ is not specified, we assume $E\in O(\poly(\varsigma))$. Letting \be\text{Pr}\left[\hat{n}_1 = n\right] = \norm{(\bra{n}_1\otimes\II_{\mathrm{rest}})\ket{\psi_x}}^2\ee be the probability of measuring $n$ photons in the first mode of $\ket{\psi_x}$:
    \begin{itemize}
        \item (Completeness) If $x\in\Ayes$, then $\text{Pr}\left[\hat{n}_1 \in [a, b] \right] \ge 2/3$.
        \item (Soundness) If $x\in\Ano$, then $\text{Pr}\left[\hat{n}_1 \le a-1 \right] \ge 2/3$.
    \end{itemize}
    With polynomial overhead and repetition, we may improve the completeness and soundness parameters to exponentially close to $1$ and $0$ as usual, with the catch that energy $E$ will also need to increase polynomially for such parallel runs\footnote{Strictly speaking, since the model is required to output a single bit, we need to implement a majority vote which can be done using the compilation tools developed in \cref{sec:compilation}.}.
\end{definition}
\noindent \emph{Remarks:} (1) If a parameter, e.g. $\mathcal G$, is not specified, the default definitions of $\model$ are assumed. (2) For clarity, we write $\BEQP{}{}$ to indicate $\BEQP{}{\poly}$, i.e. with polynomial energy. (3) When we wish to additionally constrain the number of modes, such as in \Cref{thm:classical_upper_bounds_precise}, we write $\BEQP{}{}$ with $O(f(\varsigma))$ modes, where $\varsigma$ is the input size. If not specified, the number of modes is assumed to be $\poly(\varsigma)$. (4) We define the precise version, $\PrecBEQP$, as $\BEQP{}{}$ but with $c(\varsigma) - s(\varsigma) = \frac{1}{\exp}$, where $c$ and $s$ are the completeness and soundness parameters, computable to $p$ bits of precision in time $\poly(\varsigma,p)$. (5) Unlike $\CVBQP_\poly$~\cite{CJMMM26,CGMMNRS25}, in \cref{thm:cutoff-SK} we show that $\BEQP{}{}$ has universal gate sets (Kerr and linear optics $\GKerr$ or $\GcrossKerr$) obeying a Solovay--Kitaev theorem, and we therefore drop the specific gate set when assuming universality.

\section{Parametric approximation of active gates by energy-preserving gates}
\label{sec:parametric-approx}

In this section, we show that \model\ with an ancillary coherent state of high enough energy can simulate arbitrary gates. The model is defined formally in \cref{scn:model-definition}.

Let us briefly sketch the idea behind simulating arbitrary gates in \model, which is a folklore technique in the quantum optics literature known as \emph{parametric approximation}~\cite{mollow1967quantum,mollow1967quantum2}. We write any given Hamiltonian in the Wick-ordered form (\cref{def:wick-ordering}). This is the form where each term in the Hamiltonian has the form $\prod_j (\hat a_j^\dagger)^{k_j} \prod_j \hat{a}_j^{\ell_j}$. We then add additional annihilation/creation factors of the form $\hat b/\alpha$ and $\hat b^\dag/\alpha$ (for some large and real $\alpha$) to keep each monomial energy-preserving. For instance, the Hamiltonian of a displacement operator $H_{\mathrm{disp}}=\hat a+\hat a^\dag$ will be replaced by $\widetilde H=(\hat a\hat b^\dag+\hat a^\dag\hat b)/\alpha$, which in this case is the Hamiltonian of a beamsplitter. We then show that applying $\widetilde H$ on the input state with the ancilla mode $\ket\alpha$ for the same time produces the desired dynamics up to an error of $O(1/\alpha)$, where we have hidden factors related to energy for simplicity. This simulation strategy is proven in \cref{sec:parametric-approx2}.

Later, in \cref{sec:robus-parametric}, we show that this model is robust against thermal noise in the input states. In other words, we show that replacing coherent states with displaced thermal states (with large enough displacement) can still simulate arbitrary bosonic dynamics.

\subsection{Parametric approximation}\label{sec:parametric-approx2}

In this section, we show that \model\ allows us to efficiently simulate quantum computations with polynomial moment bounds on energy throughout the evolution, as specified in \cref{eq:supsE}. These can include CV, DV, and hybrid CV-DV quantum computations.

Concretely, let $\Halg \in \CC[\a_1,\a_1^\dagger,\dots,\a_m,\a_m^\dagger]$ be symmetric on $\Dfin$ (the finite-particle subspace) and not number-preserving.
Fix the Wick-ordered (\cref{def:wick-ordering}) monomial expansion
\begin{equation}
\label{eq:Halg}
\Halg = H_0 + \sum_{k}\left(M_k + M_k^\dagger\right),
\end{equation}
where $H_0$ is the energy-preserving part and each $M_k$ is a Wick monomial of $\Halg$ (including coefficient) that increases particle number by $n_k\ge1$.
Let $d \coloneq \max_k \deg(M_k)$. Let $H$ be a self-adjoint extension of $\Halg$\footnote{Note that Gaussian and $\Q^k$ Hamiltonians are already essentially self-adjoint and therefore they have a unique self-adjoint extension $H$. However, $\Halg$ does not need to be essentially self-adjoint here.} on $L^2(\RR^m)$, so that $\Dfin \subseteq \calD(H)$.
For $\alpha \ge 1$, define
\begin{equation}
\label{eq:preserving_H}
\wtH_{\mathrm{alg}} \coloneq H_0 \otimes\II + \sum_{k} \left(M_k \otimes \biggl(\frac{\b}{\alpha}\biggr)^{n_k} + M_k^\dagger \otimes \biggl(\frac{\b^\dagger}{\alpha}\biggr)^{n_k}\right),
\end{equation}
on $\Dfin$, where $\b^\dagger$ is the creation operator on an auxiliary mode initialized to the coherent state $\ket{\alpha}$. We denote its self-adjoint closure by $\wtH$, as justified below. In the following theorem, we show that $e^{-i(H\otimes\II)t}\ket{\psi,\alpha}\approx e^{-it\wtH}\ket{\psi,\alpha}$, where the approximation error decays as $O(1/\alpha)$, under mild conditions on the maximum energy acquired in the evolution of $\ket\psi$ under $H$.

\begin{lem}
$\wtH_{\mathrm{alg}}$ is essentially self-adjoint on $\Dfin$, i.e., it has a unique self-adjoint extension.
\end{lem}
\begin{proof}
Observe that $\wtH_{\mathrm{alg}}$ commutes with $\Ntot$ (total number operator on all $m+1$ modes) on $\Dfin$.
    Therefore, $\wtH_{\mathrm{alg}}$ preserves each $n$-photon number sector
    \begin{equation}
        \calH_n \coloneq \ker(\Ntot - n),\qquad n\in\NN_0.
    \end{equation}
    Define the restrictions
    \begin{equation}
        \wtH_{n} \coloneq \wtH_{\mathrm{alg}}|_{\calH_n},
    \end{equation}
    which are finite-dimensional Hermitian operators.
    Thus, we can define
    \begin{equation}
        \wtH \coloneq \bigoplus_{n=0}^\infty \wtH_n
    \end{equation}
    which is self-adjoint on the domain
    \begin{equation}
        \label{eq:calD}
        \calD(\wtH) = \left\{\ket{\phi}=\sum_{n=0}^\infty \ket{\phi_n} \ \middle|\ \ket{\phi_n}\in\calH_{n},\; \sum_{n=0}^\infty \norm{\wtH_n\ket{\phi_n}}^2<\infty \right\}.
    \end{equation}
    For every $\phi\in\calD(\wtH)$, its finite number-sector truncations converge to $\phi$, and their images under $\wtH_{\mathrm{alg}}$ converge to $\wtH\phi$. Therefore, $\wtH$ is the closure of $\wtH_{\mathrm{alg}}$.
    Since $\wtH$ is also self-adjoint, it is the unique self-adjoint extension.
\end{proof}

\begin{theorem}[Parametric approximation of polynomial Hamiltonians]\label{thm:passive-compilation}
Let $\ket{\psi}\in\calD(H)$ and define $\ket{\psi(s)} = e^{-isH}\ket{\psi}$ such that
    \begin{equation}\label{eq:supsE}
        \sup_{s\in[0,t]} \norm{(1+\hN)^{d/2}\ket{\psi(s)}} \le M.
    \end{equation}
    Identifying $H$ with $H\otimes\II$ on the system with ancilla, there exists a constant $C_H$ depending on the Wick coefficients of $\Halg$, such that
    \begin{equation}\label{eq:deltat}
        \delta(t) \coloneq \norm{\left(e^{-itH} - e^{-it\wtH}\right)\ket{\psi,\alpha}} \le \int_{0}^t \norm{(H-\wtH)\ket{\psi(s),\alpha}}\,ds
        \le \frac{C_H M t}{\alpha},
    \end{equation}
    where the term $(H-\wtH)\ket{\psi(s),\alpha}$ is understood via the continuous extension $\Delta$ constructed in the proof.
\end{theorem}
\begin{proof}
First, note that since $\deg(M_k)\le d$, there exists a constant $C_H$ such that\footnote{See \cite[Lemma A.4]{CGMMNRS25} for bounding polynomial operators in terms of the number operator.}
\begin{equation}\label{eq:Mkphi-bound}
\sum_{k}\norm{M_k^\dagger\ket{\phi}} \le C_H \norm{(1+\hN)^{d/2}\ket{\phi}} \qquad \forall\ket\phi\in\calD((1+\hN)^{d/2}).
\end{equation}
Here the monomials are extended from $\Dfin$ by continuity in the norm $\norm{(1+\hN)^{d/2}\ket{\phi}}$.
    Next, we analyze the error in the displaced frame (i.e., we apply the change of frame defined by $\hD(\alpha)^\dagger$).
    Recall that for real $\alpha$,
    $\hD(\alpha) \coloneq e^{\alpha\b^\dagger - \alpha\b}$
    satisfies $\hD(\alpha)\ket{0}=\ket{\alpha}$ and $\hD(\alpha)^\dagger \b\,\hD(\alpha) = \b+\alpha$ on the Schwartz space (see \cref{app:heisenberg}), which includes $\Dfin$.
    The same polynomial bound, applied to the full degree of $\wtH$, and closedness of $\wtH$ show that its domain contains the Schwartz space: approximate a Schwartz vector by finite-particle cutoffs in the corresponding number-weighted norm.
    Since $\hD(\alpha)$ preserves the Schwartz space, $\ket{\phi,0}\in\calD(\hD(\alpha)^\dagger \wtH\,\hD(\alpha))$, or equivalently $\ket{\phi,\alpha}\in\calD(\wtH)$, for all $\phi\in\Dfin$.
    Thus, on $\Dfin$, $\hD(\alpha)^\dagger \wtH\,\hD(\alpha)$ is again a finite polynomial in creation and annihilation operators, and
    \begin{equation}\label{eq:DwtHD}
        \hD(\alpha)^\dagger \wtH\, \hD(\alpha) = H_0 \otimes\II + \sum_{k}\left(M_k\otimes \left(1+\frac{\b}{\alpha}\right)^{n_k} + M_k^\dagger\otimes \left(1+\frac{\b^\dagger}{\alpha}\right)^{n_k} \right).
    \end{equation}
    Additionally, we have on all of $\calD(H\otimes\II)$
    \begin{equation}\label{eq:DHD}
        \hD(\alpha)^\dagger (H\otimes\II)\hD(\alpha) = H\otimes\II
    \end{equation}
    because $\hD(\alpha)$ only acts non-trivially on the last mode.
    Combining \cref{eq:DwtHD} and \cref{eq:DHD} gives for all $\ket{\phi}\in\Dfin$
    \begin{equation}
        \begin{aligned}
            \hD(\alpha)^\dagger \bigl((H\otimes\II)-\wtH\bigr)\hD(\alpha)\ket{\phi,0} &=-\sum_k M_k^\dagger\ket{\phi}\otimes \ket{v_{k,\alpha}},\\
            \ket{v_{k,\alpha}} &\coloneq \left(\left(1+\frac{\b^\dagger}{\alpha}\right)^{n_k}-1\right)\ket{0} = \sum_{j=1}^{n_k}\binom{n_k}{j}\alpha^{-j}\sqrt{j!}\ket{j},
        \end{aligned}
        \label{eq:diffrence_H}
    \end{equation}
    where the $M_k$ terms vanish as $\b\ket{0}=0$.
    Due to the assumption $\alpha \ge 1$,
    \begin{equation}
        \norm{\ket{v_{k,\alpha}}}^2 = \sum_{j=1}^{n_k} \binom{n_k}{j}^2 j! \alpha^{-2j} \le \frac{c_k^2}{\alpha^2},\qquad c_k \coloneq \sqrt{\sum_{j=1}^{n_k}\binom{n_k}{j}^2j!}.
    \end{equation}
    Since $\hD(\alpha)$ is unitary and using \cref{eq:Mkphi-bound},
    \begin{equation}\label{eq:diff:HtH}
        \norm{\bigl((H\otimes\II) - \wtH\bigr)\ket{\phi,\alpha}} \le \sum_{k} \norm{M_k^\dagger\ket{\phi}\otimes \ket{v_{k,\alpha}}} \le \frac{1}{\alpha}\sum_{k} c_k\norm{M_k^\dagger\ket{\phi}} \le \frac{C_H}{\alpha}\norm{(1+\hN)^{d/2}\ket{\phi}}
    \end{equation}
    for all $\ket{\phi}\in\Dfin$ and constant $C_H$ depending only on the coefficients of $M_k$.
    Since $\Dfin$ is a core for $(1+\hN)^{d/2}$, we can extend the map
    \begin{equation}\label{eq:Delta}
        \ket{\phi} \longmapsto \bigl((H\otimes\II) - \wtH\bigr)\ket{\phi,\alpha}
    \end{equation}
    uniquely by continuity in the norm $\norm{(1+\hN)^{d/2}\ket{\phi}}$ to a linear map
$
\Delta : \calD((1+\hN)^{d/2})\to L^2(\RR^{m+1}),
$
such that
\begin{equation}\label{eq:Deltaphi-bound}
\norm{\Delta\ket{\phi}} \le \frac{C_H}{\alpha}\norm{(1+\hN)^{d/2}\ket{\phi}}.
\end{equation}
    So far, we have introduced the machinery to justify the second inequality of \cref{eq:deltat}.
    It remains to justify the first inequality, for which we also need $\Delta$.

Since $\ket{\psi}\in\calD(H)$ and $H$ is self-adjoint, $\ket{\psi(s)}$ is continuously differentiable \cite[Lemma II.1.3(ii)]{EN2000} with
\begin{equation}
\frac{d}{ds}\ket{\psi(s)} = -iH\ket{\psi(s)},
\end{equation}
and $\Psi(s) \coloneq\ket{\psi(s),\alpha}$ has $\Psi'(s) = -i(H\otimes\II)\Psi(s)$.
Now we would like to use Duhamel's principle with the variation-of-constants formula, but the continuous differentiability of the difference $(e^{-itH}-e^{-it\wtH})\ket{\psi,\alpha}$ is a priori not clear.\footnote{Also, we would like to avoid formally introducing the Bochner integral here to integrate over the infinite-dimensional Hilbert space.}
Instead, we will argue sector-wise.
Let $\Pi_n$ denote the projector on $\calH_n$ and define
\begin{equation}
F_n(s) \coloneqq e^{is\wtH_n}\Pi_n \Psi(s) \in\calH_n.
\end{equation}
Since $\Psi$ is continuously differentiable and $\wtH_n$ is bounded on $\calH_n$, $F_n$ is continuously differentiable with
\begin{subequations}
\begin{align}
F_n'(s) &= i\wtH_n e^{is\wtH_n}\Pi_n \Psi(s) + e^{is\wtH_n}\Pi_n\Psi'(s)\label{eq:Fn':a}\\
&= i\wtH_n e^{is\wtH_n}\Pi_n \Psi(s) -i e^{is\wtH_n}\Pi_n(H\otimes\II)\Psi(s)\label{eq:Fn':b}\\
&= i e^{is\wtH_n}\left(\wtH_n\Pi_n\Psi(s) - \Pi_n(H\otimes\II)\Psi(s)\right)\label{eq:Fn':c}\\
&= -i e^{is\wtH_n}\Pi_n \Delta \ket{\psi(s)}\label{eq:Fn':d}.
\end{align}
\end{subequations}
Here \cref{eq:Fn':a} is the product rule and \cref{eq:Fn':c} uses that $\wtH_n$ commutes with its exponential.
To justify \cref{eq:Fn':d}, note that $(H\otimes\II)\Pi_n$ is bounded, since $\Pi_n$ has finite-dimensional range contained in $\calD(H\otimes\II)$.
For $\ket{\phi}\in\Dfin$, \cref{eq:Delta} and self-adjointness give
\be
\Pi_n\Delta\ket{\phi}
=\left(\bigl((H\otimes\II)\Pi_n\bigr)^\dag - \wtH_n\Pi_n\right)\ket{\phi,\alpha}.
\ee
The operators on the right are bounded. Approximating $\ket{\phi}\in\calD((1+\hN)^{d/2})$ by finite-particle cutoffs in the $(1+\hN)^{d/2}$-weighted norm therefore extends this identity to all such $\ket{\phi}$.
If also $\ket{\phi}\in\calD(H)$, then $\bigl((H\otimes\II)\Pi_n\bigr)^\dag\ket{\phi,\alpha}=\Pi_n(H\otimes\II)\ket{\phi,\alpha}$.
Applying this to $\ket{\phi}=\ket{\psi(s)}$ proves \cref{eq:Fn':d}, without requiring the cutoffs to converge in the graph norm of $H$.
So
\begin{equation}
\sum_{n=0}^\infty \norm{F_n'(s)}^2 = \sum_{n=0}^\infty \norm{\Pi_n\Delta\ket{\psi(s)}}^2 = \norm{\Delta\ket{\psi(s)}}^2 \le \frac{C_H^2M^2}{\alpha^2},
\end{equation}
where the last step uses \cref{eq:Deltaphi-bound} and \cref{eq:supsE}.
Moreover, each $\Pi_{\le K}\Delta\ket{\psi(s)}$ is continuous in $s$ by \cref{eq:Fn':d}, so $\norm{\Delta\ket{\psi(s)}}$ is measurable as the pointwise limit of $\norm{\Pi_{\le K}\Delta\ket{\psi(s)}}$.
The displayed bound consequently makes its scalar integral finite.

Let $\Pi_{\le K} = \sum_{n=0}^K \Pi_n$, and $F(s) \coloneq \bigoplus_{n=0}^\infty F_n(s)=e^{is\wtH}\Psi(s)$. Then
\begin{equation}
\begin{aligned}
\delta(t) &= \norm{\Psi(t) - e^{-it\wtH}\Psi(0)} = \norm{e^{-it\wtH}F(t) - e^{-i t\wtH}F(0)} =
\norm{F(t) - F(0)}\\
&= \lim_{K\to\infty} \norm{\Pi_{\le K}\bigl(F(t) - F(0)\bigr)}
\end{aligned}
    \end{equation}
    as $\Pi_{\le K}\to\II$ strongly.
    Then, by the fundamental theorem of calculus,
    \begin{equation}
        \begin{aligned}
            \Pi_{\le K} \bigl(F(t) - F(0)\bigr) &= \sum_{n=0}^K \bigl(F_n(t) - F_n(0)\bigr) = -i\sum_{n=0}^K \int_{0}^t e^{is\wtH_n}\Pi_n\Delta\ket{\psi(s)}\,ds \\
            &= -i\int_{0}^t\sum_{n=0}^K e^{is\wtH_n}\Pi_n \Delta\ket{\psi(s)}\,ds = -i\int_{0}^t \Pi_{\le K} e^{is\wtH} \Delta\ket{\psi(s)}\,ds.
        \end{aligned}
    \end{equation}
    All vector integrals here are finite-dimensional.
    Taking norms and then the limit gives
    \begin{equation}
        \delta(t) = \lim_{K\to\infty}\norm{\int_{0}^t \Pi_{\le K} e^{is\wtH} \Delta\ket{\psi(s)}\,ds} \le \int_{0}^t\norm{\Delta \ket{\psi(s)}}\,ds \le \frac{C_H M t}{\alpha},
    \end{equation}
    which is \cref{eq:deltat}.
\end{proof}

Explicit examples for common active gates can be found in \cref{app:examples}.

\subsection{Robust parametric approximation}\label{sec:robus-parametric}

The simulation scheme above treats the ancillary coherent state as an energy resource that mediates effective active dynamics through an energy-preserving Hamiltonian.
For physical implementations, however, this resource state will not be prepared perfectly and may contain thermal or other preparation noise. We therefore show that the construction is stable under such imperfections: replacing the ideal coherent resource by a noisy displaced state only perturbs the simulated evolution by a controlled error. The overall message of this section is that perturbing the input or changing the choice of inputs does not affect the computational model. Informally, we prove that the coherent ancilla states can be replaced by any low-energy states that are displaced. The remark below provides some intuition on this.

\begin{remark}[Displaced Fock states as approximate resource states]
    We define the $\alpha$-displaced $n$-photon state
    \begin{equation}
        \ket{\alpha,n}:=\hat D(\alpha)|n\rangle,
    \end{equation}
    where $\hat D(\alpha)$ is the displacement operator on the resource mode and
    take $\alpha>0$.
    So, we have
    \begin{equation}
        \hat b \ket{\alpha,n} =\hat{b} \hat{D}(\alpha)\ket{n}=\hat{D}(\alpha)(\hat b+\alpha)\ket n=
        \alpha \ket{\alpha,n}+\sqrt n \ket{\alpha,n-1},
    \end{equation}
    \begin{equation}
        \hat b^\dagger \ket{\alpha,n}=\hat{b}^\dagger \hat{D}(\alpha) \ket{n} =\hat D(\alpha)(\hat b^\dagger +\alpha)\ket n=
        \alpha \ket{\alpha,n}+\sqrt{n+1}\ket{\alpha,n+1},
    \end{equation}
    where we used $\II=\hat{D}(\alpha) \hat{D}^{\dagger} (\alpha)$ and $\hat{D}^\dagger (\alpha)\hat{b} \hat{D}(\alpha)=\hat{b}+\alpha$. Thus, displaced Fock states with $n\ll |\alpha|^2$ are approximate resource states. It means that they are not exactly the resource coherent state, but they can be used in its place with a small error when the added Fock excitation $n$ is much smaller than the coherent state energy.
\end{remark}

Let us denote the computational system by $S$ and the resource ancilla by $A$.
The system $S$ consists of the original $m$ bosonic modes, while $A$ is the
single bosonic mode used in the construction of $\wtH$.
We also allow an optional reference system $R$: this represents any external
system initially entangled with $S$, and neither Hamiltonian acts on it.
Tensor factors below are identified by their subsystem labels and, where
necessary, canonically reordered into $S\otimes A\otimes R$.
We use $\norm{\cdot}_1$ for the trace norm. Consider the same setup as introduced in \cref{eq:Halg} and \cref{eq:preserving_H}.
Our goal is to show that one can replace the coherent state $\ket\alpha$ with the displaced state $\sigma_\alpha=D(\alpha) \sigma D(\alpha)^\dag$, if $\alpha$ is large enough and $\sigma$ has low energy. In other words, for any (possibly entangled) state $\rho_{SR}$, we have
$(e^{-it \wtH}\otimes\II)(\rho_{SR}\otimes \sigma_\alpha)(e^{it\wtH}\otimes\II) \approx (e^{-itH}\rho_{SR}e^{iHt})\otimes \sigma_\alpha$. We will use $\rho_{SR}(t)=(e^{-itH}\otimes\II_R)\rho_{SR}(e^{itH}\otimes\II_R)$ for the evolution that we want to simulate, and $\widetilde\rho_{SAR}(t)
=((e^{-it\wtH})_{SA}\otimes\II_R)
(\rho_{SR}\otimes\sigma_\alpha)((e^{it\wtH})_{SA}\otimes\II_R)$ for the simulated evolution. The ideal joint output is $\rho_{SAR}^{\mathrm{id}}(t)=\rho_{SR}(t)\otimes\sigma_\alpha$, with the tensor factors reordered as above.

\begin{theorem}[Robustness to low-energy resource noise]
\label{thm:robust-parametric-approximation}
Assume $\bar n=\operatorname{Tr}(\sigma\N_A)<\infty$ and
\begin{equation}\label{eq:robust-system-moment}
\sup_{s\in[0,t]}
\left(\operatorname{Tr}\bigl[\rho_S(s)(1+\hN_S)^d\bigr]\right)^{1/2}
\le M.
\end{equation}
Then there is a constant $C_H$, depending only on the fixed active monomials
of $\Halg$, such that
\begin{equation}\label{eq:finite-mean-resource-bound}
\norm{\widetilde\rho_{SAR}(t)-\rho_{SAR}^{\mathrm{id}}(t)}_1
\le C_H(1+tM)\frac{\sqrt{\bar n+1}}{\alpha}.
\end{equation}
\end{theorem}

\begin{corollary}[Thermal resource noise]
\label{cor:thermal-resource-noise}
Under the assumptions of \cref{thm:robust-parametric-approximation}, suppose
$\sigma$ is a thermal state of mean photon number $\bar n$ and
$\alpha\ge\sqrt{\bar n+1}$.
Then the bound improves to
\begin{equation}\label{eq:thermal-resource-bound}
\norm{\widetilde\rho_{SAR}(t)-\rho_{SAR}^{\mathrm{id}}(t)}_1
    \le C_HtM\frac{\sqrt{\bar n+1}}{\alpha},
\end{equation}
where $C_H$ again depends only on the fixed active monomials of $\Halg$.
\end{corollary}

\begin{proof}[Proof of \cref{thm:robust-parametric-approximation}]
Write $W=(1+\hN_S)^{d/2}$.
Throughout the proof, constants denoted by $C_H$ may be enlarged and depend
only on the fixed active monomials. Similar to the proof of \cref{thm:passive-compilation},
we have
\begin{equation}\label{eq:robust-monomial-bound}
\sum_k\left(\norm{M_k\ket\phi}+\norm{M_k^\dagger\ket\phi}\right)
\le C_H\norm{W\ket\phi},\qquad \ket\phi\in\calD(W).
\end{equation}
The monomials here are extended from $\Dfin$ by continuity in the
$W$-weighted norm. We can purify $\rho_{SR}$ to $\ket\psi$ on $SRF_S$ and $\sigma$ to $\ket\xi$
on $AF_A$.
Set $\ket{\psi(s)}=(e^{-isH}\otimes\II_{RF_S})\ket\psi$.
Then \cref{eq:robust-system-moment} gives
$\norm{(W\otimes\II_{RF_S})\ket{\psi(s)}}\le M$.
Now, for an integer $K\ge0$, define
\begin{equation}\label{eq:resource-cutoff-tail}
    Q_K=\sum_{n=0}^{K}\ket{n}\bra{n}_A,\qquad
    p_K=\operatorname{Tr}\bigl[\sigma(\II_A-Q_K)\bigr],\qquad
    \ket{\xi_K}=(Q_K\otimes\II_{F_A})\ket\xi.
\end{equation}
Thus $p_K$ is the photon-number tail probability of the undisplaced resource,
and $\norm{\ket\xi-\ket{\xi_K}}=\sqrt{p_K}$. Intuitively, the truncation is done as a bound on the first moment does not necessarily bound higher moments (this will become clearer in the rest of the proof).
Note that we do not normalize $\ket{\xi_K}$.
For $j\ge1$, the squared norms of $\b^j\ket{\xi_K}$ and
$(\b^\dagger)^j\ket{\xi_K}$ are, respectively, the truncated falling and
rising factorial moments of $\N_A$.
On $0\le n\le K$, both factorials are at most
\begin{equation}
    (n+1)\cdots(n+j)
    \le j(n+1)(K+j)^{j-1},
\end{equation}
where the falling factorial is understood to be $0$ for $n<j$.
Therefore,
\begin{equation}\label{eq:resource-truncated-ladder-bound}
    \max\left\{
        \norm{\b^j\ket{\xi_K}},
        \norm{(\b^\dagger)^j\ket{\xi_K}}
    \right\}
    \le\sqrt{j(\bar n+1)}(K+j)^{(j-1)/2}.
\end{equation}
Identities on the purifying factors are suppressed here and below. Similar to the proof of \cref{thm:passive-compilation}, we go to the displaced frame, and write the polynomial difference on the Schwartz space
\begin{equation}\label{eq:resource-polynomial-residual}
\begin{aligned}
\hD(\alpha)^\dagger(\wtH-H\otimes\II_A)\hD(\alpha)
=\sum_k\sum_{j=1}^{n_k}\binom{n_k}{j}\alpha^{-j}
\left(M_k\otimes\b^j+M_k^\dagger\otimes(\b^\dagger)^j\right).
\end{aligned}
\end{equation}
The polynomial bounds extend this expression continuously to
$\calD(W\otimes(1+\N_A)^{d/2})$ in the corresponding weighted norm.
Since $\ket{\xi_K}$ has bounded resource photon number, this extension
acts on $\ket{\psi(s)}\otimes\ket{\xi_K}$.
Choose
\begin{equation}
    K=\lfloor\alpha^2\rfloor.
\end{equation}
Then $K+j\le(j+1)\alpha^2$, and
\cref{eq:resource-truncated-ladder-bound} implies
\begin{equation}
    \alpha^{-j}
    \max\left\{
        \norm{\b^j\ket{\xi_K}},
        \norm{(\b^\dagger)^j\ket{\xi_K}}
    \right\}
    \le\sqrt{j}(j+1)^{(j-1)/2}\frac{\sqrt{\bar n+1}}{\alpha}.
\end{equation}
Combining this with \cref{eq:robust-monomial-bound}, and absorbing the
finitely many binomial and degree-dependent factors into $C_H$, gives
\begin{equation}\label{eq:resource-residual-norm}
\norm{(\wtH-H\otimes\II_A)\hD(\alpha)
\bigl(\ket{\psi(s)}\otimes\ket{\xi_K}\bigr)}\le C_HM\frac{\sqrt{\bar n+1}}{\alpha}.
\end{equation}

We highlight that the evolution comparison can be justified by testing against vectors
with finite photon-number support in $SA$, with arbitrary reference
factors. We suppress identities on reference systems. For any vector $v$
with finite photon-number support in $SA$, consider
\begin{align}
    \left\langle v,\,
       e^{is\wtH}\hD(\alpha)
       (\ket{\psi(s)}\otimes\ket{\xi_K})
    \right\rangle.
\end{align}
Since $\wtH$ preserves total photon number, $e^{-is\wtH}v$
has finite photon-number support and belongs to the domains of
both $\wtH$ and $H\otimes\II_A$. We can therefore differentiate
this scalar overlap by letting the Hamiltonians act on the test
vector. The polynomial identity
\cref{eq:resource-polynomial-residual} extends to these pairings
by approximation in the weighted norm used above.
Integrating, applying the residual bound, and taking the supremum
over unit test vectors gives
\begin{equation}\label{eq:resource-vector-comparison}
    \norm{
       \bigl(e^{-it\wtH}-e^{-itH}\otimes\II_A\bigr)
       \hD(\alpha)(\ket\psi\otimes\ket{\xi_K})
    }
    \le C_HtM\frac{\sqrt{\bar n+1}}{\alpha}.
\end{equation}
Here we used the density of finite photon-number test vectors.
\end{proof}

\begin{proof}[Proof of \cref{cor:thermal-resource-noise}]
For a purification $\ket\xi$ of a thermal state of mean photon number
$\bar n$, its factorial moments give
\begin{equation}
    \norm{\b^j\ket\xi}^2=j!\bar n^j,
    \qquad
    \norm{(\b^\dagger)^j\ket\xi}^2=j!(\bar n+1)^j.
\end{equation}
In particular, all resource moments are finite, so the preceding domain
and finite-sector comparison argument applies directly to $\ket\xi$,
without a resource cutoff.
Since $\varepsilon=\sqrt{\bar n+1}/\alpha\le1$, for every $j\ge1$,
\begin{equation}
    \alpha^{-j}\max\left\{
        \norm{\b^j\ket\xi},\norm{(\b^\dagger)^j\ket\xi}
    \right\}
       \le\sqrt{j!}\,\varepsilon^j
       \le\sqrt{j!}\,\varepsilon.
\end{equation}
\cref{eq:robust-monomial-bound} and
\cref{eq:resource-polynomial-residual} therefore bound the residual norm
by $C_HM\varepsilon$.
Integrating over $[0,t]$, converting vector distance to trace norm distance,
and discarding the purifying systems yields
\cref{eq:thermal-resource-bound}.
\end{proof}

\begin{remark}[Why the time dependences differ]
Finite mean photon number alone does not imply an error linear in $t$.
For example, with $H=\hat a^{\dagger2}+\hat a^2$ and vacuum system input, the pure
resource vector proportional to $\sum_{n\ge1}n^{-3/2}\ket n$ has finite
mean photon number, but the joint-output error is not $O(t)$ near zero at
fixed $\alpha$.
Thus the result of \cref{cor:thermal-resource-noise} (linear dependence of error on time) cannot hold for all
finite-mean resource states.
\end{remark}

\subsection{Simulating existing models within \model}
\label{sec:sim_existing_models}

In this section, we show that the model of computation we introduce (\model) can reproduce a variety of quantum computational models, either DV or CV. In particular, we show that these various models can be simulated with polynomial overheads in \model, if we impose a polynomial energy-moment bound, i.e., assuming $M=\poly(\varsigma)$ ($\varsigma$ the input size) in \eqref{eq:supsE}, for all gates we want to simulate. However, we highlight that alternatively, \cref{thm:compile_new} provides
a compilation route under a polynomial mean-energy bound (provided
suitable finite-cutoff approximations of the gates have efficiently
constructible qubit circuit descriptions).

Firstly, as discussed in \cref{sec:complexity}, \model\ with polynomial energy and constant-degree gates is equivalent to $\mathsf{BQP}$. Going to the CV models, we have the cubic+Gaussian computation introduced by Lloyd and Braunstein \cite{lloyd_quantum_1999} and formalized in \cite{CJMMM26}. We highlight that both Gaussians and cubic gates can be simulated within \model, using \cref{thm:passive-compilation}. In particular, some examples of Gaussian gates are discussed in \cref{app:examples}, and the cubic gate can be simulated by writing
\begin{align}
\hat q^3 = \frac{1}{2\sqrt2}\left[ (\hat a^\dag)^3 + 3(\hat a^\dag)^2 \hat a + 3\hat a^\dag \hat a^2 + \hat a^3 + 3(\hat a^\dag + \hat a)\right].
\end{align}
We then replace this Hamiltonian (which generates the cubic unitary $V(t) = \exp(it\hat q^3)$) by
\begin{align}
\widetilde{H}_{\mathrm{cubic}} = \frac{1}{2\sqrt 2}\left[(\hat a^\dag)^3 \hat b^3/\alpha^3 + 3(\hat a^\dag)^2 \hat a \hat b/\alpha + 3\hat a^\dag \hat a^2 \hat b^\dag/\alpha + \hat a^3(\hat b^\dag)^3/\alpha^3 + 3(\hat a^\dag \hat b+ \hat a \hat b^\dag)/\alpha\right].
\end{align}
According to \cref{thm:passive-compilation}, it suffices to choose $\alpha$ large enough, so that $\exp(it\widetilde H)$ acting on $\ket{\psi,\alpha}$ simulates the action of $V(t)$ on $\ket\psi$.

We also have hybrid CV-DV models, studied in \cite{brenner2025factoring} and \cite{liu2026hybrid}. Table III in \cite{liu2026hybrid} lists a comprehensive set of CV gates, which can all be simulated via \cref{thm:passive-compilation} as well. For the DV gates, we can simply employ dual-rail encoding, encoding qubit states $\ket0, \ket1$ into two-mode single-photon states $\ket{10},\ket{01}$. Linear optical gates and self-Kerr gates can realize arbitrary DV gates in this model. To see this, note that we can encode Paulis as $\bar Z = \hat c^\dag \hat c - \hat d^\dag \hat d$, $\bar X = \hat c^\dag \hat d + \hat d^\dag \hat c$, and $\bar Y = -i\hat c^\dag \hat d + i\hat c \hat d^\dag$. For CV-DV operations, we can use this DV encoding together with \cref{thm:passive-compilation}. Table IV in \cite{liu2026hybrid} lists a set of CV-DV interactions. We simulate these gates by using dual-rail encoding to encode the qubit. For instance, take the conditional rotation gate
\begin{align}\label{eq:CR}
\mathrm{CR}(\theta) = \exp(-i\frac{\theta}2 \hat \sigma_z \hat a^\dag \hat a).
\end{align}
Using the dual-rail encoding with modes that have annihilation operators $\hat c,\hat d$, we get that
\begin{align}
\exp(-i\frac{\theta}2 (\hat c^\dag \hat c-\hat d^\dag \hat d) \hat a^\dag \hat a),
\end{align}
exactly simulates the conditional rotation gate in \cref{eq:CR}, and is an energy-preserving operation. Similarly, the Jaynes--Cummings unitary (defined as $\mathrm{JC}(\theta,\phi)=\exp(-i\theta(e^{i\phi}\hat \sigma_- \hat a^\dag + e^{-i\phi}\hat \sigma_+\hat a))$) can be simulated in \model\ via parametric approximation. For a slightly less trivial example, we can consider the conditional displacement gate
\begin{align}\label{eq:cd}
\mathrm{CD}(\beta) = \exp(\hat\sigma_z(\beta \hat a^\dag - \beta^\ast \hat a)).
\end{align}
Using $\hat c,\hat d$ for modes encoding the qubit, we get that $\widetilde{H} = i(\hat c^\dag \hat c-\hat d^\dag \hat d)(\beta \hat a^\dag \hat b-\beta^\ast \hat a \hat b^\dag)/\alpha$ can achieve dynamics close enough to \cref{eq:cd} in unit time, by choosing $\alpha$ large enough. Overall, the gates generated by polynomial Hamiltonians are covered by \cref{thm:passive-compilation} under the assumptions stated there, and we can handle CV-DV interactions via dual-rail encodings. More general gates can instead be
treated using \cref{thm:compile_new} when their finite-cutoff versions
admit efficient qubit circuit descriptions. Our explanation here also trivially covers the computational model introduced in \cite{knill2001scheme}.

Fusion-based quantum computation (FBQC) \cite{bartolucci2023fusion} can also be simulated within \model. Its constant-sized resource states can be prepared using the dual-rail gates described above and the approximate single-photon preparation of \cref{prop:prepare-Fock-states}. Fusion measurements are implemented by linear optics followed by photon counting. To remove the intermediate measurements, we retain the modes that would have been detected and use their photon numbers to control subsequent operations (i.e., we use the principle of deferred measurement). For concreteness, consider an ordinary fusion of two dual-rail qubits. Since its input contains exactly two photons and its interferometer preserves photon number, each detector mode has occupation in $\{0,1,2\}$. As a simple example, suppose that detecting exactly one photon in a mode $d$ should trigger the phase shift $U=e^{-i\theta\hat n_b}$ on another mode $b$, while detecting zero or two photons should trigger no operation. This conditional operation is implemented, without any measurement, by evolving for unit time under
\begin{align}
\widetilde H=\theta\,\hat n_d(2-\hat n_d)\hat n_b.
\end{align}
Indeed, $n(2-n)$ equals one for $n=1$ and zero for $n=0,2$. Thus,
\begin{align}
e^{-i\widetilde H}
\left(
|0\rangle_d|\psi_0\rangle
+|1\rangle_d|\psi_1\rangle
+|2\rangle_d|\psi_2\rangle
\right)
=
|0\rangle_d|\psi_0\rangle
+|1\rangle_dU|\psi_1\rangle
+|2\rangle_d|\psi_2\rangle.
\end{align}
Here the unnormalized states \(|\psi_n\rangle\) describe the remaining modes, including \(b\). Leaving \(d\) untouched thereafter and tracing it out gives
$
\rho_{\mathrm{out}}
=
|\psi_0\rangle\langle\psi_0|
+U|\psi_1\rangle\langle\psi_1|U^\dagger
+|\psi_2\rangle\langle\psi_2|
$, which exactly reproduces photon counting followed by the prescribed feed-forward. The Hamiltonian $\widetilde H$ is energy-preserving and has degree six, so it is allowed in \model. More generally, the polynomials
\begin{align}
P_0(x)=\frac{(x-1)(x-2)}{2},\qquad
P_1(x)=x(2-x),\qquad
P_2(x)=\frac{x(x-1)}{2}
\end{align}
select the three possible occupations and therefore allow arbitrary outcome-dependent energy-preserving gates. Applying this construction to every detector, retaining the outcome modes, and implementing polynomial-time classical processing reversibly on dual-rail workspace removes all intermediate measurements with polynomial overhead. The decoded output bit can then be read by \model’s single final number measurement.

Later in \cref{sec:gkp}, we demonstrate how one can prepare finite-energy approximate GKP states of high quality with a polynomial-size circuit and a coherent state ancilla of polynomial energy in \model.

Finally, we note that we can reproduce fermionic computations~\cite{bravyi2002fermionic} by simulating qubit computations instead via the Jordan--Wigner mapping~\cite{jordan1928paulische}.

\section{Circuit compilation and state synthesis}

Having introduced \model\ and given evidence of its relevance using parametric approximation, we now turn to its algorithmic properties. In \cref{sec:cutoff-sk}, we show that \model\ has universal gate sets, which satisfy an energy-preserving Solovay--Kitaev theorem. \cref{sec:compilation,sec:synthesis} are devoted to applications of this result for circuit compilation and state synthesis, respectively.

\subsection{Efficient Solovay--Kitaev compilation with linear optics and Kerr gates}
\label{sec:cutoff-sk}

This section proves a Solovay--Kitaev theorem for energy-preserving unitaries.
Specifically, we can approximate any energy-preserving $U$ on $m=O(1)$ modes with error $\epsilon$ (according to the cutoff operator seminorm $\norm{\bullet}_{(K)}$) using $\poly(K,\log(\frac 1 \epsilon))$ Kerr and linear optical gates. We extend this result to unitaries that are not energy-preserving in the next section (see \cref{thm:compile_new}).

\subsubsection{Setup and main theorem}

We fix a unitary $U$ on $m\ge1$ modes specified up to total-photon-number cutoff $K\ge 1$ as in \cref{def:cutoff-norm}.
$U$ effectively acts on the Hilbert space
\be
    \calH_{\leq K}^{(m)}
    \coloneqq
    \Span\Biggl\{
        \ket{n_1,\ldots,n_m}\Biggm|
        n_j\in\NN_0,\;
        \sum_{j=1}^m n_j\leq K
    \Biggr\}.
\ee
For classical compilation, assume a classical description of $U|_{\calH_{\leq K}^{(m)}}$ in the Fock basis is supplied to sufficient precision.

Our construction works with either self-Kerr or cross-Kerr gates together with linear optics.
For convenience, we restate the following sets of elementary unitary gates, where $j,k$ range over all available modes, including any ancillas:
\begin{equation}
    \begin{aligned}
        \GLO &\coloneq \Bigl\{e^{i\theta\hat n_j}\Bigm|\theta\in\RR\Bigr\}
        \cup\Bigl\{e^{\theta(\ad_j\a_k-\ad_k\a_j)}\Bigm|\theta\in\RR,\ j<k\Bigr\}\\
        \GKerr &\coloneq \GLO\cup\Bigl\{e^{i\theta\hat n_j^2}\Bigm|\theta\in\RR\Bigr\},\\
        \GcrossKerr &\coloneq \GLO\cup\Bigl\{e^{i\theta\hat n_j\hat n_k}\Bigm|\theta\in\RR,\ j<k\Bigr\}.
    \end{aligned}
\end{equation}
Although our model allows continuously parametrized gates, a finite gate set also suffices:
Appendix~\ref{app:finite-alphabet} shows that each of $\GKerr$ and $\GcrossKerr$
contains a finite subset, independent of cutoff and precision, that retains
the Solovay--Kitaev guarantees of \cref{thm:cutoff-SK}.

Each element of these sets counts as one gate. We distinguish a gate $e^{i\theta H}$ from its generating Hamiltonian $H$.
The construction requires one extra ancilla initialized to the vacuum state.
The gate set $\GcrossKerr$ also assumes $m\ge2$ so that at least $3$ modes are available in total.
We will work in the Hilbert space $\calH^{(m+1)}_{\le K}$ which has dimension
\be
    D\coloneqq\dim\calH_{\leq K}^{(m+1)}=\binom{m+1+K}{K}.
    \label{eq:D}
\ee

\begin{theorem}[Energy-preserving Solovay--Kitaev theorem]\label{thm:cutoff-SK}
Every energy-preserving unitary $U$ can be approximated by a circuit
$C_U$ satisfying
\be
    \max_{\substack{\ket{\psi}\in \calH_{\le K}^{(m)}\\\norm{\ket\psi}=1}}\norm{C_U\ket{\psi}\ket{0} -e^{-i\gamma}(U\ket{\psi})\otimes \ket{0}}
    \leq\epsilon
\ee
for $U\ket{0^m} = e^{i\gamma}\ket{0^m}$, using
\be
    O\!\left(
        m^6K^7D^{14}
        \log^{c_{\rm SK}}\frac{1}{\epsilon}
    \right)
\ee
gates in $\GKerr$ or\/ $\GcrossKerr$, where $0<\epsilon<1/2$ and
$c_{\rm SK}=\log 5/\log(3/2)<4$.
$C_U$ can be constructed deterministically in $\poly(m,K,D,\log(1/\epsilon))$
bit operations.
\end{theorem}

\subsubsection{On the ancilla}\label{sec:ancilla}

\cref{thm:cutoff-SK} is stated with a vacuum ancilla, which is rather inconvenient for a Solovay--Kitaev theorem.
The general way to achieve polylogarithmic dependency on error is to first build a \emph{coarse compiler} that realizes unitaries up to a fixed precision $\epsilon_0$ and then recursively refine the error.
However, a coarse version of \cref{thm:cutoff-SK} does not directly permit recursive refinement since the input could become slightly entangled with the ancilla and so there would no longer be a unitary correction term on the data modes.

Unfortunately, ancillas cannot be avoided if we want to compile arbitrary energy-preserving unitaries acting on several photon-number sectors.
To see this formally, we decompose the unitaries into their $n$-photon sectors, writing
\begin{equation}
    U = \bigoplus_{n=0}^K U^{(n)},
\end{equation}
where $U^{(n)}$ acts on the $n$-photon subspace, and $U^{(0)} = e^{i\gamma}$.
Suppose $C_U = G_L\dotsm G_1$.
Then we have
\begin{equation}\label{eq:CUN}
\det C_U^{(n)} = \prod_{j=1}^L \det G_j^{(n)} = \exp(i\,D_n (\theta_1 n + \theta_2 n^2)),
\end{equation}
where $\theta_1,\theta_2$ are fixed linear functions of the phase shift and Kerr gate parameters (see \cref{eq:GKerr}), and $D_n$ is the dimension of the $n$-photon subspace.
Beamsplitters have determinant $1$ and can therefore be ignored here. Therefore the pattern of determinants on each block of \cref{eq:CUN} only depends on two real parameters $\theta_1, \theta_2$. For general patterns we cannot produce the correct global phase value on each block and therefore inclusion of an additional ancilla (to be traced out) is necessary.

For non-negative integers $d_1,\ldots,d_K$, write
\be
    \mathrm{SU}(d_1,\ldots,d_K)
    \coloneq \left\{\bigoplus_{j=1}^K U_j : U_j\in\mathrm{SU}(d_j)\right\},
\ee
with zero-dimensional blocks omitted.
It turns out that we \emph{can} avoid ancillas for $U\in\mathrm{SU}(D_0,\ldots,D_K)$, where $D_n=\dim\calH_n^{(m)}$.
Equivalently, $\det U^{(n)}=1$ on every photon-number sector, including $U^{(0)}=1$ on the vacuum.
As long as $\theta_1=\theta_2=0$, \cref{eq:CUN} causes no contradiction.
The following theorem implies \cref{thm:cutoff-SK} almost directly and permits recursive refinements because \cref{eq:CUU} bounds the operator norm of the entire circuit.

\begin{theorem}\label{thm:cutoff-SK2}
    Let $U\in\mathrm{SU}(D_0,\ldots,D_K)$ on $\calH^{(m)}_{\le K}$, where $D_n=\dim\calH_n^{(m)}$.
    Then $U$ can be approximated by a circuit $C_U$ satisfying
    \begin{equation}\label{eq:CUU}
        \norm{C_U - U}_{(K)} \le \epsilon
    \end{equation}
    using gates in $\GKerr$ (assuming $m\ge 2$) or $\GcrossKerr$ (assuming $m\ge 3$)
    with complexity parameters as in \cref{thm:cutoff-SK}, taking $D=\sum_{n=0}^K D_n$.
\end{theorem}

\begin{proof}
    Choose $\varepsilon_0=\min\{10^{-5},1/D\}$.
    The coarse compiler of \cref{lem:coarse-compiler} has gate count
    \begin{equation}
        L(\varepsilon_0)
        =O(m^6K^7D^{12}\varepsilon_0^{-2})
        =O(m^6K^7D^{14}).
    \end{equation}
    Applying \cref{thm:block-sk} gives
    $O(m^6K^7D^{14}\log^{c_{\rm SK}}(1/\epsilon))$ gates.
    Together with the polynomial-time coarse compiler, this gives
    $\poly(m,K,D,\log(1/\epsilon))$ bit operations.
\end{proof}

\begin{proof}[Proof of \cref{thm:cutoff-SK}]
    Assume without loss of generality $\gamma=0$, i.e., $U\ket{0^m}=\ket{0^m}$.
    For $n\ge1$, decompose the $n$-photon subspace as
    \begin{equation}
    \mathcal H_n^{(m+1)}
    =
    \left(\mathcal H_n^{(m)}\otimes\operatorname{span}\{|0\rangle\}\right)
    \oplus
    \operatorname{span}\{|0^m\rangle\otimes|n\rangle\}
    \oplus
    \bigoplus_{r=1}^{n-1}
    \left(\mathcal H_{n-r}^{(m)}\otimes\operatorname{span}\{|r\rangle\}\right),
    \end{equation}
    where the first sector has vacuum ancilla, the second sector has all $n$ photons in the ancilla, and the third sector contains the remaining states.
    Let
    \be
   \widetilde U
=
1\oplus\bigoplus_{n=1}^{K}
\left(
U^{(n)}
\,\oplus\,
\bigl(\det U^{(n)}\bigr)^{-1}
\,\oplus\,
\II
\right),
    \ee
    where the first summand is the vacuum sector.
    Then $\det \widetilde{U}^{(n)}=1$ for all $n$ and $\widetilde U\ket{\psi,0} =(U\ket{\psi})\ket0$ for all $\ket{\psi}$.
    Conclude the proof by applying \cref{thm:cutoff-SK2} to $\widetilde{U}$.
\end{proof}

\subsubsection{Proof outline}
The proof of \cref{thm:cutoff-SK} proceeds in the following steps:

\begin{enumerate}[label=(\arabic*)]
    \item Decompose $U\in\mathrm{SU}(D_0,\dots,D_K)$ into a sequence of $2$-level unitaries $U=U_L\dots U_1$ with $L=O(D^2)$, so that each $U_i$ acts non-trivially only on one adjacent pair of Fock states. Two Fock states are \emph{adjacent} if one can be obtained from the other by moving a single photon from one mode to another.
    For this step, we first order the Fock states within each photon-number sector so that consecutive states are adjacent \cite{klingsberg1982gray}, and then use the Givens QR decomposition in \cref{lem:adjacent-decomposition}.
    \item Obtain independent $\mathrm{SU}(2)$ control on adjacent Fock states $\ket{\bm u},\ket{\bm v}$.
    By \cref{lem:2level-rotations}, we can approximate any $\mathrm{SU}(2)$ unitary on $\Span\{\ket{\bm u},\ket{\bm v}\}$ while leaving the other Fock states approximately invariant.
    The proof uses an Euler decomposition into $X$ and $Y$ rotations and Fourier filtering to isolate the desired transition from a beamsplitter.
    \item Combine steps (1) and (2) to build a \emph{coarse compiler} for energy-preserving unitaries with determinant one on each sector (\cref{lem:coarse-compiler}). It achieves a sufficiently small precision $\epsilon_0$ and serves as the base case in the recursive Solovay--Kitaev algorithm.
    \item Apply the recursive Solovay--Kitaev refinement of \cref{thm:block-sk} to the coarse compiler to obtain \cref{thm:cutoff-SK2}.
    Here, it is essential that the coarse compiler works on the entire cutoff space without creating entanglement with the environment.
    \cref{thm:cutoff-SK} then follows via the vacuum ancilla as discussed in \cref{sec:ancilla}.
\end{enumerate}

\subsubsection{Rotations on adjacent Fock states}

We begin with step (2) of the proof outline, since that is where we actually deal with the continuous-variable gates.
In particular, the below lemma approximately implements an arbitrary $\mathrm{SU}(2)$ unitary on an \emph{adjacent} pair of Fock states $\ket{\bm u}$ and $\ket{\bm v}$ while leaving the remaining inputs approximately invariant.
We say $\ket{\bm u}$ and $\ket{\bm v}$ are adjacent if they only differ in the position of a single photon.

\begin{lem}\label{lem:2level-rotations}
    Let $\ket{\bm u},\ket{\bm v} \in \calH_{\le K}^{(m)}$ be adjacent Fock states.
    Let $U$ act as an arbitrary $\mathrm{SU}(2)$ unitary on $\Span\{\ket{\bm u},\ket{\bm v}\}$ and as the identity on its orthogonal complement.
    For $0<\eta<1$, there is a circuit $C_U$ satisfying
    \begin{equation}
        \norm{C_U-U}_{(K)} \le \eta,\qquad \det C_U^{(n)} = 1 \ \forall n\in\{0,\dots, K\},
    \end{equation}
    using $O(m^6K^7 D^6\eta^{-2})$ gates from $\GKerr$ (for $m\ge2$) or $\GcrossKerr$ (for $m\ge3$).
    The circuit can be constructed in time $\poly(m,K,D,\eta^{-1})$.
\end{lem}

An Euler decomposition reduces \cref{lem:2level-rotations} to implementing the rotations
\be
    X_{\bm u,\bm v} = \ketbra{\bm u}{\bm v}+\ketbra{\bm v}{\bm u},\qquad Y_{\bm u,\bm v} = -i(\ketbra{\bm u}{\bm v}-\ketbra{\bm v}{\bm u}).
\ee
For that, we isolate the $\bm u$-$\bm v$ transition from a beamsplitter hopping term.
Suppose $\bm u = \bm v+\bm e_j-\bm e_k$, where $\bm e_j,\bm e_k$ are binary vectors with a single $1$ in modes $j,k$, respectively.
Then the hopping operator
\begin{equation}\label{eq:hop-cutoff}
    A \coloneq \ad_j\a_k\Big|_{\calH_{\le K}^{(m)}} = \sum_{\subalign{w_k&>0\\|\bm w|&\le K}}\sqrt{(w_j+1)w_k}\ketbra{\bm w+\bm e_j-\bm e_k}{\bm w}
\end{equation}
does couple $\ket{\bm u}$ and $\ket{\bm v}$, but it still acts non-trivially on other Fock states.

Damas et al.~\cite{damas2026engineered} also use Kerr nonlinearities to selectively couple Fock states, analyzing the suppression of unwanted transitions through a rotating-frame expansion.
Here, we explicitly construct a number-preserving Hamiltonian $H$, quadratic in the number operators, for which the transitions generated by $\ad_j\a_k$ in the cutoff space have distinct energy differences.
Then we cancel all unwanted transitions exactly at the Hamiltonian level inside the cutoff space through the finite Fourier sum in \cref{eq:fourier-sum}.
Thus, we only need to account for the errors of the Trotter decompositions used to implement the resulting Hamiltonian evolutions inside the cutoff space, which we bound rigorously.
A similar Fourier argument was used earlier by Wocjan et al.~\cite{wocjan2002simulating} for coupled harmonic oscillators and by Nielsen et al.~\cite{nielsen2002universal} for finite-dimensional quantum systems.

\begin{claim}\label{claim:fourier-sum}
    Let $H$ be a diagonal integer Hamiltonian on $\calH_{\le K}^{(m)}$, such that
    \be
        H(\bm u') - H(\bm v')\pmod q
    \ee
    is unique among the directed pairs $(\bm u',\bm v')$ with $\bm u'=\bm v'+\bm e_j-\bm e_k$ appearing in $A$.
    Set
    \be
        V = e^{2\pi i H/q},\qquad \omega = H(\bm u) - H(\bm v) \pmod q, \qquad \alpha = \sqrt{(v_j+1)v_k}.
    \ee
    Then
    \begin{equation}\label{eq:fourier-sum}
        \frac1q\sum_{r=0}^{q-1} e^{-2\pi i r\omega/q} V^r A V^{-r} = \alpha \ketbra{\bm u}{\bm v}.
    \end{equation}
\end{claim}
\begin{proof}
    Since $H$ is diagonal, conjugation by $V^r$ gives
    \be
        V^r\ketbra{\bm u'}{\bm v'}V^{-r} = e^{2\pi i r(H(\bm u')-H(\bm v'))/q}\ketbra{\bm u'}{\bm v'}.
    \ee
    Consequently, for $\zeta= e^{2\pi i (H(\bm u')-H(\bm v')-\omega)/q}$,
    \begin{equation}
        \frac1q\sum_{r=0}^{q-1}
        e^{-2\pi i r\omega/q}
        V^r\ketbra{\bm u'}{\bm v'}V^{-r}
        = \left(\frac1q\sum_{r=0}^{q-1} \zeta^r\right) \ketbra{\bm u'}{\bm v'}
        = \begin{cases}
            \ketbra{\bm u'}{\bm v'},&
            H(\bm u')-H(\bm v')\equiv\omega\pmod q,\\
            0,&\text{otherwise}.
        \end{cases}
    \end{equation}
    Since $H(\bm u')-H(\bm v') \pmod q$ is unique among the pairs appearing in $A$, only $\ketbra{\bm u}{\bm v}$ survives from $A|_{\calH_{\le K}^{(m)}}$ (see \cref{eq:hop-cutoff}) and thus \cref{eq:fourier-sum} holds.
\end{proof}

We construct a suitable Hamiltonian $H$ from Kerr nonlinearities.

\begin{claim}\label{claim:quadratic-H}
    Fix distinct modes $j,k$. There exists a Hamiltonian
    \begin{equation}\label{eq:quadratic-H}
    H = \sum_r b_r\hat n_r + \sum_{r\le s} c_{rs} \hat n_r\hat n_s,\qquad b_r,c_{rs}\in\{0,\dots,q-1\},
    \end{equation}
    for some prime $q = O(D^2)$, such that $H(\bm u) - H(\bm v)\pmod q$ is unique among the pairs $(\bm u,\bm v)$ of Fock states in $\calH_{\le K}^{(m)}$ with $\bm u=\bm v+\bm e_j-\bm e_k$.
    $H$ can be computed in $\poly(m,D)$ time.
\end{claim}
\begin{proof}
    Let $d$ be the number of monomials in $H$.
    For each Fock state $\bm w$, assign a quadratic label
    \be \bm f(\bm w) \coloneq \bigl( (w_r)_r, (w_rw_s)_{r\le s} \bigr) \in \ZZ^d,\ee
    and for a transition $(\bm u,\bm v)$, let $\bm d_{\bm u,\bm v} \coloneq \bm f(\bm u) - \bm f(\bm v)$.
    We argue that $\bm d_{\bm u,\bm v}$ uniquely determines each adjacent transition.
    If $\bm u=\bm v+\bm e_j-\bm e_k$, the linear coordinates identify $j$ and $k$, and the quadratic coordinates recover the occupation numbers of $\bm v$:
    \begin{align}
        2v_j + 1 &= u_j^2 - v_j^2\\
        1-2v_k &= u_k^2 - v_k^2\\
        v_\ell &= u_ju_\ell - v_j v_\ell \quad (\ell\notin\{j,k\}).
    \end{align}
    For the fixed modes $j,k$, there are $T\leq D$ transitions, since each source state $\bm v$ determines at most one target $\bm u=\bm v+\bm e_j-\bm e_k$.
    Choose a prime $q = O(K+T^2)$ such that $q > \max\{4K-2,2T^2\}$.
    Such a prime exists by Bertrand's postulate.
    Write $\bm c=((b_r)_r, (c_{rs})_{r\le s})$, so that $H(\bm w) = \bm c\cdot\bm f(\bm w)$.
    It remains to pick $\bm c\in\{0,\dots,q-1\}^d$ so that for all distinct pairs $(\bm u,\bm v)$ and $(\bm x,\bm y)$ among these $T$ transitions:
    \begin{subequations}
    \begin{align}
        \bm c\cdot\bm d_{\bm u,\bm v} &\ne 0 \pmod q\label{eq:cduv}\\
        \bm c\cdot (\bm d_{\bm u,\bm v}-\bm d_{\bm x,\bm y}) &\ne 0\pmod q.\label{eq:cduvdxy}
    \end{align}
    \end{subequations}
    Each of these conditions is of the form $\bm a\cdot\bm c \ne 0\pmod q$ with $\bm a \ne 0$ over the integers and $\norm{\bm a}_\infty \le 4K-2$.
    To see this, consider first \cref{eq:cduv} with $\bm a=\bm d_{\bm u,\bm v}$.
    A straightforward calculation using $|\bm u|,|\bm v|\le K$ gives $\norm{\bm d_{\bm u,\bm v}}_\infty \le 2K-1$.
    Moreover, $\bm d_{\bm u,\bm v} \ne 0$ because its linear coordinates include $+1$ and $-1$.
    In case \cref{eq:cduvdxy}, we have $\bm a=\bm d_{\bm u,\bm v}-\bm d_{\bm x,\bm y}$ and thus $\norm{\bm a}_\infty \le 4K-2$.
    Since $\bm d_{\bm u,\bm v}$ is unique for each adjacent pair, $\bm a \ne 0$.
    The choice $q>4K-2$ therefore ensures $\bm a\ne 0\pmod q$.

    We give a simple deterministic algorithm to assign $c_1,\dots,c_d$ so that \cref{eq:cduv} and \cref{eq:cduvdxy} are satisfied.
    To choose $c_i$, suppose $c_1,\dots,c_{i-1}$ have already been chosen.
    Consider only the conditions $\bm a \cdot\bm c \ne 0\pmod q$ for which $a_i$ is the \emph{last nonzero entry} of $\bm a$ modulo $q$.
    Let $s = \sum_{j<i}a_jc_j$.
    We need $s+a_ic_i\ne 0\pmod q$.
    Since $q$ is prime, each condition forbids exactly one value, namely $c_i=-a_i^{-1}s\pmod q$.
    There are at most $T+T(T-1)=T^2<q$ conditions, so we can always choose the smallest allowed value of $c_i$.
    Since $K+1\le D$ and $T\leq D$, we have $q=O(D^2)$.
    With $d=O(m^2)$ and at most $T^2$ conditions, the algorithm takes $\poly(m,D)$ time.
\end{proof}

Next, we argue how to decompose the quadratic Hamiltonians from \cref{claim:quadratic-H} into a sum of implementable Hamiltonians in either gate set, so that we can approximate their evolution later via a Trotter decomposition.

\begin{claim}\label{claim:quadratic-decomposition}
    Let $H$ be of the form \cref{eq:quadratic-H}, with real coefficients.
    For either $\GKerr$ or $\GcrossKerr$ (assuming $m\ge3$ for the latter),
    there is a decomposition
    \be
        H=\sum_{\ell=1}^{L}H_\ell,\qquad L=O(m^2),
    \ee
    such that each $e^{itH_\ell}$ can be implemented using $O(1)$ gates
    from the chosen gate set.
\end{claim}
\begin{proof}
    Phase-shift terms $\hat n_j$ are already available in both gate sets.
    It remains to express cross-Kerr terms using self-Kerr and vice versa.
    For distinct modes $j,k$, define ``virtual modes''
    \be
    \hat c_\phi = \frac{\a_j + e^{i\phi} \hat a_k}{\sqrt2},\qquad \hat d_\phi = \frac{\hat a_j - e^{i\phi}\hat a_k}{\sqrt2}.
    \ee
    Write $N_\phi \coloneq \hat c_\phi^\dagger \hat c_\phi$ and $M_\phi \coloneq \hat d_\phi^\dagger \hat d_\phi$ for the number operators in the virtual modes.
    Both modes $\hat c_\phi$ and $\hat d_\phi$ are obtained via linear optical conjugation from $\a_j,\a_k$.
    \begin{equation}
        W \a_j W^\dagger = \hat c_\phi,\qquad W \a_k W^\dagger = -\hat d_\phi, \qquad W \coloneq e^{-i\phi \hat n_k} e^{\pi(\ad_k\a_j - \ad_j\a_k)/4}.
    \end{equation}
    Since conjugation respects products and exponentials, we have
    \begin{equation}\label{eq:kerr-conjugate}
        e^{it N_\phi^2} = W e^{it\hat n_j^2} W^\dagger,\qquad e^{it N_\phi M_\phi} = W e^{it\hat n_j\hat n_k} W^\dagger.
    \end{equation}
    \proofstep{Self-Kerr gate set.}
    Put $S = \hat n_j + \hat n_k$ and $A = \hat a_j^\dagger \hat a_k$.
    Then 
    \be N_\phi = \frac12\bigl(S + e^{i\phi} A + e^{-i\phi}A^\dagger\bigr).\ee
    Then we square and sum over four phases to eliminate the cross terms to obtain
    \be
    \sum_{r=0}^3 N_{r\pi/2}^{\,2} = S^2 + AA^\dagger + A^\dagger A = \hat n_j^2 + \hat n_k^2 + 4\hat n_j \hat n_k + \hat n_j + \hat n_k,
    \ee
    where we use $AA^\dagger = \hat n_j(\hat n_k+1)$ and $A^\dagger A = \hat n_k(\hat n_j+1)$.
    Rearranging gives
    \be
    \hat n_j \hat n_k  = \frac14\left(\sum_{r=0}^3 N_{r\pi/2}^{\,2} -\hat n_j^2 - \hat n_k^2 - \hat n_j - \hat n_k\right),
    \ee
    where the evolution of each term is implementable with $O(1)$ gates in $\GKerr$ (see \cref{eq:kerr-conjugate} for $N_\phi^2$).

    \proofstep{Cross-Kerr gate set.}
    Put $R_{jk} = N_0 M_0 + N_{\pi/2}M_{\pi/2}$, where $N_\phi,M_\phi$ refer to the pair of modes $j,k$.
    Write
    \be
    N_\phi=\frac12(S+B_\phi),\qquad
    M_\phi=\frac12(S-B_\phi),\qquad
    B_\phi=e^{i\phi}A+e^{-i\phi}A^\dagger.
    \ee
    Then
    \be
        N_{\phi}M_\phi = \frac14\left(S+B_\phi\right)\left(S- B_\phi\right) = \frac14\left(S^2 - B_\phi^2\right).
    \ee
    Expanding the two phases in $R_{jk}$, we have 
    \be
    B_0^2=A^2+AA^\dagger+A^\dagger A+(A^\dagger)^2,\qquad
    B_{\pi/2}^2=-A^2+AA^\dagger+A^\dagger A-(A^\dagger)^2.
    \ee
    Hence,
    \begin{align}
    R_{jk}
    &=N_0M_0+N_{\pi/2}M_{\pi/2}
    =\frac12\left(S^2-AA^\dagger-A^\dagger A\right)\notag\\
    &=\frac12\left[
    (\hat n_j+\hat n_k)^2
    -\hat n_j(\hat n_k+1)-\hat n_k(\hat n_j+1)
    \right]
    =\frac12\left[
    \hat n_j(\hat n_j-1)+\hat n_k(\hat n_k-1)
    \right].
    \end{align}
    Choose a third mode $\ell\notin\{j,k\}$, which is possible since $m\ge3$.
    Applying the same identity to the pairs $(j,\ell)$ and $(k,\ell)$ gives
    \be
        R_{jk}+R_{j\ell}-R_{k\ell}=\hat n_j(\hat n_j-1),
    \ee
    and therefore
    \be
        \hat n_j^2=\hat n_j+R_{jk}+R_{j\ell}-R_{k\ell}.
    \ee
    Each $R_{ab}$ is a sum of two Hamiltonians whose evolutions can each be implemented using $O(1)$ gates from $\GcrossKerr$ by \cref{eq:kerr-conjugate}.
    In either gate set, replacing each monomial of $H$ by the corresponding constant-size decomposition yields $L=O(m^2)$.
\end{proof}

Next, we recall the first-order Trotter formula.
Its application to energy-preserving Hamiltonians in the cutoff norm is straightforward.

\begin{lem}[Trotter approximation {\cite[Proposition~9]{ChildsEtAl2021}}]
\label{lem:cutoff-trotter}
    Let $H=\sum_{\ell=1}^L H_\ell$, where each $H_\ell$ is an
    energy-preserving Hamiltonian on $m$ modes.
    For every $t\in\RR$ and integer $r\ge1$,
    \begin{equation}
        \norm{
            e^{itH}
            -\left(\prod_{\ell=1}^L e^{itH_\ell/r}\right)^r
        }_{(K)}
        \le \frac{t^2}{2r}
            \sum_{1\le j<\ell\le L}
            \norm{[H_j,H_\ell]}_{(K)}
        \le \frac{t^2}{2r}
            \left(\sum_{\ell=1}^L\norm{H_\ell}_{(K)}\right)^2.
    \end{equation}
\end{lem}

We now have the tools to prove \cref{lem:2level-rotations}.
\cref{claim:fourier-sum} realizes the $\ketbra{\bm u}{\bm v}$ and $\ketbra{\bm v}{\bm u}$ transitions via conjugation with an $m$-mode energy-preserving Hamiltonian $H$.
\cref{claim:quadratic-H} explicitly constructs $H$ to be quadratic in the number operators.
\cref{claim:quadratic-decomposition} gives a decomposition of $H$ into Hamiltonians whose evolution is implementable with both $\GKerr$ and $\GcrossKerr$ gate sets.
\cref{lem:cutoff-trotter} gives the error of splitting $H$ into the implementable terms.
The remaining proof does the rather technical work of assembling these steps rigorously.

\begin{proof}[Proof of \cref{lem:2level-rotations}]
    By the Euler decomposition for $\mathrm{SU}(2)$ \cite[Lemma~3.3]{hamada2014minimum}, write
    \be
        U=e^{iaX_{\bm u,\bm v}}e^{ibY_{\bm u,\bm v}}e^{icX_{\bm u,\bm v}},
        \qquad |a|,|b|,|c|\le2\pi.
    \ee
    It suffices to approximate each factor to error $\eta/3$: the product then approximates $U$ to error $\eta$ and preserves the sector determinants.
    We therefore construct a circuit for $e^{itP}$ with $P\in\{X_{\bm u,\bm v},Y_{\bm u,\bm v}\}$ and $|t|\le2\pi$.

    Write $\bm u=\bm v+\bm e_j-\bm e_k$ and choose $H,q$ from \cref{claim:quadratic-H} for this pair of modes.
    Recall
    \be
        A=\ad_j\a_k,\qquad V=e^{2\pi iH/q},\qquad
        \omega=H(\bm u)-H(\bm v)\pmod q,\qquad
        \alpha=\sqrt{(v_j+1)v_k}\ge1.
    \ee
    Set $B_\phi=e^{i\phi}A+e^{-i\phi}A^\dagger$, so $\norm{B_\phi}_{(K)}\le2K$.
    With $\phi=2\pi\omega/q$, \cref{claim:fourier-sum} states that
    \be
        \ketbra{\bm u}{\bm v}=\frac1{\alpha q}\sum_{r=0}^{q-1}
        e^{-ir\phi}V^rAV^{-r}.
    \ee
    Together with the adjoint, we obtain
    \begin{alignat}{3}
        X_{\bm u,\bm v}
        &= \ketbra{\bm u}{\bm v}+\ketbra{\bm v}{\bm u}&&=\frac1{\alpha q}\sum_{r=0}^{q-1}
        V^r\left(e^{-ir\phi}A+e^{ir\phi}A^\dagger\right)V^{-r}
        &&=\frac1{\alpha q}\sum_{r=0}^{q-1}V^rB_{-r\phi}V^{-r},\\
        Y_{\bm u,\bm v}
        &=-i(\ketbra{\bm u}{\bm v}-\ketbra{\bm v}{\bm u})&&=\frac1{\alpha q}\sum_{r=0}^{q-1}
        V^r\left(-ie^{-ir\phi}A+ie^{ir\phi}A^\dagger\right)V^{-r}
        &&=\frac1{\alpha q}\sum_{r=0}^{q-1}V^rB_{-r\phi-\pi/2}V^{-r}.\notag
    \end{alignat}
    Thus, setting $\theta=0$ for $P=X_{\bm u,\bm v}$ and $\theta=\pi/2$ for $P=Y_{\bm u,\bm v}$, we have
    \begin{equation}\label{eq:hermitian-fourier-sum}
        P=\frac1{\alpha q}\sum_{r=0}^{q-1}V^rB_{-r\phi-\theta}V^{-r}.
    \end{equation}

    We next approximate the conjugating unitaries $V^r$.
    We write the Hamiltonian $H$ constructed in \cref{claim:quadratic-H,claim:quadratic-decomposition}
    as $H=\sum_{\ell=1}^L H_\ell$ with $L=O(m^2)$.
    Each of the $H_\ell$ terms can be implemented exactly with $\GKerr$ or $\GcrossKerr$ gate sets, and so it remains to estimate the error of the Trotter decomposition via \cref{lem:cutoff-trotter}.
    We introduce the bound
    \be
        \Lambda\coloneq\sum_{\ell=1}^L\norm{H_\ell}_{(K)}
        =O(m^2qK^2),
    \ee
    which holds because the coefficients of $H$ are bounded by $q$ and there are $O(m^2)$ monomials of degree $2$ in the number operators, which contribute $O(K^2)$ each.
    We want to approximate each $V^r$ to error at most $\xi$, where $0<\xi<1$ will be chosen below.
    By \cref{lem:cutoff-trotter}, at time $2\pi r/q<2\pi$ it suffices to use
    \be
        J=\left\lceil\frac{2\pi^2\Lambda^2}{\xi}\right\rceil
    \ee
    Trotter steps.
    This produces energy-preserving circuits
    \be
        Q_r=\left(\prod_{\ell=1}^L
        e^{i(2\pi r/q)H_\ell/J}\right)^J,
        \qquad \norm{Q_r-V^r}_{(K)}
        \le\frac{(2\pi r/q)^2\Lambda^2}{2J}
        \le\frac{2\pi^2\Lambda^2}{J}\le\xi,
    \ee
    each using $O(LJ)$ gates, since each of the $J$ Trotter steps contains $L$ evolutions costing $O(1)$ gates each. Thus the gate count for each $Q_r$ is
    $
        O(LJ)=O(m^2\Lambda^2/\xi)=O(m^6q^2K^4/\xi).
    $

    Now that we can approximate the $V^r$ unitaries, we need to apply \cref{lem:cutoff-trotter} again to decompose \cref{eq:hermitian-fourier-sum}.
    Its summand norms sum to at most $2K/\alpha\le2K$.
    We allocate $\eta/6$ of each rotation's $\eta/3$ error budget to this Trotter approximation.
    Since $|t|\le2\pi$, the Trotter error is at most $8\pi^2K^2/M$ for $M$ steps.
    To make this at most $\eta/6$, choose
    \be
        M=\left\lceil\frac{48\pi^2K^2}{\eta}\right\rceil,
        \qquad \tau=\frac{t}{\alpha qM}.
    \ee
    This gives
    \be
        C=\left(\prod_{r=0}^{q-1}
        V^r e^{i\tau B_{-r\phi-\theta}}V^{-r}\right)^M,
        \qquad
        \norm{C-e^{itP}}_{(K)}
        \le\frac{t^2}{2M}(2K)^2
        \le\frac{8\pi^2K^2}{M}\le\eta/6.
    \ee
    Replace each $V^r$ by $Q_r$ and each $V^{-r}$ by $Q_r^\dagger$, obtaining
    \be
        \widetilde{C}=\left(\prod_{r=0}^{q-1}
        Q_r e^{i\tau B_{-r\phi-\theta}}Q_r^\dagger\right)^M.
    \ee
    To bound the replacement error, write $W_r=e^{i\tau B_{-r\phi-\theta}}$ to obtain
    \begin{align}
        \norm{Q_rW_rQ_r^\dagger-V^rW_rV^{-r}}_{(K)}
        &=\norm{Q_r(W_r-\II)Q_r^\dagger-V^r(W_r-\II)V^{-r}}_{(K)}\notag\\
        &\le2\xi\norm{W_r-\II}_{(K)}
        \le4\xi|\tau|K,
    \end{align}
    where we use the fact that $W_r \approx\II$ for small $|\tau|$ for a minor improvement.
    Telescoping over the $qM$ factors of $C$ and choosing $\xi=\eta/(48\pi K)$ yields
    \be
        \norm{\widetilde{C}-C}_{(K)}
        \le4qM\xi|\tau|K
        =\frac{4|t|K\xi}{\alpha}\le\eta/6.
    \ee
    The triangle inequality therefore gives $\norm*{\widetilde{C}-e^{itP}}_{(K)}\le\eta/3$.

    The gate count for each rotation is
    \be
        O\!\left(Mq\,\frac{m^6q^2K^4}{\xi}\right)
        =O(m^6q^3K^7\eta^{-2})
        =O(m^6K^7D^6\eta^{-2}),
    \ee
    using $q=O(D^2)$.
    To see that $\det \widetilde{C}^{(n)}=1$ for all $0\le n\le K$, note that $B_\phi$ has zero diagonal in the Fock basis.
    Hence,
    \be
        \det e^{i\tau B_\phi^{(n)}}=e^{i\tau\Tr B_\phi^{(n)}}=1.
    \ee
    By multiplicativity of the determinant, the determinants of the $Q_r,Q_r^\dagger$ terms cancel to $1$.

    Finally, the construction is efficient, given $U$ to sufficient precision.
    To accommodate numerical error, run the construction with tolerance $\eta/2$ and reserve the remaining $\eta/2$ for computing the Euler angles and rounding gate parameters.
    For a circuit of $G$ gates, parameter accuracy $O(\eta/(GK^2))$ suffices, since the elementary generators have cutoff norm $O(K^2)$.
    The Euler decomposition and gate parameters can be computed efficiently to sufficient precision using standard elementary-function algorithms.
    These computations require only polynomially many bits and $\poly(m,K,D,\eta^{-1})$ time.

\end{proof}

\subsubsection{Coarse compiler}

We now combine adjacent two-level decomposition with \cref{lem:2level-rotations} to obtain a coarse compiler.

\begin{lem}[Adjacent two-level decomposition]\label{lem:adjacent-decomposition}
    Let $d\ge2$, $U\in U(d)$, and $0<\varepsilon<1$.
    Given $U$ to sufficient precision, one can compute in time $\poly(d,\log(1/\varepsilon))$ a sequence of $N=O(d^2)$ two-level unitaries $V_1,\ldots,V_N$, each acting non-trivially only on $\Span\{\ket j,\ket{j+1}\}$ for some $j$ and as the identity on its orthogonal complement, such that
    \be
        \norm{U-V_1\cdots V_N}\le\varepsilon.
    \ee
    If $U\in\mathrm{SU}(d)$, each $V_i$ can be chosen to act as an $\mathrm{SU}(2)$ unitary on its two-dimensional subspace.
\end{lem}
\begin{proof}
    We apply the Givens QR decomposition to $U$ to write it as $U=QR$, where $Q$ is a product of two-level rotations and $R$ is upper triangular.
    Bottom-to-top elimination uses at most $d(d-1)/2$ complex determinant-one rotations on consecutive rows \cite[Section~5.1.13 and Algorithm~5.2.4]{golubMatrixComputations2013}.
    Since $U$ and $Q$ are unitary, $R$ is also unitary and hence diagonal.
    Write $R=\operatorname{diag}(\lambda_1,\ldots,\lambda_d)$ and decompose into $d-1$ additional adjacent two-level factors.
    For $U\in\mathrm{SU}(d)$, we also have $\det R=1$, and these factors can be chosen with determinant one as follows.
    Put $z_j=\prod_{k=1}^j\lambda_k$.
    The product, for $j=1,\ldots,d-1$, of the adjacent diagonal blocks $\operatorname{diag}(z_j,z_j^{-1})$ equals $R$.
    Thus all factors can be chosen in $\mathrm{SU}(2)$.

    Standard backward stability of Givens QR, extended to complex arithmetic \cite[Theorem~19.10 and Section~3.6]{Higham22}, together with rounding the resulting rotation parameters, gives error $\poly(d)2^{-b}$ at $b$ working bits.
    Hence $b=O(\log(d/\varepsilon))$ suffices, and the computation takes $\poly(d,\log(1/\varepsilon))$ time.
\end{proof}

\begin{lem}[Coarse compiler]\label{lem:coarse-compiler}
    Let $U\in\mathrm{SU}(D_0,\ldots,D_K)$ on $\calH_{\le K}^{(m)}$, where $D_n=\dim\calH_n^{(m)}$.
    For $0<\epsilon_0<1$, there is a circuit $C_U$ such that
    \be
        \norm{C_U-U}_{(K)}\le\epsilon_0,
        \qquad \det C_U^{(n)}=1\quad(0\le n\le K),
    \ee
    using $O(m^6K^7D^{12}\epsilon_0^{-2})$ gates from $\GKerr$ (for $m\ge2$) or $\GcrossKerr$ (for $m\ge3$). Here $D = \sum_{j=0}^K D_j$.
    Given $U$ to sufficient precision, the circuit can be computed in $\poly(m,K,D,\epsilon_0^{-1})$ time.
    In particular, for fixed $\epsilon_0$ this gives a polynomial-size coarse compiler.
\end{lem}
\begin{proof}
    Order the Fock states in each nonvacuum sector so that consecutive states differ by moving one photon, using the efficient construction of a Gray ordering of compositions into nonnegative parts given by Klingsberg~\cite{klingsberg1982gray}.
    Apply \cref{lem:adjacent-decomposition} to each $U^{(n)}$ with tolerance $\epsilon_0/2$ to obtain a product $W=V_1\cdots V_N$ with $N=O(D^2)$ satisfying
    $
        \norm{U-W}_{(K)}\le\epsilon_0/2
    $,
    where each $V_j$ acts as an $\mathrm{SU}(2)$ unitary on adjacent Fock states and as the identity elsewhere.

    For $N\ge1$, apply \cref{lem:2level-rotations} to each $V_j$ with accuracy $\eta=\epsilon_0/(2N)$, obtaining circuits $C_j$.
    Telescoping gives
    \be
        \norm{C_1\cdots C_N-W}_{(K)}
        \le\sum_{j=1}^N\norm{C_j-V_j}_{(K)}\le\epsilon_0/2.
    \ee
    Thus $C_U=C_1\cdots C_N$ has the required accuracy, and its sector determinants are one by multiplicativity.
    Its gate count is
    \be
        O\!\left(Nm^6K^7D^6\left(\frac{N}{\epsilon_0}\right)^2\right)
        =O(m^6K^7D^{12}\epsilon_0^{-2}).
    \ee
    The Gray ordering, numerical decomposition, and individual rotation compilations all take polynomial time.
\end{proof}

\subsubsection{Solovay--Kitaev details}
\label{sec:SK-details}

We use the block-diagonal group $\mathrm{SU}(d_1,\ldots,d_K)$ defined above. Let $L, T : \mathbb{R}_{\geq 0} \rightarrow \mathbb N$. We say a universal gate set $\mathcal{G}$ has a $(L, T)$-coarse compiler if for any $U \in \mathrm{SU}(d_1,\ldots,d_K)$ and $\epsilon_0 > 0$ there exists a circuit $U_0\in\mathrm{SU}(d_1,\ldots,d_K)$ of length $L(\varepsilon_0)$ computable in time $T(\varepsilon_0)$ such that
\be
\norm{U - U_0} \leq \varepsilon_0.
\ee
For the classical runtime bounds below, we additionally assume that the block matrices of any circuit returned by the coarse compiler can be computed to $p$ bits of precision in $T(\varepsilon_0)\poly(p+\log d_{\max})$ bit operations, where $d_{\max}=\max_j d_j$.
Let
\be
c_{\rm SK} := \log_{3/2} (5) \approx 3.97.
\ee

\begin{theorem}\label{thm:block-sk}
    Let $U \in \mathrm{SU} (d_1, \ldots, d_K)$ and let $\mathcal{G}$ be a gate set that is universal for $\mathrm{SU} (d_1, \ldots, d_K)$ containing inverses with an $(L,T)$ coarse compiler. Let $0<\varepsilon<1/2$. For any $0 < \varepsilon_0 \leq \min\{10^{-5}, 1/d_1,\ldots, 1/d_K\}$, there exists a circuit $\widehat U$ as a sequence of $L_\varepsilon = O( L(\varepsilon_0) \log^{c_{\rm SK}} (1/\varepsilon))$ gates from $\mathcal{G}$, such that \be
    \norm{U - \widehat U} \leq \varepsilon.
    \ee
    Assuming $1/\varepsilon_0 \leq \poly(\max_j d_j)$, there is furthermore a classical algorithm with running time $t_\varepsilon = \widetilde O(T(\varepsilon_0) + \sum_j d_j^4) \mathrm{polylog}(1/\varepsilon)$ which produces a succinct representation of $\widehat U$.
\end{theorem}

\begin{algorithm}[tbp]
\caption{Recursive Solovay--Kitaev compilation}
\label{alg:solovay-kitaev}
\begin{algorithmic}%
    \Function{\texttt{SK}}{$U, n$}
        \If{$n = 0$}
            \State \Return the coarse approximation for $U$
        \Else
            \State $U_{n-1} \gets \texttt{SK}(U, n-1)$
            \State $(V, W) \gets \texttt{GC-Decompose}(U U^\dagger_{n-1})$
            \State $V_{n-1} \gets \texttt{SK}(V, n-1)$
            \State $W_{n-1} \gets \texttt{SK}(W, n-1)$
            \State $U_n \gets V_{n-1} W_{n-1} V^\dagger_{n-1} W^\dagger_{n-1} U_{n-1}$
            \State \Return $U_n$
        \EndIf
    \EndFunction
\end{algorithmic}
\end{algorithm}

\begin{proof}
We follow the Solovay--Kitaev algorithm outlined in \cite{dawson2005solovay}, shown in \cref{alg:solovay-kitaev}.
Let $\varepsilon_n$ be an upper bound on the approximation error obtained for level $n$ input to the algorithm, $L_n$ be the size of the corresponding circuit and $t_n$ be the time complexity for performing this algorithm.

For $n = 0$, apply the assumed $(L,T)$-coarse compiler at accuracy $\varepsilon_0$ to obtain a circuit ${U}_0$ such that
\be
\norm{U - U_0} \leq \varepsilon_0.
\ee
This is the coarse estimator in the first step. 

If $n > 0$, the algorithm recursively calls the procedure. It first obtains ${U}_{n-1}$ satisfying
\be
\norm{U - {U}_{n-1}} \leq \varepsilon_{n-1}.
\ee
Let $\Delta_{n-1} := U U^\dagger_{n-1}$ which satisfies $\norm{\Delta_{n-1} - \II} \leq \varepsilon_{n-1}$.
In exact arithmetic, \texttt{GC-Decompose} takes $\Delta_{n-1}$ and outputs a pair of matrices $V, W \in \mathrm{SU} (d_1, \ldots, d_K)$ such that $\norm{V - \II} ,\norm{W-\II} \leq 2 \sqrt{2\varepsilon_{n-1}}$ and
\be
\norm{\Delta_{n-1} - V W V^\dagger W^\dagger} \leq C \varepsilon_{n-1}^{3/2}.
\ee
This is achieved in \cref{thm:GC-Decompose}.

Next we use two additional calls to the \texttt{SK} procedure to obtain $V_{n-1}$ and $W_{n-1}$ which satisfy
\be
\norm{V - V_{n-1}} \leq \varepsilon_{n-1}, \quad \norm{W - W_{n-1}} \leq \varepsilon_{n-1}.
\ee
We finally set $U_n = V_{n-1} W_{n-1} V^\dagger_{n-1} W^\dagger_{n-1}U_{n-1}$. We can now see
\begin{align}
\begin{split}
\norm{U - U_{n}} &= \norm{\Delta_{n-1} - V_{n-1} W_{n-1} V^\dagger_{n-1} W^\dagger_{n-1}}\\
&\leq \norm{\Delta_{n-1} - VWV^\dagger W^\dagger} +\norm{VWV^\dagger W^\dagger- V_{n-1} W_{n-1} V^\dagger_{n-1} W^\dagger_{n-1}}\\
&\leq C \varepsilon_{n-1}^{3/2} + 21 \varepsilon_{n-1}^{3/2},
\end{split}
\end{align}
where in the last line we have used \cite[Lemma 1]{dawson2005solovay} and assumed $\varepsilon_{n-1} \leq 1$.

Choose $b$ in \cref{lem:GC-decomp} so that $d_{\max}2^{-b}\leq c\varepsilon_{n-1}^{3/2}$, where $d_{\max}=\max_j d_j$ and $c>0$ is a sufficiently small absolute constant. The decomposition and rounding errors then add at most $\varepsilon_{n-1}^{3/2}$ to the bound on $\norm{U-U_n}$. Let $A = C + 22$ (which can be taken to be $A = 294$ using the value of $C$ indicated in \cref{lem:GC-decomp}). We will have
\be
\varepsilon_n \coloneq A \varepsilon_{n-1}^{3/2}.
\ee
Therefore
\be
\varepsilon_n \leq A^{-2} (A^2 \varepsilon_0)^{(3/2)^n}.
\ee
This procedure works for any 
\be
\varepsilon_0 \leq 10^{-5} < 1/A^2.
\ee
We, however, need to ensure $\varepsilon_0 \leq \frac{1}{\max_j d_j}$ alongside the assumption above in order to make sure we can use \cref{thm:GC-Decompose} in every iteration $n > 0$. The length of the circuit satisfies $L_n \leq 5 L_{n-1}$. Therefore
\be
L_n \leq L(\varepsilon_0) 5^n.
\ee
Standard reasoning implies the dependency $L_\varepsilon$ stated in the theorem.

We finally analyze the classical running time $t_n$. Let
\begin{equation}
    d_{\max}:=\max_j d_j.
\end{equation}
At recursion level $n$, the \texttt{GC-Decompose} routine is applied
blockwise to a matrix $\Delta$ satisfying
\begin{equation}
    \norm{\Delta-\II}\leq\varepsilon_{n-1}.
\end{equation}
We choose the working precision so that the numerical errors contribute at most $\varepsilon_{n-1}^{3/2}$ to the error budget above.

By \cref{lem:GC-decomp}, it suffices to compute the matrices produced by
\texttt{GC-Decompose} to
\begin{equation}
    p_n
    =
    O\left(
        \log d_{\max}
        +
        \log\frac{1}{\varepsilon_{n-1}}
    \right)
\end{equation}
bits of output precision. The additional working precision required
internally by \cref{lem:GC-decomp} changes this quantity only by a
constant factor and an additive $O(\log d_{\max})$ term.

Since
\begin{equation}
    \log\frac{1}{\varepsilon_n}
    =
    O\left(
        (3/2)^n\log\frac{1}{\varepsilon_0}
    \right),
\end{equation}
we may choose
\begin{equation}
    p_n
    =
    \Theta\left(
        (3/2)^n
        \log\frac{d_{\max}}{\varepsilon_0}
    \right).
\end{equation}

By \cref{lem:GC-decomp}, the blockwise implementation of
\texttt{GC-Decompose} therefore has running time
\begin{equation}
    T_{\rm GC}(n)
    =
    \left(\sum_{j=1}^K d_j^4\right)
    \poly(p_n+\log d_{\max}).
\end{equation}
Let $r\geq1$ be a fixed exponent covering the bit complexity of
commutator decomposition and matrix evaluation of the circuits returned by recursive calls.
The running time consequently satisfies
\begin{equation}
    t_n
    \leq
    3t_{n-1}
    +
    \left(T(\varepsilon_0)+\sum_{j=1}^K d_j^4\right)
    \widetilde O(p_n^r).
\end{equation}
Iterating the recurrence gives
\begin{align}
    t_n
    &\leq
    3^n T(\varepsilon_0)
    +
    \left(T(\varepsilon_0)+\sum_{j=1}^K d_j^4\right)
    \widetilde O\left(
        \sum_{k=1}^n
        3^{n-k}(3/2)^{rk}
    \right)
    \log^r\left(\frac{d_{\max}}{\varepsilon_0}\right).
\end{align}
The geometric sum satisfies
\begin{equation}
    \sum_{k=1}^n
    3^{n-k}(3/2)^{rk}
    =
    \widetilde O\left(
        \max\{3^n,(3/2)^{rn}\}
    \right).
\end{equation}
Using
\begin{equation}
    \frac{1}{\varepsilon_0}
    \leq
    \poly(d_{\max}),
\end{equation}
the factor
\begin{equation}
    \log^r\left(\frac{d_{\max}}{\varepsilon_0}\right)
\end{equation}
can be absorbed into the $\widetilde O$ notation. Hence
\begin{equation}
    t_n
    =
    \widetilde O\left(
        T(\varepsilon_0)+\sum_{j=1}^K d_j^4
    \right)
    \max\{3^n,(3/2)^{rn}\}.
\end{equation}

Since
\begin{equation}
    (3/2)^n
    =
    O\left(\log\frac{1}{\varepsilon}\right),
\end{equation}
we have
\begin{equation}
    3^n
    =
    O\left(
        \log^{\log_{3/2}3}\frac{1}{\varepsilon}
    \right)
\end{equation}
and
\begin{equation}
    (3/2)^{rn}
    =
    O\left(
        \log^r\frac{1}{\varepsilon}
    \right).
\end{equation}
Therefore, setting
\begin{equation}
    c_{\rm CL}
    :=
    \max\left\{
        \log_{3/2}3,r
    \right\},
\end{equation}
we obtain
\begin{equation}
    t_\varepsilon
    =
    \widetilde O\left(
        \left(
            T(\varepsilon_0)+\sum_jd_j^4
        \right)
        \log^{c_{\rm CL}}(1/\varepsilon)
    \right).
\end{equation}

 \end{proof}

For unitary matrices we use the notation $\gcomm{V}{W} := V W V^{-1} W^{-1}$ to denote the group commutator. 
\begin{theorem}
Suppose $\Delta \in \mathrm{SU} (d_1, \ldots, d_K)$ and $\norm{\Delta - \II} \leq \delta$ and suppose $\delta \leq \frac{1}{\max_j d_j}$. Then there exist $V, W\in\mathrm{SU}(d_1,\ldots,d_K)$ and an absolute constant $C >0$ such that $\norm{V - \II}, \norm{W - \II} \leq 2 \sqrt{2\delta}$ and
\be
\norm{\Delta - \gcomm{V}{W}} \leq C \delta^{3/2}.
\ee
\label{thm:GC-Decompose}
\end{theorem}

\begin{proof}
    We first prove this for $K = 1$. If $d_1=1$, take $V=W=\II$. Otherwise, since $\norm{\Delta - \II} \leq \delta$ and $\delta \leq 1/d_1$, \cref{lem:GC-decomp} gives $V,W\in\mathrm{SU}(d_1)$ with $\norm{V-\II},\norm{W-\II}\leq 2\sqrt{2\delta}$ and
    \be
    \norm {\Delta - \gcomm{V}{W}} \leq C \delta^{3/2}
    \ee
    for some constant $C$ which can be taken to be $272$.

    The same bound naturally holds for $K > 1$. To see this, suppose $\Delta = \Delta_1 \oplus \ldots \oplus \Delta_K$ satisfies the conditions of the theorem. Using the $K =1$ case, we can find $V_1, \ldots, V_K$ and $W_1, \ldots, W_K$ such that
    \be
    \norm {V_j - \II} \leq 2 \sqrt{2\delta}, \quad \norm {W_j - \II} \leq 2 \sqrt{2\delta}, \quad 1 \leq j \leq K
    \ee
    and 
    \be
    \norm{\Delta_j - \gcomm{V_j}{W_j}} \leq C \delta^{3/2}.
    \ee
    Now set $V = V_1 \oplus \ldots \oplus V_K$, $W = W_1 \oplus \ldots \oplus W_K$. Then since $\norm{\bigoplus_j A_j} = \max_j \norm{A_j}$, we obtain exactly the same bounds for $V, W, \Delta$.

\end{proof}

\begin{lem}
    Suppose $A, B$ are Hermitian matrices satisfying $\norm{A} , \norm{B} \leq s$. Then
    \be
    \norm{\gcomm{e^{iA}}{e^{iB}} - e^{-[A,B]}} = O(s^3).
    \ee
    \label{lem:GC-approx}
    We can choose the constant factor $12$.
\end{lem}

\begin{proof}
    Let 
    \be
    F(t) = e^{i t A} e^{i t B}e^{-i t A} e^{-i t B}.
    \ee
    Therefore
    \be
    F (0) = \II, \quad F'(0) = 0, \quad F^{''}(0) = - 2 [A,B].
    \ee
    Therefore 
    \be
    \gcomm{e^{iA}}{e^{iB}} = F(1) = \II - [A,B] + R_3,
    \ee
    such that 
    \be
    \norm{R_3} \leq 1/6 \sup_{0\leq t \leq 1} \norm{F^{'''}(t)}.\ee
    Using the multinomial product rule, we have for any $t$
    \be
    \|F^{'''} (t)\| \leq (\norm{A} + \norm{B} + \norm{A} + \norm{B})^3 \leq (4s)^3.
    \ee
    Therefore 
    \be
    \norm {\gcomm{e^{iA}}{e^{iB}} - (\II - [A,B])} \leq \frac{32}{3} s^3.
    \ee
    Finally 
    \be
    \norm{e^{- [A,B]} - (\II - [A,B])} \leq \frac{1}{2} \norm{[A,B]}^2 \leq 2s^4.
    \ee
    Here we used that $[A,B]$ is anti-Hermitian. Using the triangle inequality, for $s\leq 2/3$ we obtain the upper bound $12s^3$. For $s\geq 2/3$, the bound follows from the trivial upper bound $2$ on the distance between unitaries.
\end{proof}

The following lemma is an essential tool in establishing the Solovay--Kitaev theorem over qudits.

\begin{lem} [The cost of group commutator decomposition]
    Given $\Delta \in \mathrm{SU} (d)$ for $d \geq 2$, satisfying $\norm{\Delta - \II} \leq \delta \leq 1/d$, we can find matrices $V,W\in\mathrm{SU}(d)$ with $\norm{V - \II}, \norm{W -\II} \leq 2\sqrt{2\delta}$ such that
    \be
    \norm{\Delta - \gcomm{V}{W}} \leq 272 \delta^{3/2}.
    \ee

    Furthermore, for $b\geq1$, if the entries of $\Delta$ are supplied to
$2b+O(\log d)$ bits of precision, one can compute entrywise $b$-bit
approximations to $V^{(b)},W^{(b)}\in\mathrm{SU}(d)$ satisfying
\begin{align}
    \norm{V^{(b)}-\II},\norm{W^{(b)}-\II}&\leq2\sqrt{2\delta}+2^{-b},
    \label{eq:gc-numerical-factor-bounds}\\
    \norm{\Delta-\gcomm{V^{(b)}}{W^{(b)}}}&\leq272\delta^{3/2}+2^{-b},
    \label{eq:gc-numerical-commutator-bound}
\end{align}
in $d^4\poly(b+\log d)$ bit operations.
    \label{lem:GC-decomp}
\end{lem}

\begin{proof}
    Using \cref{lem:log-Delta-branch}, we can find a traceless Hermitian matrix $H$ such that $\norm{H} \leq 2 \delta$ and $\Delta = e^{i H}$. Next, using \cref{lem:Ra-Ra}, there exist traceless Hermitian matrices $A$ and $B$ such that $i H  = [A, B]$ and $\norm{A}, \norm{B} \leq 2 \sqrt{2 \delta}$. Therefore $\Delta = e^{ [A, B]}$. Finally, letting $W = e^{i A}$ and $V = e^{iB}$ and using \cref{lem:GC-approx}, we can show
    \be
    \norm{\Delta - \gcomm{V}{W}} \leq 272 \delta^{3/2}.
    \ee
    Furthermore, suppose the eigenvalues of $A$ are $\lambda_1, \ldots , \lambda_d$. Then
    \be
    \norm{e^{iA} - \II} \leq 2 \max_j |\sin (\lambda_j/2)|  \leq \max_j |\lambda_j| \leq 2 \sqrt{2 \delta}.
    \ee
    We can similarly show $\norm{e^{iB} - \II} \leq 2 \sqrt{2 \delta}$.

We next study the computational complexity of the construction.
Fix a target output precision of $b$ bits and choose
\begin{equation}
    q=2b+O(\log d).
\end{equation}
We first use \cref{lem:log-Delta-branch} to compute a traceless Hermitian
matrix $\widetilde H\in\mathbb Q(i)^{d\times d}$ satisfying
\begin{equation}
    \norm{H-\widetilde H}\leq 2^{-q}.
    \label{eq:gc-log-approx}
\end{equation}
To achieve this, it suffices to take
\begin{equation}
    m=\left\lceil\frac{q+O(1)}{\log_2 d}\right\rceil
\end{equation}
and use $q+O(\log d)$ bits of precision in
\cref{lem:log-Delta-branch}. Since $d\geq2$, we have
$m=O(q)$, and hence this step requires
\begin{equation}
    d^3\poly(q+\log d)
\end{equation}
bit operations.

We then apply \cref{thm:finite-precision-diagonalization} to
$\widetilde H$. This produces a matrix $\widehat U$ and a real diagonal
matrix
\begin{equation}
    D=\diag(\lambda_1,\ldots,\lambda_d)
\end{equation}
such that
\begin{equation}
    \norm{\widehat U^\dagger\widehat U-\II}\leq 2^{-q},
    \qquad
    \norm{\widetilde H-\widehat U D\widehat U^\dagger}
    \leq 2^{-q}.
    \label{eq:gc-finite-diag}
\end{equation}
Since $\widetilde H$ is traceless, we may replace $D$ by its traceless
part
\begin{equation}
    D_0
    :=
    D-\frac{\Tr(D)}{d}\II.
\end{equation}
The bounds in \cref{eq:gc-finite-diag} imply
\begin{equation}
    \norm{\widetilde H-\widehat U D_0\widehat U^\dagger}
    =O(2^{-q}).
    \label{eq:gc-traceless-diag}
\end{equation}

We now apply the construction of \cref{lem:Ra-Ra} to the diagonal
entries of $D_0$. After ordering the diagonal entries as in that lemma,
relabel them as $\lambda_1,\ldots,\lambda_d$ and let
\begin{equation}
    s_j=\sum_{k=1}^j\lambda_k,
    \qquad
    \alpha_j=\sqrt{\frac{s_j}{2}},
    \qquad
    1\leq j\leq d-1,
\end{equation}
and define
\begin{equation}
    \Lambda
    =
    \sum_{j=1}^{d-1}
    \alpha_j\ket{e_j}\bra{e_{j+1}}.
\end{equation}
The construction of \cref{lem:Ra-Ra} then gives
\begin{equation}
    A_0=i(\Lambda-\Lambda^\dagger),
    \qquad
    B_0=\Lambda+\Lambda^\dagger,
\end{equation}
with
\begin{equation}
    [A_0,B_0]=iD_0.
\end{equation}
Set
\begin{equation}
    A=\widehat U A_0\widehat U^\dagger
    -\frac{\Tr(\widehat U A_0\widehat U^\dagger)}{d}\II,
    \qquad
    B=\widehat U B_0\widehat U^\dagger
    -\frac{\Tr(\widehat U B_0\widehat U^\dagger)}{d}\II.
\end{equation}
The trace subtractions do not change the commutator. Writing $E=\widehat U^\dagger\widehat U-\II$, we have
\begin{equation}
\begin{aligned}
    [A,B]
    &=\widehat U\bigl(A_0(\widehat U^\dagger\widehat U)B_0
      -B_0(\widehat U^\dagger\widehat U)A_0\bigr)\widehat U^\dagger\\
    &=\widehat U[A_0,B_0]\widehat U^\dagger
      +\widehat U(A_0EB_0-B_0EA_0)\widehat U^\dagger\\
    &=i\widehat U D_0\widehat U^\dagger
      +\widehat U(A_0EB_0-B_0EA_0)\widehat U^\dagger.
\end{aligned}
\end{equation}
Consequently,
\begin{equation}
\begin{aligned}
    \norm{[A,B]-iH}
    &\leq\norm{\widehat U D_0\widehat U^\dagger-H}
      +2\norm{\widehat U}^2\norm{A_0}\norm{B_0}\norm{E}\\
    &\leq\norm{\widehat U D_0\widehat U^\dagger-\widetilde H}
      +\norm{\widetilde H-H}
      +2\norm{\widehat U}^2\norm{A_0}\norm{B_0}\norm{E}\\
    &=O(2^{-q}),
\end{aligned}
\end{equation}
since $\norm{E}\leq2^{-q}$ and
$\norm*{\widehat U},\norm{A_0},\norm{B_0}=O(1)$.
Moreover, $\norm*{\widehat U^\dagger\widehat U-\II}\leq2^{-q}$ implies
$\norm*{\widehat U^{-1}}^2\leq(1-2^{-q})^{-1}$. Using
\cref{eq:gc-traceless-diag,eq:gc-log-approx} and $\norm{H}\leq2\delta$, we obtain
\begin{equation}
    \norm{D_0}
    =\norm{\widehat U^{-1}(\widehat U D_0\widehat U^\dagger)(\widehat U^\dagger)^{-1}}
    \leq\norm{\widehat U^{-1}}^2\norm{\widehat U D_0\widehat U^\dagger}
    \leq\frac{2\delta+O(2^{-q})}{1-2^{-q}}
    =2\delta+O(2^{-q}).
\end{equation}
The trace corrections have norm $O(2^{-q})$, since $A_0,B_0$ are traceless. Thus
$\norm{A},\norm{B}\leq2\sqrt{2\delta}+O(2^{-q/2})$.
For Hermitian $X,Y$, Duhamel's identity gives
\begin{equation}
    e^{iX}-e^{iY}
    =i\int_0^1 e^{i(1-t)X}(X-Y)e^{itY}\,dt,
    \qquad
    \norm{e^{iX}-e^{iY}}\leq\norm{X-Y}.
\end{equation}
Applying this with
$X=H$, $Y=-i[A,B]$, and using \cref{lem:GC-approx}, we obtain
\begin{equation}
\begin{aligned}
    \norm{\Delta-\gcomm{e^{iB}}{e^{iA}}}
    &\leq\norm{\Delta-e^{[A,B]}}
      +\norm{e^{[A,B]}-\gcomm{e^{iB}}{e^{iA}}}\\
    &\leq\norm{iH-[A,B]}+12\max\{\norm{A},\norm{B}\}^3\\
    &\leq272\delta^{3/2}+O(2^{-q/2}).
\end{aligned}
\end{equation}

All quantities above are computed to finite precision. Suppose the
diagonal entries of $D_0$ are known to additive error $O(2^{-q})$.
Then each partial sum $s_j$ is known to additive error
\begin{equation}
    O(d2^{-q}).
\end{equation}
Using
\begin{equation}
    |\sqrt{x}-\sqrt{y}|
    \leq \sqrt{|x-y|}
    \qquad
    (x,y\geq0),
\end{equation}
the corresponding coefficients $\alpha_j$ can be computed with error
\begin{equation}
    O\left(\sqrt{d}\,2^{-q/2}\right).
\end{equation}
Consequently,
\begin{equation}
    \norm{\Lambda-\Lambda^{(q)}}
    =
    O\left(\sqrt{d}\,2^{-q/2}\right),
\end{equation}
and hence the finite-precision matrices $A^{(q)}$ and $B^{(q)}$,
made Hermitian and traceless after rounding, satisfy
\begin{equation}
    \norm{A-A^{(q)}},
    \norm{B-B^{(q)}}
    =
    O\left(\poly(d)2^{-q/2}\right).
    \label{eq:gc-AB-rounding}
\end{equation}
Choosing
\begin{equation}
    q=2b+O(\log d)
\end{equation}
therefore makes both errors in \cref{eq:gc-AB-rounding} at most
$O(2^{-b})$.

Finally, set
\begin{equation}
    V^{(b)}:=e^{iB^{(q)}},
    \qquad
    W^{(b)}:=e^{iA^{(q)}}.
\end{equation}
These belong to $\mathrm{SU}(d)$ because $A^{(q)},B^{(q)}$ are
traceless Hermitian. By Duhamel's identity and
\cref{eq:gc-AB-rounding}, replacing $e^{iB},e^{iA}$ by
$V^{(b)},W^{(b)}$ changes the preceding norm and commutator
bounds by at most $O(\poly(d)2^{-q/2})$.
Choosing $q=2b+O(\log d)$ sufficiently large therefore gives
\cref{eq:gc-numerical-factor-bounds,eq:gc-numerical-commutator-bound}.
The algorithm returns their entries to $b$ bits of precision,
with operator-norm error $O(d2^{-b})$.

The dominant step is the finite-precision diagonalization of
$\widetilde H$, which by
\cref{thm:finite-precision-diagonalization} takes
$d^4\poly(q+\log d)
    =
    d^4\poly(b+\log d)$
bit operations. All remaining matrix multiplications, square-root
computations, and matrix exponentials require at most
$d^3\poly(b+\log d)$
bit operations. Thus the entire construction can be performed in
\begin{equation}
    d^4\poly(b+\log d)
\end{equation}
bit operations.

\end{proof}

Let 
    \be R_\delta = \{ X \in \mathbb{C}^{d\times d} : \Tr (X) = 0, X^\dagger = X, \quad \|X\|_\infty \leq \delta\}
    \ee 
    be the ball of radius $\delta$ with respect to operator norm for Hermitian $d \times d$ complex matrices.
 For sets of operators $A,B$, we have defined $[A , B] = \{[x,y] : x \in A , y \in B\}$.

\begin{lem}[Problem 8.15 of \cite{kitaev2002classical}]
\label{lem:Ra-Ra}
    \be
    i R_{a^2/4} \subseteq   [R_a, R_a].
    \ee
\end{lem}

\begin{proof}
    Let $H$ be a Hermitian matrix with $\norm{H} = h$. We will produce traceless Hermitian matrices $A, B \in R_{2 \sqrt{h}}$ such that $iH = [A, B]$. Consider the decomposition
    \be
    H = Q \diag (\lambda_1, \ldots, \lambda_d) Q^\dagger
    \ee
    and let $\ket {j}$ be the eigenbasis of $H$. Since $\Tr (H) = 0$, we can always order the eigenvalues in such a way that
    \be
    s_k := \sum_{j = 1}^{k} \lambda_j \in [0 , 2h], \quad \quad s_0 = s_d = 0.
    \ee
    Let
    \be
    T := \sum_{j =1}^{d-1} \sqrt{\frac{s_j}{2}} \ket{j}\bra{j+1}, \quad  A = i(T-T^\dagger),\quad B = T +T^\dagger.
    \ee
    We can show
    \be
    [T , T^\dagger] = 1/2 \diag (\lambda_1, \ldots, \lambda_d), \quad [A , B] = i H.
    \ee
    Now observe that $\norm{A}, \norm{B} \leq 2\sqrt{h}$.
\end{proof}

\begin{lem}[Finding a traceless operator logarithm]
    Let $\Delta$ be a $d\times d$ unitary matrix with unit determinant, and let
    $0\leq \delta\leq 1/d$. Suppose
    \be
        \norm{\Delta-\II}\leq \delta.
    \ee
    Then there exists a traceless Hermitian matrix $H$ such that
    \be
        \Delta=e^{iH},
        \qquad
        \norm{H}\leq 2\delta.
    \ee

    Furthermore, for every $m,p\geq 1$, one can deterministically compute a traceless
Hermitian matrix $\widetilde H\in\mathbb Q(i)^{d\times d}$, whose entries
are specified to $p$ bits of precision, such that
\begin{equation}
    \norm{H-\widetilde H}
    =
    O\left(\frac{1}{d^m}+\frac{d}{2^p}\right).
\end{equation}
The computation can be performed using
\begin{equation}
    m d^3 \poly(p+\log d+\log m)
\end{equation}
bit operations.
    \label{lem:log-Delta-branch}
\end{lem}

\begin{proof}
    Let $e^{i\theta_1},\ldots,e^{i\theta_d}$ be the eigenvalues of $\Delta$.
    Since $\norm{\Delta-\II}\leq\delta$,
    we have
    \be
        |e^{i\theta_j}-1|\leq\delta
    \ee
    for every $j$. Since $\delta\leq 1$, for each $j$ there exists
    $\delta_j\in[-\pi,\pi]$ such that
    \be
        e^{i\theta_j}=e^{i\delta_j},
        \qquad
        |\delta_j|\leq 2\delta.
    \ee
    Since $\det(\Delta)=1$,
    \be
        1=\prod_{j=1}^d e^{i\delta_j}
        =e^{i\sum_{j=1}^d\delta_j}.
    \ee
    Hence $\sum_{j=1}^d\delta_j=2\pi n$
    for some $n\in\mathbb Z$. On the other hand,
    \be
        \left|\sum_{j=1}^d\delta_j\right|
        \leq\sum_{j=1}^d|\delta_j|
        \leq 2d\delta
        \leq 2<2\pi.
    \ee
    Therefore $n=0$, and thus $\sum_{j=1}^d\delta_j=0$.
Writing
    \be
        \Delta
        =U\,\diag(e^{i\delta_1},\ldots,e^{i\delta_d})\,U^\dagger,
    \ee
    define
    \be
        H
        =U\,\diag(\delta_1,\ldots,\delta_d)\,U^\dagger.
    \ee
    Then $H$ is Hermitian,
    \be
        \Delta=e^{iH},
        \qquad
        \Tr(H)=\sum_{j=1}^d\delta_j=0,
    \ee
    and
    \be
        \norm{H}
        =\max_j|\delta_j|
        \leq 2\delta.
    \ee

   Next, we prove the second part of the lemma about the computational complexity of computing $\widetilde H$. Write
\begin{equation}
    \Delta=\II+E,
    \qquad
    \norm{E}\leq\delta<1.
\end{equation}
Since the power series for the logarithm converges absolutely for
$\norm{E}<1$,
\begin{equation}
    iH
    =
    \log(\II+E)
    =
    \sum_{k=1}^{\infty}
    \frac{(-1)^{k+1}}{k}E^k.
\end{equation}
Define the truncated logarithm
\begin{equation}
    iH_m
    :=
    \sum_{k=1}^{m}
    \frac{(-1)^{k+1}}{k}E^k.
\end{equation}
Then
\begin{align}
    \norm{H-H_m}
    &\leq
    \sum_{k=m+1}^{\infty}\frac{\delta^k}{k}
    \notag\\
    &\leq
    \frac{\delta^{m+1}}{(m+1)(1-\delta)}
    \leq
    \frac{1}{d^m}.
    \label{eq:truncated-log-error}
\end{align}

We next compute $H_m$ using finite-precision arithmetic. Let
\begin{equation}
    w=p+O(\log d+\log m)
\end{equation}
be the working precision. Performing all scalar arithmetic with $w$ bits
of precision and rounding after each arithmetic operation gives a matrix
$H_{p,m}\in\mathbb Q(i)^{d\times d}$ satisfying (see \Cref{lem:b-precision}):
\begin{equation}
    \norm{H_m-H_{p,m}}
    =
    O\left(\frac{d}{2^p}\right).
    \label{eq:log-rounding-error}
\end{equation}
Indeed, $H_m$ can be evaluated iteratively by computing
$E,E^2,\ldots,E^m$. Since $\norm{E}<1$, all these matrices have operator
norm at most $1$. There are $O(md^3)$ scalar arithmetic operations, and
choosing $O(\log m+\log d)$ guard bits ensures that the accumulated
rounding error is $O(d2^{-p})$.

We now enforce Hermiticity. Define
\begin{equation}
    \widetilde H_{p,m}
    :=
    \frac{H_{p,m}+H_{p,m}^\dagger}{2}.
\end{equation}
Since $H$ is Hermitian,
\begin{align}
    \norm{H-\widetilde H_{p,m}}
    &=
    \frac12
    \norm{
        (H-H_{p,m})
        +(H-H_{p,m})^\dagger
    }
    \notag\\
    &\leq
    \norm{H-H_{p,m}}
    \notag\\
    &\leq
    \norm{H-H_m}
    +
    \norm{H_m-H_{p,m}}
    \notag\\
    &=
    O\left(
        \frac{1}{d^m}
        +
        \frac{d}{2^p}
    \right).
    \label{eq:hermitianized-log-error}
\end{align}

Finally, we enforce tracelessness by setting
\begin{equation}
    \widetilde H
    :=
    \widetilde H_{p,m}
    -
    \frac{\Tr(\widetilde H_{p,m})}{d}\II.
\end{equation}

Since $\Tr (H) = 0$, it follows that
\begin{equation}
    \Tr(\widetilde H_{p,m})
    =
    \Tr(\widetilde H_{p,m}-H).
\end{equation}
Using $|\Tr(A)|\leq d\norm{A}$, we obtain
\begin{align}
    \norm{\widetilde H_{p,m}-\widetilde H}
    &=
    \frac{|\Tr(\widetilde H_{p,m})|}{d}
    \notag\\
    &\leq
    \norm{\widetilde H_{p,m}-H}.
\end{align}
Therefore,
\begin{align}
    \norm{H-\widetilde H}
    &\leq
    \norm{H-\widetilde H_{p,m}}
    +
    \norm{\widetilde H_{p,m}-\widetilde H}
    \notag\\
    &\leq
    2\norm{H-\widetilde H_{p,m}}
    \notag\\
    &=
    O\left(
        \frac{1}{d^m}
        +
        \frac{d}{2^p}
    \right).
\end{align}

Computing the truncated logarithm requires $m-1$ matrix
multiplications, each costing $O(d^3)$ scalar arithmetic operations.
At working precision
\begin{equation}
    w=p+O(\log d+\log m),
\end{equation}
each scalar operation requires $\poly(w)$ bit operations.
Hermitianization and trace subtraction require only $O(d^2)$ additional
scalar operations. Hence the total running time is
\begin{equation}
    m d^3\poly(p+\log d+\log m).
\end{equation}
\end{proof}

\begin{lem}
    Suppose $M$ is a $d \times d$ matrix and let $M_b$ be $M$ with each entry specified using $b$ bits of precision (in the binary representation). Then
    \be
    \norm{M - M_b} \leq \frac{d}{2^b}.
    \ee
    \label{lem:b-precision}
\end{lem}

\begin{proof}
    Let $E = M - M_b$. Then
    \be
    \norm{E} \leq \norm{E}_{\mathrm F} \leq \sqrt{\sum_{j,k} |E_{j,k}|^2} \leq d/2^b.
    \ee
\end{proof}

\subsection{Compiling general unitaries in Fock space}
\label{sec:compilation}

The Solovay--Kitaev theorem (\cref{thm:cutoff-SK}) derived in the previous section allows us to efficiently compile energy-preserving unitary operations acting on a constant number of modes into a sequence of linear-optical and Kerr/cross-Kerr gates. Combined with parametric approximation (\cref{thm:passive-compilation}), it also allows us to efficiently compile unitary gates acting on a constant number of modes generated by polynomial Hamiltonians of constant degree, under the energy-moment assumptions of that theorem.

In this section, we extend this result to more general unitary operations that act on more than a constant number of modes and are not necessarily energy-preserving. Concretely, we consider any unitary on Fock space that admits an efficient qubit quantum circuit description. To define this class of operations formally, we denote by
\be
    \mathcal H_{m,K}\coloneqq\mathrm{span}\{\ket{k_1\dots k_m}\,|\,0\le k_j\le K\}
\ee
the Hilbert space of states with at most $K$ photons per mode, and we introduce for $q\ge mK$:
\be
    \mathcal E:\;\ket{k_1\dots k_m}\mapsto\bigotimes_{j=1}^m\left(\ket1^{\otimes k_j}\otimes\ket0^{\otimes(K-k_j)}\right)\otimes\ket0^{\otimes q-mK}.
\ee
This is an isometry from $\mathcal H_{m,K}$ to the space of $q$ qubits, which flattens Fock states into a unary encoding of their photon number. We now define the qubit quantum circuit description of a unitary operation over Fock space:

\begin{definition}[Qubit circuit description of Fock unitary operations]\label{def:circdesc}
    Given a unitary $U$ over Fock space preserving $\mathcal H_{m,K}$ and $\delta>0$, we say that $U$ has a $\delta$-approximate circuit description on $\mathcal H_{m,K}$ with $L$ gates over $q\ge mK$ qubits if there exists a unitary qubit circuit $V$ of $L$ single-qubit gates and $CZ$ gates over $q$ qubits such that
    \be
    \sup_{\substack{\ket\psi\in\mathcal H_{m,K}\\\norm{\ket\psi}=1}}\|(V\mathcal E-\mathcal EU)\ket\psi\|\le\delta.
    \ee
\end{definition}

With this definition in place, we show in \cref{thm:compile_new_passive} how to efficiently compile number-preserving unitaries admitting an efficient qubit circuit description. Compared to \cref{thm:cutoff-SK}, this includes unitaries such as those generated by constant-degree number-preserving polynomial Hamiltonians that preserve $\mathcal H_{m,K}$. Then, we extend this result to non-number-preserving unitaries in \cref{thm:compile_new}.

\subsubsection{Encoding gadgets}

Given a target unitary $U$, the proofs of \cref{thm:compile_new_passive,thm:compile_new} involve the following steps: mapping Fock states to a unary encoding in which $U$ has an efficient qubit circuit description, then mapping that unary encoding to a dual-rail encoding where the circuit can be efficiently compiled using gates from $\GKerr$ or $\GcrossKerr$, and finally decoding back from dual-rail to Fock space.

To that end, we first introduce energy-preserving unitary compilation gadgets for the Fock-to-unary encoding $\mathcal E$ (\cref{lem:encodingfu}), as well as for two unary-to-dual-rail encodings (\cref{lem:encodingudr}) defined as
\be\label{eq:Fisometry}
    \mathcal F:\;\ket{\bm x}\mapsto\bigotimes_{j=1}^q(\ket{1-x_j}\otimes\ket{x_j}),
\ee
and
\be\label{eq:Jisometry}
    \mathcal J:\;\ket{\bm x}\mapsto\mathcal F\ket{\bm x}\otimes\ket{|\bm x|},
\ee
for all $\bm x\in\{0,1\}^q$, where $|\bm x|$ denotes the Hamming weight of $\bm x$.
These gadgets may be of independent interest, since they provide simple tools to import qubit-based results to the bosonic setting.

For the Fock-to-unary encoding:

\begin{lem}[Efficient Fock-to-unary encoding]\label{lem:encodingfu}
    Let $0<\eta<1$ and $q\ge mK$. There exists a bosonic unitary circuit $C_{\mathcal E}$ such that for all unit $\ket\psi\in\mathcal H_{m,K}$,
    \be
        \|C_{\mathcal E}(\ket\psi\otimes\ket0^{\otimes(q+1)})-\mathcal E\ket\psi\otimes \ket0^{\otimes(m+1)}\|\le\frac\eta2.
    \ee
    Moreover, $C_{\mathcal E}$ has $\poly(q,\mathrm{log}(\frac1\eta))$ gates from $\GKerr$ or $\GcrossKerr$.
\end{lem}

\begin{proof}[Proof of \cref{lem:encodingfu}]
Let us define the 2-mode degree-4 energy-preserving Hamiltonian:
\be\label{eq:Hspread}
R_{kl}\coloneqq i\left[\hat a_k\hat a_l^\dag(1-\hat n_l)-(1-\hat n_l)\hat a_l\hat a_k^\dag\right].
\ee
For all $n\ge1$, we have:
\begin{align}
    R_{kl}\ket n_k\otimes\ket0_l&=i\sqrt n\ket{n-1}_k\otimes\ket1_l,\\
    R_{kl}\ket{n-1}_k\otimes\ket1_l&=-i\sqrt n\ket n_k\otimes\ket0_l,
\end{align}
so that
\be\label{eq:Hspread2}
    e^{-itR_{kl}}\ket n_k\otimes\ket0_l=\cos(t\sqrt n)\ket n_k\ket0_l+\sin(t\sqrt n)\ket{n-1}_k\otimes\ket1_l.
\ee
We now apply Grover's fixed-point quantum search algorithm~\cite{grover2005fixed} in order to amplify the amplitude of the state $\ket{n-1}\otimes\ket1$: let $T_0=e^{-i\frac\pi{4\sqrt K}R_{kl}}$ and define recursively
\be
T_{s}=T_{s-1}(e^{-i\frac\pi3\hat n_l})T_{s-1}^\dag(e^{i\frac\pi3\hat n_l})T_{s-1},
\ee
so that $T_s$ has $2\cdot3^s-1$ elementary energy-preserving gates. For $s=0$ and all $n\ge1$, we have
\be
    T_0\ket n\otimes\ket0=\cos\left(\frac\pi4\sqrt{\frac nK}\right)\ket n\otimes\ket0+\sin\left(\frac\pi4\sqrt{\frac nK}\right)\ket{n-1}\otimes\ket1,
\ee
and since $\sin x\ge\frac2\pi x$ for $x\in[0,\frac\pi2]$, the probability for $\ket{n-1,1}$ satisfies
\be
    \sin^2\!\left(\frac\pi4\sqrt{\frac nK}\right)\ge\frac n{4K}\;,
\ee
for $1\le n\le K$. Following~\cite{grover2005fixed}, the application of $T_s$ amplifies this probability exponentially and we obtain
\be
    T_s\ket n\otimes\ket0=\alpha_n(s)\ket n\otimes\ket0+\beta_n(s)\ket{n-1}\otimes\ket1,
\ee
for all $n\ge1$, where $|\alpha_n(s)|^2+|\beta_n(s)|^2=1$ and $|\alpha_n(s)|^2\le e^{-3^sn/(4K)}$ for all $1\le n\le K$. We also have $T_s\ket{0,0}=\ket 0\otimes\ket0$ since $R_{kl}\ket 0\otimes\ket0=0$.

Up to an exponentially small error in $s$ and a global phase, the sequence $T_s$ transfers exactly one photon from a given Fock state to an auxiliary vacuum mode. Hence, given an input Fock state $\ket n$ with $1\le n\le K$, repeatedly applying $T_s$ using $K$ different auxiliary vacuum modes and correcting the phase by applying
\be
    \prod_{j=1}^Ke^{-i\arg(\beta_j(s))\hat n_j}
\ee
allows us to approximately map $\ket n\otimes\ket0^{\otimes K}$ to $\ket0\otimes\ket1^{\otimes n}\otimes\ket0^{\otimes(K-n)}$. In doing so, the total number of elementary gates is given by $K(2\cdot3^s-1)+K=2K\cdot3^s$, and the total error satisfies
\be
\sqrt2\sum_{k=1}^ne^{-\frac{3^sk}{8K}}\le\frac{\sqrt2}{e^{\frac{3^s}{8K}}-1}.
\ee
To reach a precision of $\eta/{3m}$, we thus take $s$ satisfying
\be
    3^s=\Theta\left(K\log(\frac{m}\eta)\!\right).
\ee
Repeating this construction across the $m$ input modes and recompiling all gates reusing the same vacuum ancilla mode via $\GKerr$ or $\GcrossKerr$ using \cref{thm:cutoff-SK} completes the proof.
\end{proof}

For the unary-to-dual-rail encodings, we introduce two different constructions. The first one is used in the proof of \cref{thm:compile_new_passive} and implements the isometry $\mathcal J$ (\cref{eq:Jisometry}) using an auxiliary Fock state, while the second one, used in the proof of \cref{thm:compile_new}, implements the isometry $\mathcal F$ (\cref{eq:Fisometry}) using an auxiliary coherent state.

\begin{lem}[Efficient unary-to-dual-rail encodings]\label{lem:encodingudr}
    Let $0<\eta<1$ and $q\ge1$.
    There exists a bosonic unitary circuit $C_{\mathcal J}$ with $\poly(q,\mathrm{log}(\frac1\eta))$ gates from $\GKerr$ or $\GcrossKerr$ such that for all unit $\ket\phi\in\mathcal H_{q,1}$,
    \be\label{eq:F1}
        \|C_{\mathcal J}(\ket\phi\otimes\ket q\otimes\ket0^{\otimes(q+1)})-(\mathcal J\ket\phi)\otimes\ket0\|\le\frac\eta2,
    \ee
    where $\ket q$ denotes a Fock state with $q$ photons.
    
    Moreover, for every $\alpha\ge1$, there exists a bosonic unitary circuit $C_{\mathcal F}$ with $\poly(q,\mathrm{log}(\frac1\eta))$ gates generated by local constant-degree Hamiltonians such that for all unit $\ket\phi\in\mathcal H_{q,1}$,
    \be\label{eq:F2}
        \|C_{\mathcal F}(\ket\phi\otimes\ket\alpha\otimes\ket0^{\otimes(q+1)})-(\mathcal F\ket\phi)\otimes\ket\alpha\otimes\ket0\|\le\frac{q}{\alpha}+\frac\eta2.
    \ee
\end{lem}

\begin{proof}[Proof of \cref{lem:encodingudr}]

We first prove \cref{eq:F1}.
On a unary mode, an ancillary mode, and a ``reservoir'' mode containing the ancillary Fock state, define a number-preserving unitary $J$ by
\be
    \begin{aligned}
    &J\ket0\otimes\ket0\otimes\ket r=\ket0\otimes\ket1\otimes\ket{r-1},\\
    &J\ket0\otimes\ket1\otimes\ket{r-1}=\ket0\otimes\ket0\otimes\ket r,
    \end{aligned}
\ee
for all $1\le r\le q$, and acting as the identity on all other Fock basis states. This unitary preserves photon number and maps the vacuum to itself.
We apply $J$ successively to each unary mode and a new ancillary vacuum mode, always using the same reservoir mode initially in $\ket q$. Each $\ket0$ state in the unary mode transfers one photon from the reservoir mode to the ancillary mode, while each $\ket1$ state leaves the ancillary mode in the vacuum. On input $\ket{\bm x}$ in the unary modes, the reservoir mode thus ends in the state $\ket{q-(q-|\bm x|)}=\ket{|\bm x|}$, and, up to reordering, the ancillary modes end up in the state $\mathcal F\ket{\bm x}$.
We then compile each application of $J$ via gates from $\GKerr$ or $\GcrossKerr$, using \cref{thm:cutoff-SK} with a total-photon cutoff $2q$ and an error $\frac\eta{2q}$ and reusing the same vacuum ancilla mode, leading to a total error of at most $\frac\eta2$. The Fock basis matrix of $J$ is a permutation matrix of polynomial dimension, so each compilation costs $\poly(q,\log(q/\eta))$, and the final reordering of the modes costs $O(q)$ linear-optical gates.

To prove \cref{eq:F2}, we first apply a parametric approximation to the Hamiltonian $H_Y=i[\hat a^\dag(1-\hat n)-(1-\hat n)\hat a]$, which implements a Pauli $Y$ gate on $\mathrm{span}\{\ket0,\ket1\}$. The corresponding $2$-mode energy-preserving Hamiltonian is given by $\frac1\alpha R_{kl}$, with the $2$-mode degree-$4$ Hamiltonian $R_{kl}$ defined in \cref{eq:Hspread}. Using \cref{thm:passive-compilation} with $\ket\psi=\ket0$ we thus obtain $\|e^{-i\frac\pi{2\alpha}R_{kl}}\ket\alpha_k\otimes\ket0_l-\ket\alpha_k\otimes\ket1_l\|\le\frac C{\alpha}$, for some constant $C$. To get a tighter bound, we write $\ket{\psi(s)}=e^{-isH_Y}\ket0$ and we have
\be
\left\|\left(\frac1\alpha R_{kl}-\II\otimes H_Y\right)\ket\alpha_k\otimes\ket{\psi(s)}_l\right\|=\frac{|\sin s|}\alpha,
\ee
so a direct use of Duhamel's formula gives
\be
\label{eq:para0to1}
\left\|e^{-i\frac\pi{2\alpha}R_{kl}}\ket\alpha_k\otimes\ket0_l-\ket\alpha_k\otimes\ket1_l\right\|\le\frac 1\alpha.
\ee

To map unary data to dual-rail data, we need to do the following: on input $\ket1$, we append a register in state $\ket0$, while on input $\ket0$, we append a register in state $\ket1$. By \cref{eq:para0to1}, this can be achieved approximately by appending a register in state $\ket\alpha_0\otimes\ket0_k$ and applying $CR=e^{-i\frac{\pi}{2\alpha}(1-\hat n_j)R_{0k}}$, where $j$ denotes the unary input register and $j,k$ the two output registers. Repeating this operation for $j=1\dots q$ using the same coherent state in register $0$ and different auxiliary output registers for $k=q+j$, we obtain an approximate implementation of the isometry $\mathcal F$, with a total error of $O(\frac q{|\alpha|})$ and using $q$ gates. The final reordering of the modes costs $O(q)$ linear-optical gates, and all gates are generated by local energy-preserving polynomial Hamiltonians of
constant degree.
\end{proof}

\subsubsection{Compiling number-preserving unitaries}

With these encoding results in place, we now consider the compilation of number-preserving unitaries with an efficient qubit circuit description:

\begin{theorem}[Efficient compilation of number-preserving unitaries]\label{thm:compile_new_passive}
Let $0<\varepsilon<1/2$. Let $U$ be a number-preserving unitary over $\mathcal H_{m,K}$, and suppose that $U$ has a $\delta$-approximate qubit circuit description with $L$ gates over $q\ge mK$ qubits as in \cref{def:circdesc}, with $\delta\le\varepsilon/4$. Then, there exist a bosonic unitary circuit $C_U$ and $\alpha\in\mathbb R$ such that
\be
    \sup_{\substack{\ket\psi\in\mathcal H_{m,K}\\\norm{\ket\psi}=1}}\|C_U(\ket\psi\otimes\ket\alpha\otimes\ket0^{\otimes(2q+2)})-(U\ket\psi)\otimes\ket\alpha\otimes\ket0^{\otimes(2q+2)}\|\le\varepsilon.
\ee
Moreover, $C_U$ has $\poly(m,K,L,\log(1/\varepsilon))$ gates from $\GKerr$ or $\GcrossKerr$ and the coherent state $\ket\alpha$ has energy $O(mK+L +\log(1/\varepsilon))$.
\end{theorem}

\noindent \cref{thm:compile_new_passive} allows us to compile more general number-preserving unitaries compared to \cref{thm:cutoff-SK}. For instance, the unitary dynamics generated by the Bose--Hubbard Hamiltonian at total photon number $K$ (extended by the identity to an energy-preserving unitary on $\mathcal H_{m,mK}$) can be compiled to exponential precision with a polynomial-size circuit over $\GKerr$ or $\GcrossKerr$, since the corresponding unitary has an efficient qubit circuit description via Hamiltonian simulation \cite{berry2015hamiltonian}.

\begin{proof}[Proof of \cref{thm:compile_new_passive}]
We first assume that we can use an ancillary mode with a Fock state $\ket q$. \cref{prop:prepare-Fock-states} below shows that this is without loss of generality.

Let $V_1$ denote the $\delta$-approximate qubit circuit description of $U$, with $L$ gates over $q$ qubits, which by \cref{def:circdesc} satisfies
\be\label{eq:UtoV1passive}
\|V_1\mathcal E-\mathcal EU\|\le\delta.
\ee
Since each of the $L$ gates of $V_1$ acts on at most $2$ qubits, we have $q\le 2L+mK$ without loss of generality, by removing the qubits on which $V_1$ acts as the identity.

Setting $C_\mathrm{enc}=C_\mathcal JC_\mathcal E$ from \cref{lem:encodingfu,lem:encodingudr} up to identity operators and a reordering of the ancillary registers, we obtain for all $0<\eta<1$, $q\ge mK$, and unit $\ket\psi\in\mathcal H_{m,K}$,
\be\label{eq:FEtoCencpassive}
    \|C_\mathrm{enc}(\ket\psi\otimes\ket q\otimes\ket0^{\otimes(2q+1)})-(\mathcal J\mathcal E\ket\psi)\otimes\ket0^{\otimes(m+1)}\|\le\eta,
\ee
and $C_\mathrm{enc}$ has $\poly(m,K,L,\logarithm(\frac1\eta))$ gates from $\GKerr$ or $\GcrossKerr$.

The isometry $\mathcal F$ (\cref{eq:Fisometry}) defines a logical unitary circuit $V_2\mathcal F=\mathcal FV_1$, where $V_2$ acts on dual-rail qubits instead of unary qubits and can be reproduced exactly by linear-optical gates (for single-qubit unitaries) and cross-Kerr gates (for $CZ$ gates). Let $W$ denote the corresponding bosonic circuit. By construction, on the dual-rail subspace,
\be\label{eq:V1toWpassive}
    W\mathcal F=V_2\mathcal F=\mathcal FV_1,
\ee
and $W$ has $O(L)$ gates.

For all $\ket\psi\in\mathcal H_{m,K}$, we let
\be
\mathcal A:\;\ket\psi\mapsto\ket\psi\otimes\ket q\otimes\ket0^{\otimes(2q+1)}
\ee
be the isometry which appends the ancillary states to the input $\ket\psi$. We have (omitting identity operators on the $\ket q$ register for brevity)
\be\label{eq:comberrpassive}
C_\mathrm{enc}^\dag WC_\mathrm{enc}\mathcal A-\mathcal AU=C_\mathrm{enc}^\dag[W(C_\mathrm{enc}\mathcal A-\mathcal J\mathcal E)+(W\mathcal J\mathcal E-\mathcal J\mathcal EU)+(\mathcal J\mathcal E-C_\mathrm{enc}\mathcal A)U].
\ee
Since $C_\mathrm{enc}^\dag$, $W$ and $U$ are unitary, the first and last terms are bounded by $\eta$ in operator norm by \cref{eq:FEtoCencpassive}. Moreover, for all $\ket\psi\in\mathcal H_{m,K}$ we have $\ket\psi=\sum_N\Pi_N\ket\psi$, where $\Pi_N$ is the projector onto the subspace of states with total photon number $N$, so
\begin{align}
    \nonumber\mathcal J\mathcal E\ket\psi&=\sum_N\mathcal J\mathcal E\Pi_N\ket\psi\\
    &=\sum_N\mathcal F\mathcal E\Pi_N\ket\psi\otimes\ket N.\label{eq:JE}
\end{align}
Hence, the second term in \cref{eq:comberrpassive} satisfies, for all unit $\ket\psi\in\mathcal H_{m,K}$
\begin{align}
\|(W\mathcal J\mathcal E-\mathcal J\mathcal EU)\ket\psi\|^2&=\sum_N\|(W\mathcal J\mathcal E-\mathcal J\mathcal EU)\Pi_N\ket\psi\|^2\notag\\
&=\sum_N\|W\mathcal J\mathcal E\Pi_N\ket\psi-\mathcal J\mathcal E\Pi_NU\ket\psi\|^2\notag\\
&=\sum_N\|W\mathcal F\mathcal E\Pi_N\ket\psi\otimes\ket N-\mathcal F\mathcal E\Pi_NU\ket\psi\otimes\ket N\|^2\notag\\
&=\sum_N\|(W\mathcal F\mathcal E-\mathcal F\mathcal EU)\Pi_N\ket\psi\otimes\ket N\|^2\notag\\
&=\sum_N\|\mathcal F(V_1\mathcal E-\mathcal EU)\Pi_N\ket\psi\|^2\notag\\
&\le\delta^2\sum_N\|\Pi_N\ket\psi\|^2\notag\\
&=\delta^2,
\end{align}
where we used the orthogonality of the photon-number sectors in the first and last lines, the fact that $U$ is energy-preserving in the second and fourth lines, \cref{eq:JE} in the third line, \cref{eq:V1toWpassive} in the fifth line, the fact that $\mathcal F$ is an isometry and \cref{eq:UtoV1passive} in the sixth line. Summing the error terms with the triangle inequality, we obtain
\be\label{eq:compilerr}
\|C_\mathrm{enc}^\dag WC_\mathrm{enc}\mathcal A-\mathcal AU\|\le\delta+2\eta.
\ee
By construction, $C_\mathrm{enc}^\dag WC_\mathrm{enc}$ is a bosonic circuit with local energy-preserving gates of constant degree. Moreover, $C_\mathrm{enc}$ has $\poly(m,K,L,\log(1/\eta))$ gates, with $q\le 2L+mK$, while $W$ has $O(L)$ gates, so $C_\mathrm{enc}^\dag WC_\mathrm{enc}$ also has $\poly(m,K,L,\log(1/\eta))$ gates. Finally, we may compile all the gates with \cref{thm:cutoff-SK} using gates from $\GKerr$ or $\GcrossKerr$.

To conclude the proof, we show that the ancillary Fock state $\ket q$ may be efficiently prepared to precision $\eta$ using a coherent state with energy $O(q+\log(1/\eta))$:

\begin{proposition}\label{prop:prepare-Fock-states}
For any $q\in\mathbb N$ and $0<\eta<1$, there exist $\alpha\in\mathbb R$, a normalized single-mode state $\ket \chi$, and a three-mode energy-preserving unitary circuit $P$ such that
\begin{align}
\norm{P(\ket0\otimes\ket\alpha\otimes\ket0)
-\ket q\otimes\ket\chi\otimes\ket0}&\leq\eta.
\label{eq:preparation}
\end{align}
Moreover, $P$ has $\poly(q,\log(1/\eta))$ gates from $\mathcal G_{\mathrm{Kerr}}$ or $\mathcal G_{\mathrm{Kerr}}^{\times}$, and the coherent state has energy $|\alpha|^2=O(q+\log(1/\eta))$.
Discarding the last two modes yields a state $\rho$ satisfying
\begin{align}
\frac12\norm{\rho-\ket q\bra q}_1&\leq\eta.
\label{eq:reduced}
\end{align}
\end{proposition}

\begin{proof}
The case $q=0$ is immediate by taking $\alpha=0$, $\ket \chi=\ket0$, and $P=\II$.
Assume $q\geq1$.
The idea is to choose a coherent ancilla whose photon number lies almost entirely between $q$ and a polynomial cutoff $K$, and to move $q$ photons to the output mode.

We pick $\alpha=2\sqrt{q+\log(8/\eta)}$ and $K=\lceil4|\alpha|^2\rceil$, and we define the set $\mathcal S=\{q,\ldots,K\}$ and the two-mode unitary $Q$ which swaps the states
\begin{align}
\ket{0,j}&\longleftrightarrow\ket{q,j-q},
\label{eq:unitary}
\end{align}
for all $j\in \mathcal S$, and leaves invariant all the other Fock basis states.
These swaps act in distinct photon-number sectors, so $Q$ is energy-preserving and thus maps the vacuum state to itself.
Write $\Pi_{\mathcal S}=\sum_{j\in \mathcal S}\ket j\bra j$, and $p=\norm{(\II-\Pi_\mathcal S)\ket\alpha}^2$,
and define the normalized state
\begin{align}
\ket \chi&\coloneqq
\frac{e^{-|\alpha|^2/2}}{\sqrt{1-p}}
\sum_{j=q}^{K}\frac{\alpha^j}{\sqrt{j!}}\ket{j-q}.
\label{eq:residual}
\end{align}
The unitary $Q$ maps $\ket0\otimes\Pi_{\mathcal S}\ket\alpha$ to $\sqrt{1-p}\ket q\otimes\ket \chi$, whereas the remaining part of the state is orthogonal to $\ket q\otimes\ket \chi$.
Hence,
\begin{align}
\norm{Q(\ket0\otimes\ket\alpha)-\ket q\otimes\ket \chi}^2=2-2\sqrt{1-p}\leq2p.
\label{eq:ideal-error}
\end{align}
It remains to bound $p$: for $N\sim\operatorname{Poisson}(|\alpha|^2)$, the coherent-state expansion gives
\begin{align}
p&=\sum_{j\notin \mathcal S}e^{-|\alpha|^2}\frac{|\alpha|^{2j}}{j!}
=\Pr[N<q]+\Pr[N>K].
\label{eq:poisson}
\end{align}
Since $\mathbb E[2^N]=e^{|\alpha|^2}$ and $\mathbb E[2^{-N}]=e^{-|\alpha|^2/2}$, Markov's inequality implies
\begin{align}
\Pr[N<q]&\leq2^q e^{-|\alpha|^2/2}\leq\frac{\eta^2}{64},\\
\Pr[N>K]&\leq2^{-K}e^{|\alpha|^2}\leq\frac{\eta^2}{64}.
\label{eq:tails}
\end{align}
Thus $p\leq\eta^2/32$, and \cref{eq:ideal-error} bounds the ideal preparation error by $\eta/4$.

We compile the unitary $Q$ using \cref{thm:cutoff-SK} with cutoff $K=O(q+\log(1/\eta))$ and error $\eta/2$, using one additional vacuum ancilla, and denote the resulting circuit by $P$.
Separating the state into two components based on its photon number, we obtain
\begin{align}
\norm{\bigl(P-Q\otimes\II\bigr)
(\ket0\otimes\ket\alpha\otimes\ket0)}
&\leq\frac\eta2+2\sqrt{\Pr[N>K]}.
\label{eq:compilation}
\end{align}
Combining \cref{eq:ideal-error,eq:tails,eq:compilation} gives
\begin{align}
\norm{P(\ket0\otimes\ket\alpha\otimes\ket0)
-\ket q\otimes\ket \chi\otimes\ket0}
&\leq\sqrt{2p}+\frac\eta2+2\sqrt{\Pr[N>K]}\nonumber\\
&\leq\frac\eta4+\frac\eta2+\frac\eta4=\eta.
\label{eq:total-error}
\end{align}
Finally, the trace distance between pure states is bounded by their $2$-norm distance, and partial trace cannot increase the trace distance, completing the proof of \cref{eq:reduced}.
\end{proof}

\noindent With \cref{prop:prepare-Fock-states}, we set $C_U=P^\dag C_\mathrm{enc}^\dag WC_\mathrm{enc}P$. Picking $\eta=\varepsilon/12$ and combining the compilation error in \cref{eq:compilerr} with the state preparation error from \cref{prop:prepare-Fock-states} completes the proof of \cref{thm:compile_new_passive}.

\end{proof}

\subsubsection{Compiling non-number-preserving unitaries}\label{ssscn:nonnumberpreserving}

Next, we consider non-number-preserving unitaries with an efficient qubit circuit description:

\begin{theorem}[Efficient compilation of non-number-preserving unitaries]\label{thm:compile_new}
Let $0<\varepsilon<1/2$. Let $U$ be a unitary over Fock space, and suppose that $U$ has a $\delta$-approximate qubit circuit description over $\mathcal H_{m,K}$ with $L$ gates over $q\ge mK$ qubits as in \cref{def:circdesc}, with $\delta\le\varepsilon/4$. Then, there exist a bosonic unitary circuit $C_U$ and $\alpha\in\mathbb R$ such that
\be
    \sup_{\substack{\ket\psi\in\mathcal H_{m,K}\\\norm{\ket\psi}=1}}\norm{C_U(\ket\psi\otimes\ket\alpha\otimes\ket0^{\otimes(2q+1)})-(U\ket\psi)\otimes\ket\alpha\otimes\ket0^{\otimes(2q+1)}}\le\varepsilon.
\ee
Moreover, $C_U$ has $\poly(m,K,L,\log(1/\varepsilon))$ gates generated by local constant-degree energy-preserving polynomial Hamiltonians and the coherent state $\ket\alpha$ has energy $O(q^2/\varepsilon^2)$.
\end{theorem}

\noindent Before proving the result, we discuss its significance. \cref{thm:compile_new} implies that unitary operations with an efficient qubit circuit description can be compiled efficiently, using a single coherent state as the source of energy. Compared to \cref{thm:compile_new_passive}, \cref{thm:compile_new} includes unitaries which are not necessarily energy-preserving, but the price to pay is that the required energy scales polynomially with the desired (inverse) precision instead of logarithmically, and that the gates used include gates generated by local constant-degree Hamiltonians beyond $\GKerr$ or $\GcrossKerr$.

On the other hand, \cref{thm:compile_new} allows us to compile a larger class of non-number-preserving unitaries compared to a simple combination of \cref{thm:passive-compilation,thm:cutoff-SK}. A concrete example is a photon-number multi-controlled gate $e^{i\prod_{j=1}^m\hat n_j H_0}$ for which parametric approximation is not efficient since the total degree is growing with $m$, where $H_0$ is a well-behaved non-number-preserving Hamiltonian of constant degree.

Recall that we showed in \cref{sec:sim_existing_models} that computations from existing models of quantum computations can be reproduced efficiently by \model\ using parametric approximation via \cref{thm:passive-compilation}, as long as energy moments are reasonably bounded throughout the computation. \cref{thm:compile_new} allows us to strengthen this conclusion, which now remains valid when only assuming that the mean energy $E$ remains polynomially bounded throughout the computation.\footnote{We assume that a suitable cutoff-preserving approximation to the target evolution is supplied through an efficient approximate qubit circuit description, as in \cref{def:circdesc}. For a constant number of modes, supplying its Fock-basis matrix to sufficient precision suffices. We do not claim an efficient procedure for constructing this description from the original Hamiltonian in general.} To see this, truncate each state throughout the computation to a total photon number $EL^2/\varepsilon^2$ and apply \cref{thm:compile_new}, which leads to a simulation of the computation using $\poly(m,L,E,1/\varepsilon)$ resources. Since all gates used in \cref{thm:compile_new} are generated by local energy-preserving constant-degree Hamiltonians, we can further apply \cref{thm:cutoff-SK} at total-photon cutoff $O((mK+|\alpha|^2)/\varepsilon^2)$ to recompile all gates using $\GKerr$ or $\GcrossKerr$, using once again $\poly(m,L,E,1/\varepsilon)$ resources.

\begin{proof}[Proof of \cref{thm:compile_new}]
The proof is similar to that of \cref{thm:compile_new_passive}, using the unary-to-dual-rail isometry $\mathcal F$ instead of $\mathcal J$.

Let $V_1$ denote the $\delta$-approximate qubit circuit description of $U$, with $L$ gates over $q$ qubits, which by \cref{def:circdesc} satisfies
\be\label{eq:UtoV1}
\|V_1\mathcal E-\mathcal EU\|\le\delta.
\ee
The unitary $U$ is not necessarily energy-preserving and $V_1$ acts on a unary encoding of Fock states, so $V_1$ is also not necessarily energy-preserving, when seen as a unitary acting on $\mathcal H_{q,1}$. To compile the target unitary $U$ over $\mathcal H_{m,K}$ using only energy-preserving unitary gates, we further encode each unary qubit that $V_1$ acts on as a dual-rail qubit in Fock space via the isometry $\mathcal F$ (see \cref{eq:Fisometry}).
This isometry defines a logical unitary circuit $V_2\mathcal F=\mathcal FV_1$, where $V_2$ acts on dual-rail qubits instead of unary qubits. Seen as a unitary acting on the dual-rail subspace of $\mathcal H_{2q,1}$, $V_2$ is now energy-preserving and can be reproduced exactly gate-by-gate by linear-optical gates (for single-qubit unitaries) and cross-Kerr gates (for $CZ$ gates). We denote by $W$ the corresponding bosonic circuit. By construction, on the dual-rail subspace,
\be\label{eq:V1toW}
    W\mathcal F=V_2\mathcal F=\mathcal FV_1,
\ee
and $W$ has $O(L)$ gates.

We now use the fact that the combination of both encoding isometries ($\mathcal E$ from Fock space to unary and $\mathcal F$ from unary to dual-rail) can be implemented approximately using a coherent state and a fixed set of unitary gates generated by local energy-preserving polynomial Hamiltonians of constant degree: setting $C_\mathrm{enc}=C_\mathcal FC_\mathcal E$ up to identity operators and a reordering of the ancillary registers and combining \cref{lem:encodingfu,lem:encodingudr}, we obtain for all $0<\eta<1$, $q\ge mK$, $\alpha\in\mathbb R$ with $\alpha\ge1$, and unit $\ket\psi\in\mathcal H_{m,K}$,
\be\label{eq:FEtoCenc}
    \|C_\mathrm{enc}(\ket\psi\otimes\ket\alpha\otimes\ket0^{\otimes(2q+1)})-(\mathcal F\mathcal E\ket\psi)\otimes\ket\alpha\otimes\ket0^{\otimes(m+1)}\|\le\eta+\frac{q}{\alpha}.
\ee
Moreover, $C_\mathrm{enc}$ has $\poly(m,K,L,\log(1/\eta))$ energy-preserving gates generated by local polynomial Hamiltonians of constant degree (using once again $q\le 2L+mK$ without loss of generality).

For all $\ket\psi\in\mathcal H_{m,K}$, we let
\be
\mathcal A:\;\ket\psi\mapsto\ket\psi\otimes\ket\alpha\otimes\ket0^{\otimes(2q+1)}
\ee
be the isometry which appends the ancillary states to the input $\ket\psi$. We have (omitting identity operators on the $\ket\alpha$ register for brevity)
\be
C_\mathrm{enc}^\dag WC_\mathrm{enc}\mathcal A-\mathcal AU=C_\mathrm{enc}^\dag[W(C_\mathrm{enc}\mathcal A-\mathcal F\mathcal E)+(W\mathcal F\mathcal E-\mathcal F\mathcal EU)+(\mathcal F\mathcal E-C_\mathrm{enc}\mathcal A)U].
\ee
Since $C_\mathrm{enc}^\dag$, $W$ and $U$ are unitary, the first and last terms are bounded by $\eta+\frac{q}{|\alpha|}$ in operator norm by \cref{eq:FEtoCenc}. Moreover, the second term satisfies
\begin{align}
\|W\mathcal F\mathcal E-\mathcal F\mathcal EU\|&=\|\mathcal FV_1\mathcal E-\mathcal F\mathcal EU\|\notag\\
&\le\|V_1\mathcal E-\mathcal EU\|\notag\\
&\le\delta,
\end{align}
where we used \cref{eq:V1toW} in the first line, the fact that $\mathcal F$ is an isometry in the second line, and \cref{eq:UtoV1} in the last line. Setting $C_U=C_\mathrm{enc}^\dag WC_\mathrm{enc}$ and summing the error terms with the triangle inequality, we obtain
\be
\|C_U\mathcal A-\mathcal AU\|\le\delta+2\left(\eta+\frac{q}{\alpha}\right)=\delta+2\eta+\frac{2q}{\alpha}.
\ee
By construction, $C_U=C_\mathrm{enc}^\dag WC_\mathrm{enc}$ is a bosonic circuit with local energy-preserving gates of constant degree. Moreover, $C_\mathrm{enc}$ has $\poly(m,K,L,\log(1/\eta))$ gates, while $W$ has $O(L)$ gates, so $C_U$ also has $\poly(m,K,L,\log(1/\eta))$ gates generated by local constant-degree Hamiltonians. Picking $\eta=\varepsilon/8$ and $\alpha=4q/\varepsilon$ concludes the proof.
\end{proof}

\subsection{State synthesis}
\label{sec:synthesis}

We give deterministic state-preparation results to inverse-polynomial accuracy in trace distance. We show how one can prepare finite superpositions of Fock states, and high-quality Gottesman--Kitaev--Preskill (GKP) states efficiently via \model.
All gates that we use
belong to either $\GKerr$ or
$\GcrossKerr$, as defined in \cref{sec:cutoff-sk}.
Importantly, there are no intermediate measurements or
postselections required. This is in contrast with usual protocols for GKP state preparation \cite{brenner2025complexity}.

\subsubsection{Preparing finite Fock superpositions}

States of finite Fock cutoff form a rich subset of states (indeed, they form a dense subset of the entire Hilbert space). Here, we show that a state of up to total photon number $K$, over $O(1)$ many modes, can be prepared with $\poly(K,\frac{1}{\varepsilon})$ many linear-optical and Kerr or cross-Kerr gates. Note that we keep the number of modes constant (this also allows the state to be efficiently describable).

Write \(\mathcal H^{(m)}_{\le K}=\operatorname{span}\{
|\bm n\rangle:|\bm n|\le K\}\).
The following construction proves state synthesis directly from
\cref{thm:cutoff-SK}. First, we introduce an auxiliary lemma, which is a direct generalization of \cref{prop:prepare-Fock-states}:

\begin{lem}[Preparing a state using a coherent reservoir]\label{cor:arbitrary states}
Let \(K\ge1\), let \(|\psi\rangle\in\mathcal H^{(m)}_{\le K}\)
be normalized, and let \(\alpha>0\) be real. There is an energy-preserving
unitary \(U\) on the \(m\) target modes and one reservoir mode such that
\begin{equation}\label{eq:new-lifting-error}
\big\|U|0^m,\alpha\rangle-|\psi\rangle|\alpha\rangle\big\|
\le \frac{\sqrt{\langle\psi|\hat N^2|\psi\rangle}}{\alpha}
       +2\sqrt{\Pr\{X<K\}},
\end{equation}
where $X\sim\operatorname{Poisson}(\alpha^2)$.
In particular, the first term is at most \(K/\alpha\), and whenever $\alpha^2\ge4K\log2$, we have that $\Pr\{ X<K\}\le e^{-\alpha^2/4}$.
\end{lem}
\begin{proof}
Write \(|\psi\rangle=\sum_{|\bm n|\le K}
 c_{\bm n}|\bm n\rangle\).
In each total-number sector \(j\ge K\), prescribe
\begin{equation}\label{eq:new-sector-map}
 U|0^m,j\rangle=
 \sum_{|\bm n|\le K}c_{\bm n}
 |\bm n,j-|\bm n|\rangle.
\end{equation}
The right-hand side is a unit vector in the same sector, so this prescription
extends to a unitary on that finite-dimensional sector. Take the identity
on sectors \(j<K\). The direct sum is an energy-preserving unitary fixing
the global vacuum.

Let \(B=\sum_{r\ge0}|r\rangle\langle r+1|\) be the lowering shift on
the reservoir. Since \(B\sqrt{\hat n}=\hat a\) and \(\hat a|\alpha\rangle
=\alpha|\alpha\rangle\),
\begin{align}
\|(B-\II)|\alpha\rangle\|^2
&=\|B(\II-\sqrt{\hat n}/\alpha)|\alpha\rangle\|^2\notag\\
&\le \mathbb E(1-\sqrt X/\alpha)^2
 \le \frac{\mathbb E(X-\alpha^2)^2}{\alpha^4}
 =\frac1{\alpha^2}.\label{eq:new-shift}
\end{align}
Here \(|\sqrt{x}-\alpha|\le |x-\alpha^2|/\alpha\), and \(B\) is a
contraction. Telescoping gives
\(\|(B^r-\II)|\alpha\rangle\|\le r/\alpha\).
Put \(P=\sum_{j\ge K}|j\rangle\langle j|\) and
\(p=\|(\II-P)|\alpha\rangle\|^2\). \cref{eq:new-sector-map} gives
\be
 U(|0^m\rangle P|\alpha\rangle)
 =\sum_{\bm n}c_{\bm n}|\bm n\rangle
 B^{|\bm n|}P|\alpha\rangle.
\ee
Replacing \(P|\alpha\rangle\) by \(|\alpha\rangle\) on the right costs
at most \(\sqrt p\), by orthogonality of the target Fock states and
contractivity of each \(B^r\). Restoring the omitted input costs another
\(\sqrt p\). Finally, the squared norm of the remaining difference is
at most \(\sum_{\bm n}|c_{\bm n}|^2
 |\bm n|^2/\alpha^2\). This proves the claim.
\end{proof}

\begin{theorem}[Approximate synthesis with a fixed number of modes]
\label{thm:new-finite}
Let \(K\ge1\), \(0<\epsilon<1/2\), and
\(|\psi\rangle\in\mathcal H^{(m)}_{\le K}\).
There is a circuit on \(m+2\) modes whose reduced target output \(\rho\)
satisfies \(\frac12\norm{\rho-|\psi\rangle\langle\psi|}_1\le\epsilon\).
One can take a single coherent input of mean photon number
\begin{equation}\label{eq:new-energy}
 \mu=\alpha^2=64K^2/\epsilon^2
\end{equation}
and all other inputs vacuum. Let
$K_{\mathrm{joint}}=\left\lceil4\mu+4\log(16/\epsilon)\right\rceil$.
A sufficient gate count is
\begin{equation}\label{eq:new-gate-count}
 O\!\left(\poly(m,K_{\mathrm{joint}},\frac{1}{\varepsilon}) \, K_{\mathrm{joint}}^{O(m)}\right).
\end{equation}
Thus, for fixed \(m\), the gate count and input energy are polynomial in
\(K,\epsilon^{-1}\). Given the target amplitudes to sufficient precision,
classical compilation takes \(\poly(m,K_{\mathrm{joint}}) \, K_{\mathrm{joint}}^{O(m)}
\log(1/\epsilon)\) time.
\end{theorem}
\begin{proof}
Apply \cref{cor:arbitrary states} with \(\alpha=8K/\epsilon\).
The shift error is at most \(\epsilon/8\). For
\(X\sim\operatorname{Poisson}(\mu)\), exponential Markov bounds give
\be
 p_-:=\Pr\{X<K\}\le 2^K e^{-\mu/2},\qquad
 p_+:=\Pr\{X>K_{\mathrm{joint}}\}\le 2^{-K_{\mathrm{joint}}}e^{\mu}.
\ee
Our choices imply \(2\sqrt{p_-}\le\epsilon/8\) and
\(2\sqrt{p_+}\le\epsilon/8\). For example, \(\mu\ge256K^2\), so
\(K\log2-\mu/2\le-\mu/4\), and
\(\mu\ge8\log(16/\epsilon)\); the tail bound follows immediately
from the definition of \(K_{\mathrm{joint}}\).

Use \cref{thm:cutoff-SK} to approximate \(U\) on the total-number cutoff \(K_{\mathrm{joint}}\)
with vector error \(\epsilon/4\), adding its one vacuum ancilla. The
vacuum phase in that theorem is zero here. Splitting the coherent input
at \(K_{\mathrm{joint}}\), the error on the full input is at most
\(\epsilon/4+2\sqrt{p_+}\). Combined with
\cref{eq:new-lifting-error}, this is less than \(\epsilon\).
Pure-state trace distance is at most vector distance, and partial trace
contracts trace distance. The gate count is precisely the bound in
\cref{thm:cutoff-SK} applied to \(m+1\) modes. A unitary completion in each sector
can be obtained by finite-dimensional linear algebra in polynomial time.
\end{proof}

\subsubsection{Preparing Gottesman--Kitaev--Preskill states}\label{sec:gkp}
For \(0<\kappa,\Delta\le1\), define
\begin{equation}\label{eq:new-gkp}
 |G_{\kappa,\Delta}\rangle
 =Z^{-1/2}\sum_{j\in\mathbb Z}e^{-\kappa^2j^2/2}|g_{j,\Delta}\rangle,
 \qquad
 g_{j,\Delta}(x)=(\pi\Delta^2)^{-1/4}
                e^{-(x-j)^2/(2\Delta^2)},
\end{equation}
where \(Z\) is for normalization. We first show that these states are highly concentrated in the low photon number sector.

\begin{lem}[Exponential photon-number tail]\label{lem:new-gkp-tail}
There is a universal constant \(C\ge1\) such that, with
\(R=\kappa^{-2}+\Delta^{-2}\),
\begin{equation}\label{eq:new-gkp-tail}
 \langle G_{\kappa,\Delta}|(\II-\Pi_{\le K})|G_{\kappa,\Delta}\rangle
 \le \frac C\kappa\exp\!\left(-\frac K{8R}\right)
 \qquad(K\ge0).
\end{equation}
\end{lem}
\begin{proof}
We bound an exponential number moment. In the Bargmann representation,
\(|g_{j,\Delta}\rangle\) has the entire function
\be
 f_j(z)=A_j e^{a z^2/2+b_jz},\quad
 a=\frac{\Delta^2-1}{\Delta^2+1},\quad
 b_j=\frac{\sqrt2j}{1+\Delta^2},\quad
 A_j=\sqrt{\frac{2\Delta}{1+\Delta^2}}
       e^{-j^2/[2(1+\Delta^2)]}.
\ee
The Bargmann norm is
\(\|f\|^2=\pi^{-1}\int_{\mathbb C}|f(z)|^2e^{-|z|^2}\,d^2z\).
The map \(e^{t\hat n/2}\) acts as \(f(z)\mapsto f(e^{t/2}z)\).
Integrating separately over the real and imaginary parts of \(z\) yields
\begin{equation}\label{eq:new-gauss-mgf}
 \|e^{t\hat n/2}|g_{j,\Delta}\rangle\|^2
 =\sqrt{\frac{1-a^2}{1-a^2e^{2t}}}
 \exp\!\left(
 \frac{j^2(e^t-1)}{(1+\Delta^2)-(\Delta^2-1)e^t}\right),
\end{equation}
provided \(|a|e^t<1\).

For \(0\le t\le\Delta^2/8\), note that
\(1-a^2=4\Delta^2/(1+\Delta^2)^2\ge\Delta^2\) and
\(e^{2t}-1\le3t\). Thus
\(1-a^2e^{2t}\ge\tfrac12(1-a^2)\).
The denominator in the exponent in \cref{eq:new-gauss-mgf} is at
least \(2\), and \(e^t-1\le2t\). Consequently,
\begin{equation}\label{eq:new-gauss-mgf-bound}
 \|e^{t\hat n/2}|g_{j,\Delta}\rangle\|
 \le2^{1/4}e^{tj^2/2}.
\end{equation}
Take \(t=\tfrac18\min\{\kappa^2,\Delta^2\}\).
The overlaps \(\langle g_{j,\Delta}|g_{k,\Delta}\rangle
=e^{-(j-k)^2/(4\Delta^2)}\) are nonnegative, so
\be
 Z\ge\sum_j e^{-\kappa^2j^2}\ge c/\kappa
\ee
for a universal \(c>0\). The last inequality follows, for example,
by retaining all \(|j|\le1/\kappa\).
On the other hand, the triangle inequality and
\cref{eq:new-gauss-mgf-bound} give
\be
 \|e^{t\hat n/2}|G_{\kappa,\Delta}\rangle\|
 \le\frac{2^{1/4}}{\sqrt Z}
       \sum_j e^{-(\kappa^2-t)j^2/2}
 \le\frac{C'}{\sqrt\kappa}.
\ee
Here \(\kappa^2-t\ge7\kappa^2/8\), and the Gaussian sum is at most
\(1+\int_{\mathbb R}e^{-7\kappa^2x^2/16}\,dx=O(1/\kappa)\).
The same summable estimate shows convergence in the graph norm of
\(e^{t\hat n/2}\), justifying its application to the infinite sum.
Therefore \(\langle e^{t\hat n}\rangle\le C/\kappa\).
Exponential Markov gives a tail at most
\((C/\kappa)e^{-tK}\), and
\(t\ge1/(8R)\), proving \cref{eq:new-gkp-tail}.
\end{proof}

\GKPprep*
\begin{proof}
Choose
\be
 K=\left\lceil 8R\log\frac{16C}{\kappa\epsilon^2}\right\rceil,
 \qquad
 |G^{(K)}\rangle=
 \frac{\Pi_{\le K}|G_{\kappa,\Delta}\rangle}
      {\|\Pi_{\le K}|G_{\kappa,\Delta}\rangle\|}.
\ee
\cref{lem:new-gkp-tail} implies an omitted probability at most
\(\epsilon^2/16\). The exact pure-state trace-distance formula gives
\be
 \frac12\norm{|G^{(K)}\rangle\langle G^{(K)}|-
            |G_{\kappa,\Delta}\rangle\langle G_{\kappa,\Delta}|}_{1}
 =\sqrt{1-\|\Pi_{\le K}|G_{\kappa,\Delta}\rangle\|^2}
 \le\epsilon/4.
\ee
Apply \cref{thm:new-finite} with \(m=1\) and synthesis error
\(\epsilon/2\), allowing another \(\epsilon/4\) for numerical
specification of the target. Its coherent energy is \(O(K^2/\epsilon^2)\),
which proves \cref{eq:new-gkp-energy}. So a gate bound
is \(O(\mathrm{poly}(K_{\rm joint}\mathrm{polylog}(1/\epsilon)))\), where
\(K_{\mathrm{joint}}=O(K^2/\epsilon^2)\).

To verify that classical state specification is also efficient, first
truncate the peak index in \cref{eq:new-gkp} to \(|j|\le J\).
Because each peak is normalized and \(Z\ge c/\kappa\),
\be
 \frac1{\sqrt Z}\left\|\sum_{|j|>J}
 e^{-\kappa^2j^2/2}|g_{j,\Delta}\rangle\right\|
 \le C_1\kappa^{-1/2}e^{-\kappa^2J^2/4}.
\ee
Indeed, split each Gaussian weight into two equal exponential factors
and bound the remaining Gaussian sum by \(O(1/\kappa)\).
Thus \(J=O(\kappa^{-1}\sqrt{\log(C_2/(\kappa\epsilon))})\)
makes this error as small as any prescribed constant multiple of
\(\epsilon\). Orthogonal projection does not increase it, and
normalization changes it by at most another constant factor, since the
projected target has norm bounded below.

For each retained peak, its Fock coefficients \(u_{j,n}\) are computed
from the Bargmann function above using
\be
 u_{j,0}=A_j,\quad u_{j,1}=b_jA_j,\qquad
 u_{j,n+1}=\frac{b_ju_{j,n}+a\sqrt n\,u_{j,n-1}}{\sqrt{n+1}}.
\ee
Summing \(e^{-\kappa^2j^2/2}u_{j,n}\) over \(|j|\le J\), for
\(0\le n\le K\), and normalizing gives the required vector. This takes
\(O(JK)\) arithmetic steps. Since \(|a|\le1\) and
\(|b_j|\le\sqrt2J\), a deliberately conservative roundoff amplification
bound is \([C_3(J+1)]^{K+1}\). Hence polynomially many guard bits in
\(K,\log(J+1),\log(1/\epsilon)\), together with the bits needed to
specify the parameters, suffice. Elementary functions and square roots
can be evaluated to those precisions in polynomial time. As usual, the
uniformity statement assumes \(\kappa,\Delta\) are supplied as rational
parameters or are efficiently computable to the requested precision.
\end{proof}

Compared to the GKP preparation protocol of Brenner, Caha,
Coiteux-Roy, and K\"onig~\cite{brenner2025complexity}, we highlight that we do not need a coherent implementation of $\lfloor\hat q\rfloor$, which would be required to implement their protocol non-adaptively. As a downside, the circuit depth is polynomial in our case, whereas \cite{brenner2025complexity} has a logarithmic-depth state preparation circuit.

\section{Complexity of energy-preserving bosonic computations}\label{sec:complexity}

Having established algorithmic primitives for \model, we now turn to its complexity-theoretic foundations. Recall that the complexity class $\BEQP{}{E}$ associated to \model\ with energy $E$ and its variants are defined in \cref{def:cc-beq}. We first show that $\BEQP{}{}=\BQP$ in \cref{sec:eqBQP}. Then, we show in \cref{sec:disptherm,sec:therm} that the same result holds when replacing the input coherent state in \model\ by noisy variants, namely displaced thermal states or simply thermal states. Finally, we obtain several complexity upper bounds in \cref{sec:bound_complexity} when varying the parameters of \model.

\subsection{\texorpdfstring{$\BQP$}{BQP}-completeness for coherent states}
\label{sec:eqBQP}

In this section, we show that the power of bosonic energy-preserving computations with input coherent states of polynomial energy ($\BEQP{}{}$) is equivalent to $\BQP$.

We first simulate $\BQP$ in $\BEQP{}{}$:
\begin{theorem}\label{thm:bqp-universal}
$\BQP\subseteq\BEQP{}{}$.
\end{theorem}
\begin{proof}
This result is a direct corollary of \cref{thm:compile_new}. For the sake of completeness, we include a direct proof below.

Simulation of a \BQP\ circuit has been proposed with linear optics, Kerr non-linearities and dual-rail encoding before in \cite{Chuang_1995}, which we briefly review:
\begin{enumerate}
    \item Every two modes encode a qubit. Logical-zero is $\ket{01}$ and logical-one is $\ket{10}$ written in the Fock basis. We label the modes' annihilation operators by $\hat a_1$ and $\hat a_2$. If there is another qubit involved, we use operators $\hat b_1, \hat b_2$ to address the annihilation operators of the second pair of modes.
    \item Single-qubit gates are linear optical elements: an arbitrary $\mathrm{SU}(2)$ operation on the $\ket{10}$ and $\ket{01}$ subspace can be achieved via beamsplitters and phase shifters as their representation in this 2-dimensional subspace is given by
        \begin{align}
            (\mathrm{BS}(\theta, \phi))_{[01,10]} =
            \begin{bmatrix}
                \cos(\theta) & e^{i\phi} \sin(\theta)\\
                -e^{-i\phi}\sin(\theta) & \cos(\theta)
            \end{bmatrix},
        \end{align}
        and the phase shifter on the first rail is
        \begin{align}
            \mathrm{PS}_1(\theta) = e^{i\theta\hat n_1}.
        \end{align}
        For example, a $50:50$ beamsplitter followed by a $PS(\pi)_1$ applies the Hadamard gate in dual-rail encoding.
    \item Consider two logical qubits. It is the case that $\overline{CZ} = \exp(i \pi \hat a_1^\dag \hat a_1 \cdot \hat b_1^\dag \hat b_1)$.
\end{enumerate}
Also note that these gates implement the qubit gates exactly on the logical subspace. As a result, having the capability of generating $\ket{1}$ to $\frac{1}{\poly}$ error enables us to perform arbitrary qubit computations within \model, so invoking \cref{prop:prepare-Fock-states} completes the proof.

For a $\BQP$ simulation, we need to simulate a qubit quantum computation with $n$ qubits and depth $L = \poly(n)$ to error $\varepsilon = \frac{1}{\poly(n)}$. Therefore, we need to prepare each Fock state to error $\frac{\varepsilon}{n}$, which requires (from \cref{prop:prepare-Fock-states}) $\poly(n)$ ancilla coherent states with energy $O(\log(n/\varepsilon))$.
Therefore $\BQP$ can be simulated within \model\ with $\poly(n)$ modes with $\poly(n)$ energy, and polynomial depth, implying that $\BQP\subseteq\BEQP{}{}$.
\end{proof}

We now prove the reverse statement:

\begin{theorem}\label{thm:in_BQP}
$\BEQP{}{}\subseteq\BQP$.
\end{theorem}

\begin{proof}
    It was shown in \cite[Theorem 5.4]{CGMMNRS25} that any bosonic computation with constant-locality gates and the promise of polynomial energy (throughout) can be simulated in $\mathsf{BQP}$, if the Fock entries of the gate set, up to photon cutoff $K$, can be computed in time $\poly(K, \frac{1}{\varepsilon})$ to accuracy $\varepsilon$. Hence, it is left to show that this assumption holds. This is straightforward to show, as energy-preserving gates are block-diagonal in the Fock basis.

\begin{lem}\label{lem:assumption-satisfied}
    Considering a fixed $m=O(1)$-mode energy-preserving Hamiltonian $H$ of degree at most $d\ge1$, with coefficient magnitudes bounded by $C_{\mathrm{max}}>1$,\footnote{\label{fn:rational-inputs}Here, we assume for simplicity that the coefficients have rational real and imaginary parts and that $t$ is rational. Otherwise, we need to compute sufficiently accurate approximations.} the entry $(\exp(iHt))_{\bm k,\bm\ell}$ can be computed to precision $2^{-p}$ in
    \be
        O\!\left((1+n^{3m})\bigl(p + \log(1+|t|) + \log C_{\mathrm{max}} + \log d + d\log(1+n)\bigr) + n^m d^{2m}\right)
    \ee
    arithmetic steps for $n = |\bm k| = |\bm\ell|\ge0$, using numbers of
    \be
        O\!\left(p + \log(1+|t|) + \log C_{\mathrm{max}} + (m+d)\log(1+n) + (2m+1)\log d\right)
    \ee
    bits. Entries with $|\bm\ell|\ne|\bm k|$ are zero.

    Therefore, for constant locality and $d = \poly(n), t = \poly(n), C_{\text{max}} = \exp(n), p = \poly(n)$ and polynomial-length rational inputs, the total runtime and space are bounded by $\poly(n)$.
\end{lem}
\begin{proof}
    This can be shown by noting that the $n$-photon block has dimension $O(1+n^m)$, and the exponentiation can be done via a \textit{scaling and squaring} method~\cite{Higham22,higham05scaling,fasihigham19}. For Hermitian Hamiltonians, we make the complexity and forward-error analysis rigorous in \cref{lem:mat-exp-in-photon-block} in \cref{app:assumption-satisfied-long}.
    From this result, setting $m$ as fixed $O(1)$ gives the stated bounds.
\end{proof}

Therefore, by \cref{lem:assumption-satisfied} the amplitudes of the gate set (up to polynomial energy) can be computed in polynomial time.
\end{proof}

\subsection{\texorpdfstring{$\BQP$}{BQP}-completeness for displaced thermal states}
\label{sec:disptherm}

In this section, we probe how robust $\BEQP{}{}$ is to thermal noise in the ancilla states. We show that we can concentrate the signal-to-noise ratio by taking polynomially many copies of noisy ancillas with polynomial thermal noise and just one vacuum mode, and improve the error to $\frac{1}{\poly}$.

\begin{theorem}[Coherent amplitude concentration]
Given $N$ single-mode displaced thermal states with mean thermal photon number $\bar{n}$ and amplitude $\alpha$, with $|\alpha| > 0$, and one vacuum mode, there exists a sequence of energy-preserving operations that can ``concentrate'' the coherent state purity into one mode, with amplitude $\alpha$ and mean thermal photon number $\frac{\bar{n}}{N}$.
\label{thm:coherent-distillation}
\end{theorem}
\begin{proof}[Proof]
    Suppose we start with $N$ displaced thermal states $\rho(i) = \rho(\alpha, \bar{n})$. Let $\bm{\alpha} = [\alpha \dots \alpha]^T\in \mathbb{C}^N$.
Tracking the moments gives
\be \bm d_i = \bm{\alpha}, \sigma_i = (\bar{n} + 1/2)\II.\ee
From \cref{lem:coher-gen}, there exists a linear-optical network $A$, such that $\bm{\psi}_1 = [\sqrt{N}\alpha, 0 \dots 0]^T \in \mathbb{C}^N$ is the displacement vector of the output state. Under this operation,
$\bm d_1 = \bm\psi_1, \sigma_1 = A\sigma_iA^T = (\bar{n} + 1/2) A\II A^T = \sigma_i$.
Therefore, we have pumped the target mode with all the energy, while not increasing the noise. In the next step, we will dump the excess energy and the noise into a vacuum mode. Using a beam splitter with transmittance $1/N$, we couple the target mode with the clean vacuum state. The resulting state is:
\be \bm d_0 = B
\begin{bmatrix}
    \sqrt{N}\alpha\\
    0
\end{bmatrix},\quad
\sigma_0 = B
\begin{bmatrix}
    \bar{n} + 1/2 & 0\\ 0 & 1/2
\end{bmatrix} B^T,\ee
where
\be B =
\begin{bmatrix}
    \frac{1}{\sqrt{N}} & \sqrt{\frac{N-1}{N}} \\
    -\sqrt{\frac{N-1}{N}} & \frac{1}{\sqrt{N}} \\
\end{bmatrix}.
\ee
Solving these, and tracing out the vacuum mode, we get a displaced thermal state $\rho(\alpha, \bar{n}/N)$ as required.

\end{proof}

\begin{theorem}
    Energy-preserving bosonic computation on displaced thermal states of polynomial thermal photon number is $\BQP$-complete.
\end{theorem}
\begin{proof}
    From the definition of thermal states \cref{def:thermal-state}, 
    \be
    \frac12\norm{\rho(\alpha, \bar{n})-\ketbra{\alpha}{\alpha}}_1 = \frac12\norm{\rho(0, \bar{n})-\ketbra{0}{0}}_1 = \frac{\bar{n}}{1 + \bar{n}}
    \ee

    From \Cref{thm:coherent-distillation}, we can use an extra vacuum input and $\poly$ many copies of ancillas with $\poly$ thermal noise to get a state within $1/\poly$ error to the required input state.
    Therefore, we need only $2n$ vacuum states ($n$ for the data, $n$ as a energy dumping state), and $\poly(n)$ copies of noisy ancillas to simulate the $\BEQP{}{}$ output probability to $1/\poly$ error.
    The proof of completeness then follows exactly as in \cref{thm:bqp-universal} and \cref{thm:in_BQP}.
\end{proof}

\subsection{\texorpdfstring{$\BQP$}{BQP}-completeness for thermal states}
\label{sec:therm}

While the preceding results use displaced states as coherent states, can we similarly supply the energy resource directly through thermal states?
Can we still simulate $\BQP$ in this new model, where the input coherent states are replaced by thermal states of possibly different temperatures?
In this section we prove that thermal states of differing mean photon number are enough to maintain $\BQP$-completeness. The main idea is to encode the $\BQP$ computation into pairs of thermal states with different energies. The trade-off is that we either need polynomial-degree Hamiltonians, or error correction with constant-degree gates and a constant energy separation.
Throughout this section, ``energy'' of a thermal mode means its mean photon number, as in \cref{def:thermal-state}.

\begin{definition}[\thermalmodel: \model~with thermal states]
    \label{defn:thermal-model}
    We define the thermal model by replacing all the input states and ancillas of \model\ (\cref{scn:model-definition}) with thermal modes of average photon numbers $\bar{n}_i$. The corresponding promise class is denoted $\ThermalBEQP{G}{\bar{n}_i}$, defined analogously to \cref{def:cc-beq}. When we wish to utilize gates of polynomial degree, we write $\ThermalBEQP{\poly}{\bar{n}_i}$.
\end{definition}

\subsubsection{Energy alone is not enough}
If all $\bar{n}_i$ are equal in the thermal model (\cref{defn:thermal-model}), then the resulting state is a Gibbs state of the total photon number operator. Since the operations are energy-preserving, the input state commutes with the operations. Therefore, the input state is invariant under energy-preserving operations, and the photon number of a single mode can be calculated from the probability distribution of a single thermal state. Therefore, thermal inputs all at the same temperature are trivially simple.

\subsubsection{Energy difference is enough}
\label{sec:thermal-poly}

We now show that a polynomial separation of thermal temperatures supplies a reliable logical input.

\begin{restatable}[\thermalmodel~with two temperatures of thermal states is $\BQP$-complete]{claim}{thermalbqp}
    \label{thm:thermal-poly-bqp}
    Let
    \be
        \bar n_L\leq b=O(1),\qquad \bar n_H=\poly(N)
    \ee
    be such that $\bar n_H - \bar n_L = \poly(N)$.
    Then, $\ThermalBEQP{\poly}{\nb_L, \nb_H}$ is $\BQP$-complete. %
\end{restatable}

To prove this claim, we first encode the qubit state in the thermal state input. We use a pair of cold ($\bar{n}_L$) and hot ($\bar{n}_H$) modes and a multi-photon dual-rail encoding, which is described below:

\begin{definition}[Multiphoton dual-rail encoding]
    \label{defn:multi-photon-dual-rail}
    For integer $d$, identify

    \begin{subequations}\label{eq:code}
        \be
        \calC_0 = \{\ket{a,b} \mid a<b, a,b\in\{0,\dots,d\}\}
        \ee
        \be
        \calC_1 = \{\ket{a,b} \mid a>b, a,b\in\{0,\dots,d\}\}
        \ee
        \be
        \Cinv = \{\ket{a,b} \mid a=b \vee a>d \vee b>d\}
        \ee
    \end{subequations}

    The implementation of the $\sigma_x$ gate is trivial:
    \begin{align}
    U_X = e^{-i \pi (\hat{n}_1+\hat{n}_2) / 2} e^{i\frac{\pi}{2}(\a_1\adag_2+\adag_1\a_2)}.
    \end{align}
    For the $T$-gate, pick a polynomial $p_1(n)$ via interpolation, such that
    \begin{subequations}
    \be
    \forall n\in\{1,\dots,d\}\colon p_1(n) = 1
    \ee
    \be
    \forall n\in\{-d,\dots,0\}\colon p_1(n)=0.
    \ee
    \end{subequations}
    Then we can realize the $T$-gate on the codespace \cref{eq:code} as
    \begin{equation}
    U_T = e^{i\frac{\pi}{4}p_1(\hat n_1-\hat n_2)}.
    \end{equation}
    Similarly, we realize the $CZ$ gate with modes $1,2$ for control, and $3,4$ for target, as
    \begin{equation}
    U_{CZ} = e^{i\pi p_1(\hat n_1-\hat n_2)p_1(\hat n_3-\hat n_4)}.
    \end{equation}
    For the Hadamard gate, we can use the idea of \cite[Theorem 2]{arzani2025can} and \cref{sec:passive-description} to implement
    \begin{equation}
    U_H^{(d)} = \frac{1}{\sqrt2}\sum_{0\le a<b\le d}\Bigl(\bigl(\ket{a,b}+\ket{b,a}\bigr)\bra{a,b} + \bigl(\ket{a,b}-\ket{b,a}\bigr)\bra{b,a} \Bigr) + \sum_{a=0}^d \ketbra{a,a}{a,a}.
    \end{equation}

    Conditioned on all pairs being valid and below the cutoff, we measure the first mode of the logical circuit by measuring only the first mode of the physical circuit. For a fixed integer $1\le A\le d$, we output $0$ if $\hat n_1<A$ and $1$ otherwise, i.e.\ the interval test of $\ThermalBEQP{}{}$ (\cref{defn:thermal-model}) with $[a,b]=[A,d]$. Since the logical gates preserve the unordered pair $\{a,b\}$, this reads out the logical value correctly whenever the output pair is initialized with $a<A\le b\le d$.
\end{definition}

To show that the multiphoton encoding is valid for \thermalmodel, we prove the following lemma:

\begin{lem}\label{lem:hadamard-passive}
$U_H^{(d)}$ can be constructed as an energy-preserving gate.
\end{lem}
\begin{proof}
    Since we require that $U_H^{(d)}$ only operate correctly on the $\Pi_{\le d}^{\otimes 2}$ subspace, we can use \cref{thm:passive-description}
with parameters $H = \frac{\log(U_H)}{i}$, $m = 2$, $N_{\mathrm{max}} = d+1$ to find an energy-preserving polynomial Hamiltonian $\hat{P}_{H}$ of degree at most $12(d+1)$ such that $\hat{P}_H=H_H^{d} \oplus H'^{\perp}$.
Then $U_H^{(d)} = (e^{i\hat{P}_H})^{(d)}$.
\end{proof}

Since we initialize the input pair to thermal states, we may initialize to a state which lies outside $\calC_0$. We bound this error in the following lemma:

\begin{lem}
\label{lem:thermal-error-estimate}
For the state $\rho = (\rho_L\otimes\rho_H)$, where $\rho_L,\rho_H$ are thermal states of mean photon numbers $0\le \nb_L\ll \nb_H$, the probabilities of measuring the state in the different code spaces are
\begin{subequations}
\be
P(\calC_0) = \la(1-\la^d) - (1-\la)\lz\la\frac{1-(\lz\la)^d}{1-\lz\la}
\ee
\be
P(\calC_1) = \lz(1-\lz^d) - (1-\lz)\lz\la\frac{1-(\lz\la)^d}{1-\lz\la}
\ee
\be
P(\Cinv) = \frac{(1-\lz)(1-\la)}{1-\lz\la}\left(1-(\lz\la)^{d+1}\right) + \lz^{d+1} + \la^{d+1} - (\lz\la)^{d+1},
\ee
\end{subequations}
where $\lambda_k = \frac{\nb_k}{\nb_k + 1}$.
\end{lem}
\begin{proof}
For $k\in\{L,H\}$, define $P_k(n)=(1-\lambda_k)\lambda_k^n$ as the probability of measuring $n$ photons in the cold ($k=L$) or hot ($k=H$) thermal state. Therefore, the probability of measuring the state in $\calC_1$ is
\begin{subequations}
\begin{align}
    P(\calC_1) &= \sum_{a=1}^d \sum_{b=0}^{a-1} P_L(a) P_H(b) = \sum_{a=1}^d P_L(a) \sum_{b=0}^{a-1} P_H(b)\\
    &= \sum_{a=1}^d P_L(a) \sum_{b=0}^{a-1} (1-\la)\la^b = \sum_{a=1}^d (1-\lz)\lz^a (1-\la^a)\\
    &= (1-\lz) \left(\lz\frac{1-\lz^d}{1-\lz} - \lz\la \frac{1- (\lz\la)^d}{1-\lz\la}\right)\\
    &\le \frac{\lz (1- \la)}{1 - \lz\la} = \frac{\nb_L}{\nb_L + \nb_H + 1}.
\end{align}
\end{subequations}
To calculate $P(\Cinv)$, first note $P(\Cinv) = P(a=b \le d) + P(a>d \lor b>d)$
\begin{subequations}
\begin{align}
    P(a = b \le d) &= \sum_{a=0}^d P_L(a) P_H(a) = \frac{(1-\lz)(1-\la)}{1-\lz\la} (1- (\lz\la)^{d+1})\\
    P(a \le d) &= \sum_{a=0}^d P_L(a) = (1- \lz^{d+1})\\
    P(b \le d) &= \sum_{b=0}^d P_H(b) = (1- \la^{d+1})\\
    P(a>d \lor b>d) &= 1 - P(a\le d)P(b\le d)\\
    P(\Cinv) &= P(a = b \le d) + P(a>d \lor b>d) \notag\\
             &= \frac{(1-\lz)(1-\la)}{1-\lz\la}\left(1-(\lz\la)^{d+1}\right) + \lz^{d+1} + \la^{d+1} - (\lz\la)^{d+1} \notag\\
             &\le \frac{1}{\nb_L +\nb_H + 1} + \lz^{d+1} + \la^{d+1}.
\end{align}
\end{subequations}
\end{proof}

Therefore the probability that a logical qubit is initialized incorrectly is

\be
    \label{eqn:thermal-error-estimate}
    p_{\mathrm{init}} \le \frac{1+\nb_L}{1+\nb_L +\nb_H} + \lz^{d+1} + \la^{d+1}.
\ee

\begin{lem}
\label{lem:thermal-readout-error}
Let the output pair be initialized in $\rho = (\rho_L\otimes\rho_H)$ as in \cref{lem:thermal-error-estimate}, and let $1\le A\le d$. Then the probability that the output pair is initialized with $a\ge A$ or $b<A$, so that the first-mode readout of \cref{defn:multi-photon-dual-rail} may err, is
\be
    \label{eqn:thermal-readout-error}
    p_{\mathrm{read}} \le \lz^A + (1-\la^A) \le \lz^A + \frac{A}{1+\nb_H}.
\ee
\end{lem}
\begin{proof}
For a thermal mode, $P(n\ge A) = \sum_{n=A}^\infty (1-\lambda_k)\lambda_k^n = \lambda_k^A$. Hence $P(a\ge A) = \lz^A$ and $P(b<A) = 1-\la^A$, and the first inequality follows from the union bound. For the second, Bernoulli's inequality gives $\la^A = (1-(1-\la))^A \ge 1 - A(1-\la) = 1 - \frac{A}{1+\nb_H}$.
\end{proof}

\begin{theorem}
    \label{thm:thermal-bqphard}
    $\ThermalBEQP{\poly}{\{\nb_L, \nb_H\}}$ is $\BQP$-hard.
\end{theorem}
\begin{proof}

    While we have an expression for the input error, we also need a bound on the output error.
    However, we show that the physical gates leave the $\Cinv$ codespace invariant (\cref{lem:code-space-inv}) and therefore, the output error due to the encoding is upper bounded by $Np_{\mathrm{init}}$.

    \begin{lem}
    \label{lem:code-space-inv}
    $\Cinv$ is invariant under the action of $U_X, U_H, U_T, U_{CZ}$
    \end{lem}
    \begin{proof}
    $U_X \ket{a,b} = \ket{b,a}$, which implies that if $a=b \vee a>d \vee b>d$, the resulting state is also in $\Cinv$.
    $U_T$ and $U_{CZ}$ only act to add a phase, since they are functions of $\hat{n}_1$ and $\hat{n}_2$, and $\ket{a,b}$ is an eigenvector of these.
    The polynomial construction of $U_{H}$ specifically has a block diagonal structure on the $\le d$ and $> d$ local photon number subspaces. For states $\ket{a,a}$, $U_H$ acts as the identity by construction.
    \end{proof}

    In addition, by \cref{lem:thermal-readout-error}, the first-mode readout adds an error of at most $\lz^A + \frac{A}{1+\nb_H}$.

    For some error $\varepsilon$ of the output probability of the ideal $\BQP$ circuit, we need

    \be
        N \left(\frac{1+\nb_L}{1+\nb_L +\nb_H} + \lz^{d+1} + \la^{d+1}\right) + \lz^A + \frac{A}{1+\nb_H} \le \varepsilon
    \ee

    Now for a $\BQP$ problem of input size $\varsigma$, we need $\varepsilon = \frac{1}{\poly(\varsigma)} \equiv N\varsigma^{-c}$, and $N = \poly(\varsigma)$.

    Therefore

    \begin{subequations}
        \be
        \frac{1 + \nb_L}{1+ \nb_L + \nb_H} \le \varsigma^{-c}
        \ee

        \be
        \nb_H = \Omega(\varsigma^c (1 + \nb_L))
        \ee
    \end{subequations}

    These equations force $\nb_H = \poly(\varsigma)$, and we can pick $\nb_L = O(1)$.

    Further, from the definition of the thermal state,

    \be
        \la^{d+1} = \frac{\nb_H}{1+\nb_H}^{d+1} \le \varsigma^{-c}.
    \ee

    Since $\log(1+1/\nb_H)\ge 1/(1+\nb_H)$, we have

    \begin{subequations}
        \begin{align}
            \left(d+1\right)\ge (1+\nb_H)\,c\log\varsigma
            &\implies \frac{d+1}{1+\nb_H}\ge c\log\varsigma \\
            &\implies \left(1+\frac{1}{\nb_H}\right)^{d+1}\ge \varsigma^c,
        \end{align}
    \end{subequations}
    so some $d= \poly(\varsigma)$ suffices.

    Finally, $A = \lceil c\log\varsigma/\log(1/\lz)\rceil = O(\log\varsigma)$ gives $\lz^A\le \varsigma^{-c}$, and $\frac{A}{1+\nb_H}\le \varsigma^{-c}$ holds for $\nb_H = \Omega(\varsigma^c\log\varsigma)$. Both are compatible with $A\le d$.

    Therefore, if we pick $d = \poly(\varsigma),\ \nb_H = \poly(\varsigma),\ \nb_L = O(1),\ A = O(\log\varsigma)$, we get the required error scaling, and the qubit universality of $\{U_X, U_H, U_T, U_{CZ}\}$ implies we can simulate a $\BQP$ circuit to $\frac{1}{\poly}$ error.
\end{proof}

\begin{theorem}
    \label{thm:thermal-inbqp}
    $\ThermalBEQP{\poly}{\{\nb_L, \nb_H\}}$ is in $\BQP$.
\end{theorem}
\begin{proof}
    Containment follows from the simulation argument of \cref{thm:in_BQP}. Indeed, the total expected input energy is $\poly(N)$ and is preserved throughout the circuit. Thermal states have geometric Fock distributions with efficiently computable probabilities, so after truncating their tails at a polynomial cutoff they are efficiently preparable polynomial-dimensional mixed states. Every gate acts on $O(1)$ modes and has polynomial degree, and its truncated Fock entries are efficiently computable by \cref{lem:assumption-satisfied}. Thus the proof of \cref{thm:in_BQP} applies even for thermal states, giving containment in $\BQP$.
\end{proof}

\cref{thm:thermal-bqphard} and \cref{thm:thermal-inbqp} prove the claim, which we restate below.

\thermalbqp*

\begin{remark}[Constant-degree gates]
\label{rem:thermal-poly-SK}
The polynomial degree of $U_H$, $U_T$ and $U_{CZ}$ is not essential. Each acts on $m\le4$ modes which, on runs where every pair is initialized in $\calC_0$, hold at most $K=4d$ photons. Hence \cref{thm:cutoff-SK} compiles each into $\poly(d,\log(1/\epsilon))$ gates from $\GKerr$ that are $\epsilon$-accurate on $\calH^{(m)}_{\le K}$. So a circuit of $L$ logical  gates incurs at most $L\epsilon$ additional error, and $\epsilon=1/\poly$ suffices.

However, \cref{thm:cutoff-SK} requires a vacuum ancilla, which the thermal model does not provide unless $\nb_L=0$. We therefore use the ancilla-free \cref{thm:cutoff-SK2}, which requires unit determinant on every photon-number sector. To this end, replace each gate $U$ by
\be
    U' = U\,e^{-i\varphi(\hat N)},\qquad \varphi(n) = \frac{\arg\det U^{(n)}}{D_n},
\ee
where $\hat N$ is the total photon number of its $m$ modes, so that $U'\in\mathrm{SU}(D_0,\dots,D_K)$. Every gate preserves the photon number of each pair, and the thermal input is diagonal in these numbers. The correction $e^{-i\varphi(\hat N)}$ therefore acts as a global phase on each branch of the initial mixture and leaves the output distribution unchanged. Consequently, \cref{thm:thermal-poly-bqp} also holds with constant-degree gates from $\GKerr$ (or $\GcrossKerr$), at polynomial overhead.
\end{remark}

\subsubsection{Thermal energy-preserving with constant per-mode energy is \texorpdfstring{$\BQP$}{BQP}-complete}

In \cref{sec:thermal-poly}, a $\frac{1}{\poly}$ error is achieved by letting the hot mean photon number grow with the problem size, $\nb_H=\Omega(n^c\log n)$, together with a cutoff $d=\poly(n)$. Consequently, both the energy per mode and the logical gates depend on the input size: every gate must act correctly on a Fock-space truncation that grows with $n$, whether it is implemented directly with polynomial degree or compiled into $\GKerr$ via \cref{rem:thermal-poly-SK}.

It is therefore natural to ask whether two \emph{fixed} temperatures already suffice, with $\nb_L$, $\nb_H$ and the cutoff $d$ all independent of $n$. In this regime the energy per mode is constant, and the gates form a single fixed family acting on a fixed truncation, independent of the problem size.

A constant energy difference would lead to a constant error rate (\cref{eqn:thermal-error-estimate}), which cannot simulate the $\BQP$ circuit accurately enough. However, we can use error correction schemes to bound the error growth back to the required $O\left(\frac{1}{\poly}\right)$.

In this section, we show that there exists a constant threshold temperature difference above which an error correction scheme exists. We first state a corollary of \cref{lem:thermal-error-estimate}:

\begin{corollary}
\label{cor:arbitrary-error}
For an $\varepsilon > 0$ with $\varepsilon = O(1)$, there exists a pair of mean photon numbers $\nb_H = O(1),\nb_L =O(1)$, such that $P(\calC_1) + P(\Cinv) \le \varepsilon$.
\end{corollary}

When the state is in $\mathcal{C}_{\mathrm{inv}}$, primarily in the local photon $> d$ subspace, the gates are not well defined on the computational space. We interpret this as the logical qubit ``failing.'' This physical behavior motivates a mapping to a formal noise model where qubits are initialized with a static error status.

\begin{definition}[Static Initialization Noise Model]
    \label{def:sine-model}
Let $Q$ be an arbitrarily large pool of physical qubits initialized at time $t = 0$. Each physical qubit in $Q$ is assigned an i.i.d. status $s \in \{\text{Fail}, \text{Ideal}\}$ with probabilities $p$ and $1 - p$.

All gates are assumed to execute with perfect fidelity, subject only to the initial state of the target qubits. For an arbitrary ideal $\Delta$-local operation $\mathcal{G}$ acting on a subsystem $S \subseteq Q$:
\begin{enumerate}
    \item If $s_q \neq \text{Fail}$ for all $q \in S$, then $\mathcal{G}$ is executed perfectly.
    \item If there exists $q \in S$ such that $s_q = \text{Fail}$, then $\mathcal{G}$ is replaced with an adversarial CPTP map $\mathcal{E}$ acting on $S$, with the constraint that $\mathcal{E}$ does not change the status of any qubit. Explicitly, the CPTP maps may be implemented via a shared environment, allowing for correlated noise.
\end{enumerate}
\end{definition}

To prevent the spatial propagation and accumulation of these failures, we introduce a randomized routing gadget that periodically swaps the computational state onto fresh qubits.

\begin{definition}[\DATASWAP]
\be \DATASWAP_{1,2} = (\II\otimes H) CZ_{1,2} (H\otimes H) CZ_{1,2} (H\otimes H) CZ_{1,2} (\II\otimes H)\ee
\end{definition}

\begin{definition}[Routing Gadget $\mathcal{T}$]
    \label{def:routing-gadget}
At each time step $t$, if we are to apply the $m$-local operation $\mathcal{G}_t$ on the active target register $\bm{q}_{t-1}$, we replace $\mathcal{G}_t$ with the routing gadget $\mathcal{T}$:
\begin{enumerate}
    \item For each active qubit $q_{t-1} \in \bm{q}_{t-1}$:
        \begin{enumerate}
            \item Draw a fresh qubit $a$ from the pool $Q$.
            \item Apply the $\DATASWAP$ gate\footnote{We emphasize that $\text{DATA-SWAP}$ must be implemented using compiled polynomial constructions rather than a simple linear-optical swap using two beam splitters. Because the compiled gates leave the high-photon-number subspace $\mathcal{C}_{\mathrm{inv}}$ invariant (as guaranteed by \cref{lem:code-space-inv}), a failed mode with an arbitrarily high photon number in $\mathcal{C}_{\mathrm{inv}}$ cannot leak photons into the fresh mode. A beamsplitter swap, by contrast, would physically mix the modes and allow high photon counts to propagate to the rest of the circuit, violating the non-poisoning property of the noise model.} between $q_{t-1}$ and $a$.
            \item Permanently discard $q_{t-1}$ and label $a$ as the active register qubit $q_{t-1}$.
        \end{enumerate}
    \item Apply $\mathcal{G}_t$ on the fresh register $\bm{q}_{t-1}$.
\end{enumerate}
\end{definition}

We note that if all qubits participating in the gadget are Ideal, then the gadget emulates the action of $\mathcal{G}$. If any qubit is Failing, then an arbitrary gate is applied, but the status of the participating qubits is randomly, independently set. Thus, the effect of a failing qubit is limited to only two rounds that it participates in. Therefore, we have mapped static initialized noise to independently sampled wire section noise. However, the gate noise is still correlated, as the two ends of the wire both fail.

To apply standard quantum fault-tolerance threshold theorems, we must show that the correlated distribution of gate failures is stochastically dominated by a distribution of completely independent gate failures.

We formalize the quantum circuit under the Static Initialization Noise Model (\cref{def:sine-model}) equipped with the Routing Gadget as a graph $G = (V,E)$.

\begin{definition}[Circuit Graph]
Let $G = (V,E)$ be a graph representing the execution of the quantum circuit.
\begin{itemize}
    \item \textbf{Vertices ($V$):} Each vertex $v \in V$ corresponds to a Routing Gadget.
    \item \textbf{Edges ($E$):} Each edge $e \in E$ corresponds to a physical qubit drawn from the pool $Q$. Thus, $e = \{u, v\}$ connects operations $u$ and $v$, if the qubit participates in these gates.
    \item \textbf{Maximum Degree ($\Delta$):} The maximum number of physical qubits involved in any single gate operation, corresponding to the maximum degree of $G$.
\end{itemize}
\end{definition}

\begin{definition}[Random Wire Error Noise Model]
    For each edge $e \in E$, let $E_e \in \{0, 1\}$ be the indicator that the physical qubit $e$ is initialized to \texttt{Fail}. By \cref{def:routing-gadget} on the static initialization noise model, the family $(E_e)_{e \in E}$ consists of independent and identically distributed (i.i.d.) Bernoulli random variables with parameter $p$:
\be E_e \sim \text{Ber}(p) \ee
A gate operation $v \in V$ is executed perfectly if and only if all incident qubits are \texttt{Ideal} (${E}_e = 0$). If any incident qubit fails, the gate undergoes an adversarial CPTP map. Let $F_v \in \{0, 1\}$ indicate the physical failure of gate $v$:
\be F_v = \max_{e \ni v} E_e \ee
Note that the random variables $(F_v)_{v \in V}$ are correlated because adjacent gates share edges.
\end{definition}

\begin{theorem}[Stochastic domination of correlated errors]
    \label{def:stoc-dom}
Let $G=(V,E)$ be a circuit graph of maximum degree $\Delta$, subject to the physical edge-error distribution $E_e \sim \text{Ber}(p)$ i.i.d. There exists an i.i.d.\ virtual gate error distribution $(V_v)_{v \in V}$ with local failure rate $p_{\mathrm{gate}}$ such that:
\begin{enumerate}
    \item The variables $V_v$ are mutually independent.
    \item The local failure rate is exactly $p_{\text{gate}} = 1 - (1 - \sqrt{p})^\Delta$.
    \item The virtual error model stochastically dominates the physical error model, meaning $F_v \le V_v$ for all $v \in V$.

        For small $p$, $p_{\text{gate}} = \Delta \sqrt{p} + O(p)$.
\end{enumerate}
\end{theorem}
\begin{proof}

We establish stochastic domination via an explicit marking on vertices mapped from the marking on edges.

For every edge $e = (u, v)$, construct two ``socket'' variables $S_{u, e}, S_{v, e}$ conditioned on the value of $E_e$.

If $E_e = 1$, we define $(S_{u, e}, S_{v, e}) = (1, 1) \text{ w.p. } 1$.

If $E_e = 0$, we define $(S_{u, e}, S_{v, e}) \sim
\begin{cases}
    (1, 0) & \text {wp }\frac{\sqrt{p}}{1 + \sqrt{p}}\\
    (0, 1) & \text {wp }\frac{\sqrt{p}}{1 + \sqrt{p}}\\
    (0, 0) & \text {wp }\frac{1 - \sqrt{p}}{1 + \sqrt{p}}\\
\end{cases}$.

It is easy to see these are well-defined probability distributions. Let us calculate the marginals of the socket variables.
\begin{subequations}
\begin{align}
    P(S_{u,e} = 1) &= P(S_{u,e} = 1 | E_e = 0) P(E_e = 0) + P(S_{u,e} = 1 | E_e = 1)P(E_e = 1) \\
    &= \frac{\sqrt{p}}{1 + \sqrt{p}} \cdot (1 - p) + 1 \cdot p \\
    &= \sqrt{p}.
\end{align}
\end{subequations}
From symmetry, $S_{v, e}$ has a similar marginal.

We now show that $S_{u,e}$ and $S_{v, e}$ are independent.

\be P(S_{u,e} = 1) P(S_{v,e} = 1) =  \sqrt{p} \sqrt{p} = p = P(E_e = 1) = P(S_{u,e} = 1, S_{v,e} = 1)\ee
\be P(S_{u,e} = 0) P(S_{v,e} = 1) =  (1 - \sqrt{p}) \sqrt{p} = (1-p) \frac{\sqrt{p}}{1+\sqrt{p}} = \frac{\sqrt{p}}{1+\sqrt{p}} P(E_e = 0) = P(S_{u,e} = 0, S_{v,e} = 1)\ee
\be P(S_{u,e} = 1) P(S_{v,e} = 0) =  (1 - \sqrt{p}) \sqrt{p} = (1-p) \frac{\sqrt{p}}{1+\sqrt{p}} = \frac{\sqrt{p}}{1+\sqrt{p}} P(E_e = 0) = P(S_{u,e} = 1, S_{v,e} = 0)\ee
\be P(S_{u,e} = 0) P(S_{v,e} = 0) =  (1 - \sqrt{p})^2  = (1-p) \frac{1 - \sqrt{p}}{1+\sqrt{p}} = \frac{1-\sqrt{p}}{1+\sqrt{p}} P(E_e = 0) = P(S_{u,e} = 0, S_{v,e} = 0)\ee

Further, due to the independence of $E_e$ and $E_{e'}$, the socket variables are all mutually independent.

Now let us define $V_v$ for each vertex $v$. Let $d(v)$ be the number of edges incident on $v$. Then define dummy socket variables $D_{j, v} \sim \text{Ber}(\sqrt p)$ for $1 \le j \le \Delta - d(v)$. Then define

\be V_v = \max(\max_{e \ni v} S_{v, e}, \max_j D_{j, v})\ee

Firstly, due to the mutual independence of the underlying socket and dummy socket variables, $V_v$ are mutually independent.

Since $V_v = 0$ if and only if all the underlying variables are $0$,

\be P(V_v = 1) = 1 - (1- \sqrt{p})^\Delta = p_{\text{gate}}\ee

If $V_v$ has to stochastically dominate $F_v$, we need to show that $F_v = 1 \implies V_v = 1$. If $F_v = 1$, then by definition, $\exists e \text{ such that } (e = (u, v) \lor e = (v, u)) \land E_e = 1$.
Then by construction, $S_{v, e} = 1$. Therefore, $V_v = 1$, and therefore $F_v \le V_v$ under the new mapping, completing our proof.
Because we have constructed a valid marking where the physical errors are a subset of the independent virtual errors, the correlated physical error model is stochastically dominated by the i.i.d.\ virtual error model.
\end{proof}

We now show that the thermal computation is an instance of the Static Initialization Noise Model (\cref{def:sine-model}). Since the thermal states are diagonal in the Fock basis, the input is a classical mixture over the initial occupations sampled from the thermal distribution (\cref{def:thermal-state}). If the resulting state $\ket{a, b}\in \calC_0$, we mark it Ideal, or Fail otherwise. Furthermore, the gates we have constructed explicitly leave the $\Cinv$ invariant (\cref{lem:code-space-inv}). Therefore, the modes do not leak into or out of $\Cinv$, allowing us to keep the status static. If all the modes a gate operates on are Ideal, then the gate operates correctly on the logical subspace. If any mode is Failing, then the action of the gate is unknown. Therefore, the error model of the logical computation encoded by the thermal computation is an instance of the Static Initialization Noise Model.

To complete the mapping from edge error to gate error, we need to construct an adversary that can simulate the actual bosonic computation under the following constraints:
\begin{enumerate}
\item The adversary has access to all the qubits participating in a failing gate.
\item The adversary has access to an environment.
\item The adversary can apply an arbitrary joint unitary on the qubits participating in failing gates at that time step and the environment.
\end{enumerate}

The adversary fixes a large photon cutoff $K_{\mathrm{env}} > d$. The adversary starts the computation by sampling the thermal states by sampling photon numbers $0\le a,b\le K_{\mathrm{env}}$ from the thermal distribution. The adversary then stores the photon numbers in its environment. For every failing gate, the adversary simulates the underlying physical computation by evaluating the appropriate polynomial that generates the gate, and then updating the corresponding entries in its environment and the qubits. The $K_{\mathrm{env}}$ truncation costs at most $Q(\lz^{K_{\mathrm{env}}+1}+\la^{K_{\mathrm{env}}+1})$ in trace distance for $Q$ pairs. At fixed temperatures, $K_{\mathrm{env}} = O(\log(Q/\varepsilon))$ suffices.

Therefore, through a sequence of mappings, we have converted the highly correlated Static Initialization Noise model to an Uncorrelated Gate Model which is stochastically dominant (\cref{fig:thermal_schematic}). If we are able to correct the gate error model for some values of $\lambda_L, \lambda_H, d$, then we can correct the underlying physical computation, allowing us to simulate $\BQP$ circuits.

\begin{corollary}
\label{cor:threshold-thm}
Let $p_{\mathrm{th}}$ be the threshold of \cite{AB08}. If $p_{\mathrm{gate}} = 1 - (1 - \sqrt p)^\Delta < p_{\mathrm{th}}$, there exists a fault-tolerant quantum error correction protocol that corrects Static Initialization Noise Error (\cref{def:sine-model}) with the routing gadgets (\cref{def:routing-gadget}). Since $p_{\mathrm{gate}} \le \Delta \sqrt p$, it suffices that
\be p < \left(\frac{p_{\mathrm{th}}}{\Delta}\right)^2. \ee
\end{corollary}

\begin{proof}
    By \cref{def:stoc-dom}, there is an i.i.d.\ virtual fault pattern $V_v(u)$ with rate $p_\mathrm{gate}$ such that every physically faulty gadget is virtually faulty. Conditioned on a fault path $\bm{V}$, the physical computation done by the adversary is an allowed strategy, and locations outside the fault path are ideal. Since the threshold theorem holds for every such adversary, it also holds for a mixture over the adversaries conditioned on all fault paths. Since $V_v$ is i.i.d., its distribution satisfies \cite[Section~10.1, Eq.~(10.1)]{AB08} with $c = 1$ and $\eta = p_{\mathrm{gate}}$, satisfying all the conditions required for the threshold theorem to hold.
\end{proof}

Now we can combine \cref{cor:arbitrary-error} and \cref{cor:threshold-thm}:

\begin{corollary}\label{cor:thermalBQP}
    $\ThermalBEQP{\poly}{\nb_L = O(1), \nb_H = O(1)}$ is $\BQP$-complete.
\end{corollary}
\begin{proof}
Start with a $\BQP$ circuit whose error is at most $1/6$, using the standard equivalent bounded-error convention.

Measure the first output mode using the interval $[A,d]$, as in \cref{sec:thermal-poly}. This reads the logical bit correctly whenever the output pair is \texttt{Ideal} and $a<A\le b$. Fix $\nb_L$, then choose constants $A$, $\nb_H$, and $d\ge A$, in that order, so that $p_{\mathrm{gate}}<p_{\mathrm{th}}$ and the readout error satisfies
\be p+\lz^A+\frac{A}{1+\nb_H}<\frac1{12}. \ee
Such choices exist by \cref{eqn:thermal-error-estimate,lem:thermal-readout-error}.

Apply \cref{cor:threshold-thm}, choosing the simulation and truncation errors to sum to less than $1/12$. The total error is then less than $1/6+1/12+1/12=1/3$, meeting the promise inequalities of $\BQP$. Therefore, we simulate a $\BQP$ circuit, and show $\BQP$-hardness. Containment follows from \cref{thm:thermal-inbqp}.

Since the construction uses polynomially many pairs and gates, with constant per-mode input energy and gate degree, %
$\ThermalBEQP{\poly}{\nb_L = O(1), \nb_H = O(1)} = \BQP$.
\end{proof}

\begin{figure}[htbp]
\centering
\resizebox*{!}{0.8\textheight}{\input{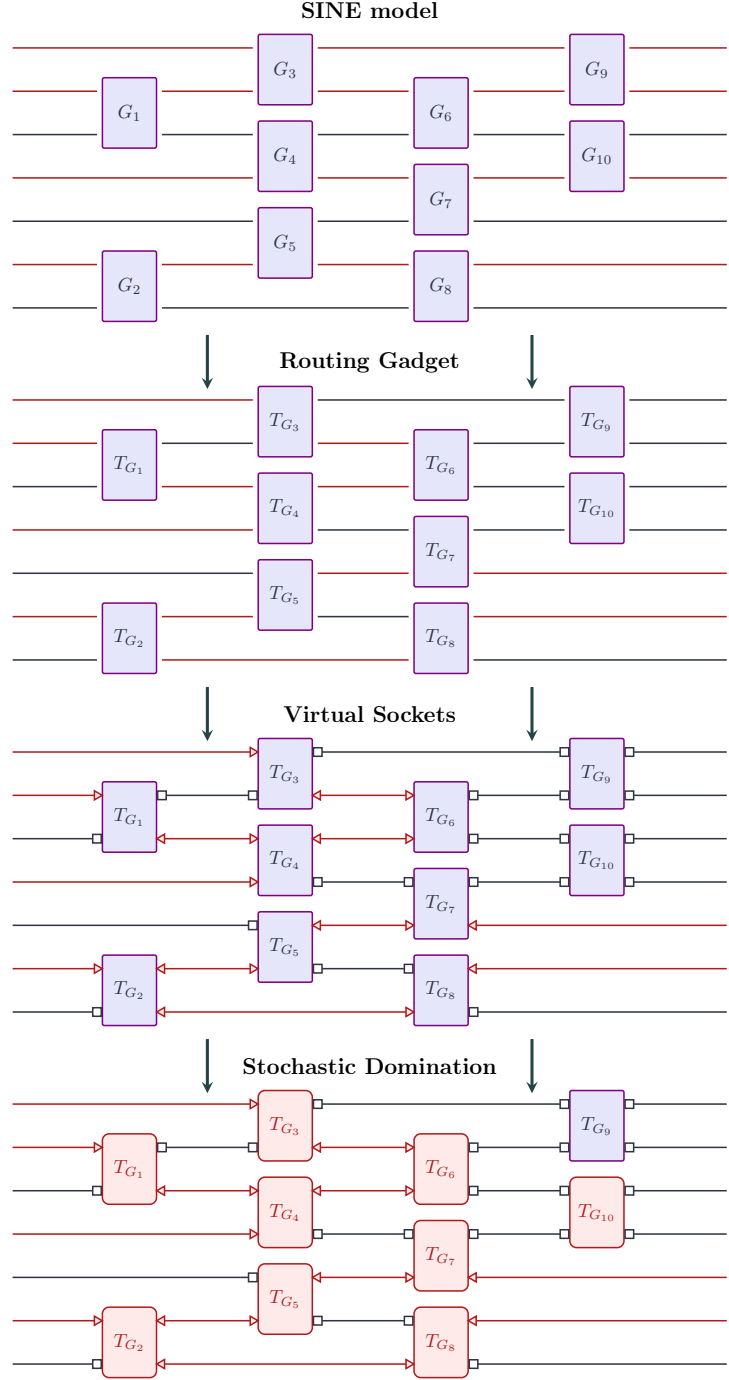}}
\caption{A schematic of the sequence of mappings from the Static Initialization Noise Error (SINE) model (\cref{def:sine-model}) to the Uncorrelated Gate Error model. Red components signify failing components. The Routing Gadgets make the qubit failures into wire section failures. Virtual sockets are then introduced, which are marked failing if the edge they connect to fails. Finally, gates are marked such that the gate must fail if any of its sockets fail. Note the failure of $T_{G_{10}}$, which occurs in spite of all the sockets passing, indicative of the stochastic dominance. The virtual socket failure rate is $\sqrt p$; the resulting virtual gate failure rate is $1-(1-\sqrt p)^\Delta$, where $\Delta$ is the maximum gate degree.}
\label{fig:thermal_schematic}
\end{figure}

\subsection{Computations with non-local energy-preserving gates of polynomial degree are \texorpdfstring{$\PSPACE$}{PSPACE}-complete}

In this section, we explore the computational power of nonlocal energy-preserving gates, i.e.\ gates that can operate on a number of modes that scales with the input size. We retain exactly the input model of \cref{scn:model-definition}: the computation starts with vacuum modes and product coherent states, and the total expected input photon number is polynomial. We prove that allowing polynomially described, nonlocal gates generated by polynomial-degree energy-preserving Hamiltonians makes this model $\PSPACE$-complete. This class of gates includes a single evolution generated by a sum of constant-local Hamiltonians, such as those in nearest-$k$-neighbor interaction models.

We previously saw that circuits composed of polynomially many constant-local gates of polynomial degree can be simulated in $\BQP$.
Therefore, an implication of the results of this section is that there is no efficient decomposition of nonlocal unitaries to constant-local bosonic gates unless $\PSPACE = \BQP$ (\cref{cor:no-compile}).
This result is interesting for physical models of the form $\exp (i (\sum_j h_j) )$ where $h_j$ is a constant-local Hamiltonian of polynomial degree.

\newcommand{\NLP}{\mathsf{NLP}}
\begin{theorem}
    \label{claim:pspace-complete}
    Let $\NLP$ (non-local $\poly$) be the set of polynomially described (polynomial-bit coefficients and polynomial-size monomial list), possibly nonlocal, energy-preserving polynomial Hamiltonians of degree $\poly(n)$. Then $\BEQP{\NLP}{}=\PSPACE$.
\end{theorem}

\begin{theorem}
\label{thm:pspace-hardness}
    The simulation of $\NLP$ in \cref{claim:pspace-complete}, with vacuum and coherent states as its only input states, is $\PSPACE$-hard.
\end{theorem}

We prove this by exponentially fast forwarding a Universal Hamiltonian simulation introduced in \cite[Theorem 3]{bohdanowicz2017universalhamiltoniansexponentiallylong} of a $\BQP$ circuit for an exponential amount of time using $\poly$ energy and $\poly$ degree, allowing us to show that we can solve $\PSPACE$-hard problems.

\begin{proof}
Let $V$ be a polynomial-size circuit on $n$ qubits, and let $t$ be an evolution time. From \cite[Theorem 3]{bohdanowicz2017universalhamiltoniansexponentiallylong}, there exists a family of translationally invariant Hamiltonians $\{H_m\}$, and $\poly$-size unitaries $D$ and $E$ on $m$ qudits of dimension $14580$, such that
\be \norm{V^t - \left(\II^{\otimes n} \otimes \bra {0^{m-n}}\right) D e^{i H_m t'} E \left(\II^{\otimes n} \otimes \ket{0^{m-n}}\right)} < 1/\poly(n),\ee
where $m = \poly(n, \log t)$ and $t' = \poly(t,n)$, and
\be \hat{H}_m = \frac{1}{m} \sum_{l=1}^{m-1} \hat{h}_{l,l+1}.\ee

However, their construction does not provide a deterministic way to find $D, E$ and $t'$, and rather provides a run and sample construction. We use the time-averaged completion guarantee in the reference, and construct a coherent time register which results in an average over all runtimes, rather than assuming an efficiently computable successful individual evolution time.

We encode each qudit in the dual-rail encoding of its binary representation. For example, the level $\ket{21}$, whose binary representation is $0010101$, is encoded as $\ket{01011001100110}$. The encoded register has $2\lceil\log(14580)\rceil m=\poly(n,\log t)$ modes and exactly $\lceil\log(14580)\rceil m$ photons.

Each two-qudit term $\hat h_{l,l+1}$ induces an energy-preserving Hamiltonian $\hat h_{B_{l,l+1}}$ on the corresponding $4\lceil\log(14580)\rceil=O(1)$ bosonic modes. On the encoded subspace every local occupation number is at most one. Applying \cref{thm:passive-description} on these $4\lceil\log(14580)\rceil$ modes with $N_{\mathrm{max}}=2$ therefore gives an $O(1)$-degree energy-preserving polynomial description of each encoded term. Their sum
\be
    H_B=\frac{1}{m}\sum_{l=1}^{m-1}h_{B_{l,l+1}}
\ee
has a polynomial-size description and implements $H_m$ on the encoded subspace. While the encoding itself is exact, $H_m$ has algebraic irrational coefficients, which need to be approximated to $\poly$ bits of precision for exponentially long applications.

The construction in \cite{bohdanowicz2017universalhamiltoniansexponentiallylong} guarantees high completion probability averaged over $t \in[0,t']$,\cite[Lemma 7]{bohdanowicz2017universalhamiltoniansexponentiallylong} rather than at an explicitly specified individual time. Fix a constant $0<\delta<1/12$, the error in the simulation.

Choose $M=2^r$ and prepare an $r$-qubit dual-rail register in
$M^{-1/2}\sum_{j=0}^{M-1}\ket j$. Define
\begin{subequations}
    \be
    R=\sum_{k=0}^{r-1}2^k\hat n_{k,1},
    \ee
    \be
    H_{\mathrm{coh}}
      =H_B\otimes(R/M),
    \ee
\end{subequations}
where $\hat n_{k,1}$ counts photons in the logical-one rail.
When evolved for a time $t'$, this produces
\[
    \frac1{\sqrt M}\sum_{j=0}^{M-1}
    e^{iH_Bt'j/M}\ket{\psi_0}\otimes\ket j.
\]
Leaving the time register unmeasured therefore implements the discrete time average. If $\|H_B\| \le B$ on the encoding space ($H_B$ may be unbounded on the whole fock space), the discretization error is at most $Bt'/M$. Thus $M\ge Bt'/\delta$ suffices, requiring only $r=\poly(n)$ qubits when $t'=O(\exp(n))$. $H_{\mathrm{coh}}$ remains energy-preserving constant degree and polynomially described.

By the run-time analysis in \cite[Section 5.4, 5.5] {bohdanowicz2017universalhamiltoniansexponentiallylong}, the simulator parameters can be chosen so that evolution for a uniformly sampled time in $[0,t']$ has probability at least $1-\delta$ on completed computations, where $t'\le\poly(t,n)$.
For $t=2^{\poly(n)}$, this gives
\[
    t'\le\poly(2^{\poly(n)},n)=2^{\poly(n)}.
\]
We use this $t'$ as the averaging cutoff in the coherent time-averaging construction. The discrete average then completes with probability at least $1-2\delta$. On completed histories the work register contains the desired output; unfinished histories have total weight at most $2\delta$. Consequently, measuring only the logical-one rail of the output qubit changes the ideal acceptance probability by at most $2\delta$.

Introduce a control mode and fix an integer $q\geq2$. For $d=\poly(n)$ define
\be
H_{\mathrm{FF}}=H_{\mathrm{coh}} \otimes \hat n_c^d.
\ee
This Hamiltonian is energy-preserving, has degree $\poly(n)$, and has a polynomial-size description. Crucially, the control must be the number state $\ket q$, not a coherent state: since $\hat n_c^d\ket q=q^d\ket q$,
\be
    e^{iTH_{\mathrm{FF}}}
    \left(
    \ket{\psi_0}\otimes
    \frac1{\sqrt M}\sum_{j=0}^{M-1}\ket j\otimes\ket q
    \right)
    =
    \frac1{\sqrt M}\sum_{j=0}^{M-1}
    e^{iH_Bt'j/M}\ket{\psi_0}\otimes\ket j\otimes\ket q.
\ee
Choosing $d$ so that $q^d\geq t'$ and setting $T=t'/q^d\leq1$ implements $e^{iH_{\mathrm{coh}}t'}$. For $t=2^{\poly(n)}$, both $m$ and $d$ are polynomial, while the encoded register and control contain only $\poly(n)$ photons.

It remains to reconcile these finite Fock states with the input rule in \cref{scn:model-definition}. Start with vacuum modes and a coherent ancilla of polynomial energy. We can directly use \cref{cor:arbitrary states} to prepare the initial data, time registers, ancillas and the control state mode-by-mode to inverse-polynomial trace distance. Thus, the simulation circuit is an instance of $\BEQP{\NLP}{}$. The preparation circuit, time register, fast-forwarded evolution, together form a polynomial-size circuit of energy-preserving gates.

Finally, choose $V$ to be the quantum circuit implementing an exponentially long, polynomial quantum space computation and measure the output qudit. The state-preparation error changes this acceptance probability by at most inverse polynomial, so the usual constant completeness--soundness gap is preserved. Since every $\PSPACE$ computation has such an exponential-time, polynomial-space circuit, the coherent-input simulation problem is $\PSPACE$-hard.
\end{proof}

\begin{corollary}
\label{cor:no-compile}
Non-local unitaries, even if generated by $k$-local $\poly$-degree energy-preserving Hamiltonians, cannot be compiled into unitaries generated by local $\poly$-degree polynomials to constant additive precision efficiently, even when restricted to $\poly$ energy states, unless $\BQP = \PSPACE$.
\end{corollary}
\begin{proof}
Let us assume that there is an efficient decomposition of $H_{\mathrm{FF}}$ to local $\poly$ degree gates. Therefore, we have local $\poly$ degree energy-preserving gates being applied to $\poly(n)$ energy states, allowing us to use \cref{lem:assumption-satisfied} to simulate it in $\BQP$. Therefore, unless $\BQP \supseteq \PSPACE$, an efficient decomposition of arbitrary $\poly$ degree gates to $\poly$ degree local gates is not possible, even when restricted to $\poly(n)$ energy states.
\end{proof}

\begin{remark}
    Note that we can define $\hat{H}_{\mathrm{FF}} = \hat{H}_{B} \otimes \hat{n}$, and operate it on $\ket{\text{encoded data}}\ket{E^d}$, and still get fast-forwarding.
    Therefore, a natural extension to \cref{thm:pspace-hardness} is that simulating constant degree non-local energy-preserving operations on $\poly$ modes and exponential energy to constant additive precision is also $\PSPACE$-hard. A similar no-go can be shown for exponential energy states, if we can upper bound unitaries generated by simulating local constant-degree Hamiltonians on exponential-energy states to a class smaller than $\PSPACE$. Such a simulation remains open as of yet (\cref{sec:open-questions}).
    However, this does not neatly fall into $\BEQP{\mathsf{NLC}}{\exp}$ ($\mathsf{NLC}$ means non local constant degree) since we do not have an efficient bound on the preparation circuit for $\ket{q}$ where $q$ is an exponentially large number, as \cref{thm:cutoff-SK} is only efficient for $K = \poly$ cutoffs.
\end{remark}

\begin{theorem}\label{thm:inpspace}
    The coherent-input simulation problem in \cref{claim:pspace-complete} is in $\PSPACE$.
\end{theorem}
\begin{proof}
We give an explicit classical polynomial-space simulation. Consider first one gate $U=e^{iHt}$, where $H$ is specified by a polynomial-size monomial list with polynomial-bit coefficients and $t$ has a polynomial-bit description. The input is exactly that of \cref{scn:model-definition},
\be
    \ket{\psi}=\ket{\bm0}\otimes\ket{\bm\alpha},
    \qquad E=\|\bm\alpha\|^2=\poly(\varsigma),
\ee
where the real and imaginary parts of every $\alpha_j$ have polynomial-bit descriptions. Its Fock amplitude is
\be
    \langle\bm y|\psi\rangle
    =e^{-E/2}\prod_j\frac{\alpha_j^{y_j}}{\sqrt{y_j!}},
\ee
and is computable to polynomially many bits in polynomial space.

Fix a constant target error $\varepsilon>0$, and let $\Pi_{\leq K}$ be the projector onto the subspace of total photon number at most
\be
    K=\left\lceil \frac{E}{\delta^2}\right\rceil,
\ee
where $\delta>0$ is a sufficiently small constant chosen below. By Markov's inequality,
\be
    \langle\psi|(\II-\Pi_{\leq K})|\psi\rangle\leq \frac{E}{K+1}\leq \delta^2.
\ee
Consequently, the normalized truncated state $\ket{\psi_K}$ differs from $\ket{\psi}$ by $O(\delta)$ in trace distance. Since $H$ is energy-preserving, $[H,\Pi_{\leq K}]=0$, so this same bound holds after applying $U$. Thus it suffices to compute the desired output probabilities for $U\ket{\psi_K}$ to additive error $\varepsilon-O(\delta)$.

The truncated space $\mathcal H_{\le K}^{(m)}$ has dimension
\be
    D=\binom{m+K}{K}.
\ee
Although $D$ can be exponential, its basis states are occupation vectors $\bm r$ with $|\bm r|\leq K$, and hence can be indexed using $O(\log D)=\poly(\varsigma)$ bits. Number preservation makes $U|_{\mathcal H_{\le K}^{(m)}}$ block diagonal, with its $r$-photon block equal to $e^{iH_rt}$.

We next show that an entry of this truncated unitary is computable in polynomial space. Given two occupation vectors $\bm k,\bm\ell$, the corresponding entry is zero unless $|\bm k|=|\bm\ell|$. Otherwise, the matrix entries of $H_r$ are computed by iterating through the polynomial-size monomial description of $H$ and applying the ladder-operator formula. By \cref{lem:mat-exp-in-photon-block}, a scaling-and-squaring Taylor expansion computes entries of $e^{iH_rt}$ to $p$ bits using working precision polynomial in $p$, $\log D$, and the input length.

For completeness, this calculation need not store the exponentially large matrix $H_r$ or any of its powers. Evaluate every matrix product appearing in the Taylor expansion and the subsequent squaring recursively, using
\be\label{eq:recursion1}
    (AB)_{\bm k,\bm\ell}
    =\sum_{|\bm q|=r} A_{\bm k,\bm q}B_{\bm q,\bm\ell}.
\ee
A depth-first evaluation keeps only the two output indices, the current summation index, a polynomial-precision accumulator, and a recursion stack. The stack depth is the Taylor degree plus the number of squaring steps, both polynomial in the input length and $p$; each index uses $O(\log D)$ bits. Hence each entry of $U|_{\mathcal H_{\le K}^{(m)}}$ is computable to $p$ bits in polynomial space (and, possibly, exponential time). For a polynomial-size circuit $U_L\cdots U_1$, recursively use
\be\label{eq:recursion2}
    \langle\bm x|U_L\cdots U_1|\bm y\rangle
    =\sum_{|\bm z|\leq K}
      \langle\bm x|U_L|\bm z\rangle
      \langle\bm z|U_{L-1}\cdots U_1|\bm y\rangle.
\ee
This adds only polynomial recursion depth and polynomial-size indices, so the same polynomial-space bound holds for the full circuit.

It remains to compute a measurement probability. For each occupation vector $\bm x$ with $|\bm x|\leq K$, calculate
\be\label{eq:recursion3}
    a_{\bm x}=\sum_{|\bm y|\leq K}
    \langle\bm x|U|\bm y\rangle\langle\bm y|\psi_K\rangle,
\ee
by enumerating $\bm y$ and maintaining one polynomial-precision accumulator. The normalizing factor $\|\Pi_{\leq K}\psi\|^{-1}$ is computed in the same way by enumerating the truncated coherent-state probabilities. To obtain the probability of any fixed outcome on a constant set of modes, enumerate the occupation vectors $\bm x$ consistent with that outcome and accumulate $|a_{\bm x}|^2$. All enumeration counters use $O(\log D)$ bits. Choosing the gate-entry and input-amplitude precision to be $2^{-p}=O(\varepsilon/(L D^{L+4}))$ for a circuit of $L=\poly(\varsigma)$ gates makes the total numerical error from the nested sums at most $\varepsilon/2$; although the running time may be exponential, $p=O(L\log D+\log(L/\varepsilon))=\poly(\varsigma)$, so the space remains polynomial. Combining this with the $O(\delta)$ truncation error, and choosing $\delta$ sufficiently small, gives the required constant additive precision. Repeating the same procedure for the polynomially many outcomes on a constant number of modes yields their entire photon-number distribution.
\end{proof}

\noindent This completes the proof of \cref{claim:pspace-complete}.

\subsection{Bounding the complexity of energy-preserving computations}
\label{sec:bound_complexity}

\begin{table}[t]
    \centering
    \begin{tabular}{lllll}
        \toprule
        Precision & Energy     & Space      & Time Complexity          \\
        \midrule
        $O(1)$    & $\poly(\varsigma)$ & $O(1)$     & $\in\Pp$                 \\
        $O(1)$    & $O(1)$     & $\poly(\varsigma)$ & $\in\Pp$                 \\
        $O(1)$    & $\log(\varsigma)$  & $\log(\varsigma)$  & $\in\Pp$                 \\
        $O(1)$    & $\log(\varsigma)$  & $\poly(\varsigma)$ & $\in\QPp\cap\mathsf{BQP}$\\
        $O(1)$    & $\poly(\varsigma)$ & $\log(\varsigma)$  & $\in\QPp\cap\mathsf{BQP}$\\
        $O(1)$    & $\poly(\varsigma)$ & $\poly(\varsigma)$ & $\in\mathsf{BQP}$        \\
        \midrule
        $1/\exp(\varsigma)$ & $\poly(\varsigma)$ & $O(1)$     & $\in\Pp$               \\
        $1/\exp(\varsigma)$ & $\poly(\varsigma)$ & $\poly(\varsigma)$ & $\in\mathsf{PSPACE}$   \\
        \bottomrule
    \end{tabular}
    \caption{Computational trade-offs obtained in \cref{thm:classical_upper_bounds_precise} between required precision (additive error on output probabilities), energy (total average particle number), space (number of modes), and time complexity, when simulating \model, as a function of input size $\varsigma$. $\QP$ denotes quasi-polynomial time, i.e.\ $2^{\mathrm{polylog}(\varsigma)}$ time.}%
    \label{tab:tradeoffs}
\end{table}

In this section, we provide upper bounds on the computational complexity of simulating energy-preserving bosonic quantum computations (estimating output probabilities), depending on the number of modes, available energy, and required additive precision. The results illustrating the corresponding computational trade-offs are summarized in \cref{tab:tradeoffs}. Any computational complexity upper bound obtained is also valid for a lower precision, lower energy, or lower space.

Recall that the notation $\BEQP{}{E}$ on $f$ modes (\cref{def:cc-beq}) denotes the set of promise problems which can be solved in polynomial time by \model\ with input energy $E$ over $f$ modes using a universal gate set (local, constant-degree), with a constant promise gap, whereas $\PrecBEQP$ additionally assumes an inverse-exponential promise gap. If $E$ is not specified, it is assumed to be polynomial. %
We now state complexity upper bounds for \model, as energy, space and precision parameters are varied:

\begin{theorem}\label{thm:classical_upper_bounds_precise}
We have the following upper bounds on energy-preserving bosonic quantum computations:
\begin{itemize}
    \item $\BEQP{}{}$ with $O(1)$ modes is in $\Pp$.

    \item $\BEQP{}{\log}$ with $O(\log)$ modes is in $\Pp$.

    \item $\BEQP{}{O(1)}\subseteq\Pp$.

    \item $\BEQP{}{\log}\subseteq\QPp\cap\BQP$.

    \item $\BEQP{}{}=\BQP$.

    \item $\BEQP{}{}$ with $O(\log)$ modes is in $\QPp\cap\BQP$.

    \item $\PrecBEQP$ on $O(1)$ modes is in $\Pp$.

    \item $\PrecBEQP\subseteq\PSPACE$.
\end{itemize}
\end{theorem}

\begin{proof} Let $\varsigma$ denote the input size and $m$ the number of modes.
The gentle measurement lemma for the number operator shows that it is enough to cut off the Fock state description of the input state at $K=E/\varepsilon^2$ total photon number to retain $\varepsilon$ additive precision in trace distance. This amounts to a total of $D=\binom{m+K}K$ Fock coefficients indexed by $\bm n$ such that $|\bm n|\le K$. Then, since all gates in the computation are local, the Fock amplitudes of the output state up to photon-number $K$ can be obtained in time $\poly(\varsigma,K)\binom{m+K}K$, which is also sufficient to guarantee an $\varepsilon$ precision in trace distance since the gates are energy-preserving.

The running time of this classical algorithm thus is $\poly(\varsigma,E)\binom{m+144E}{144E}$ for constant precision smaller than half the constant promise gap, which together with \cref{thm:in_BQP} proves all the $\BEQP{}{}$-related results stated in the theorem.

For higher precision ($\PrecBEQP$), we can refine these results by using the fact that Fock basis amplitudes of states in energy-preserving computations are decaying exponentially fast \cite{upreti2026exponentially}. This allows us to utilize the gentle measurement lemma for an operator $f(\hat N)$, where $f$ is a rapidly increasing non-negative function.

Let $\ket\psi=\sum_{\bm n}\psi_{\bm n}\ket{\bm n}$ and assume $\bra{\psi}f(\hat N)\ket{\psi}=f_\psi$. Then, for all $K\ge0$,
\be
f(K)\sum_{|\bm n|>K}|\psi_{\bm n}|^2\le\sum_{|\bm n|>K}f(|\bm n|)|\psi_{\bm n}|^2\le f_\psi,
\ee
so, whenever $f(K)>0$,
\be
\left\|\Pi_{\le K}\ket\psi-\ket\psi\right\|_2^2\le\frac{f_\psi}{f(K)}.
\ee
For coherent states, with the function $f_s:x\mapsto s^x$, for $s>1$, we have $\bra{\alpha}f_s(\hat n)\ket{\alpha}=e^{(s-1)|\alpha|^2}$ and thus $\bra{\bm\alpha}f_s(\hat N)\ket{\bm\alpha}=\bra{\bm\alpha}\hat U^\dag f_s(\hat N)\hat U\ket{\bm\alpha}=e^{(s-1)\|\bm\alpha\|^2}$, where $\hat U$ is any energy-preserving unitary. Hence,
\be
\label{eq:upper_coherent}
\left\|\Pi_{\le K}\hat U\ket{\bm\alpha}-\hat U\ket{\bm\alpha}\right\|_2^2\le\frac{e^{(s-1)\|\bm\alpha\|^2}}{s^K}.
\ee
Writing $E=\|\bm\alpha\|^2$, we obtain a precision $\varepsilon$ for
\be
\label{eq:upper_number}
K=\left\lceil\frac{2\log(\frac1\varepsilon)}{\log s}+\frac{(s-1)E}{\log s}\right\rceil,
\ee
for any $s>1$.

We thus obtain a classical algorithm to compute the Fock amplitudes of the truncated output state up to photon-number $K$ in time $\poly(\varsigma,K)\binom{m+K}K$, guaranteeing an $\varepsilon$ precision in trace distance. This shows that $\PrecBEQP$ on $O(1)$ modes is contained in $\Pp$, as stated in the theorem, but only gives containment of $\PrecBEQP$ in $\EXP$ in general.

To improve this last upper bound to $\PSPACE$, we use the polynomial-space simulation of \cref{thm:inpspace}. With $\poly(\varsigma)$ modes and gates, we choose an inverse-exponential precision with $\varepsilon\leq(c(\varsigma)-s(\varsigma))/8$, where $c(\varsigma)-s(\varsigma)$ denotes the gap in $\PrecBEQP$. Taking $s=2$ in \cref{eq:upper_number} gives a cutoff $K=O(E+\log(1/\varepsilon))=\poly(\varsigma)$. The truncated space has dimension $D=\binom{m+K}{K}$, so $\log D=\poly(\varsigma)$.

We now apply the recursion in \cref{eq:recursion1,eq:recursion2,eq:recursion3} at cutoff $K$. As in the proof of \cref{thm:inpspace}, choosing $p=O(L\log D+\log(L/\varepsilon))=\poly(\varsigma)$ bits suffices to compute the acceptance probability to numerical error at most $\varepsilon$, by summing over the output occupation vectors whose first component belongs to $[a,b]$. Including truncation, the total error is at most $2\varepsilon\leq(c(\varsigma)-s(\varsigma))/4$. Comparing the resulting probability with $(c(\varsigma)+s(\varsigma))/2$ therefore decides the promise problem in polynomial space, proving $\PrecBEQP\subseteq\PSPACE$.
\end{proof}

\clearpage
\printbibliography

\appendix
\crefalias{section}{appendix}

\section{Examples of energy-preserving simulation of active gates}
\label{app:examples}

Here we produce parametric approximations of some well-known active gates for illustrative purposes. We use the notations of \cref{sec:parametric-approx} and pick $\alpha\ge1$.

\subsection{Displacement}

For displacement, the active Hamiltonian is given by

\be H = i(\a - (\ad)) \otimes\II_b.\ee
The corresponding energy-preserving Hamiltonian is
\be \wtH =\frac{i}{\alpha} (\a\bd - \ad \b),\ee
and we denote the input state by
\be \ket{\Psi} = \ket{\psi}\ket{\alpha}.\ee

We have
\begin{subequations}
\begin{align}
    \norm{(H - \wtH) \ket{\Psi}} &= \norm{\a \ket{\Psi} - \ad \ket{\Psi} - \frac{\a\bd}{\alpha}\ket{\psi, \alpha} + \frac{\ad\b}{\alpha} \ket{\psi, \alpha}}\\
    &= \norm{\a \ket{\Psi} - \ad \ket{\Psi} - \frac{\a\bd}{\alpha}\ket{\psi, \alpha} + \frac{\ad\alpha}{\alpha} \ket{\psi, \alpha}}\\
    &= \norm{\a \ket{\psi,\alpha} - \frac{\a\bd}{\alpha}\ket{\psi, \alpha}}\\
    &= \norm{\a \ket{\psi}} \norm{\left(1 - \frac{\bd}{\alpha}\right)\ket{\alpha}}.\\
\end{align}
\end{subequations}
Now, rewriting $\ket{\alpha} = D(\alpha)\ket{0}$ and using $D^\dagger(\alpha)\b D(\alpha) = \b+\alpha$, we obtain

\begin{subequations}
\begin{align}
    \norm{(H - \wtH) \ket{\Psi}} &= \norm{\a \ket{\psi}} \norm{\left(1 - \frac{(\bd + \alpha)}{\alpha}\right)\ket{0}}\\
    &= O\left(\frac{\norm{\a \ket{\psi}}}\alpha\right).
\end{align}
\end{subequations}

\subsection{Squeezing}

For squeezing, the active Hamiltonian is given by

\be H = i(\a^2 - (\ad)^2) \otimes\II_b.\ee
The corresponding energy-preserving Hamiltonian is
\be \wtH =\frac{i}{\alpha^2} (\a^2(\bd)^2 - (\ad)^2\b^2),\ee
and we denote again the input state by
\be \ket{\Psi} = \ket{\psi}\ket{\alpha}.\ee

We have
\begin{subequations}
\begin{align}
    \norm{(H - \wtH) \ket{\Psi}} &= \norm{\a^2 \ket{\Psi} - (\ad)^2 \ket{\Psi} - \frac{\a^2(\bd)^2}{\alpha^2}\ket{\psi, \alpha} + \frac{(\ad)^2\b^2}{\alpha^2}\ket{\psi,\alpha}}\\
    &= \norm{\a^2 \ket{\Psi} - (\ad)^2 \ket{\Psi} - \frac{\a^2(\bd)^2}{\alpha^2}\ket{\psi, \alpha} + \frac{(\ad)^2\alpha^2}{\alpha^2} \ket{\psi, \alpha}}\\
    &= \norm{\a^2 \ket{\psi,\alpha} - \frac{\a^2(\bd)^2}{\alpha^2}\ket{\psi, \alpha}}\\
    &= \norm{\a^2 \ket{\psi}} \norm{\left(1 - \frac{(\bd)^2}{\alpha^2}\right)\ket{\alpha}}.\\
\end{align}
\end{subequations}

Now, rewriting $\ket{\alpha} = D(\alpha)\ket{0}$ and using $D^\dagger(\alpha)\b D(\alpha) = \b+\alpha$, we obtain

\begin{subequations}
\begin{align}
    \norm{(H - \wtH) \ket{\Psi}} &= \norm{\a^2 \ket{\psi}} \norm{\left(1 - \frac{(\bd + \alpha)^2}{\alpha^2}\right)\ket{0}}\\
    &= \norm{\a^2 \ket{\psi}} \norm{\left(\frac{(\bd)^2}{\alpha^2} +2\frac{\bd}{\alpha} \right)\ket{0}}\\
    &= \norm{\a^2 \ket{\psi}} \norm{\frac{\sqrt{2!}}{\alpha^2} \ket{2} +2\frac{\sqrt{1!}}{\alpha}\ket{1}}\\
    &= \norm{\a^2 \ket{\psi}} \sqrt{\frac{2!}{\alpha^4} + \frac{2^2\cdot 1!}{\alpha^2}}\\
    &= O\left(\frac{\norm{(1+\hat n)\ket{\psi}}}{\alpha}\right).
\end{align}
\end{subequations}

\subsection{Arbitrary single term}

For an arbitrary monomial, the active Hamiltonian is given by
\be H = ((\ad)^k\a^l + h.c.) \otimes\II_b ,\qquad k>l.\ee
The corresponding energy-preserving Hamiltonian is
\be \wtH = (\ad)^k\a^l\otimes\left(\frac{\b}{\alpha}\right)^{k-l} + (\ad)^l\a^k\otimes\left(\frac{\bd}{\alpha}\right)^{k-l} ,\qquad k>l,\ee
and we denote again the input state by
\be \ket{\Psi} = \ket{\psi}\ket{\alpha}.\ee

The same analysis follows: the first term cancels since $\b\ket\alpha=\alpha\ket\alpha$, and for the second we use $D^\dagger(\alpha)\bd D(\alpha) = \bd+\alpha$, which gives
\be
\norm{(H - \wtH) \ket{\Psi}} = \norm{(\ad)^l\a^k\ket\psi}\sqrt{\sum_{i=1}^{k-l}\binom{k-l}{i}^2 \frac{i!}{\alpha^{2i}}} = O\left(\frac{\norm{(1+\hat N)^{d/2}\ket\psi}}{\alpha}\right),
\ee
where $d=k+l$.

\section{Finite gate sets}
\label{app:finite-alphabet}

Our model allows tunable evolution times
(\cref{def:model,def:cc-beq}), so the gate families in
\cref{eq:GKerr} are admissible as stated.
Nevertheless, we show here that a finite subset of $\GKerr$ or $\GcrossKerr$ suffices, independent of cutoff and precision.
Let $\varphi=(\sqrt5-1)/2$ and, for each generator $H\in\{\hat n_1,\hat n_1^2, \hat n_1\hat n_2,-i(\hat a_1^\dagger\hat a_2-\hat a_2^\dagger\hat a_1)\}$, keep only $e^{2\pi i\varphi H}$ and its inverse.
These generators have integer spectra, so their rotations are $2\pi$-periodic.
There is a constant $C>0$ such that, for every $0<\eta<1$ and
$\theta\in\mathbb R$, there exist integers
$0\le j\le\lceil C/\eta\rceil$ and $\ell\in\mathbb Z$ satisfying
\begin{equation}
    |2\pi j\varphi-\theta-2\pi\ell|\le\eta.
\end{equation}
This follows from \cite[Theorem~3.1]{badziahin2023badly}, since
$\varphi=[0;1,1,\ldots]$ has bounded partial quotients.
Since the generators have cutoff norm $O(K^2)$, any tunable elementary gate
can therefore be approximated to cutoff error $\delta$ using
$O(K^2/\delta)$ fixed gates.
Applying this replacement only in the coarse compiler (\cref{lem:coarse-compiler}), before
Solovay--Kitaev refinement, retains the $\mathrm{polylog}(1/\epsilon)$ dependence on precision.

\section{Heisenberg picture with bosonic gates}\label{app:heisenberg}

Prior work (e.g.\ \cite{kalajdzievski2021exact,CJMMM26,CGMMNRS25}) relies on the following formula
\begin{equation}
e^{X} Y e^{-X} = \sum_{n=0}^\infty\frac{\adj^n_X(Y)}{n!}, \qquad \adj_X(Y) \coloneq [X,Y]
\end{equation}
to work with bosonic gates in the Heisenberg picture, which is a special case of the Baker--Campbell--Hausdorff formula.
Unfortunately, it does not apply as is to unbounded operators due to domain and convergence issues, and the cited works do not explicitly justify its use for the non-Gaussian gates considered.
In the following, we will give a justification.
With that, the energy bounds for Gaussian + $\Q^3$ circuits in \cite{CJMMM26,CGMMNRS25} become rigorous. This also justifies on Schwartz space the corresponding Heisenberg calculations in the \emph{exact}\footnote{The approximate decompositions in \cite[Section 3]{kalajdzievski2021exact} have the same underlying issues as \cite{lloyd_quantum_1999} we discuss in \cref{sec:pathology}.} synthesis part of \cite{kalajdzievski2021exact}.

\begin{theorem}\label{thm:Heisenberg}
Let $A = C(\Q_1,\P_1,\dots,\Q_m,\P_m)$ for $C\in\CC[x^{2m}]$, and (i) $H=\Q_j^2+\P_j^2$ or (ii) $H=R(\Q_1,\dots,\Q_m)$ for $R\in\RR[x^m]$,\footnote{These include the gates used in the exact synthesis results of \cite{kalajdzievski2021exact}.} such that
\begin{equation}\label{eq:ad-span}
    \dim V_{H,A} <\infty, \qquad V_{H,A}\coloneq \Span\bigl\{\adj^n_H(A)\bigm| n\in\NN_0\bigr\}
\end{equation}
as a linear vector space in $\End(\calS)$, where $\calS\coloneq \calS(\RR^m)$.
Then for all $\psi\in\calS$
\begin{equation}\label{eq:A(t)}
    e^{itH}A e^{-itH}\psi\ =\ \sum_{n=0}^\infty \frac{(it)^n}{n!}\adj_H^n(A)\psi.
\end{equation}
\end{theorem}
\begin{proof}
Let $B_1,\dots,B_d$ be a basis of $V_{H,A}$ with $B_1=A$.
Define $M\in\CC^{d\times d}$ as
\begin{equation}\label{eq:BM}
    i[H,B_j] = \sum_{k=1}^d M_{kj} B_k.
\end{equation}
Then treating $M$ as a $V_{H,A}\to V_{H,A}$ map, we get
\begin{equation}
    F(t) \coloneq \sum_{n=0}^\infty \frac{(it)^n}{n!} \adj_H^n(A) = \sum_{n=0}^\infty \frac{t^n M^n(A)}{n!} = (e^{tM}A).
\end{equation}
Then $F(0)=A$ and for all $\psi\in\calS$
\begin{equation}
    \frac{d}{dt} F(t)\psi = M F(t)\psi,
\end{equation}
where differentiability in $L^2(\RR^m)$ follows as $e^{tM}A$ is linear in $B_1,\dots,B_d$ with smooth coefficients in $t$, so we can write
\begin{equation}\label{eq:Ft}
    F(t) = \sum_{j=1}^d c_{j}(t) B_j,
\end{equation}
for smooth functions $c_j(\cdot)$, with $c_j(t) = (e^{tM})_{j1}$.
Thus, $\norm{F(t)\psi}$ is bounded on finite intervals, and we obtain the derivative
\begin{equation}
    \frac{d}{dt} F(t)\psi = \frac{d}{dt}\sum_{j=1}^d c_{j}(t) B_j\psi = \sum_{k=1}^d c_k(t)\sum_{j=1}^d M_{jk}B_j\psi = \sum_{k=1}^d c_{k}(t)i[H,B_k]\psi = i[H,F(t)]\psi,
\end{equation}
using \cref{eq:BM} and \cref{eq:Ft} in the last two steps.
Thus, $F(t)$ satisfies the generalized Heisenberg equation of motion \cref{eq:gen-Heisenberg} on $\calS$ with initial condition $F(0)=A$.
Applying \cref{claim:continuous} below to $A$, $A^\dagger$, and the $B_j^\dagger$, together with \cref{eq:Ft} and invariance of $\calS$, verifies the continuity hypotheses of \cref{cor:arai} and condition (X.2)(a) of \cref{def:solution-class} for $F$.
Thus, \cref{cor:arai}, applied to the closure of $A|_{\calS}$, gives $F(t) = A_H(t)$ on $\calS$, with $A_H(t)$ as in \cref{eq:AH}.

\begin{claim}
    $e^{itH}\calS = \calS$ for all $t\in\RR$.
\end{claim}
\begin{proof}
    Case (i) is trivial due to the characterization of $\calS(\RR^m) = \bigcap_{k=1}^\infty \calD(\Ntot^k)$ (see \cite[Section 2.2.1]{KKW16}; also explicitly shown in \cite[Lemma A.9]{CGMMNRS25}), and $\Q_j^2+\P_j^2 = 2\N_j + 1$ commutes with $\Ntot = \sum_{j=1}^m \N_{j}$.

    Case (ii) holds since the derivatives of $e^{it R(\Q_1,\dots,\Q_m)}$ grow polynomially at infinity  (see \cite[Definition 1.2.20]{lerner2014elements} and discussion thereafter).
\end{proof}

Since $F(t)$ decomposes into $B_j$ with smooth coefficients \cref{eq:Ft}, and $e^{itH}$ is a bijection on $\calS$, it only remains to prove continuity.
\begin{claim}\label{claim:continuous}
    Let $B \in\CC[\Q_1,\P_1,\dots,\Q_m,\P_m]$.
    Then $\lim_{t\to0} \norm{Be^{itH}\psi-B\psi} = 0$ for all $\psi\in\calS$.
\end{claim}
\begin{proof}
    For case (i), assume $H = \Q_1^2+\P_1^2=2\N_1+1$ and write $\psi$ in the Fock/Hermite basis
    \begin{equation}
        \psi = \sum_{\bfn\in\NN_0^m} c_{\bfn}\ket{\bfn}.
    \end{equation}
    Since $\psi\in \calD((1+\Ntot)^{k})$ for any $k$, and any polynomial Hamiltonian is bounded with respect to some $(1+\Ntot)^k$, we have
    \begin{equation}
        C\coloneq \sum_{\bfn} (2n_1+1)\abs{c_{\bfn}}\cdot \norm{B\ket{\bfn}} < \infty.
    \end{equation}
    Thus, we get with $\abs{e^{it}-1}\le \abs{t}$,
    \begin{equation}
        \norm{B(e^{itH} - 1)\psi} \le \sum_{\bfn}\abs{c_{\bfn}}\abs{t(2n_1+1)}\norm{B\ket{\bfn}}\le \abs{t}C \ \xrightarrow[t\to0]{}\ 0.
    \end{equation}
    For case (ii), recall the definition of the Schwartz space
    \begin{equation}
        \calS(\RR^m) \coloneq \Bigl\{ f\in C^\infty(\RR^m)\Bigm| \sup_{x\in\RR^m} \bigl|x^\alpha\partial^\beta f(x)\bigr| < \infty\ \forall \alpha,\beta\in\NN_{0}^m\Bigr\}.
    \end{equation}
    Since $\Q_j : \psi(x)\mapsto x_j\psi(x)$ and $\P_j : \psi(x)\mapsto -i\partial_{x_j}\psi(x)$ for all $\psi\in\calS$, we may restrict to observables $B = x^\alpha\partial^\beta$ with $\alpha,\beta\in\NN_0^m$.
    Then we have by the Leibniz rule
    \begin{equation}
        x^\alpha\partial^\beta e^{itR(x)}\psi(x) = \sum_{\gamma\le\beta}x^\alpha\binom{\beta}{\gamma} \left(\partial^\gamma e^{itR(x)}\right)\left(\partial^{\beta-\gamma}\psi(x)\right) \eqcolon \sum_{\gamma\le\beta} r_{\gamma}(t,x)e^{itR(x)}\psi_{\gamma}(x),
    \end{equation}
    where $r_{\gamma}(t,x)$ is a polynomial such that $r_{\gamma}(t,x)e^{itR(x)} = \binom{\beta}{\gamma}x^{\alpha}\partial^\gamma e^{itR(x)}$ and $\psi_{\gamma} \coloneq \partial^{\beta-\gamma}\psi\in\calS$.
    For each $\gamma$, we separate $r_{\gamma}$ in terms of powers of $t$, such that $r_{\gamma}(t,x) = \sum_{j=0}^k t^j \, r_{\gamma,j}(x)$ for polynomials $r_{\gamma,j}$.
    Then for $j>0$
    \begin{equation}
        t^j\,r_{\gamma,j}(x)\,e^{itR(x)}\psi_{\gamma}(x) = e^{itR(x)}\,t^j \bigl(r_{\gamma,j}(x)\psi_{\gamma}(x)\bigr)\  \ \xrightarrow[t\to0]{}\ 0
    \end{equation}
    in $L^2$ as $r_{\gamma,j}(x)\psi_{\gamma}(x)\in \calS$ and $e^{itR(x)}$ is a bounded operator.
    For the remaining $j=0$ term, the pointwise inequality $\abs{e^{iu}-1}\le\abs{u}$ gives
    \begin{equation}
        \norm{r_{\gamma,0}\,\bigl(e^{itR}-1\bigr)\psi_\gamma}_2
        \le \abs{t}\,\norm{R\,r_{\gamma,0}\,\psi_\gamma}_2
        \ \xrightarrow[t\to0]{}\ 0,
    \end{equation}
    since $Rr_{\gamma,0}\psi_\gamma\in\calS\subsetneq L^2$.
    Now we have shown that all terms of
    \begin{equation}
        x^\alpha\partial^\beta \bigl(e^{itR(x)}-1\bigr)\psi = \sum_{\gamma\le\beta} \left(r_{\gamma}(t,x)e^{itR(x)} - r_{\gamma}(0,x)\right)\psi_{\gamma}(x)
    \end{equation}
    converge to $0$ as $t\to0$ and thus continuity is proven.
\end{proof}

\end{proof}

The exact synthesis of \cite{kalajdzievski2021exact} also uses the following identity, which we also verify here for completeness.

\begin{corollary}
Let $A,H$ be as in \cref{thm:Heisenberg} such that $A$ is essentially self-adjoint on $\calS$.
Then
\begin{equation}\label{eq:unitary-exponent}
    e^{isH} e^{itA} e^{-isH} = e^{it A_H(s)}.
\end{equation}
\end{corollary}
\begin{proof}
First, note that $A_H(s)$ is essentially self-adjoint on $\calS$ since $A$ is essentially self-adjoint on $\calS$ and $e^{isH}\calS=\calS$.
Let $\psi\in\calS$, $F(t) = e^{isH} e^{itA} e^{-isH}\psi$, and $G(t) =  e^{it A_H(s)}\psi$.
Then $F(0)=G(0)=\psi$ and $F'(t)= i(e^{isH} A e^{-isH})F(t) = i A_H(s) F(t)$ and $G'(t)= iA_H(s)G(t)$.
Then $F(t) = G(t)$ due to the uniqueness of the classical solution to the abstract Cauchy problem \cite[Proposition II.6.2]{EN2000} for $\psi\in\calD(A_H(s))$.
\cref{eq:unitary-exponent} extends to all of $L^2$ by density.
\end{proof}

\begin{corollary}[Mixed position and momentum generators]\label{cor:Heisenberg-mixed}
\Cref{thm:Heisenberg} also holds for $H=R(X_1,\dots,X_m)$ with $R$ a real polynomial and $X_j\in\{\Q_j,\P_j\}$ for each mode $j$. In particular, it applies to $H=\Q_{j-1}^2\P_j$ for $j\ge2$.
\end{corollary}
\begin{proof}
Let
\begin{equation}
    U=\exp\!\left(-\frac{i\pi}{2}\sum_{j:\,X_j=\P_j}\N_j\right),
    \qquad UX_jU^\dagger=\Q_j.
\end{equation} Fourier transforms preserve $\calS$ and exchange position and momentum operators up to sign. Consequently,
\begin{equation}
    \widetilde H\coloneq UHU^\dagger=R(\Q_1,\dots,\Q_m),
    \qquad \widetilde A\coloneq UAU^\dagger
\end{equation}
is a position-polynomial generator together with a polynomial observable. Since $\widetilde H=R(\Q_1,\dots,\Q_m)$ is essentially self-adjoint on $\calS$, and $U\calS=\calS$, the operator $H=U^\dagger\widetilde H U$ is also essentially self-adjoint on $\calS$. Moreover,
\begin{equation}
    e^{itH}=U^\dagger e^{it\widetilde H}U,
    \qquad \adj_{\widetilde H}^n(\widetilde A)=U\adj_H^n(A)U^\dagger
\end{equation}
on $\calS$. The commutator span remains finite-dimensional, because
\begin{equation}
    V_{\widetilde H,\widetilde A}=UV_{H,A}U^\dagger,
    \qquad \dim V_{\widetilde H,\widetilde A}=\dim V_{H,A}<\infty.
\end{equation} Each factor in $e^{itH}=U^\dagger e^{it\widetilde H}U$ preserves $\calS$, so $e^{itH}\calS=\calS$. For $\psi\in\calS$, we also have $U\psi\in\calS$, and the position-polynomial case of \cref{thm:Heisenberg} gives
\begin{equation}
    e^{itH}Ae^{-itH}\psi
    =U^\dagger e^{it\widetilde H}\widetilde A e^{-it\widetilde H}U\psi
    =U^\dagger\sum_{n=0}^\infty\frac{(it)^n}{n!}\adj_{\widetilde H}^n(\widetilde A)U\psi
    =\sum_{n=0}^\infty\frac{(it)^n}{n!}\adj_H^n(A)\psi.
\end{equation}
\end{proof}

\subsection{Uniqueness of the Generalized Heisenberg Equation}

The following summarizes Arai's Corollary 4.5 \cite{arai07} with several simplifying assumptions for our setting.
The \emph{Heisenberg operator} of $A$ with respect to $H$ is formally defined as
\begin{equation}\label{eq:AH}
A_H(t) := e^{itH} A e^{-itH}, \quad t \in \mathbb{R}.
\end{equation}

\begin{definition}[{\cite[Definition 4.1]{arai07}}]
Let $\mathcal{D}$ be a dense subspace of $L^2(\RR^m)$, and let $t \mapsto X(t)$ be an operator-valued function on $\mathbb{R}$ such that $\calD(X(t))$ is dense and $\mathcal{D} \subseteq \calD(X(t)) \cap \calD(X(t)^*) \cap \calD(H)$ for all $t \in \mathbb{R}$. We say $X(\cdot)$ obeys a \emph{generalized Heisenberg equation of motion} on $\calD$ if, for all $\phi, \psi \in \mathcal{D}$, the scalar function $t \mapsto \langle \phi, X(t)\psi \rangle$ is differentiable in $t \in \mathbb{R}$ with:
\begin{equation}\label{eq:gen-Heisenberg}
    \frac{d}{dt}\langle \phi, X(t)\psi \rangle = i \big( \langle H\phi, X(t)\psi \rangle - \langle X(t)^*\phi, H\psi \rangle \big).
\end{equation}
\end{definition}

\begin{definition}[{\cite[Definition 4.3]{arai07}}]\label{def:solution-class}
Assume $e^{itH}\mathcal{D} = \mathcal{D}$ for all $t\in\mathbb{R}$.
We say that an operator-valued function $X(\cdot)$ belongs to the admissible solution class $\sfX_{\mathcal{D}}$ if:
\begin{enumerate}[label=(X.\arabic*),itemindent=2em]
    \item $\mathcal{D} \subseteq \bigcap_{t\in\mathbb{R}} \big( \calD(X(t)) \cap \calD(X(t)^*) \cap \calD(H) \big)$.
    \item Either (a) holds for all $\psi\in\mathcal D$, or (b) holds for all $\psi\in\mathcal D$:
        \begin{enumerate}[label=(\alph*)]
            \item the mappings $t \mapsto X(t)\psi$ and $t \mapsto X(t)^* e^{itH}\psi$ are strongly continuous in $L^2(\mathbb{R}^m)$;
            \item the mappings $t \mapsto X(t)e^{itH}\psi$ and $t \mapsto X(t)^*\psi$ are strongly continuous in $L^2(\mathbb{R}^m)$. 
        \end{enumerate}
\end{enumerate}
\end{definition}

\begin{corollary}[{\cite[Corollary 4.5]{arai07}}]\label{cor:arai}
Let $H$ and $A$ be closed operators whose restrictions to $\mathcal{S}\equiv\calS(\RR^m)$ are polynomials in $\Q$ and $\P$ in $m$ modes. Assume $H$ is self-adjoint and that $e^{itH}\mathcal{S} \subseteq \mathcal{S}$ for all $t \in \mathbb{R}$.

Furthermore, assume that for all $\psi \in \mathcal{S}$, the $L^2(\mathbb{R}^m)$-valued mappings:
\begin{equation}
    t \mapsto A e^{itH}\psi \quad \text{and} \quad t \mapsto A^\dagger e^{itH}\psi
\end{equation}
are strongly continuous in $t \in \mathbb{R}$.

Then, the Heisenberg operator $X(t) = A_H(t)$ is the \emph{unique solution on $\mathcal S$} to the generalized Heisenberg equation of motion \cref{eq:gen-Heisenberg} on $\mathcal{S}$ that belongs to the class $\sfX_{\mathcal{S}}$ and satisfies the initial condition $X(0)|_{\mathcal{S}} = A|_{\mathcal{S}}$.
\end{corollary}

\section{Effective energy-preserving descriptions of energy-preserving operations}
\label{sec:passive-description}

In this appendix, we show that for any arbitrary energy-preserving unitary, we can find a unitary generated by an energy-preserving polynomial which operates on an effective truncated photon space.
This result is inspired by \cite{arzani2025can}, and provides a constructive proof for the energy-preserving setting.

Let $m \ge 1$ be the number of modes, and let $N_{\mathrm{max}}\in \mathbb{N}$ denote the local dimension cutoff, and let the truncated finite-dimensional bosonic subspace be defined as $\mathcal{S} = \bigotimes_{k=1}^{m} \mathcal{H}_{N_{\mathrm{max}}}$ where $\mathcal{H}_{N_{\mathrm{max}}} = \mathrm{span} \{\ket{0}_k \dots \ket{N_{\mathrm{max}} - 1 }_k\}$. Here $N_{\mathrm{max}}$ is the local dimension, so the local photon cutoff is $N_{\mathrm{max}}-1$.
\begin{theorem}[Effective energy-preserving descriptions of energy-preserving operations]
\label{thm:passive-description}
If $\hat{H}$ is a Hermitian operator over $\mathcal{S}$ that is energy-preserving (i.e.\ $[\hh, \hat{N}] = 0$ on $\mathcal{S}$), then there exists a polynomial Hamiltonian $\ph(\hat{\bm a}, \hat{\bm a}^\dagger)$ such that:
\begin{enumerate}
    \item $\ph = \hh \oplus \hh'$, for some $\hat{H}'$ on $\mathcal{S}^{\perp}$.
    \item $[\ph, \hat{N}] = 0$ over $\mathcal H^{\otimes m}$.
    \item The degree of $\ph$ is at most $6mN_{\mathrm{max}}$.
    \item $\ph$ can be efficiently computed in the \textit{size} of $\hh$. Explicitly, it takes $O(N_{\mathrm{max}}^{4m})$ steps for a dense $\hh$, $sN_{\mathrm{max}}^{2m}$ for a $\hh$ with $s$ nonzero entries.
\end{enumerate}
\end{theorem}
\begin{proof}

Because $\hh$ is a linear operator on $\mathcal{S}$, we can write it in the Fock basis as:

\be
\hh = \sum_{\vn, \vm \in I_{\mathcal{S}}} H_{\vn\vm} \ketbra{\vn}{\vm},
\ee
where $I_{\mathcal{S}}=\{0,\dots,N_\text{max}-1\}^m$.

We have $H_{\vn, \vm} = H_{\vm, \vn}^*$ and $H_{\vm, \vn} = 0$ for all $|\vn| \ne |\vm|$.

Construct the restricted projector

\be \hat{L}^{(N_{\mathrm{max}})}_{n}(\hat{n}_k) = \prod_{j=0, j\ne n}^{N_{\mathrm{max}} - 1} \frac{\hat{n}_k - j}{n-j},\ee
and note that for any Fock state $\ket{r}_k \in \mathcal{H}_{N_{\mathrm{max}}}$ where $0\le r<N_{\mathrm{max}}$, $\hat{L}^{(N_{\mathrm{max}})}_{n}(\hat{n}_k) \ket{r} = \delta_{n,r} \ket{r}$. Since $\LL$ is defined as a polynomial function of $\hat{n}_k$, it commutes with $\hat{n}_k$. The degree of $\LL$ is at most $2N_{\mathrm{max}}$, and it can be computed as a function of $\hat a_k, \hat a^\dagger_k$ in time $N_{\mathrm{max}}^2$, and contains at most $N_{\mathrm{max}}$ terms.

Now, we can build the $m$-mode subspace projector

\be \hat{\Pi}_{\vm} = \prod_{k=1}^{m}\LL_{m_k}(\hat{n}_k)\ee
and check that for any multimode Fock state $\ket{\vm} \in \mathcal{S}$, we have $\hat{\Pi}_{\vn}\ket{\vm} = \delta_{\vn, \vm} \ket{\vm}$. Furthermore, $\hat{\Pi}_{\vn}$ is a polynomial of degree at most $2mN_{\mathrm{max}}$ and commutes with $\hat{N}$, which can be computed in time $O(N_{\mathrm{max}}^m)$ containing at most $N_{\mathrm{max}}^m$ terms.

We now define the bare multimode ladder operator corresponding to a transition from state $\ket{\vm}$ to $\ket{\vn}$:
\begin{equation}
    \Anm = \prod_{k=1}^m (\hat{a}_k^\dagger)^{n_k} (\hat{a}_k)^{m_k}.
\end{equation}
When applied to $\ket{\vm}$, standard bosonic algebra yields
\begin{equation}
    \Anm \ket{\vm} = C_{\vn, \vm} \ket{\vn},
\end{equation}
where the scalar coefficient is $C_{\vn, \vm} = \prod_{k=1}^m \sqrt{n_k! m_k!} \neq 0$.
Taking the adjoint gives $\hat{A}_{\vn, \vm}^\dagger = \hat{A}_{\vm, \vn}$.

Now we investigate the global commutation relation between $\Anm$ and the total number operator $\hat N$. Using the standard identities $[\hat N, \hat{a}_k^\dagger] = \hat{a}_k^\dagger$ and $[\hat N, \hat{a}_k] = -\hat{a}_k$, we find:
\begin{equation}[\hat N, \Anm] = (|\vn| - |\vm|) \Anm.
\end{equation}
By the hypothesis that $\hat{H}$ is energy-preserving, non-zero matrix elements are restricted to states with the same total particle number ($|\vn| = |\vm|$). It follows that for all \textit{required} transition pairs $(\vn, \vm)$ where $H_{\vn, \vm} \neq 0$, we have:
\begin{equation}
    [\Anm, \hat{N}] = 0
\end{equation}
over the full infinite-dimensional Hilbert space $\mathcal{H}^{\otimes m}$. By construction, this is a polynomial of degree $|\vn|+|\vm| \le 2mN_{\mathrm{max}}$, with $1$ term calculated in $m$ time.

Finally, we define a polynomial operator $\Tnm$ to represent the exclusive transition $\ket{\vn}\!\bra{\vm}$ strictly within the truncated subspace $\mathcal{S}$:
\begin{equation}
    \Tnm = \frac{1}{C_{\vn, \vm}} \hat{\Pi}_{\vn}(\hat{\bm{n}}) \Anm \hat{\Pi}_{\vm}(\hat{\bm{n}}).
\end{equation}

By sandwiching the bare ladder operators $\Anm$ between the local number projectors $\hat{\Pi}_{\vn}$ and $\hat{\Pi}_{\vm}$, we ensure that this operator only acts non-trivially on the specific state $\ket{\vm}$ and correctly maps it to $\ket{\vn}$. The projectors perfectly suppress the spurious transitions that would otherwise occur when $\Anm$ acts on higher-excited states. Also note, $\Tnm^\dagger = \frac{1}{C_{\vn\vm}}\hat{\Pi}_{\vm} \Anm^\dagger \hat{\Pi}_{\vn} = \hat{T}_{\vm, \vn}$, as expected. It is clear to see that the degree of $\Tnm$ is at most $6mN_{\mathrm{max}}$. The cost of calculating this requires multiplying the two $N_{\mathrm{max}}^m$ term projectors, which requires time $N_{\mathrm{max}}^{2m}$.

Finally, we construct the full, globally defined polynomial Hamiltonian by summing over all matrix elements of the target finite-dimensional Hamiltonian $\hat{H}$:
\begin{equation}
    \ph(\hat{\bm{a}}, \hat{\bm{a}}^\dagger) = \sum_{\vn, \vm \in I_{\mathcal{S}}} H_{\vn, \vm} \hat{T}_{\vn, \vm}.
\end{equation}

By construction, $\ph|_\mathcal{S} =\hat H$, and $\ph$ acts as some arbitrary $\hat H'$ on $\mathcal{S}^\perp$, as required by the theorem.

If there are $s$ nonzero terms in $\hh$, the total cost of evaluating $\ph$ is $sN_{\mathrm{max}}^{2m}$. For dense energy-preserving Hamiltonians, $s = O(N_{\mathrm{max}}^{2m})$. Since the size of $\hh$ is $N_{\mathrm{max}}^{2m}$, the cost of evaluating $\ph$ is $\poly(\text{size}(\hh))$.
\end{proof}

\section{Evaluating Fock entries of energy-preserving Hamiltonians}
\label{app:assumption-satisfied-long}
\renewcommand{\va}{{\bm{\alpha}}}
\renewcommand{\vb}{{\bm{\beta}}}

\begin{lem}[Direct matrix-exponential computation on a photon-number block]
    \label{lem:mat-exp-in-photon-block}
    Let
    \be
        H=\sum_{\va,\vb} c_{\va,\vb}
        (\hat a^\dagger)^\va \hat a^\vb
    \ee
    be an \(m\)-mode Hermitian energy-preserving Hamiltonian of total degree at most
    $d\ge1$, where $|c_{\va, \vb}| \le C_{\mathrm{max}}$, with $C_{\mathrm{max}} > 1$.\textsuperscript{\hyperref[fn:rational-inputs]{\ref*{fn:rational-inputs}}}
    Given \(\bm k,\bm\ell\in\mathbb N_0^m\), \(t \in\mathbb Q^+\), and \(p\ge 1\):
    If \(|\bm k|\neq |\bm\ell|\), the entry is zero. For \(n=|\bm k|=|\bm\ell|\ge0\), the following bounds hold.
    Let
    \be
        D_n=\binom{n+m-1}{m-1}
    \ee
    be the dimension of the \(n\)-photon subspace.
    One can compute a number \(z\) satisfying
    \be
        \left|
        z-(\exp(iHt))_{\bm k,\bm\ell}
        \right|\le 2^{-p}
    \ee
    using
    \be
        O\!\left(D_n T_{m,d} d+D_n^3(p+s)\right)
    \ee
    arithmetic operations on numbers of
    \be
        O\!\left(p+s+\log D_n+\log(p+s)\right)
    \ee
    bits, where
    \be
        T_{m,d}
        =
        \sum_{r\le d/2}
        \binom{r+m-1}{m-1}^2
        =
        O_m(d^{2m-1})
    \ee
    is the number of energy-preserving monomials, and
    \be
        s=
        O\!\left(
            1
            +\log(1+|t|)
            +\log(C_{\max})
            +\log T_{m,d}
            +d\log(1+n)
        \right).
    \ee
    In particular, for fixed \(m,d,C_{\text{max}}\), this is
    \be
        O\!\left(
            (1+n^{3m})\bigl(p+\log(1+|t|)+ d \log(1+n)\bigr)
        \right)
    \ee
    arithmetic operations.
\end{lem}

\begin{proof}
    Because \(H\) is energy-preserving, every monomial
    \be
        (\hat a^\dagger)^{\va} \hat a^{\vb}
    \ee
    appearing in \(H\) satisfies \(|\va|=|\vb|\). Therefore it preserves total
    photon number. Hence \(H\) is block diagonal with respect to the decomposition
    \be
        \mathcal H=\bigoplus_{n\ge 0}\mathcal H_n,\qquad
        \mathcal H_n=\operatorname{span}\{|\bm r\rangle:|\bm r|=n\}.
    \ee
    Consequently
    \be
        (\exp(iHt))_{\bm k,\bm\ell}=0,
    \ee
    unless \(|\bm k|=|\bm\ell|\). The vacuum block is scalar and handled trivially. Assume from now on that
    \be
        |\bm k|=|\bm\ell|=n>0.
    \ee
    Let \(H_n\) denote the restriction of \(H\) to \(\mathcal H_n\). The dimension of
    this block is
    \be
        D_n=\binom{n+m-1}{m-1}.
    \ee

    We first bound \(\|H_n\|\). For a monomial
    \be
        M_{\va,\vb}=(\hat a^\dagger)^\va \hat a^\vb,
    \ee
    with \(|\va|=|\vb|=r\), its action on a Fock basis vector is
    \be
        M_{\va,\vb}|\bm{q}\rangle
        =
        \left[
            \prod_{j=1}^m
            \sqrt{
                \frac{q_j!}{(q_j-\vb_j)!}
                \frac{(q_j-\vb_j+\va_j)!}{(q_j-\vb_j)!}
            }
        \right]
        |\bm{q}-\vb+\va \rangle,
    \ee
    with the convention that the coefficient is zero unless \(q_j \ge \beta_j\ \forall\ j\).
    Since $M$ is in normal form, the maximum coefficient on the application
    of an annihilation or creation operator gives a coefficient of at most
    $\sqrt{n}$. Therefore,
    \be
        \|M_{\va, \vb}|_{\mathcal H_n}\|
        \le (n)^r
        \le (n)^{d/2}.
    \ee
    The number of energy-preserving monomials of total degree at most \(d\) is
    \be
        T_{m,d}
        =
        \sum_{r\le d/2}
        \binom{r+m-1}{m-1}^2
        =
        O_m(d^{2m-1}).
    \ee
    Thus, using \(|c_{\va,\vb}|\le C_{\max}\), we have
    \be
        \|H_n\|
        \le
        C_{\max}T_{m,d}n^{d/2}.
    \ee
    Set
    \be
        A:=itH_n.
    \ee
    Since \(H_n\) is Hermitian, \(A\) is skew-Hermitian. Therefore \(U = \exp(A)\) is
    unitary.
    To evaluate $U$, we use a scaling-and-squaring method, which approximates
    $U$ by approximating $\left(\exp\left(\frac{iH_nt}{2^s}\right)\right)^{2^s}$.
    Choose
    \be
        s=\max\left\{
            0,
            \left\lceil
            \log_2\bigl((|t|+1)C_{\max}T_{m,d}n^{d/2}\bigr)
            \right\rceil+1
        \right\}.
    \ee
    Then
    \be
        B:=2^{-s}A
    \ee
    satisfies
    \be
        \|B\|\le \tfrac12.
    \ee
    This allows us to approximate $\exp(B)$ by its Taylor expansion.

    Let
    \be
        T_K(B)=\sum_{j=0}^K\frac{B^j}{j!}
    \ee
    be the degree-\(K\) Taylor approximation to \(e^B\). Since \(\|B\|\le 1\),
    \be
        \|e^B-T_K(B)\|
        \le
        \sum_{j>K}\frac{\|B\|^j}{j!}
        \le
        \sum_{j>K}\frac{1}{j!}.
    \ee
    Using the elementary bound \(j!\ge 2^{j-1}\) for \(j\ge 1\), we get
    \be
        \sum_{j>K}\frac{1}{j!}
        \le
        \sum_{j>K}2^{1-j}
        =
        2^{1-K}.
    \ee
    We now need to pick a large enough $K$ such that the error growth during squaring is controlled.

    We now account for finite-precision arithmetic. Let \(w\) be the working precision,
    and after evaluating the Taylor polynomial and after every matrix squaring, round
    each entry to \(w\) bits. Thus the implemented squaring stage has the form
    \be
        Y_{j+1}=\operatorname{round}_w(Y_j^2)=Y_j^2+R_j .
    \ee
    For matrices of norm \(O(1)\), a dense matrix multiplication followed by entrywise
    rounding to \(w\) bits has the conservative rounding bound, using the matrix-product error analysis of \cite[Sections 3.5 and 3.6]{Higham22},
    \be
        \|R_j\|\le D_n^2 2^{-w}
    \ee

    and requires $O(\log(D_n))$ guard bits.
    We choose
    \be
        w=p+\lceil s\log_2 3\rceil + \lceil 2\log_2 D_n\rceil+O(\log(p+s))+O(1),
    \ee
    so that
    \be
        \|R_j\|\le 2^{-(p+3)}3^{-s}
    \ee
    throughout the squaring stage.

    We approximate the matrix $A$ to \(\widehat A\) whose entries on and above the diagonal are within \(2^{-w}\) of those of \(A\),
    with purely imaginary diagonal, and set \(\widehat A_{\bm\ell\bm k}=-\overline{\widehat A_{\bm k\bm\ell}}\).
    Each entry of \(A\) is a sum of at most \(T_{m,d}\) products of \(t\), a coefficient, and a
    square root. Each product is of modulus at most \(2^{s}\), so this only requires computing
    these factors to \(O(w+s+\log T_{m,d})\) bits. Then \(\widehat A\) is exactly skew-Hermitian and, bounding
    the operator norm by the Frobenius norm \(\|\cdot\|_{\mathrm F}\),
    \be
        \|\widehat A-A\|\le\|\widehat A-A\|_{\mathrm F}\le D_n2^{-w}\le 2^{-(p+3)}3^{-s}.
    \ee
    In particular, \(\widehat B:=2^{-s}\widehat A\) satisfies \(\|\widehat B\|\le\frac12+2^{-(p+3)}\le 1\),
    so the Taylor bound above also holds for \(\widehat B\).

    Let \(U_j=e^{2^j\widehat B}\). Since \(\widehat B\) is skew-Hermitian, each \(U_j\) is unitary. Writing
    \be
        e_j=\|Y_j-U_j\|,
    \ee
    we have
    \be
        Y_j^2-U_j^2
        =
        U_j(Y_j-U_j)+(Y_j-U_j)U_j+(Y_j-U_j)^2,
    \ee
    and therefore, as long as \(e_j\le 1\),
    \be
        e_{j+1}
        \le
        3e_j+\|R_j\|.
    \ee
    Choosing the Taylor degree
    \be
        K=p+\lceil s\log_2 3\rceil+5
    \ee
    makes the truncation error \(\|e^{\widehat B}-T_K(\widehat B)\|\le 2^{1-K}\le 2^{-(p+4)}3^{-s}\).

    We evaluate \(T_K(\widehat B)\) by Horner's rule, \(Q_0=\II\), \(Q_i=\II+\frac{1}{K-i+1}\widehat BQ_{i-1}\), so that
    \(Q_K=T_K(\widehat B)\), rounding each step to \(w\) bits to obtain \(\widetilde Q_i\). The exact terms
    satisfy \(\|Q_i\| = O(1)\), and since \(\|\widehat B/(K-i+1)\|\le 1\), the errors
    \(f_i=\|\widetilde Q_i-Q_i\|\) obey \(f_i\le f_{i-1}+D_n^22^{-w}\).
    Hence \(f_K\le KD_n^22^{-w}\le 2^{-(p+4)}3^{-s}\), using the \(O(\log(p+s))\) guard bits in \(w\).
    Taking \(Y_0=\widetilde Q_K\) therefore gives
    \be
        e_0\le\|\widetilde Q_K-T_K(\widehat B)\|+\|T_K(\widehat B)-e^{\widehat B}\|\le 2^{-(p+3)}3^{-s}.
    \ee
    Iterating the recurrence gives
    \be
        e_s
        \le
        3^s e_0+\frac{3^s-1}{2}\max_j\|R_j\|
        \le
        2^{-(p+2)}.
    \ee
    Finally, \(U_s=e^{\widehat A}\), and for skew-Hermitian \(X,Y\), \(\|e^X-e^Y\|\le\|X-Y\|\). Hence
    \be
        \|Y_s-e^{A}\|
        \le
        e_s+\|\widehat A-A\|
        \le
        2^{-(p+2)}+2^{-(p+3)}
        \le
        2^{-p}.
    \ee
    Thus rounding after each squaring gives a valid \(2^{-p}\)-additive approximation
    to every entry, while using only \(O(p+s+\log D_n+\log T_{m,d}+\log(p+s))\)-bit numbers.

    It remains to count arithmetic operations. The block \(H_n\) has dimension
    \(D_n\). Constructing it explicitly by evaluating all entries costs at most
    \be
        O(D_n T_{m,d} d)
    \ee
    arithmetic operations. The Taylor polynomial \(T_K(\widehat B)\) can be evaluated by Horner's method \cite{horner1819}\footnote{Although the result was previously known to Chinese and Arab mathematicians \cite{libbrecht2005chinese}.} using \(K\) dense matrix multiplications,
    and the final squaring stage uses \(s\) further dense matrix multiplications. Hence
    the matrix-exponential part costs
    \be
        O(D_n^3(K+s))
        =
        O(D_n^3(p+s)).
    \ee
    Combining these gives
    \be
        O(D_nT_{m,d}d+D_n^3(p+s)).
    \ee
\end{proof}

\section{Application: gate decomposition for Bose--Hubbard dynamics}
\label{app:bose-hubbard}
In this appendix, we apply the \model\ to the simulation of Bose--Hubbard dynamics. We first consider the standard Bose--Hubbard Hamiltonian and derive an energy-preserving Lie--Trotter decomposition with explicit error and resource bounds on finite photon-number sectors, extending the result to states with bounded mean photon number. We then treat a coherently driven, non-number-preserving Bose--Hubbard model by introducing a coherent-state ancilla that converts the dynamics into an energy-preserving simulation, and finally compare our construction with related bosonic simulation approaches.
\subsection{Standard Bose--Hubbard Hamiltonian}

The Bose--Hubbard Hamiltonian describes the behavior of interacting bosonic particles trapped in a periodic lattice potential and is defined as
\begin{align}
    \label{eq:BoseHubbard}
    \hat H_{\mathrm{BH}}
    &= \hat H_{\mathrm{hop}}+\hat H_{\mathrm{ons}} \\
    \hat H_{\mathrm{hop}}
    &= -J\sum_{(i,j)\in\mathcal P}\big(\hat a_i^\dagger \hat a_j+\hat a_j^\dagger \hat a_i\big) \\
    \hat H_{\mathrm{ons}}
    &= \frac{U}{2}\sum_{i=1}^m \hat n_i(\hat n_i-1)-\mu\sum_{i=1}^m \hat n_i,
\end{align}
with $\hat n_i=\hat a_i^\dagger \hat a_i$ and $\hat N=\sum_{i=1}^m \hat n_i$. Moreover, $\hat H_{\mathrm{hop}}$ is the hopping (kinetic) term and $\hat H_{\mathrm{ons}}$ is the on-site interaction and potential term. $J$ is the hopping amplitude that describes boson mobility between modes (sites) in the lattice, $U$ is the on-site interaction strength and $\mu$ is the chemical potential, which essentially sets the number of particles. We fix real parameters $J,U,\mu$ and a set $\mathcal P$ of interacting (hopping) pairs $(i,j)$ of modes and assume bounded coordination: each mode $i$ participates in at most $\nu=O(1)$ pairs in $\mathcal P$ where $\nu$ is the maximum number of distinct hopping links per mode (site).

We work in the energy-preserving gate model generated by
\begin{equation}
    \mathcal G=\{\hat a_j^\dagger \hat a_k, \hat n_j\hat n_k: 1\le j,k\le m \}.
\end{equation}
Since $\hat n_i=\hat a_i^\dagger \hat a_i$ and $\hat n_i^2=\hat n_i\hat n_i$, the elementary Bose--Hubbard terms are generated by $\mathcal G$.

For $K\in \mathbb N$, let $\Pi_{\le K}$ denote the projector onto the subspace with total photon number at most $K$
\begin{equation}
    \mathcal H_{\le K}^{(m)}
    = \mathrm{span}\{ |n_1,\dots,n_m\rangle : n_1+\cdots+n_m\le K\}.
\end{equation}
Since $[\hat H_{\mathrm{BH}},\hat N]=0$, the subspace
$\mathcal H_{\le K}^{(m)}$ is invariant under $e^{-it\hat H_{\mathrm{BH}}}$.

\begin{theorem}[Energy-preserving simulation with photon cutoff $K$]
    \label{thm:BH-truncated}
    Fix $K\in \mathbb N$, $t\ge 0$, and simulation error $\varepsilon\in(0,1)$.
    There exists an energy-preserving circuit $\hat{\widetilde U}(t)$ generated by
    $\mathcal G$ such that for every state $\rho$ supported on $\mathcal H_{\le K}^{(m)}$, we have
    \begin{equation}
        \label{eq:BH-goal1}
        \norm{\hat{U}_{\mathrm{BH}}(t)\rho \hat{U}_{\mathrm{BH}}^{\dagger}(t)
            -
        \hat{\widetilde U}(t)\rho \hat {\widetilde U}^\dagger(t)}_1
        \le \varepsilon,
    \end{equation}
    where $ \hat{U}_{\mathrm{BH}}=e^{-it\hat H_{\mathrm{BH}}}$. Moreover, $\hat{\widetilde U}(t)$ can be chosen to be the first-order Lie–Trotter product formula with $r$ steps, where it suffices to take
 \begin{equation}
        \label{eq:choose_r}
        r\ge
        \frac{C_4 t^2 m\nu}{\varepsilon}
        \Bigl(|J|(|U|K^2+|\mu|K)+\nu |J|^2K\Bigr),
    \end{equation}
for a sufficiently large universal constant $C_4$. In particular, for fixed $\nu, J, U$ and $\mu$, the number of Trotter steps is polynomial in $m, t, K$ and $1/\epsilon$.
\end{theorem}

\begin{proof}
    We write
    \begin{equation}\label{eq:evolution-finite}
        \hat H_{\mathrm{BH}}=\sum_{(i,j)\in\mathcal P}\hat h_{ij}+\sum_{i=1}^m \hat g_i,
        \qquad
        \hat h_{ij}=-J(\hat a_i^\dagger \hat a_j+\hat a_j^\dagger \hat a_i),
        \qquad
        \hat g_i=\frac{U}{2}\hat n_i(\hat n_i-1)-\mu \hat n_i.
    \end{equation}

    We now approximate the evolution on the finite-dimensional subspace $\Pi_{\le K}\mathcal H$. Let the Trotter step size be $\delta=t/r$, where $r$ is the number of Trotter steps. The first-order Lie--Trotter formula is
    \begin{equation}\label{eq:trotter}
        \hat U_r(t)=
        \Big[
            \Big(\prod_{(i,j)\in\mathcal P} e^{-i\delta \hat h_{ij}}\Big)
            \left(\prod_{i=1}^m e^{-i\delta \hat g_i}\right)
        \Big]^r.
    \end{equation}
    We take $\hat{\widetilde U}(t)=\hat U_r(t)$. Each factor $e^{-i\delta \hat h_{ij}}$ and $e^{-i\delta \hat g_i}$ is an allowed
    energy-preserving gate, since $\hat h_{ij}$ is a real linear combination of generators
    $\hat a_j^\dagger \hat a_k$ and $\hat g_i$ is a polynomial in $\hat n_i$.

    We estimate all norms on the finite-dimensional subspace $\mathcal H_{\le K}^{(m)}$.
    Recall that $\norm{\cdot}_{(K)}$ denotes the operator norm on $\mathcal H_{\le K}^{(m)}$.
    Since $0\le \hat n_i\le \hat N\le K$ on the subspace,
    we have
    \begin{equation}
        \label{eq:normboundN}
        \norm{\hat n_i}_{(K)}\le K.
    \end{equation}
    Also, for every normalized $|\psi\rangle$ supported on $\Pi_{\le K}\mathcal H$
    \begin{align}
        \label{eq:normbound11}
        &\norm{\hat a_i\ket{\psi}}_2^2
        =
        \bra{\psi}\hat a_i^\dagger \hat a_i\ket{\psi}
        =
        \bra{\psi}\hat n_i\ket{\psi}
        \le K,
        &\norm{\hat a_i}_{(K)}\le \sqrt K.
    \end{align}
Similarly,
    \begin{align}\label{eq:normbound22}
        &\norm{\hat a_i^\dagger\ket{\psi}}_2^2
        =
        \bra{\psi}\hat a_i\hat a_i^\dagger\ket{\psi}
        =
        \bra{\psi}(\hat n_i+\II)\ket{\psi}
        \le K+1,
        &\norm{\hat a_i^\dagger}_{(K)}\le \sqrt{K+1}.
    \end{align}
Consequently,
    \begin{equation}
        \norm{\hat h_{ij}}_{(K)}
        \le
        |J|\bigl(
            \norm{\hat a_i^\dagger}_{(K)}\norm{\hat a_j}_{(K)}
            +
            \norm{\hat a_j^\dagger}_{(K)}\norm{\hat a_i}_{(K)}
        \bigr)
        \le 2|J|\sqrt{K(K+1)},
    \end{equation}
    and
    \begin{equation}
        \label{eq:normboundg}
        \norm{\hat g_i}_{(K)}
        \le \frac{|U|}{2}\norm{\hat n_i(\hat n_i-1)}_{(K)}+|\mu|\norm{\hat n_i}_{(K)}
        \le \frac{|U|}{2}K^2+|\mu|K.
    \end{equation}

    A first-order product-formula estimate, proved in \cref{lem:first-order-bound}, gives a universal constant $C>0$ such that
    \begin{equation}\label{eq:product-formula1}
        \norm{\hat{U}_{\mathrm{BH}}(t)-\hat{\widetilde U}(t)}_{(K)}
        \le
        \frac{Ct^2}{r}
        \sum_{\substack{\hat{X},\hat{Y}\in \{\hat h_{ij}\}\cup\{\hat g_i\}\\ [\hat{X},\hat{Y}]\neq 0}}
        \norm{[\hat X,\hat Y]}_{(K)}.
    \end{equation}
Since $\hat g_i$ is diagonal in the occupation basis, we have
    \begin{equation}
        [\hat g_i,\hat g_k]=0 \qquad \text{for all } i,k.
    \end{equation}
Also, by locality,
    \begin{equation}
        [\hat h_{ij},\hat g_k]=0 \quad \text{unless } k\in\{i,j\},
        \qquad
        [\hat h_{ij},\hat h_{k\ell}]=0 \quad \text{unless } \{i,j\}\cap\{k,\ell\}\neq \emptyset.
    \end{equation}
Using the canonical commutation relations and the bounds above, one obtains universal constants $C_1,C_2>0$ for large $K$ such that
    \begin{equation}
        \norm{[\hat h_{ij},\hat g_i]}_{(K)} + \norm{[\hat h_{ij},\hat g_j]}_{(K)}\le
        C_1 |J|\bigl(|U|K^2+|\mu|K\bigr),
    \end{equation}
    and for overlapping hopping terms,
    \begin{equation}
        \norm{[\hat h_{ij},\hat h_{ik}]}_{(K)}\le
        C_2 |J|^2 K.
    \end{equation}
Because each site has a degree at most $\nu$, the number of nonzero
    commutators involving a fixed edge $(i,j)$ is $O(\nu)$, and
    $|\mathcal P|\le m\nu/2$.
Hence
    \begin{equation}
        \label{eq:summationXY}
        \sum_{\substack{\hat X,\hat Y\in \{\hat h_{ij}\}\cup\{\hat g_i\}\\ [\hat X,\hat Y]\neq 0}}
        \norm{[\hat X,\hat Y]}_{(K)}
        \le
        C_3\, m\nu\Bigl(|J|\,|U|\,K^2 + |J|\,|\mu|\,K + \nu |J|^2 K\Bigr),
    \end{equation}
for a universal constant $C_3>0$. Hence, it suffices to choose $r$ as in \cref{eq:choose_r}. In particular,
    \begin{equation}\label{eq:choose_Nr}
        r=\poly(m,t,K,1/\varepsilon)
    \end{equation}
    for fixed $\nu,J,U,\mu$.  Choosing $r$ as in \cref{eq:choose_r} and considering \cref{eq:product-formula1} and \cref{eq:summationXY}, we obtain
    \begin{equation}
        \label{eq:truncation_error}
        \norm{\hat{U}_{\mathrm{BH}}(t)-\hat{\widetilde U}(t)}_{(K)}\le \varepsilon/2.
    \end{equation}

    Now let $\sigma$ be any state supported on $\Pi_{\le K}\mathcal H$. For unitaries $\hat U,\hat V$ on this finite-dimensional subspace,
    \begin{equation}\label{eq:unitary_UV}
        \norm{\hat U\sigma \hat U^\dagger-\hat V\sigma \hat V^\dagger}_1
        \le
        2\norm{\hat U-\hat V}_{(K)}.
    \end{equation}
    Indeed,
    \begin{equation}\label{eq:unitary_proof}
        \hat U\sigma \hat U^\dagger-\hat V\sigma \hat V^\dagger = (\hat U-\hat V)\sigma \hat U^\dagger+\hat V\sigma(\hat U^\dagger-\hat V^\dagger),
    \end{equation}
    and the triangle inequality together with $\|\sigma\|_1=1$ and $\|\hat U\|=\|\hat V\|=1$ gives \cref{eq:unitary_UV}.

    Applying \cref{eq:unitary_UV} with $\hat U=\hat{U}_{\mathrm{BH}}$ and $\hat V=\hat{\widetilde U}(t)$ and using \cref{eq:truncation_error} proves the claim in \cref{eq:BH-goal1}.
\end{proof}

\begin{corollary}[States with bounded expected photon number]
    \label{cor:BH-expected-number}
    Let $\rho$ be any state satisfying
    \begin{equation}
        \Tr(\rho \hat N)\le \bar E.
    \end{equation}
    Then for every cutoff $K\in \mathbb N$ and every energy-preserving circuit
    $\hat{\widetilde U}_{(K)}(t)$ satisfying the conclusion of
    \cref{thm:BH-truncated} on $\mathcal H_{\le K}$ with error $\varepsilon_K$,
    one has
    \begin{equation}
        \label{eq:BH-exp-bound}
        \norm{
            \hat{U}_{\mathrm{BH}}(t)\rho\, \hat{U}^{\dagger}_{\mathrm{BH}}(t)
            -
        \hat{\widetilde U}_{(K)}(t)\rho\, \hat{\widetilde U}_{(K)}^{\dagger}(t)}_1
        \le
        \varepsilon_K + 4 \sqrt{\frac{\bar E}{K+1}}.
    \end{equation}
    In particular, if $K+1\ge 64\bar E/\varepsilon^2$ and $\varepsilon_K\le \varepsilon/2$,
    then
    \begin{equation}
        \norm{
        \hat{U}_{\mathrm{BH}}(t)\rho\, \hat{U}^\dagger_{\mathrm{BH}}(t) - \hat{\widetilde U}_{(K)}(t)\rho\, \hat{\widetilde U}_{(K)}^{\dagger}(t)}_1
        \le \varepsilon.
    \end{equation}
    Hence, the Bose--Hubbard evolution can be simulated with complexity
    $\poly(m,t,\bar E,1/\varepsilon)$ on states whose total photon number is
    bounded in expectation.
\end{corollary}

\begin{proof}
    Let $\Pi_{\le K}$ be the projector onto $\mathcal H_{\le K}$ and write
    \begin{equation}
        \rho_{\le K}=\Pi_{\le K}\rho \Pi_{\le K},
        \qquad
        \rho_{>K}=(\II-\Pi_{\le K})\rho (\II-\Pi_{\le K}).
    \end{equation}
    Since $[\hat H_{\mathrm{BH}},\hat N]=0$, both the ideal evolution
    and the energy-preserving simulator from \cref{thm:BH-truncated} preserve the
    decomposition into number sectors.

    We set
    \begin{equation}
        p_K=\Tr\bigl((\II-\Pi_{\le K})\rho\bigr),
    \end{equation}
    which is the probability that the state has more than $K$ photons.

    By Markov's inequality,
    \begin{equation}
        p_K \le \frac{\Tr(\rho \hat N)}{K+1}\le \frac{\bar E}{K+1}.
    \end{equation}
    Let
    \begin{equation}
        \widetilde\rho_{\le K}=
        \rho_{\le K}/(1-p_K), \quad p_K<1
    \end{equation}
    be the normalized version of the low-photon part of $\rho$. If $p_K=1$, then $\rho_{\le K}=0$ and the low-sector contribution is zero. So, $\widetilde\rho_{\le K}$ is supported on $\mathcal H_{\le K}$, and
    \cref{thm:BH-truncated} implies
    \begin{equation}
        \norm{
            \hat{U}_{\mathrm{BH}}(t)\widetilde\rho_{\le K} \hat{U}^\dagger_{\mathrm{BH}}(t)
            -
        \hat{\widetilde U}_{(K)}(t)\widetilde\rho_{\le K}\hat{\widetilde U}_{(K)}^{\dagger}(t)}_1
        \le \varepsilon_K.
    \end{equation}
    Multiplying by $(1-p_K)$ gives
    \begin{equation}
        \norm{
            \hat{U}_{\mathrm{BH}}(t)\rho_{\le K} \hat{U}^\dagger_{\mathrm{BH}}(t)
            -
        \hat{\widetilde U}_{(K)}(t)\rho_{\le K}\hat{\widetilde U}_{(K)}^{\dagger}(t)}_1
        \le \varepsilon_K.
    \end{equation}
    To control the contribution outside the cutoff without assuming that $\rho$ is block diagonal with respect to $\Pi_{\le K}$, we use the gentle measurement lemma. Since
$\Tr(\Pi_{\le K}\rho)=1-p_K$, we have
\begin{equation}
\norm{\rho-\rho_{\le K}}_1 =
\norm{
\rho-\Pi_{\le K}\rho\Pi_{\le K}}_1
\le 2\sqrt{p_K}.
\label{eq:gentle-cutoff}
\end{equation}
Using the triangle inequality and invariance of the trace norm under unitary conjugation
\begin{equation}
\begin{aligned}
&
\norm{\hat U_{\mathrm{BH}}(t)\rho\hat U_{\mathrm{BH}}^\dagger(t)
-
\hat{\widetilde U}_{(K)}(t)\rho
\hat{\widetilde U}_{(K)}^{\dagger}(t)}_1
\\
&\le
\norm{\rho-\rho_{\le K}}_1
+
\norm{
\hat U_{\mathrm{BH}}(t)\rho_{\le K}\hat U_{\mathrm{BH}}^\dagger(t)
-
\hat{\widetilde U}_{(K)}(t)\rho_{\le K}
\hat{\widetilde U}_{(K)}^{\dagger}(t)}_1
+
\norm{\rho-\rho_{\le K}}_1.
\end{aligned}
\label{eq:cutoff-triangle}
\end{equation}
For $p_K<1$, the low-sector bound gives a middle contribution at most
$(1-p_K)\varepsilon_K\le\varepsilon_K$. For $p_K=1$, one has
$\rho_{\le K}=0$, so the middle term vanishes. Therefore
\begin{equation}
\begin{aligned}
&
\norm{
\hat U_{\mathrm{BH}}(t)\rho\hat U_{\mathrm{BH}}^\dagger(t)
-
\hat{\widetilde U}_{(K)}(t)\rho
\hat{\widetilde U}_{(K)}^{\dagger}(t)}_1
\le
\varepsilon_K+4\sqrt{p_K}
\le
\varepsilon_K+
4\sqrt{\frac{\bar E}{K+1}}.
\end{aligned}
\label{eq:expected-energy-tail}
\end{equation}
This proves the claimed bound.
\end{proof}

We now establish the first-order product-formula bound
used in the proof above and in the active Bose--Hubbard
simulation below. We work on the truncated subspace $\Pi_{\le K}\mathcal H$, so all operators are bounded and all norms below are understood as
$\norm{\cdot}_{(K)}$.

\begin{lem}[First-order product-formula bound]\label{lem:first-order-bound}
Let
\begin{equation}
    \hat H=\sum_{\ell=1}^L \hat X_\ell
\end{equation}
be a sum of bounded Hermitian operators on $\Pi_{\le K}\mathcal H$, and define
\begin{equation}
    \hat S(\delta)=\prod_{\ell=1}^L e^{-i\delta \hat X_\ell}.
\end{equation}
Then there exists a universal constant $C>0$ such that, for every $t\in\mathbb R$, integer $r\ge 1$, and $\delta=t/r$,
\begin{equation}
    \norm{e^{-it\hat H}-\hat S(\delta)^r}_{(K)}
    \le
    \frac{Ct^2}{r}
    \sum_{1\le k<\ell\le L}\norm{[\hat X_k,\hat X_\ell]}_{(K)}.
\end{equation}
Equivalently, one may restrict the sum to pairs with $[\hat X_k,\hat X_\ell]\neq 0$.
\end{lem}

\begin{proof}
It suffices to consider $t\ge 0$, since negative times are covered by replacing each $\hat X_\ell$ with $-\hat X_\ell$.
We first bound the error of a single Trotter step.

\proofstep{Two-term case.}
Let $\hat A$ and $\hat B$ be bounded Hermitian operators, and define
\begin{equation}
    \hat F(s)=e^{-is(\hat A+\hat B)}e^{is\hat B}e^{is\hat A},
    \qquad s\in[0,\delta].
\end{equation}
Since $\hat F(0)=\II$, setting $s=\delta$ and multiplying on the right by $e^{-i\delta \hat A}e^{-i\delta \hat B}$ gives
\begin{equation}
    e^{-i\delta(\hat A+\hat B)}-e^{-i\delta \hat A}e^{-i\delta \hat B}
    =
    \bigl(\hat F(\delta)-\II\bigr)e^{-i\delta \hat A}e^{-i\delta \hat B},
\end{equation}
and hence
\begin{equation}
    \norm{e^{-i\delta(\hat A+\hat B)}-e^{-i\delta \hat A}e^{-i\delta \hat B}}_{(K)}=
    \norm{\hat F(\delta)-\II}_{(K)}.
\end{equation}
Differentiating gives
\begin{equation}
    \hat F'(s)
    =
    e^{-is(\hat A+\hat B)}
    \Bigl(-i(\hat A+\hat B)+i\hat B+ie^{is\hat B}\hat Ae^{-is\hat B}\Bigr)
    e^{is\hat B}e^{is\hat A}.
\end{equation}
Therefore
\begin{equation}
    \norm{\hat F'(s)}_{(K)}
    =
    \norm{e^{is\hat B}\hat Ae^{-is\hat B}-\hat A}_{(K)}.
\end{equation}
Using the identity
\begin{equation}
    e^{is\hat B}\hat Ae^{-is\hat B}-\hat A
    =
    i\int_0^s e^{iu\hat B}[\hat B,\hat A]e^{-iu\hat B} du,
\end{equation}
we obtain
\begin{equation}
    \norm{\hat F'(s)}_{(K)}
    \le
    \int_0^s \norm{[\hat A,\hat B]}_{(K)}du
    =
    s \norm{[\hat A,\hat B]}_{(K)}.
\end{equation}
Integrating from $0$ to $\delta$,
\begin{equation}
    \norm{\hat F(\delta)-\II}_{(K)}
    \le
    \int_0^\delta \norm{\hat F'(s)}_{(K)}ds
    \le
    \frac{\delta^2}{2}\norm{[\hat A,\hat B]}_{(K)}.
\end{equation}
Thus
\begin{equation} \label{eq:upper_commutator}
    \norm{e^{-i\delta(\hat A+\hat B)}-e^{-i\delta \hat A}e^{-i\delta \hat B}}_{(K)}
    \le
    \frac{\delta^2}{2}\norm{[\hat A,\hat B]}_{(K)}.
\end{equation}

\proofstep{General case.}
Now let
\begin{equation}
    \hat Y_q=\sum_{\ell=1}^q \hat X_\ell,
    \qquad q=1,\dots,L.
\end{equation}
Using the telescoping identity
\begin{equation}
    e^{-i\delta \hat Y_L}-\prod_{\ell=1}^L e^{-i\delta \hat X_\ell}
    =
    \sum_{q=2}^L
    \left(
        e^{-i\delta \hat Y_q}
        -
        e^{-i\delta \hat Y_{q-1}}e^{-i\delta \hat X_q}
    \right)
    \prod_{\ell=q+1}^L e^{-i\delta \hat X_\ell},
\end{equation}
and the unitary invariance of the norm, we get
\begin{equation}
    \norm{e^{-i\delta \sum_{\ell=1}^L \hat X_\ell}-\hat S(\delta)}_{(K)}
    \le
    \sum_{q=2}^L
    \norm{e^{-i\delta \hat Y_q}-e^{-i\delta \hat Y_{q-1}}e^{-i\delta \hat X_q}}_{(K)}.
\end{equation}
Applying the two-term bound \cref{eq:upper_commutator} with $\hat A=\hat Y_{q-1}$ and $\hat B=\hat X_q$ gives
\begin{equation}
    \norm{e^{-i\delta \hat Y_q}-e^{-i\delta \hat Y_{q-1}}e^{-i\delta \hat X_q}}_{(K)}
    \le
    \frac{\delta^2}{2}\norm{[\hat Y_{q-1},\hat X_q]}_{(K)}.
\end{equation}
Since
\begin{equation}
    [\hat Y_{q-1},\hat X_q]=\sum_{k=1}^{q-1}[\hat X_k,\hat X_q],
\end{equation}
the triangle inequality yields
\begin{equation}
    \norm{[\hat Y_{q-1},\hat X_q]}_{(K)}
    \le
    \sum_{k=1}^{q-1}\norm{[\hat X_k,\hat X_q]}_{(K)}.
\end{equation}
Therefore
\begin{equation}
    \norm{e^{-i\delta \sum_{\ell=1}^L \hat X_\ell}-\hat S(\delta)}_{(K)}
    \le
    \frac{\delta^2}{2}
    \sum_{q=2}^L \sum_{k=1}^{q-1}\norm{[\hat X_k,\hat X_q]}_{(K)},
\end{equation}
that is
\begin{equation}\label{eq:upper_delta}
    \norm{e^{-i\delta \sum_{\ell=1}^L \hat X_\ell}-\hat S(\delta)}_{(K)}
    \le
    \frac{\delta^2}{2}
    \sum_{1\le k<\ell\le L}\norm{[\hat X_k,\hat X_\ell]}_{(K)}.
\end{equation}

\proofstep{$r$ Trotter steps.}
Set $\delta=t/r$. Then
\begin{equation}
    e^{-it\hat H}-\hat S(\delta)^r
    =
    \sum_{q=0}^{r-1}
    \hat S(\delta)^q
    \bigl(e^{-i\delta \hat H}-\hat S(\delta)\bigr)
    e^{-i\delta \hat H(r-1-q)}.
\end{equation}
Taking norms and using that $e^{-i\delta \hat H}$ and $\hat S(\delta)$ are unitary on $\Pi_{\le K}\mathcal H$,
\begin{equation}
    \norm{e^{-it\hat H}-\hat S(\delta)^r}_{(K)}
    \le
    r \norm{e^{-i\delta \hat H}-\hat S(\delta)}_{(K)}.
\end{equation}
Applying \cref{eq:upper_delta},
\begin{equation}
    \norm{e^{-it\hat H}-\hat S(\delta)^r}_{(K)}
    \le
    r\cdot \frac{\delta^2}{2}
    \sum_{1\le k<\ell\le L}\norm{[\hat X_k,\hat X_\ell]}_{(K)}.
\end{equation}
Since $\delta=t/r$, we have $r\delta^2=t^2/r$, and hence
\begin{equation}
    \norm{e^{-it\hat H}-\hat S(\delta)^r}_{(K)}
    \le
    \frac{t^2}{2r}
    \sum_{1\le k<\ell\le L}\norm{[\hat X_k,\hat X_\ell]}_{(K)}.
\end{equation}
This proves the claim with $C=\frac12$. Replacing the full sum by the sum over nonzero commutators only does not change the value.
Applying this lemma to $\hat H=\hat H_{\mathrm{BH}}=\sum_{(i,j)\in\mathcal P}\hat h_{ij}+\sum_{i=1}^m \hat g_i$
gives \cref{eq:product-formula1}.
\end{proof}

\subsection{Active Bose--Hubbard Hamiltonian}

In \cref{thm:BH-truncated}, we considered the standard Bose--Hubbard Hamiltonian, which is already energy-preserving.
We now show how to simulate a simple active/driven variant in the same energy-preserving model by adding a coherent drive and using a coherent state ancilla.

We define the active Bose--Hubbard Hamiltonian
\begin{equation}
    \label{eq:active-BH}
    \hat H_{\mathrm{actBH}}= \hat H_{\mathrm{BH}} + \hat H_{\mathrm{drv}},
\end{equation}
\begin{equation}
    \label{eq:BH-drive}
    \hat H_{\mathrm{drv}}=
    \eta \sum_{i=1}^m (\hat a_i + \hat a_i^\dagger),
\end{equation}
where $\eta\in\mathbb R$ is the drive strength and measures how strongly the external classical field (such as a laser) couples to the bosonic modes and $\hat H_{\mathrm{BH}}$ is the standard Bose--Hubbard Hamiltonian defined in \cref{eq:BoseHubbard}. The only non-energy-preserving part of $\hat H_{\mathrm{actBH}}$ is the coherent drive \cref{eq:BH-drive}.

\paragraph{Energy-preserving Hamiltonian with one coherent ancilla.} We introduce one ancilla mode $\hat b$, prepared in a coherent state $\ket{\alpha}_b$ with
real $\alpha>0$, and define the energy-preserving Hamiltonian
\begin{equation}
    \label{eq:preserving-active-BH}
    \hat{\widetilde H}_\alpha=
    \hat H_{\mathrm{BH}}\otimes\II_b +\hat R_\alpha,
\end{equation}
\begin{equation}
    \label{eq:Ralpha}
    \hat R_\alpha = \frac{\eta}{\alpha}\sum_{i=1}^m
    \bigl(\hat a_i \hat b^\dagger + \hat b \hat a_i^\dagger \bigr).
\end{equation}
Every monomial in \cref{eq:preserving-active-BH} has the same number of creation and annihilation
operators, so $\hat{\widetilde H}_\alpha$ is energy-preserving on the system plus ancilla. Hence, it is an allowed energy-preserving Hamiltonian in \model.

For a system state $\ket{\psi}$, we define the active evolution
\begin{equation}
    \label{eq:ideal-active-evolution}
    \ket{\psi(t)} = e^{-it\hat H_{\mathrm{actBH}}}\ket{\psi},
    \qquad
    E(t) = \sup_{0\le s\le t}\bra{\psi(s)}\hat N\ket{\psi(s)},
\end{equation}
where $\hat N=\sum_{i=1}^m \hat n_i$ is the system total photon number operator.

The quantity $E(t)$ can be bounded directly in terms of the initial
mean photon number. To make the argument rigorous, first restrict to
the finite-number subspace $\Pi_{\le K}\mathcal H$ and define
\begin{equation}
\hat H_{\mathrm{actBH}}^{(K)}
=
\Pi_{\le K}\hat H_{\mathrm{actBH}}\Pi_{\le K}.
\end{equation}
On this finite-dimensional subspace, all operators are bounded, so for
$n_K(s)=\langle\psi_K(s)|\hat N|\psi_K(s)\rangle$, we may differentiate to obtain
\begin{equation}
\left|\frac{d}{ds}n_K(s)\right|=
\left|
i\langle\psi_K(s)|
[\hat H_{\mathrm{actBH}}^{(K)},\hat N]
|\psi_K(s)\rangle
\right|
\le
2|\eta|\sqrt{m\,n_K(s)}.
\end{equation}
Here we used $[\hat H_{\mathrm{BH}},\hat N]=0$ and the
Cauchy--Schwarz inequality. Hence, integrating the corresponding
inequality for $\sqrt{n_K(s)}$ gives
\begin{equation}
\sqrt{n_K(s)}
\le
\sqrt{n_K(0)}+|\eta|\sqrt{m}\,s.
\end{equation}
Taking $K\to\infty$ for any normalized input with
$\langle\psi|\hat N|\psi\rangle<\infty$ yields
\begin{equation}
E(t)
\le
\left(
\sqrt{E(0)}+|\eta|\sqrt{m}\,t
\right)^2.
\end{equation}
The same cutoff argument justifies the Duhamel identity used below:
it is first applied to the bounded truncated operators on a common
finite-number subspace and the cutoff is then removed.
\begin{theorem}[Energy-preserving simulation of active Bose--Hubbard Hamiltonian]
    \label{thm:active-BH-passive}
     For every normalized system state $|\psi\rangle$ with $\langle\psi|\hat N|\psi\rangle<\infty$ and $t\ge0$,
    \begin{equation}
        \label{eq:main-active-BH-bound}
        \norm{e^{-it\hat{\widetilde H}_\alpha}\bigl(\ket{\psi}\otimes\ket{\alpha}_b\bigr)
        -\bigl(e^{-it\hat H_{\mathrm{actBH}}}\ket{\psi}\bigr)\otimes\ket{\alpha}_b}\le \frac{|\eta| t \sqrt{m E(t)}}{\alpha}.
    \end{equation}
\end{theorem}

\begin{proof}
    We apply $\hat D^\dagger_b(\alpha)=e^{\alpha \hat b-\alpha \hat b^\dagger}$ on the ancilla as $\hat D^\dagger_b(\alpha)\ket{\alpha}=\ket{0}_b$. Since displacement is unitary, it preserves norms so that the simulation error can be bounded equivalently in this displaced frame. Consequently, it suffices to bound
    \begin{equation}
        \norm{
            e^{-it(\hat H_{\mathrm{actBH}}\otimes\II_b+\hat R_\alpha)}
            (\ket{\psi}\otimes\ket{0}_b)
        - (e^{-it\hat H_{\mathrm{actBH}}}\ket{\psi})\otimes\ket{0}_b},
    \end{equation}
    where we used $\hat D^\dagger_b(\alpha) \hat{\widetilde H}_\alpha \hat D_b(\alpha)=\hat H_{\mathrm{actBH}}\otimes\II_b+\hat R_\alpha$.

    Set $\ket{\psi(s)}=e^{-is\hat H_{\mathrm{actBH}}}\ket{\psi}.$ Let $A=\hat H_{\mathrm{actBH}}\otimes\II_b$ and $B=\hat R_\alpha$. Then Duhamel's formula gives
    \begin{equation}
        e^{-it(A+B)}-e^{-itA}
        =
        -i\int_0^t e^{-i(t-s)(A+B)}\,B\,e^{-isA}\,ds.
    \end{equation}
    Applying both sides to the initial state $ \ket{\psi}\otimes\ket{0}_b $, taking norms and using unitarity,
    \begin{equation}
        \norm{
            e^{-it(A+B)}(\ket{\psi}\otimes\ket{0}_b)
            -
        e^{-itA}(\ket{\psi}\otimes\ket{0}_b)}
        \le
        \int_0^t
        \norm{
        B e^{-isA}(\ket{\psi}\otimes\ket{0}_b)} ds.
    \end{equation}
    Since
    \begin{equation}
        e^{-isA}(\ket{\psi}\otimes\ket{0}_b)
        =
        \bigl(e^{-is\hat H_{\mathrm{actBH}}}\ket{\psi}\bigr)\otimes\ket{0}_b
        =\ket{\psi(s)}\otimes\ket{0}_b,
    \end{equation}
    we obtain

    \begin{equation}
        \label{eq:duhamel-active-BH}
        \norm{ e^{-it(\hat H_{\mathrm{actBH}}\otimes\II_b+\hat R_\alpha)} (\ket{\psi}\otimes\ket{0}_b)
        - (e^{-it\hat H_{\mathrm{actBH}}}\ket{\psi})\otimes\ket{0}_b} \le \int_0^t \norm{
        \hat R_\alpha\bigl(\ket{\psi(s)}\otimes\ket{0}_b\bigr)} ds.
    \end{equation}
    Since $\hat b\ket{0}_b=0$ and $\hat b^\dagger\ket{0}_b=\ket{1}_b$,
    \begin{equation}
        \hat R_\alpha\bigl(\ket{\psi(s)}\otimes\ket{0}_b\bigr)
        = \frac{\eta}{\alpha} \left(\sum_{i=1}^m \hat a_i \ket{\psi(s)}\right)\otimes \ket{1}_b.
    \end{equation}
    Hence
    \begin{equation}
        \norm{
        \hat R_\alpha\bigl(\ket{\psi(s)}\otimes\ket{0}_b\bigr)} = \frac{|\eta|}{\alpha}
        \norm{\sum_{i=1}^m \hat a_i \ket{\psi(s)}}.
    \end{equation}
    Using the Cauchy--Schwarz inequality, we obtain
    \begin{equation}
        \norm{ \sum_{i=1}^m \hat a_i \ket{\psi(s)}}^2 = \sum_{i,j=1}^m \bra{\psi(s)}\hat a_i^\dagger \hat a_j\ket{\psi(s)}\le
        m \sum_{i=1}^m \bra{\psi(s)}\hat a_i^\dagger \hat a_i\ket{\psi(s)}= m \bra{\psi(s)}\hat N\ket{\psi(s)}.
    \end{equation}
    Therefore

    \begin{equation}
        \norm{\hat R_\alpha\bigl(\ket{\psi(s)}\otimes\ket{0}_b\bigr)} \le \frac{|\eta| \sqrt m}{\alpha} \sqrt{\bra{\psi(s)}\hat N\ket{\psi(s)}}
        \le \frac{|\eta| \sqrt{m\,E(t)}}{\alpha}.
    \end{equation}
    Plugging this into \cref{eq:duhamel-active-BH} yields \cref{eq:main-active-BH-bound}. The final error bound is similar to the displacement example in \cref{app:examples}.
\end{proof}

\paragraph{Energy-preserving gate decomposition.}
The energy-preserving Hamiltonian \cref{eq:preserving-active-BH} decomposes into elementary energy-preserving terms:
\begin{equation}
    \label{eq:preserving-decomposition}
    \hat{\widetilde H}_\alpha = \sum_{(i,j)\in\mathcal P}\hat h_{ij} +\sum_{i=1}^m \hat g_i +\sum_{i=1}^m \hat d_i,
\end{equation}
where $\hat h_{ij}$ is an energy-preserving hopping term and $\hat g_i$ is an on-site energy-preserving polynomial, both defined in \cref{eq:evolution-finite}, and each $\hat d_i$ is a beamsplitter coupling between site $i$ and the ancilla mode $b$
\begin{align}
    \hat d_i= \frac{\eta}{\alpha}\bigl(\hat a_i\hat b^\dagger+\hat b \hat a_i^\dagger \bigr).
\end{align}
Hence, the active Bose--Hubbard simulation reduces to an energy-preserving Hamiltonian simulation problem.

Fix a total number cutoff $K$ and let $\Pi_{\le K}$ be the projector onto the joint system-ancilla subspace with at most $K$ photons.
We now approximate the evolution on the finite-dimensional subspace using the Lie--Trotter formula
\begin{equation}
    \label{eq:active-BH-trotter}
    \hat{\widetilde U}_r(t)= \Big[
        \Big(\prod_{(i,j)\in\mathcal P} e^{-i\delta \hat h_{ij}}\Big)
        \left(\prod_{i=1}^m e^{-i\delta \hat g_i}\right)
        \left(\prod_{i=1}^m e^{-i\delta \hat d_i}\right)
    \Big]^r,
    \qquad
    \delta=t/r.
\end{equation}
Since all factors are energy-preserving gates, $\hat{\widetilde U}_r(t)$ is an allowed energy-preserving circuit.
For the energy-preserving operators considered here, the cutoff norm of \cref{def:cutoff-norm} satisfies $\norm{X}_{(K)} = \norm{X\Pi_{\le K}}$. A first-order product-formula estimate in \cref{lem:first-order-bound} gives
\begin{equation}
    \label{eq:active-BH-PF-commutator}
    \norm {e^{-it\hat{\widetilde H}_\alpha}-\hat{\widetilde U}_r(t)} _{(K)} \le \frac{t^2}{2r}
    \sum_{u<v} \norm{[\hat H_u,\hat H_v]}_{(K)},
\end{equation}
where $\hat H_u$ and $\hat H_v$ denote terms from the collection in \cref{eq:preserving-decomposition}.

The Hamiltonian in Eq.\eqref{eq:preserving-decomposition} contains $L=|\mathcal P|+2m$ terms. Unlike the standard Bose--Hubbard case, the drive terms
$\hat d_i$ all share the ancillary mode $\hat b$ and therefore need
not commute even when they act on different system modes. Indeed, for
$i\neq j$,
\begin{equation}
[\hat d_i,\hat d_j] =
\frac{\eta^2}{\alpha^2}
\left(
\hat a_i^\dagger\hat a_j -
\hat a_j^\dagger\hat a_i
\right).
\end{equation}
We therefore use the conservative bound over all pairs of terms.
If $\lambda_K$ is an upper bound on the sector norm of every term in
Eq.\eqref{eq:preserving-decomposition}, then
\begin{equation}
\sum_{u<v}
\norm{
[\hat H_u,\hat H_v]}_{(K)}
\le
\frac{L(L-1)}{2}\,2\lambda_K^2
=
L(L-1)\lambda_K^2,
\label{eq:active-commutator-sum}
\end{equation}
where, for every pair $u<v$,
\begin{equation}
\norm{[\hat H_u,\hat H_v]}_{(K)}
\le
2\norm{\hat H_u}_{(K)}
\norm{\hat H_v}_{(K)}
\le
2\lambda_K^2,
\end{equation}
where $\lambda_K$ is chosen as a uniform upper bound on the sector
norms of all terms in Eq.\eqref{eq:preserving-decomposition}. In particular, we may take
\begin{equation}
\lambda_K
=
\max\left\{
2|J|K,\,
\frac{|U|}{2}K(K-1)+|\mu|K,\,
\frac{2|\eta|}{\alpha}K
\right\},
\end{equation}
so that
\begin{equation}
\|\hat H_u\|_{(K)}\le\lambda_K,
\end{equation}
where the terms in $\{...\}$ provide upper bounds on the sector norms of $\hat{h}_{ij}$, $\hat g_i$ and $\hat d_i$ for large $K$, respectively (see \cref{eq:normboundN}-\cref{eq:normboundg}). Therefore,
\begin{equation}
\label{eq:active-BH-PF-final}
\norm{e^{-it\hat{\widetilde H}_\alpha}-\hat{\widetilde U}_r(t)}_{(K)}
\le\frac{t^2}{2r}L(L-1)\lambda_K^2,
\end{equation}
where for a target Trotter error $\varepsilon_{\rm Trot}>0$ 
\begin{equation}
r=\poly(m,t,K,\eta / \alpha,1 /\varepsilon_{\mathrm{Trot}})
\end{equation}
for fixed $\nu,J,U,\mu$.

For the actual input state
$|\Phi\rangle=|\psi\rangle\otimes|\alpha\rangle_b$, let
$\Pi_{\le K}$ denote the projector onto the joint system-ancilla
subspace with total photon number at most $K$, and define
\begin{equation}
p_K
=
\langle\Phi|(\II-\Pi_{\le K})|\Phi\rangle.
\end{equation}
Since the mean total photon number of $|\Phi\rangle$ is
$E(0)+\alpha^2$, Markov's inequality gives
\begin{equation}
p_K
\le
\frac{E(0)+\alpha^2}{K+1}.
\end{equation}
Both $e^{-it\hat{\widetilde H}_\alpha}$ and $\hat{\widetilde U}_r(t)$
preserve the joint total photon number. Hence, using the product-formula
bound on $\Pi_{\le K}$ and the trivial bound
$\|\hat U-\hat V\|\le 2$ on its complement, we obtain
\begin{equation}
\begin{aligned}
\left\|
\left(
e^{-it\hat{\widetilde H}_\alpha}
-\hat{\widetilde U}_r(t)
\right)|\Phi\rangle
\right\|&\le
\frac{t^2}{2r}L(L-1)\lambda_K^2
+
2\sqrt{p_K}
\\
&\le
\frac{t^2}{2r}L(L-1)\lambda_K^2
+
2\sqrt{\frac{E(0)+\alpha^2}{K+1}}.
\end{aligned}
\end{equation}
\paragraph{Final error bound.} The total simulation error of an active Bose--Hubbard Hamiltonian using energy-preserving gates with the Trotter formula splits into the \emph{ancilla approximation error}, scaling as $O(\alpha^{-1})$ via \cref{eq:main-active-BH-bound} and the \emph{energy-preserving product-formula error}, scaling as $O(t^2/r)$ on each fixed total number sector via \cref{eq:active-BH-PF-final}. So, we have
\begin{equation}
\label{eq:final_error}
\norm{\hat U_{\mathrm{actBH}}(t)\ket{\psi}\otimes\ket{\alpha}_b-\hat{\widetilde U}_r(t)\ket{\psi}\otimes\ket{\alpha}_b} \le \frac{|\eta|t\sqrt {m E(t)}}{\alpha}
+ \frac{L(L-1)\lambda_K^2 t^2}{2r}+2\sqrt{\frac{E(0)+\alpha^2}{K+1}}.
\end{equation}

\subsection{Comparison with related approaches}
The present work is complementary to the truncation framework of Tong \emph{et al.} \cite{TongEtAl2022} where the Hamiltonian is decomposed, with respect to a local occupation number, into a part that changes the local quantum number and a part that preserves it. Their truncation bounds apply under additional assumptions: the initial state is supported, or well approximated, within a bounded local occupation sector, and the number-changing part of the Hamiltonian
satisfies the local-growth conditions of their framework. Under
these hypotheses, the required local cutoff can scale
polylogarithmically with $1/\varepsilon$ for the bosonic models
covered by their analysis. In particular, their direct truncation theorem
does not apply to arbitrary Bose--Hubbard
dynamics.
In contrast, our analysis exploits the exact total-number conservation of the standard Bose--Hubbard Hamiltonian and works directly in the invariant cutoff space $\mathcal H_{\leq K}^{(m)}$. Our approach consequently gives a first-order product-formula simulation whose number of Trotter steps is polynomial in $m,t,K$, and $1/\epsilon$, with linear dependence on $1/\epsilon$.

Kalajdzievski \emph{et al.} \cite{KalajdzievskiWeedbrookRebentrost2018} instead construct an explicit bosonic gate decomposition of the Bose--Hubbard evolution using Fourier, quadratic, cubic, quartic and controlled-phase gates. They also extend their construction to include the leading dipole-dipole interaction,
thereby showing how the same bosonic gate-decomposition strategy can be applied to an extended
Bose--Hubbard-type Hamiltonian with additional intersite interactions. Their elementary gate set therefore contains both particle energy-preserving and particle number-changing bosonic operations. Although the target Bose--Hubbard dynamics conserves the total particle number, their bosonic gate
decomposition does not preserve it gate by gate. In contrast, every factor in our standard Bose--Hubbard decomposition preserves the total particle number individually. Their implementation is also based on a first-order product-formula construction, with the required number of time slices scaling polynomially with the lattice size, evolution time, and inverse precision. Our approach yields polynomial scaling in these parameters, while additionally expressing the resource bound in terms of the conserved particle-number cutoff $K$. Since the two constructions employ different elementary gate sets and different treatments of the infinite-dimensional Hilbert space, their gate-count scalings are not directly interchangeable.
For the active Bose--Hubbard Hamiltonian, our work goes beyond the standard Bose--Hubbard setting treated in \cite{KalajdzievskiWeedbrookRebentrost2018}. While that work decomposes the standard Bose--Hubbard evolution into bosonic gates that include both active and energy-preserving operations, here we additionally consider an explicit number-changing coherent drive. Rather than implementing this drive directly by an active gate, we replace it by an energy-preserving interaction with a single coherent ancillary mode. The resulting enlarged Hamiltonian is compatible with the energy-preserving gate model, while \cref{thm:active-BH-passive} quantifies the additional approximation error through its dependence on the coherent-state amplitude $\alpha$. Thus, the coherent-ancilla construction provides an energy-preserving realization of driven Bose--Hubbard dynamics that is not considered in the gate decomposition of \cite{KalajdzievskiWeedbrookRebentrost2018}.

\section{Approximate Hermitian diagonalization}
\label{app:finite-precision-diagonalization}

The group commutator decomposition (\cref{lem:GC-decomp}) in our Solovay--Kitaev algorithm (\cref{sec:SK-details}) requires an efficient algorithm for approximate diagonalization of Hermitian matrices.
For our purposes, it suffices to give an approximately unitary matrix $\widehat{Q}$ and a diagonal matrix $\Lambda$, such that $\widehat Q\Lambda\widehat Q^\dagger\approx H$ (\cref{eq:finite-precision-diagonalization}).
It turns out the Jacobi method~\cite[Section~8.5]{golubMatrixComputations2013} suffices to achieve polynomial bit complexity.

\begin{theorem}
\label{thm:finite-precision-diagonalization}
Given a Hermitian $H\in\mathbb Q(i)^{d\times d}$ with $\norm{H}\leq1$, whose real and imaginary entry parts have numerators and denominators of at most $b$ bits, and an integer $p\geq1$, one can deterministically compute $\widehat Q\in\mathbb Q(i)^{d\times d}$ and a real rational diagonal matrix $\Lambda$ satisfying
\begin{equation}\label{eq:finite-precision-diagonalization}
    \norm{\widehat Q^\dagger\widehat Q-\II}\leq2^{-p},
    \qquad
    \norm{H-\widehat Q\Lambda\widehat Q^\dagger}\leq2^{-p},
\end{equation}
in $d^3\poly(b+p+\log d)$ bit operations.
\end{theorem}

\begin{proof}
The case $d=1$ is immediate. Suppose $d\geq2$ and put $\epsilon=2^{-p}$. We use Jacobi iteration, selecting an off-diagonal entry of largest magnitude at each step~\cite[Section~8.5]{golubMatrixComputations2013}. Rounding errors in Jacobi iteration were analyzed by Wilkinson~\cite[Chapter~5, Sections~18--20]{Wilkinson1965}. For completeness, we give a self-contained analysis for complex Hermitian matrices, with explicit bounds on the number of iterations and the working precision, yielding the stated bit-complexity bound.

We first show that each exact Jacobi rotation reduces the off-diagonal part of the matrix. We then bound the accumulated rounding errors and establish the required working precision and running time.

For a Hermitian matrix $B$, write
\begin{equation}
    \operatorname{off}(B)=\sqrt{\sum_{j\ne k}|B_{jk}|^2},
    \qquad \norm{B}_{\mathrm F}=\sqrt{\sum_{j,k}|B_{jk}|^2}.
\end{equation}
Thus, $\operatorname{off}(B)$ is the Frobenius norm of the off-diagonal part of $B$.
A Jacobi step selects an off-diagonal entry $B_{uv}$, $u<v$, of largest magnitude and applies a unitary acting only on coordinates $u$ and $v$ to eliminate this entry. If $B_{uv}=0$, set $J=\II$. Otherwise put $z=B_{uv}$, $x=|z|$, and $y=(B_{vv}-B_{uu})/2$. Let $J$ equal the identity outside its $u,v$ block, with
\begin{equation}
    t=\frac{\operatorname{sgn}(y)x}{|y|+\sqrt{y^2+x^2}},
    \qquad
    J_{\{u,v\},\{u,v\}}
    =\frac1{\sqrt{1+t^2}}
    \begin{pmatrix}
        1 & (z/x)t\\
        -(\overline z/x)t & 1
    \end{pmatrix},
    \label{eq:jacobi-bit-rotation}
\end{equation}
where $\operatorname{sgn}(0)=1$. Then $J$ is unitary.
We show that each Jacobi iteration decreases $\operatorname{off}(\cdot)$ by a factor of $\varrho$:
\begin{equation}
    \operatorname{off}(J^\dagger BJ)\leq\varrho\operatorname{off}(B),
    \qquad
    \varrho=\sqrt{1-\frac{2}{d(d-1)}},
    \qquad
    1-\varrho\geq d^{-2}.
    \label{eq:jacobi-bit-contraction}
\end{equation}
We follow the argument in~\cite[Eq.~(8.5.2) and Section~8.5.3]{golubMatrixComputations2013}, generalized to complex Hermitian matrices.
Let $C = J^\dagger B J$.
For $x>0$, direct multiplication gives
\begin{equation}
    (J^\dagger BJ)_{uv}
    =\frac{z}{x(1+t^2)}\bigl[x(1-t^2)-2yt\bigr]=0,
\end{equation}
since $t$ satisfies $xt^2+2yt-x=0$. The $vu$ entry also vanishes by Hermitian symmetry.
Hence $C_{\{u,v\},\{u,v\}}$ is diagonal (trivial for $x=0$).
Unitary invariance gives $\norm{C}_{\mathrm F}^2=\norm{B}_{\mathrm F}^2$ and
\begin{equation}
    \norm{C_{\{u,v\},\{u,v\}}}_{\mathrm{F}}^2 = C_{uu}^2+C_{vv}^2
    =B_{uu}^2+B_{vv}^2+|B_{uv}|^2+|B_{vu}|^2
    =B_{uu}^2+B_{vv}^2+2|B_{uv}|^2.
\end{equation}
Since $C_{kk}=B_{kk}$ for $k\notin\{u,v\}$, we obtain
\begin{equation}
        \operatorname{off}(C)^2
        =\norm{C}_{\mathrm F}^2-\sum_k C_{kk}^2
        =\norm{B}_{\mathrm F}^2-\left(\sum_k B_{kk}^2
          +2|B_{uv}|^2\right)
        =\operatorname{off}(B)^2-2|B_{uv}|^2.
\end{equation}
Since the largest squared magnitude is at least the average over the $d(d-1)/2$ entries above the diagonal, we have $\operatorname{off}(B)^2=2\sum_{j<k}|B_{jk}|^2\leq d(d-1)|B_{uv}|^2$, and \cref{eq:jacobi-bit-contraction} follows.

Choose
\begin{equation}
    k_{\max}=4d^2\bigl(p+\lceil\log_2d\rceil+4\bigr),
    \qquad
    \eta=\frac{\epsilon}{100(k_{\max}+d^2+1)}.
\end{equation}
Store all computed real and imaginary parts as \emph{dyadic rationals}, meaning integers divided by powers of two. Round $H$ to a Hermitian dyadic matrix $H_0$ with $\norm{H_0-H}_{\mathrm F}\leq\eta$, and set $\widehat Q_0=\II$. At each step choose an off-diagonal entry of largest magnitude in $H_k$, and approximate its exact rotation $J_k$ by a dyadic matrix $\widehat J_k$, leaving entries outside the chosen block unchanged. Compute
\begin{equation}
    \begin{aligned}
        H_{k+1}&=\operatorname{round}(\widehat J_k^\dagger H_k\widehat J_k)
        =J_k^\dagger H_kJ_k+E_k,\qquad \norm{E_k}_{\mathrm F}\leq\eta,\\
        \widehat Q_{k+1}&=\operatorname{round}(\widehat Q_k\widehat J_k)
        =\widehat Q_kJ_k+F_k,\qquad \norm{F_k}_{\mathrm F}\leq\eta.
    \end{aligned}
    \label{eq:jacobi-bit-updates}
\end{equation}
Here products are evaluated exactly before entrywise rounding, preserving Hermitian symmetry for $H_{k+1}$. Thus $E_k,F_k$ include both rotation approximation and product rounding. Below we justify these bounds using a fixed precision of $O(p+\log d)$ bits.

Stop when
\begin{equation}
    \operatorname{off}(H_k)\leq\epsilon/8.
    \label{eq:jacobi-bit-stopping}
\end{equation}
By \cref{eq:jacobi-bit-contraction,eq:jacobi-bit-updates},
\begin{equation}
    \operatorname{off}(H_{k+1})\leq\varrho\operatorname{off}(H_k)+\eta.
\end{equation}
Applying this bound repeatedly gives
\begin{equation}
    \operatorname{off}(H_k)\leq\varrho^k\operatorname{off}(H_0)+\frac{\eta}{1-\varrho}
             \leq2\sqrt d\,\varrho^k+d^2\eta.
\end{equation}
Since $\varrho\leq e^{-1/d^2}$, we have $2\sqrt d\,\varrho^{k_{\max}}\leq\epsilon/16$. Together with $d^2\eta\leq\epsilon/100$, this ensures termination after $k_*\leq k_{\max}$ rotations.

Let $Q_k=J_0\cdots J_{k-1}$. Since each exact $J_k$ is unitary, \cref{eq:jacobi-bit-updates} accumulates the errors $E_k,F_k$ without amplifying them. Including the initial rounding of $H$, this gives
\begin{equation}
    \norm{H_k-Q_k^\dagger H Q_k}_{\mathrm F}\leq(k+1)\eta,
    \qquad
    \norm{\widehat Q_k-Q_k}_{\mathrm F}\leq k\eta.
    \label{eq:jacobi-bit-accumulated-errors}
\end{equation}
In particular, $\norm{H_k},\norm*{\widehat Q_k}\leq2$ throughout. Output $\widehat Q=\widehat Q_{k_*}$ and $\Lambda=\operatorname{diag}(H_{k_*})$. Put $Q=Q_{k_*}$ and $R=\widehat Q-Q$. By \cref{eq:jacobi-bit-accumulated-errors}, $\norm{R}\leq\alpha:=k_*\eta$. Expanding $\widehat Q=Q+R$ gives
\begin{equation}
    \norm{\widehat Q^\dagger\widehat Q-\II}
    =\norm{Q^\dagger R+R^\dagger Q+R^\dagger R}
    \leq2\alpha+\alpha^2.
    \label{eq:jacobi-bit-unitarity}
\end{equation}
Using $\norm{H_{k_*}-\Lambda}\leq\operatorname{off}(H_{k_*})$, $\norm{\Lambda}\leq2$, and \cref{eq:jacobi-bit-stopping,eq:jacobi-bit-updates}, we obtain
\begin{align}
    \norm{H-\widehat Q\Lambda\widehat Q^\dagger}
    &\leq\norm{Q^\dagger HQ-H_{k_*}}+\norm{H_{k_*}-\Lambda}
      +\norm{R\Lambda Q^\dagger+Q\Lambda R^\dagger+R\Lambda R^\dagger}\notag\\
    &\leq\norm{H_0-H}_{\mathrm F}+\sum_{j=0}^{k_*-1}\norm{E_j}_{\mathrm F}
      +\operatorname{off}(H_{k_*})+2(2\alpha+\alpha^2)\notag\\
    &\leq(k_*+1)\eta+\epsilon/8+2(2\alpha+\alpha^2).
    \label{eq:jacobi-bit-reconstruction}
\end{align}
The right-hand sides of \cref{eq:jacobi-bit-unitarity,eq:jacobi-bit-reconstruction} are at most $\epsilon$, since $\alpha\leq\epsilon/100$ and $(k_*+1)\eta\leq\epsilon/100$.

To bound the required precision, note that before stopping the selected entry satisfies $x=|(H_k)_{uv}|>\epsilon/(8d)$. Hence every denominator in \cref{eq:jacobi-bit-rotation} is at least $\epsilon/(8d)$ and every square-root argument is at least $(\epsilon/(8d))^2$. The bounds $x,|y|\leq2$ also keep all intermediate magnitudes bounded by a universal constant. Binary arithmetic, using bisection for square roots, therefore computes a dyadic approximation with
\begin{equation}
    \norm{\widehat J_k-J_k}_{\mathrm F}\leq\xi:=\eta/100,
    \qquad
    w=O\!\left(\log\frac d\epsilon+\log\frac1\eta\right)=O(p+\log d)
\end{equation}
binary digits after the point.

Round each real and imaginary entry part to error at most $\eta/(4d)$, keeping the diagonal of $H_{k+1}$ real and rounding conjugate entries consistently. This contributes at most $\eta/2$ in Frobenius norm. Together with $\norm*{\widehat J_k-J_k}_{\mathrm F}\leq\xi$ and $\norm{H_k},\norm*{\widehat Q_k}\leq2$, this gives
\begin{align}
    \norm{E_k}_{\mathrm F}
    &\leq\norm{\widehat J_k^\dagger H_k\widehat J_k-J_k^\dagger H_kJ_k}_{\mathrm F}+\eta/2\notag\\
    &\leq\norm{H_k}(2\xi+\xi^2)+\eta/2
      \leq4\xi+2\xi^2+\eta/2\leq\eta,\\
    \norm{F_k}_{\mathrm F}
    &\leq\norm*{\widehat Q_k}\xi+\eta/2
      \leq2\xi+\eta/2\leq\eta.
\end{align}
Induction on $k$ justifies both the norm bounds and the rounding procedure, using \cref{eq:jacobi-bit-accumulated-errors}.

To select the largest off-diagonal entry efficiently, store the squared magnitudes $|(H_k)_{uv}|^2$, $u<v$, in a priority queue~\cite[Section~6.5]{cormen2009introduction}, implemented via an indexed max-heap that stores the heap position of each pair $(u,v)$. This data structure gives access to the largest stored value in $O(1)$ scalar operations and allows the value associated with any given pair $(u,v)$ to be updated in $O(\log d)$ scalar operations. Separately maintain the sum $\operatorname{off}(H_k)^2$. Each rotation changes only $O(d)$ matrix entries, so updating these quantities costs $O(d\log d)$ scalar operations per iteration. The matrix updates require only $O(d)$ scalar operations because each rotation acts on two coordinates. All stored entries and intermediate dyadic products have $O(w+\log d)$ bits, and each rotation is computable in $\poly(w)$ bit operations. Including initial rounding, the total cost is therefore
\begin{equation}
    d^2\poly(b+w)+k_{\max}d\log d\,\poly(w)
    =d^3\poly(b+p+\log d).\qedhere
\end{equation}
\end{proof}

\end{document}